%% file: ClassicThesis.tex
\documentclass[ twoside,
                openright,
                titlepage,numbers=noenddot,%1headlines,
                headinclude,footinclude,cleardoublepage=empty,abstract=on,
                BCOR=5mm,
                paper=a4,fontsize=11pt
                ]{scrreprt}
\input{classicthesis-config}
\input{mypackages}
\input{mymacros}

\begin{document}
\frenchspacing
\raggedbottom
\selectlanguage{american} % american ngerman
%\renewcommand*{\bibname}{new name}
%\setbibpreamble{}
\pagenumbering{roman}
\pagestyle{plain}
%********************************************************************
% Frontmatter
%*******************************************************
\include{FrontBackmatter/DirtyTitlepage}
\include{FrontBackmatter/Titlepage}
\include{FrontBackmatter/Titleback}
\cleardoublepage\include{FrontBackmatter/Dedication}
%\cleardoublepage\include{FrontBackmatter/Foreword}
\cleardoublepage\include{FrontBackmatter/Abstract}
% \cleardoublepage\include{FrontBackmatter/Publications}
\cleardoublepage\include{FrontBackmatter/Acknowledgments}
\cleardoublepage\include{FrontBackmatter/Contents}
%********************************************************************
% Mainmatter
%*******************************************************
\cleardoublepage
\pagestyle{scrheadings}
\pagenumbering{arabic}
%\setcounter{page}{90}
% use \cleardoublepage here to avoid problems with pdfbookmark
\cleardoublepage
\part{Mathematical Background}\label{pt:background}
\include{Chapters/Chapter01}

\include{Chapters/Chapter02}
\include{Chapters/Chapter03}
\include{Chapters/Chapter04}
\include{Chapters/Chapter05}
\cleardoublepage
\part{The Tarski Laplacian}\label{pt:tarski}
\include{Chapters/Chapter06}
\include{Chapters/Chapter07}
\cleardoublepage
\part{Applications}\label{pt:applications}
\include{Chapters/Chapter08}
\include{Chapters/Chapter09}
%\addtocontents{toc}{\protect\clearpage} % <--- just debug stuff, ignore
%\include{multiToC} % <--- just debug stuff, ignore for your documents
% ********************************************************************
% Backmatter
%*******************************************************
\appendix
\cleardoublepage
\part{Appendix}
\include{Chapters/Chapter0A}
\include{Chapters/Chapter0B}
%********************************************************************
% Other Stuff in the Back
%*******************************************************
\cleardoublepage\include{FrontBackmatter/Bibliography}
% \cleardoublepage\include{FrontBackmatter/Declaration}
% \cleardoublepage\include{FrontBackmatter/Colophon}
% ********************************************************************
% Game Over: Restore, Restart, or Quit?
%*******************************************************
\end{document}

%% file: classicthesis-config.tex
\PassOptionsToPackage{utf8}{inputenc}
  \usepackage{inputenc}

\PassOptionsToPackage{T1}{fontenc} % T2A for cyrillics
  \usepackage{fontenc}

\PassOptionsToPackage{
  drafting=false,    % print version information on the bottom of the pages
  tocaligned=false, % the left column of the toc will be aligned (no indentation)
  dottedtoc=false,  % page numbers in ToC flushed right
  eulerchapternumbers=true, % use AMS Euler for chapter font (otherwise Palatino)
  linedheaders=false,       % chaper headers will have line above and beneath
  floatperchapter=true,     % numbering per chapter for all floats (i.e., Figure 1.1)
  eulermath=false,  % use awesome Euler fonts for mathematical formulae (only with pdfLaTeX)
  beramono=true,    % toggle a nice monospaced font (w/ bold)
  palatino=true,    % deactivate standard font for loading another one, see the last section at the end of this file for suggestions
  style=classicthesis % classicthesis, arsclassica
}{classicthesis}

\newcommand{\myTitle}{Lattice Theory in Multi-Agent Systems\xspace}
\newcommand{\mySubtitle}{The Tarski Laplacian \& Applications\xspace}
\newcommand{\myDegree}{Doctor of Philosophy\xspace}
\newcommand{\myName}{Hans Riess\xspace}

\newcommand{\myFaculty}{Troy Olsson, Associate Professor of Electrical and Systems Engineering\xspace}
\newcommand{\myDepartment}{Electrical and Systems Engineering\xspace}
\newcommand{\myUni}{University of Pennsylvania\xspace}

\newcommand{\myTime}{2022\xspace}
\newcommand{\myVersion}{\classicthesis}

\providecommand{\mLyX}{L\kern-.1667em\lower.25em\hbox{Y}\kern-.125emX\@}

\PassOptionsToPackage{ngerman,american}{babel} % change this to your language(s), main language last
    \usepackage{babel}

\usepackage{csquotes}
\PassOptionsToPackage{%
  backend=bibtex8,bibencoding=ascii,%
  language=auto,%
  style=numeric-comp,%
  sorting=nyt, % name, year, title
  maxbibnames=10, % default: 3, et al.
  natbib=true % natbib compatibility mode (\citep and \citet still work)
}{biblatex}
    \usepackage{biblatex}

\PassOptionsToPackage{fleqn}{amsmath}       % math environments and more by the AMS
  \usepackage{amsmath}

\usepackage{graphicx} %
\usepackage{scrhack} % fix warnings when using KOMA with listings package
\usepackage{xspace} % to get the spacing after macros right
\PassOptionsToPackage{printonlyused,smaller}{acronym}
  \usepackage{acronym} % nice macros for handling all acronyms in the thesis
\usepackage{tabularx} % better tables
\usepackage{subfig}
\usepackage{listings}
\usepackage{classicthesis}

\hypersetup{%
  colorlinks=true, linktocpage=true, pdfstartpage=3, pdfstartview=FitV,%
  breaklinks=true, pageanchor=true,%
  pdfpagemode=UseNone, %
  plainpages=false, bookmarksnumbered, bookmarksopen=true, bookmarksopenlevel=1,%
  hypertexnames=true, pdfhighlight=/O,%nesting=true,%frenchlinks,%
  urlcolor=CTurl, linkcolor=CTlink, citecolor=CTcitation, %pagecolor=RoyalBlue,%
  pdftitle={\myTitle},%
  pdfauthor={\textcopyright\ \myName, \myUni, \myFaculty},%
  pdfsubject={},%
  pdfkeywords={},%
  pdfcreator={pdfLaTeX},%
  pdfproducer={LaTeX with hyperref and classicthesis}%
}

\makeatletter
\@ifpackageloaded{babel}%
  {%
    \addto\extrasamerican{%
    }%
    \addto\extrasngerman{%
    }%
    }{\relax}
\makeatother

%% file: mypackages.tex
\usepackage{babel}
\usepackage{amssymb,amsmath,amsthm,mathtools,amsfonts}
\usepackage{stmaryrd}
\usepackage{algorithm2e}
\RestyleAlgo{ruled}
\usepackage{lscape}
\usepackage{quiver}
\usepackage{xstring}

\usepackage{xcolor}

\usepackage{tikz}
\usepackage{tikz-cd}
\usetikzlibrary{babel}
\usepackage{graphicx}

%% file: mymacros.tex
\newcommand{\define}[1]{\textit{#1}}
\newcommand{\id}{\text{id}}
\newcommand{\field}{{\Bbbk}}
\renewcommand{\epsilon}{\varepsilon}
\renewcommand{\phi}{\varphi}
\renewcommand{\leq}{\leqslant}
\renewcommand{\geq}{\geqslant}
\newcommand{\Z}{\mathbb{Z}}
\newcommand{\N}{\mathbb{N}}
\newcommand{\R}{\mathbb{R}}
\newcommand{\C}{\mathbb{C}}

\newcommand{\Rext}{\R_{\max}}
\newcommand{\primes}{\mathfrak{p}}

\newcommand{\powerset}[1]{\wp (#1)}
\newcommand{\multiset}[1]{\N \left[ #1\right]}
\DeclareMathOperator*{\subspaces}{Sub}
\DeclareMathOperator*{\Inv}{Inv} %invariant subspaces

\newcommand{\monoid}[1]{#1}

\newcommand{\lattice}[1]{\mathbf{#1}}
\newcommand{\poset}[1]{\mathrm{#1}}
\newcommand{\upset}[1]{\uparrow\!{#1}}
\newcommand{\downset}[1]{\downarrow #1}
\newcommand{\Space}[1]{{\mathbb{#1}}} 
\renewcommand{\preceq}{\preccurlyeq}
\renewcommand{\succeq}{\succcurlyeq}

\newcommand{\meet}{\wedge}
\newcommand{\join}{\vee}
\newcommand{\bigmeet}{\bigwedge}
\newcommand{\bigjoin}{\bigvee}
\newcommand{\covers}{\lessdot}
\newcommand{\op}[1]{{#1}^{\mathrm{op}}}
\newcommand{\Int}[1]{\mathbb{I}{#1}}
\newcommand{\antichain}{\parallel}
\newcommand{\irred}[1]{\mathcal{J}\left(#1\right)}
\newcommand{\downsets}[1]{\mathcal{D} \left(#1 \right)}
\DeclareMathOperator*{\Closed}{Cl}
\newcommand{\resid}[1]{\mathbb{#1}}
\DeclareMathOperator*{\Alex}{Alex}

\newcommand{\sbt}{\,\begin{picture}(-1,1)(-1,-3)\circle*{3}\end{picture}\ }
\newcommand{\ladj}[1]{{#1}_{\sbt}}
\newcommand{\radj}[1]{{#1}^{\sbt}}

\newcommand{\galup}[1]{{#1}^{\uparrow}}
\newcommand{\galdown}[1]{{#1}^{\downarrow}}
\DeclareMathOperator*{\Gal}{Gal}

\newcommand{\nbhd}[1]{\mathcal{N}_{#1}}
\newcommand{\graph}[1]{#1}
\newcommand{\nodes}[1]{{\mathcal{V}_{#1}}}
\newcommand{\edges}[1]{{\mathcal{E}_{#1}}}
\newcommand{\hyperedges}[1]{\mathcal{H}_{#1}}
\newcommand{\weight}{\omega}

\newcommand{\bool}{\mathbf{B}}
\newcommand{\cat}[1]{\mathcal{\StrChar{#1}{1}} \mathit{\StrMid{#1}{2}{3}}}

\newcommand{\Ker}{\mathrm{Ker\;}}

\renewcommand{\Im}{\mathrm{Im\;}}
\DeclareMathOperator{\Hom}{\mathrm{Hom}}
\DeclareMathOperator{\End}{\mathrm{End}}

\DeclareMathOperator*{\Open}{Open}
\DeclareMathOperator*{\Mor}{Mor}
\DeclareMathOperator*{\Ob}{Ob}

\newcommand{\fc}{\trianglelefteqslant}
\newcommand{\cofc}{\trianglerighteqslant}
\DeclareMathOperator*{\Ind}{Ind}
\newcommand{\boundary}{\partial}
\newcommand{\coboundary}{\delta}
\newcommand{\sheaf}[1]{\underline{\mathcal{#1}}}
\newcommand{\cosheaf}[1]{\overline{\mathcal{#1}}}
\newcommand{\bisheaf}[1]{\mathcal{#1}}
\newcommand{\sections}[1]{\Gamma \left(#1\right)}

\DeclareMathOperator{\Face}{Face}

\newcommand{\Parallel}[2]{P^{\mathcal{#1}}_{#2}} %OLD NOTATION
\DeclareMathOperator{\Hol}{Hol}
\DeclareMathOperator{\Path}{Path}
\DeclareMathOperator{\Free}{Free}

\newcommand{\Laplacian}{L}
\newcommand{\graphLaplacian}{L}
\newcommand{\Helmholtz}{\mathfrak{L}}
\newcommand{\Closure}{E}
\DeclareMathOperator*{\fixed}{Fix}
\DeclareMathOperator*{\prefix}{Pre}
\DeclareMathOperator*{\suffix}{Post}

\newcommand{\conv}{\varoast}
\newcommand{\domain}{\Omega}
\newcommand{\channels}{\mathcal{C}}
\newcommand{\signals}{\mathcal{X}}
\newcommand{\aggregate}{\bigbox}

\newcommand{\rel}[1]{\mathcal{#1}}
\DeclareMathOperator*{\Rel}{Rel}

\newcommand{\prob}[1]{\mathbf{Pr}\left[#1\right]}
\newcommand{\Part}[1]{\Pi\left(#1\right)}
\newcommand{\Subpart}[1]{\Pi_{\ast}\left(#1\right)}

\newcommand{\Hodge}{\Delta}

\DeclareMathOperator*{\Grass}{Grass}

\newcommand{\model}{M}
\newcommand{\lang}{\mathcal{L}}
\newcommand{\AP}{\Phi}
\newcommand{\eval}{\pi}
\newcommand{\always}{\square}
\newcommand{\eventually}{\lozenge}
\newcommand{\until}{\mathcal{U}}
\newcommand{\nexttime}{\bigcirc}
\newcommand{\TS}{\mathbf{TS}}
\newcommand{\intension}[1]{\llbracket {#1} \rrbracket}
\newcommand{\Buchi}{\mathbf{B}}
\newcommand{\Models}{{\mathcal{M}}}

\newcommand{\textbuchi}{B\"{u}chi~}
\newcommand{\textpuschel}{P\"{u}schel~}
\newcommand{\pusmar}{P\"{u}schel-Maragos~}

\theoremstyle{definition}
\newtheorem{example}{Example}[chapter]
\newtheorem{assumption}{Assumption}[chapter]
\newtheorem{problem}{Problem}[chapter]

\newtheorem{examples}{Examples}[chapter]
\newtheorem{definition}{Definition}[chapter]
\newtheorem{proposition}{Proposition}[chapter]
\newtheorem{theorem}{Theorem}[chapter]
\newtheorem{lemma}{Lemma}[chapter]
\newtheorem{corollary}{Corollary}[chapter]
\newtheorem{warning}{Warning}[chapter]
\newtheorem{remark}{Remark}[chapter]

\newtheorem*{fact}{Fact}
\newtheorem*{claim}{Claim}

\newtheorem{conjecture}{Conjecture}[chapter]

%% file: FrontBackmatter/DirtyTitlepage.tex
%*******************************************************
% Little Dirty Titlepage
%*******************************************************
\thispagestyle{empty}
%\pdfbookmark[1]{Titel}{title}
%*******************************************************
\begin{center}
    \spacedlowsmallcaps{\myName} \\ \medskip

    \begingroup
        \color{CTtitle}\spacedallcaps{\myTitle}
    \endgroup
\end{center}

%% file: FrontBackmatter/Titlepage.tex
%*******************************************************
% Titlepage
%*******************************************************
\begin{titlepage}
    %\pdfbookmark[1]{\myTitle}{titlepage}
    % if you want the titlepage to be centered, uncomment and fine-tune the line below (KOMA classes environment)
    \begin{addmargin}[-1cm]{-3cm}
    \begin{center}
        \large

        \hfill

        \vfill

        \begingroup
            \color{CTtitle}\spacedallcaps{\myTitle} \\ \bigskip
        \endgroup

        \spacedlowsmallcaps{\myName}

        \vfill

        \includegraphics[width=6cm]{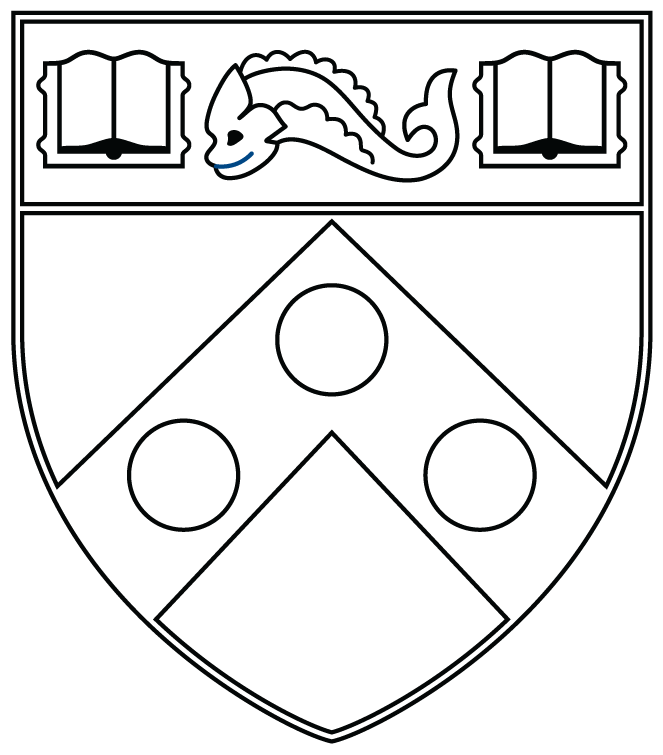} \\ \medskip

        % \mySubtitle \\ \medskip
        \myDegree \\
        \myDepartment \\
        % \myFaculty \\
        \myUni \\ \bigskip

        \myTime\

        \vfill

    \end{center}
  \end{addmargin}
\end{titlepage}

%% file: FrontBackmatter/Titleback.tex
\thispagestyle{empty}

\hfill

\vfill

\noindent\myName: \textit{\myTitle,} \mySubtitle, %\myDegree,
\textcopyright\ \myTime

%\bigskip
%
%\noindent\spacedlowsmallcaps{Supervisors}: \\
%\myProf \\
%\myOtherProf \\
%\mySupervisor
%
%\medskip
%
%\noindent\spacedlowsmallcaps{Location}: \\
%\myLocation
%
%\medskip
%
%\noindent\spacedlowsmallcaps{Time Frame}: \\
%\myTime

%% file: FrontBackmatter/Dedication.tex
%*******************************************************
% Dedication
%*******************************************************
\thispagestyle{empty}
\phantomsection
\pdfbookmark[1]{Dedication}{Dedication}

\vspace*{10cm}

\begin{center}
    for my dad who would \\
    on napkins and on airplanes \\
    teach me algebra
\end{center}

%% file: FrontBackmatter/Abstract.tex
%*******************************************************
% Abstract
%*******************************************************
%\renewcommand{\abstractname}{Abstract}
\pdfbookmark[1]{Abstract}{Abstract}
% \addcontentsline{toc}{chapter}{\tocEntry{Abstract}}
\begingroup
\let\clearpage\relax
\let\cleardoublepage\relax
\let\cleardoublepage\relax

\chapter*{Abstract}
SIn this thesis, we argue that (order-) lattice-based multi-agent information systems constitute a broad class of networked multi-agent systems in which relational data is passed between nodes. Mathematically modeled as lattice-valued sheaves, we initiate a discrete Hodge theory with a Laplace operator, analogous to the graph Laplacian and the graph connection Laplacian, acting on assignments of data to the nodes of a Tarski sheaf. The Hodge-Tarski theorem (the main theorem) relates the fixed point theory of this operator, called the Tarski Laplacian in deference to the Tarski Fixed Point Theorem, to the global sections (consistent global states) of the sheaf. We present novel applications to signal processing and multi-agent semantics and supply a plethora of examples throughout.
\vfill

\endgroup

\vfill

%% file: FrontBackmatter/Acknowledgments.tex
%*******************************************************
% Acknowledgments
%*******************************************************
\pdfbookmark[1]{Acknowledgments}{acknowledgments}

\begingroup
\let\clearpage\relax
\let\cleardoublepage\relax
\let\cleardoublepage\relax
\chapter*{Acknowledgments}
As I listen to ``Electric Sunrise'' by Plini on repeat, I would like to first thank my advisor, Rob Ghrist, for your selfless stoicism, diligent guidance, and edgy-yet-cerebral musical taste.

Secondly, I would like to thank the funding agencies that made my time at Penn possible: the Defense Advanced Research Projects Agency (DARPA), the Basic Research Office (BRO), the Simons Foundation, the National Science Foundations (NSF).
% I offer special thanks to program managers Jean-Luc Cambier, Lieutenant Colonol David Lewis, Fariba Faroo and Fredrick Leve.
I would also like to thank John Leggett III for endowing the fellowship I received during my first two years at Penn. I am humbled that some of your extreme generosity was bestowed upon me.

I offer thanks to all of Rob's students, past and present, including Iris Yoon, Huy Mai, Zo\"{e} Cooperband, Julian Gould, Darrick Lee, and Miguel Lopez. I especially want to thank Gregok Roerk for helping me appreciate lattice theory and Justin Curry for your mentorship in my undergraduate years Duke. I owe much to Jakob Hansen who introduced me to sheaf Laplacians. Thanks also to other students, especially Mikhail Hayhoe, Alp Aydinoglu, and Michael Sobrepera.

I would like to thank the several faculty members I have worked with at Penn. As I embark on my own postdoc, I only begin to fully appreciate all that you do. I would like to thank postdocs Paige North, Yiannis Kantaros, Alejandro Parada-Mayorga, as well as Dan Guaralnik. I would like to thank Alejandro Ribeiro for instilling upon me your unique perspectives on machine learning and data science. George Pappas, for teaching me Linear Systems and recommending a book.

I would like to thank some of the fellow researchers at other universities which whom I have crossed paths over the years. I owe gratitude to Sanjeevi Krishnan who introduced me to the Tarski Fixed Point Theorem, Vidit Nanda and Mike Lesnik for prolonging the life of the MacPherson Seminar at the Institute, Mike Munger for our continued collaboration on some of the economic implications of this work, Kelly Spendlove and Juan-Pablo Vigneaux for our stimulating mathematical discussion,  Gunnar Carlsson for your advice, and Michael Zavlanos for new beginnings.

I would like to thank Erich Prince, for your friendship to me in a city where I once knew no one else. Joe Wolf, for tennis matches at a moment's notice. I would like to thank my parents, for always believing in me. Most of all, I would like to thank my fiance\'{e}, Rebecca. You have been incredibly patient with me, and I love you.

\endgroup

%% file: FrontBackmatter/Contents.tex
%*******************************************************
% Table of Contents
%*******************************************************
\pagestyle{scrheadings}
%\phantomsection
\pdfbookmark[1]{\contentsname}{tableofcontents}
\setcounter{tocdepth}{2} % <-- 2 includes up to subsections in the ToC
\setcounter{secnumdepth}{3} % <-- 3 numbers up to subsubsections
\manualmark
\markboth{\spacedlowsmallcaps{\contentsname}}{\spacedlowsmallcaps{\contentsname}}
\tableofcontents
\automark[section]{chapter}
\renewcommand{\chaptermark}[1]{\markboth{\spacedlowsmallcaps{#1}}{\spacedlowsmallcaps{#1}}}
\renewcommand{\sectionmark}[1]{\markright{\textsc{\thesection}\enspace\spacedlowsmallcaps{#1}}}
%*******************************************************
% List of Figures and of the Tables
%*******************************************************
\clearpage
\begingroup
    \let\clearpage\relax
    \let\cleardoublepage\relax
    %*******************************************************
    % List of Figures
    %*******************************************************
    %\phantomsection
    %\addcontentsline{toc}{chapter}{\listfigurename}
    \pdfbookmark[1]{\listfigurename}{lof}
    \listoffigures

    \vspace{8ex}

    %*******************************************************
    % List of Tables
    %*******************************************************
    %\phantomsection
    %\addcontentsline{toc}{chapter}{\listtablename}
    % \pdfbookmark[1]{\listtablename}{lot}
    % \listoftables

    % \vspace{8ex}
    % \newpage

    %*******************************************************
    % List of Listings
    %*******************************************************
    %\phantomsection
    %\addcontentsline{toc}{chapter}{\lstlistlistingname}
    % \pdfbookmark[1]{\lstlistlistingname}{lol}
    % \lstlistoflistings

    % \vspace{8ex}

    % %*******************************************************
    % % Acronyms
    % %*******************************************************
    % %\phantomsection
    % \pdfbookmark[1]{Acronyms}{acronyms}
    % \markboth{\spacedlowsmallcaps{Acronyms}}{\spacedlowsmallcaps{Acronyms}}
    % \chapter*{Acronyms}
    % \begin{acronym}[UMLX]
    %     \acro{DRY}{Don't Repeat Yourself}
    %     \acro{API}{Application Programming Interface}
    %     \acro{UML}{Unified Modeling Language}
    % \end{acronym}

\endgroup

%% file: Chapters/Chapter01.tex
%************************************************
\chapter{Introduction}\label{ch:introduction}
%************************************************

In this thesis we offer a novel approach to understanding multi-agent systems from a lattice-theoretic point of view, a stark contrast to modern approaches relying on spectral graph theory \cite{chung1997spectral}, dynamical systems \cite{ghrist2022applied}, and control theory \cite{chen1984linear}.

%---------------------------------
\section{Multi-Agent Systems}
%---------------------------------

A system is something so fundamental it is difficult to define. While we will even not attempt to define systems here, we will define multi-agent systems as a special class of systems consisting of many interacting components called \define{agents}. Agents are differentiated  from ordinary components of systems such as the environment because agents are typically assumed to possess some degree of autonomy, intelligence, or communication capabilities, likely all three. Consider just a few examples of multi-agent systems.

\begin{enumerate}
    \item \textit{Wireless Communication}. Communication systems are naturally multi-agent systems. Agents send/receive \define{messages}. More elaborate multi-agent systems usually rely on a communication subsystem in order to gain information from other agents in the system.
    \item \textit{Swarm robotics}. Multiple (usually simple) robots (e.g.~quadrotor UAVs) collaborate to perform tasks such as surveillance, payload transportation, and more \cite{abdelkader2021aerial}. Swarms, it is widely held, outperforms smaller groups of more sophisticated robots due to advantages of scalability and robustness. 
    \item \textit{Sensor networks}. Multiple agents, equipped with one or more sensors, gather information about the environment or other agents inside or outside of the system. Applications include wearable devices, threat detection, industrial monitoring, and environmental conditions.
    \item \textit{Firms}.  Firms compete with directly with other firms, exchanging information via prices as well as through other methods such as collusion, acquisition, or mergers.
\end{enumerate}

In multi-agent systems, heterogeneous streams of information, possibly encoded in various types of data structures, are collected, processed, and exchanged among the various agents in the system. Consider the following examples of information processing tasks performed by multi-agent systems.

\begin{enumerate}
\item \textit{Consensus.} Consensus is a broad category of information processing that aims for every agent to eventually come to an agreement on a particular state or quantity of interest that depends on the state of all the agents. Some examples include flocking \cite{tanner2007flocking} and distributed formation control \cite{fax2004information} in swarm robotic systems, synchronization, and rendezvous in space \cite{cortes2006robust} and time \cite{nejad2009max}. A consensus algorithm, also called a consensus protocol, is a rule of interaction between agents specifying the exchange of information between a given agent and all of its neighbors in the network \cite{olfati2007consensus}.

% \begin{displayquote}
% In networks of agents (or dynamic systems), consensus means to reach an agreement regarding a certain quantity of interest that depends on the state of all agents. A consensus algorithm (or protocol) is an interaction rule that specifies the information exchange between an
% agent and all of its neighbors on the network.
% \end{displayquote}

\item \textit{Collaborative Filtering}. Collaborative filtering is a general technique of filtering information from multiple sources facilitated by ``interactions'' between agents. The most popular use-case of collaborative filtering is recommendation systems which take into account both the preferences of an individual agent as well as similar agents, neighbors, in a graph. To this end, graph signal processing \cite{ortega2018graph} has shown to be an effective tool for rating prediction \cite{huang2018rating}.

\item \textit{Information Fusion}. Information fusion, also known as \define{data fusion}, is the task of integrating data from multiple sources in heterogeneous formats. In one view \cite{khaleghi2013multisensor}, data fusion is the study of the transformation of various sources of information into a representation that can be interpreted by a machine or human.
% As it has been argued in the case of sensor fusion, new mathematical insights, for instance, sheaf theory \cite{robinson2017sheaves}, is needed to advance the state-of-art of information fusion.
\end{enumerate}

In each of the above information processing tasks, local information is aggregated in order to obtain global information. Sheaf theory \cite{bredon2012sheaf,curry2014sheaves}, it is our view, is an appropriate context to study such data relationships. Sheaves are data structures supplying mathematical rigor to the soft questions such as \emph{is local data consistent?}, \emph{do globally consistent assignments of data even exist?} In the agriculturally inspired nomenclature, a sheaf is a fixation of stalks (e.g.~vector spaces, sets, lattices, etc.) to a base (e.g.  graph, hypergraph, simplicial complex, manifold). A network sheaf (Definition \ref{def:network-sheaves}) is a sheaf $\sheaf{F}$ based over an undirected graph $\graph{G} = (\nodes{G}, \edges{G})$ specified by the data of objects (e.g.~vectors paces) over nodes and edges called \define{stalks}, and maps preserving the structure of objects (e.g.~linear transformations) called \define{restriction maps}. If $i$ and $j$ are nodes and $ij$ is an edge, stalks are denoted $\sheaf{F}(i)$, $\sheaf{F}(j)$, and $\sheaf{F}(ij)$, and restriction maps denoted $\sheaf{F}(i \fc ij): \sheaf{F}(i) \to \sheaf{F}(ij)$.

Sheaves facilitate the identification of consistent assignments of data called \define{global sections} (Definition \ref{def:sections}). We argue that global sections
\begin{align*}
    \Gamma(\graph{G}; \sheaf{F}) &=& \Bigl\{ \underset{\textit{assignments of data}}{\mathbf{x} \in \prod_{i \in \nodes{G}}}~\vert~ \underset{\textit{agreement along restriction maps}}{\sheaf{F}(i \fc ij)(x_i) = \sheaf{F}(j \fc ij)(x_j)} \quad \forall ij \in \edges{G}\Bigl\}
\end{align*}
 are a natural extension or generalization of consensus.\footnote{If every restriction map is the identity, then we recover consensus exactly.} In Chapter \ref{ch:semantics}, we discover a notion of consensus in multi-agent epistemic logic which we could call \define{semantic consensus}.
%------------------------
\section{Lattice Thoery}
%------------------------

Lattices are a class of ordered sets with two ``merging'' operations called \define{meet} and \define{join}. As ordered sets, finite lattices can be visualized with \define{Hasse diagrams}, directed acyclic graphs (DAGs). The nodes of a Hasse diagram are the elements of the ordered set. If $x \preceq y$, and there is no other element between $x$ and $y$, we draw a directed edge from $x$ to $y$. Orders are elementary, but lattices are less known. There may be a historical explanation.

Lattice theory \cite{birkhoff1940lattice,davey2002introduction,gratzer2002general,roman2008lattices} is a multi-faceted field of mathematics, with a rich history \cite{bilova2001lattice}. However, as a field, lattice theory never fully matured, perhaps unjustly so \cite{rota1997many}:
\begin{displayquote}
Never in the history of mathematics has a mathematical theory been the object of such vociferous vituperations as lattice theory. Dedekind, Jonsson, Kurosh, Malcev, Ore, von Neumann, Tarski, and most prominently Garrett Birkhoff have contributed a new vision of mathematics, a vision that has
been cursed by a conjunction of misunderstanding, resentment, and raw prejudice.
\end{displayquote}

Lattices, as a mathematical concept, is a relaxation of a boolean algebra \cite{boole1847mathematical}. To describe these structures, it was Klein who originally coined a term, \define{verband}, which loosely translates from German as ``association.'' However, Garret Birkhoff, popularized the English nomenclature \cite{bilova2001lattice}, ``lattice,'' due to the fact (with some speculation) that the Hasse diagrams of lattices often resemble lattices of the garden variety (pardon the pun!).

We recall this obscure history to highlight the fact that ``lattice'' has an unfortunate ambiguous meaning. To a physicist, a lattice is (likely) a discreet subset of $\R^d$ (Figure \ref{fig:integer-lattice}). To a network-scientist \cite{watts2004six}, a lattice is (likely) an undirected graph whose nodes uniformly follow a grid pattern and whose edges are drawn according to (unique!) nearest neighbors. We call these objects \define{integer lattice} or \define{lattice graphs} to differentiate them from an \define{order lattices}.

\begin{figure}[h]
\centering
\includegraphics[width=0.4\textwidth]{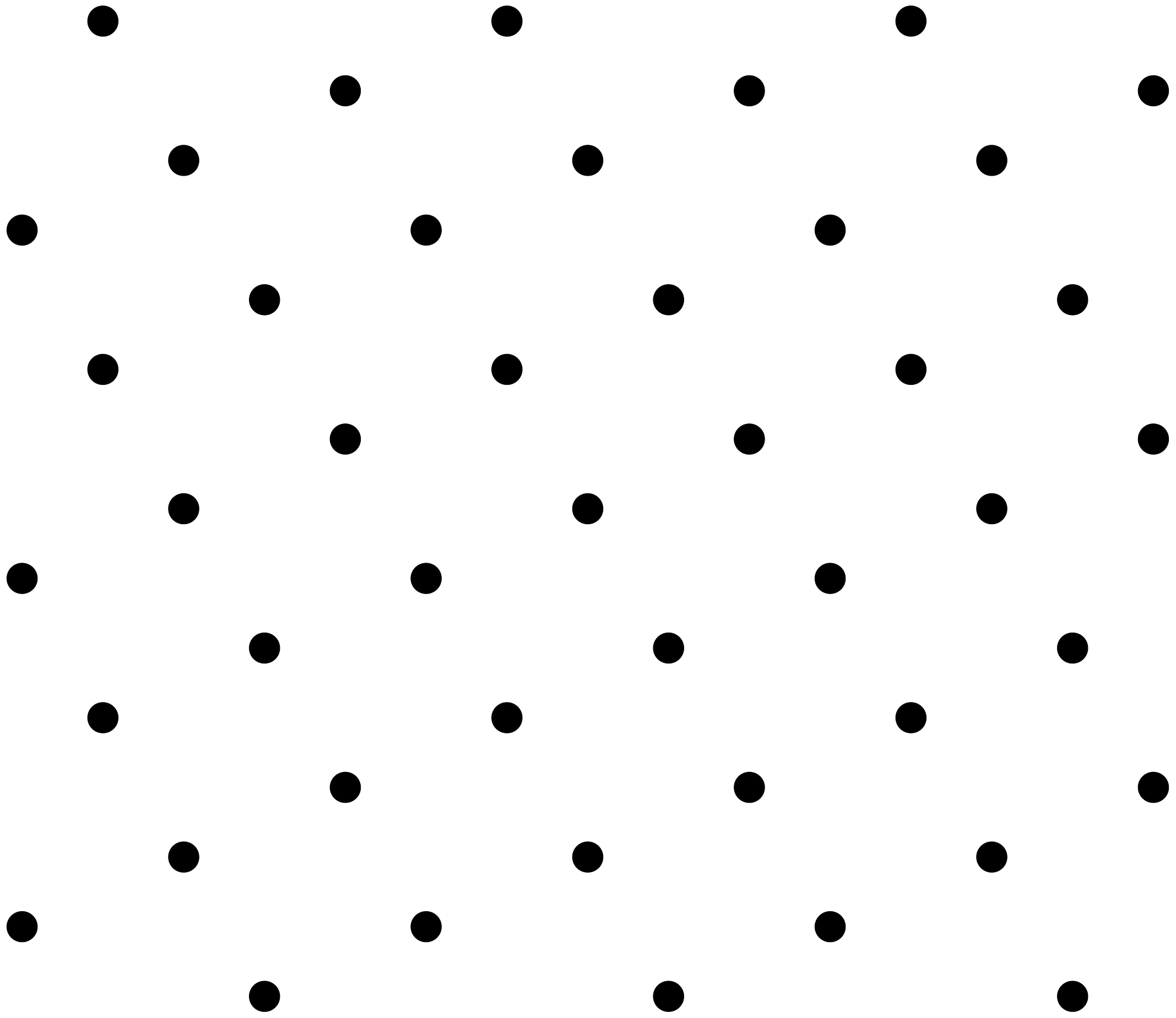} \quad \includegraphics[width=0.4\textwidth]{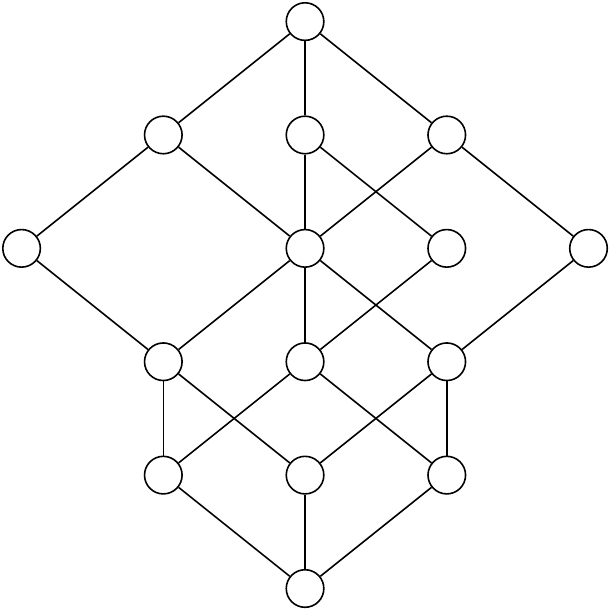}
\caption{(Left) integer lattice; (right) order lattice.} \label{fig:integer-lattice}
\end{figure}

Thus, lattice theory, as a subject, is certainly conflicted in its identity. Brought about as a weakening of the axioms boolean logic (e.g.~no longer requiring negation and distributivity), on one hand, lattice theory belongs to logic. An argument could also be made that lattice theory belongs to combinatorics \cite{rota1964foundations}. In our view, lattice theory, while establishing connections between the areas mentioned and more \cite{rota1997many}, belongs to algebra, the general study of symbol manipulation. For one, a lattice is a universal algebra \cite{sankappanavar1981course}, roughly, a set equipped with a collection of $n$-ary operations satisfying equational identities. Another compelling reason to highlight the algebraic side of lattice theory is that lattices of subgroups of a group, a collection of symmetries, are perhaps the first example studied in lattice theory. Permutations on $n$ elements form a group. You can reverse permutations by sending $j$ to $i$, as opposed to sending $i$ to $j$. Furthermore, you can apply permutations repeatedly and it doesn't matter where you place parentheses. The group of permutations on $4$ elements, a seemingly benign structure, has a rather intricate lattice of subgroups (Figure \ref{fig:symmetric}). Substructures share non-trivial relationships, even in such a simple example.

\begin{figure}[h]
\centering
\includegraphics[width=0.5\textwidth] {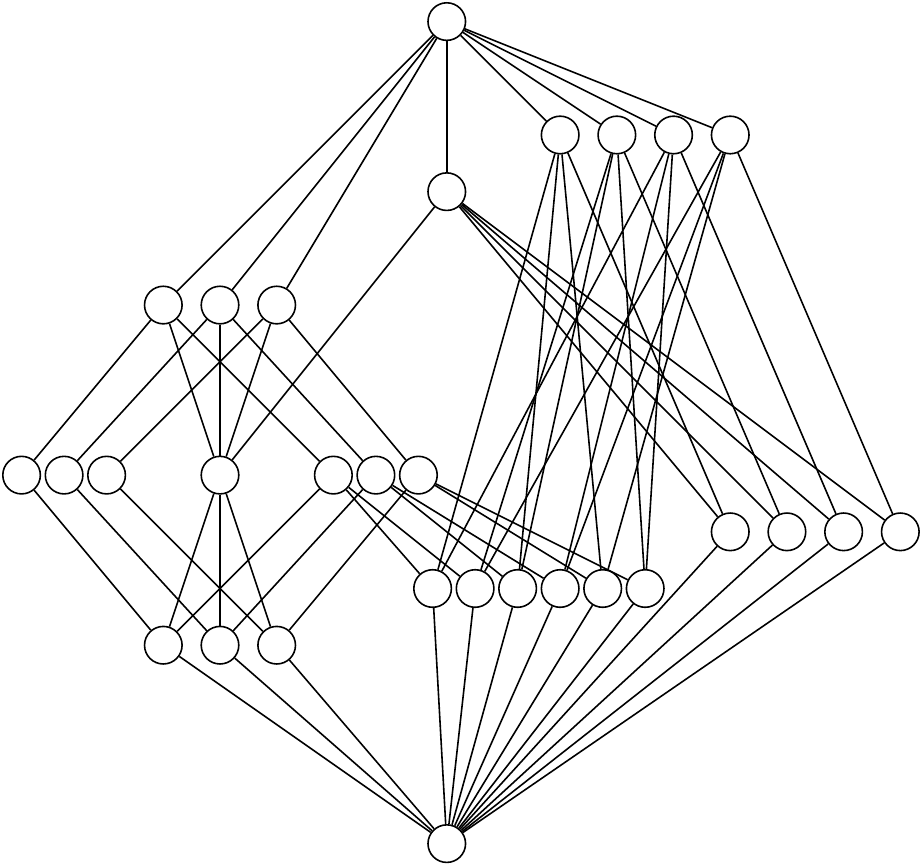}
\caption{Hasse diagram of the lattice of subgroups of $S_4$.}
\label{fig:symmetric}
\end{figure}x

Lattices function both as ordered sets and algebraic structures. Suppose $x, y \in \lattice{L}$, an ordered set with $\preceq$, then $x \join y$ (join) is the least upper bound (if it exists) and $x \meet y$ (meet) is the greatest lower bound (again, it exists). On one hand, a lattice is an ordered set such that the least upper bounds and greatest lower bounds are guaranteed to exist (Definition \ref{def:lattice-comb}). On the other hand, a lattice is, a set with binary operations $\join$ and $\meet$ satisfying axioms (Definition \ref{def:lattice-alg}). Elementary examples of lattices include the following.

\begin{examples}[Lattices]
\leavevmode
\begin{enumerate}
    \item Subsets of an arbitrary set with the operations union and intersection. Alternatively, subsets with an inclusion relation.
    \item Subgroups (subspaces) of a group (vector space) with product (sum) and intersection (Figure \ref{fig:symmetric}). Alternatively, subgroups (subspaces) with an inclusion relation.
    \item Partitions of an arbitrary set with coarsest common refinement under refinement. Alternatively, partitions under the refinement relation.
    \item Truth values $\bool = \{0,1\}$ with $\mathtt{OR}$/ $\mathtt{AND}$. Alternatively, $\bool$ with the order $0 < 1$.
\end{enumerate}
\end{examples}

\noindent Other quite general examples are embedded in other lattices such as powersets and partitions, including subpartitions (Example \ref{eg:partitions}), and lattices characterizing formal concepts \cite{wille1982restructuring} and fuzzy logic \cite{belohlavek1999fuzzy}. Less known is a ``lattice theory of information'' in which the information content (relative entropy) of a random variable establishes a lattice order. Joins and meets have interesting interpretations (Example \ref{eg:information-lattice}).

This reoccurring duality between order and algebraicity inspires a new outlook on multi-agent systems. Suppose a system consists of $\{1,2, \dots, n\}$ agents whose communication patterns are modeled by a graph $\graph{G} = (\nodes{G}, \edges{G})$ with $\nodes{G} = \{1,2, \dots, n\}$. Data collected in a lattice by individual agents models relational information about the system. On the other hand, because of the duality, this data is amalgamated with data stored by neighboring agents via the binary operations, meet and join.

As a point of clarification, it is not our goal to develop a general theory of multi-agent systems, but to make a case that certain mathematical tools are fundamental in modeling multi-agent information systems. To this end, it is our hope that the theory of lattice-valued sheaves (Chapter \ref{ch:lattice-valued}) and sheaf Laplacians (Chapter \ref{ch:tarski}) posited in this manuscript inspire novel information and control systems.

%--------------------------------
\section{Survey of Contributions}
%--------------------------------

The \define{Tarski Laplacian} is the key construction and contribution. The Tarski Laplacian was first introduced \cite{ghrist2022cellular} in an attempt to establish a Hodge theory and thus define cohomology of cellular sheaves valued in lattices.  The difficulty in defining cohomolgoy is that the construction of the cellular sheaf cochain complex initiated by Shepard \cite{shepard1985cellular} breaks apart. To make a long story short, meets and joins lack inverses. Nonetheless, some key computational and theoretical construction are still possible in \define{homological algebra} \cite{grandis2013homological}. Cohomology of mathematical objects such as simplicial complexes, topological spaces, or groups is often helpful in classifying properties of a structure. We would expect cohomology of lattice-valued sheaves would have ties to the behavior of the information systems they model. In the theory of cellular sheaves \cite{curry2014sheaves}, sheaf cohomology $H^0$ classifies the global sections. Global sections are identified as a lattice of fixed points via the Tarski Laplacian (Hodge-Tarski Theorem, Theorem \ref{thm:main}).

Homology of chain complexes in semiexact categories (the category of lattices and Galois connections is semiexact) as well as exact sequences and the snake lemma were established Marco Grandis \cite{grandis2013homological}. We suspect a proper theory of sheaf cohomology of $\cat{Sup}$-valued sheaves would require either a modified notion of projective or injective resolutions of sheaves \cite{bredon2012sheaf}, and, thus, a notion of projective or injective (complete) lattices. With the program to develop a cohomology theory for lattice-valued sheaves supplanted by novel applications of the Tarski Laplacian, our efforts are contained in the following manuscript \cite{ghrist2022cellular}.

It has been suggested  that  sheaf theory simultaneously straddles both the axes of algebra/geometry and structure/obstruction \cite{ghrist2021laplacians}. Lattice-valued sheaf theory, we place in the second quadrant of this helpful philosophical model (Figure \ref{fig:sheaf-axes}). While lattices are pure algebraic structures, in Chapter \ref{ch:tarski}, we introduce geometric notations such as Laplace operators and even parallel transport.

\begin{figure}[b]
	\centering
	\includegraphics[width=\textwidth ]{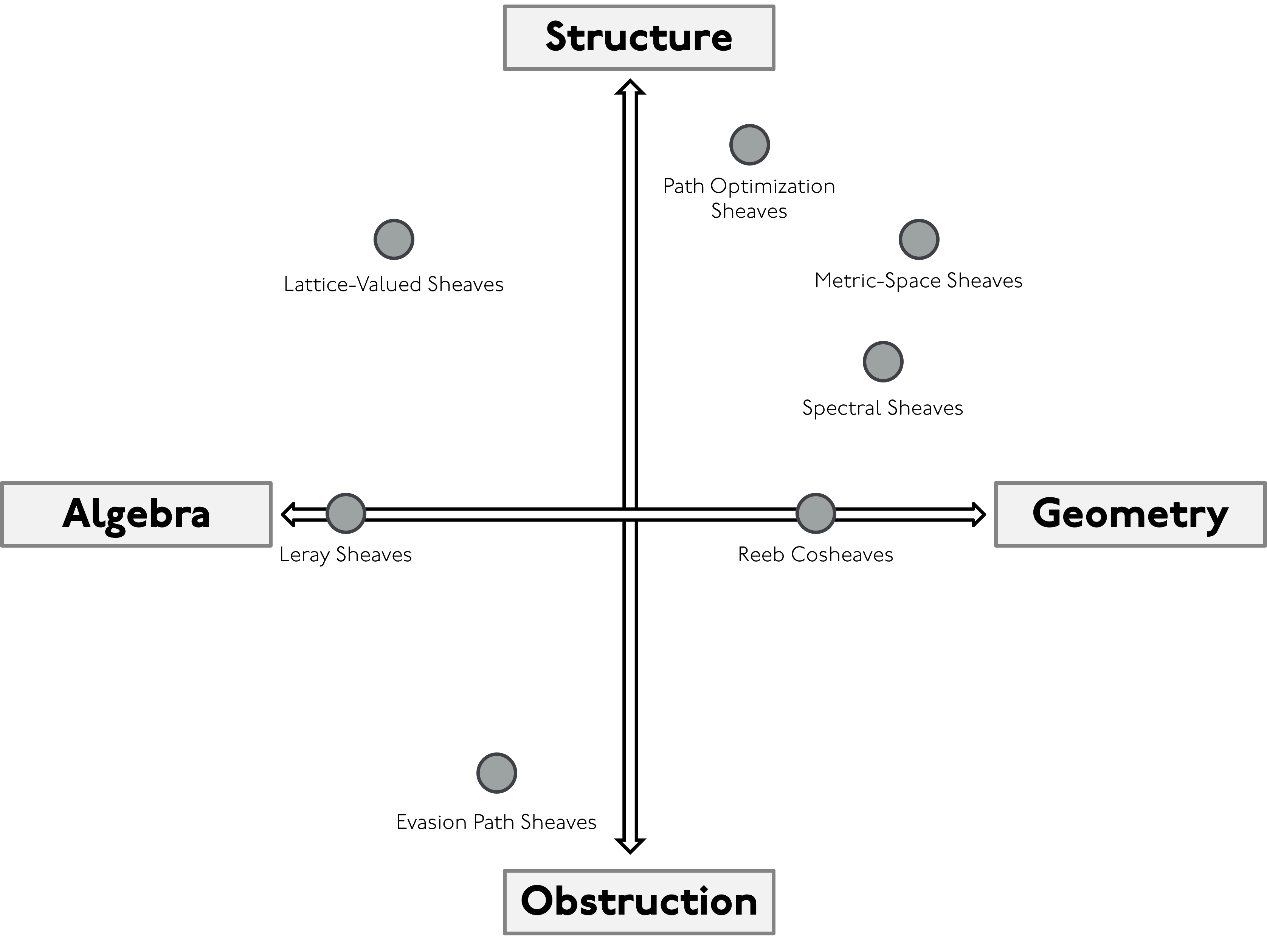}
	\caption{A philosophical model of applied sheaf theory: (from left to right) Leray sheaves \cite{curry2018dualities}, lattice-valued sheaves \cite{ghrist2022cellular}, evasion path sheaves \cite{ghrist2017positive}, path optimization sheaves \cite{moy2020path}, Reeb cosheaves \cite{de2016categorified}, spectral sheaves \cite{hansen2019toward}, and (pseudo)-metric space sheaves \cite{robinson2020assignments}.}
	\label{fig:sheaf-axes}
\end{figure}

From the standpoint of category theory, lattice-valued sheaves over (incidence) posets
% ,having the structure \[ i \fc ij \cofc j \quad \forall ij \in \edges{G},\]
are a categorification \cite{baez1998categorification} of sheaves valued in an abelian category. Abelian groups, for instance, are replaced with lattices of subgroups (categories) and homomorphisms are replaced with Galois connections (adjunctions). This fortuitous lifting of classical sheaf theory to lattice theory is not without a price. When a functor from a poset into the category of complete lattices and Galois connections factors through the functor sending an object to its lattice of (normal) subobjects, we cannot, in general, recover the homology of the abelian sheaf, from the lattice-valued sheaf \cite[Section 4.3]{ghrist2022cellular}.

The Tarski Laplacian is a local operator that acts on assignments of data to a sheaf of lattices. The Tarski Laplacian is called a ``Laplacian'' for a number of reasons. Primarily, the Hodge-Tarski Theorem identifies the global sections of a lattice-valued sheaf with fixed points. We offer a deep connection to the graph connection Laplacian \cite{singer2012vector} based on a notion of parallel transport between local sections of a network sheaf recently introduced \cite{bodnar2022neural}. We hope the Tarski Laplacian can establish connections to information geometry \cite{li2021transport} and representation theory \cite{krishnan2020invertibility}.

% \begin{align}
%     \{\mathbf{x}~\vert~\Laplacian \mathbf{x} \succeq \mathbf{x} \} &=& \sections{\graph{G}; \sheaf{F}}. \label{eq:hodge-tarski-preview}
% \end{align}
Dynamical systems on network sheaves with coefficients in Hilbert spaces (e.g.~$\R^d$ or $L_2(\R^d)$) were initiated by Hansen \cite{hansen2019distributed,hansen2021opinion,hansenconnections}. We introduce a dynamical system
\begin{align*}
    \mathbf{x}[t+1] &=& \Laplacian \mathbf{x}[t] \meet \mathbf{x}[t]
\end{align*}
called the \define{heat equation} with a time-varying version called the \define{gossip equation} \eqref{alg:gossip}. Beyond diffusion dynamics, saddle-point dynamics, as well as other local dynamics based on simple aggregation rules are of interest. In the narrow case of approximation of the constant sheaf (Definition \ref{def:approx-constant}), consensus algorithms are a byproduct of diffusion dynamics driven by the Tarski Laplacian.

We spend the better part of Chapter \ref{ch:signals} comparing approaches to lattice signal processing, as well as providing our own. We identify the convolution defined by \textpuschel \cite{puschel2021discrete} as an algebraic signal model and describe the relationship between two algebraic signal models corresponding to signals on lattices. In the final chapter, after discussing a class of networked model checking problems, we present applications of the Tarski Laplacian to semantics, offering novel models of knowledge diffusion/consensus motivated by epistemic logic \cite{fagin2004reasoning}.

%%%%%%%%%%%%%%%%%%%%%%%%%%%%%%
\section{Related Work}
\label{sec:related}
%%%%%%%%%%%%%%%%%%%%%%%%%%%%%%

We make no attempt at a comprehensive literature review, focusing on a few relevant topics, instead.

%-----------------------------------
\subsection{Quantum logic}
%-----------------------------------

In physics, Birkhoff and von Neumann proposed lattice theory as a logical model for quantum mechanics \cite{birkhoff1936logic}. Their work is based on the following observation. Suppose $H$ is a Hilbert space, then the set of closed subspaces of $H$ forms an (orthocomplemented) lattice under the closure of subspace sums (join) and intersection (meet). These subspaces are in one-to-one correspondence with projections which von Neumann reasoned could be viewed as quantum observables \cite{von2018mathematical}. However, their approach did not stand the test of time, due to Bell's Theorem being experimentally disproved \cite{aspect1981experimental}. Consequently, monoidal categories/string diagrams \cite{coecke2018picturing} and linear logic \cite{girard1987linear}, the logic of quantum information theory, are more fashionable these days. 

%------------------------------
\subsection{Supermodular games}
%------------------------------

Economists have applied lattice theory to a particular class of $n$-person games (multi-agent system). A real-valued function on a lattice is \define{supermodular} if
\begin{align*}
f(x \join y) + f(x \meet y) \geq f(x) + f(y) \quad \forall x, y \in \lattice{L}.
\end{align*}
Then, a \define{supermodular game} consists of players $i \in \{1,2,\dots,n\}$ each with a lattice of possible strategies $\lattice{L}_i$ and a utility function
\begin{align*}
    u_i(x_i, \mathbf{x}_{-i}): \prod_{i=1}^n \lattice{L}_i \to \R,\\
    x_i \in \lattice{L}_i, \quad \mathbf{x}_{-i} \in \prod_{j \neq i} \lattice{L}_i \quad i \in \{1,2,\dots, n\}
\end{align*}
such that $u_i(-,\mathbf{x}_{-i})$ is supermodular and $u_i$ satisfies an additional property of \define{increasing differences}. Then, every $n$-person supermodular game has a Nash equilibrium \cite{topkis1978minimizing} (see also \cite[Chapter 7]{vohra2004advanced}). It was later shown by Zhou that if each strategy lattice is complete, the set of Nash equilibrium forms a complete lattice \cite{zhou1994set}. While pertinent to problems in economics such as Bertrand competition \cite{edgeworth1925pure}, supermodular games have found practical use in wireless communication networks, with applications to energy-efficient power allocation \cite{liu2018supermodular}, power control \cite{altman2003supermodular}, and interference compensation \cite{huang2006distributed}. Submodular functions on a lattice, real-valued functions satisfying 
\begin{align*}
    f(x \join y) + f(x \meet y) &\leq& f(x) + f(y) \quad \forall x, y \in \lattice{L},
\end{align*}
and algorithms to minimize them were introduced by Topkis \cite{topkis1978minimizing}. Unsupervised learning problems such as principal component analysis (PCA) and generalized PCA, have been reformulated as constrained submodular maximization problems on the lattice of spaces of $\R^d$ \cite{nakashima2019subspace}, as well as meta-learning \cite{adibi2020submodular}.

%-----------------------------------
\subsection{Discrete event systems}
%-----------------------------------

In discrete event systems \cite{cassandras2008introduction}, an area of control theory, Galois connections model the triggering of events. A discrete event system is a transition system (Chapter \ref{ch:semantics}) with discrete states and transitions triggered by events. Linear systems \cite{chen1984linear} and their properties (e.g.~controlability, observability), are well-understood. In general, discrete event systems are less tractable, however, in the past two decades, much headway has been made by representing discrete-event systems with square matrices with elements in the following lattice.

Let $\Rext$ denote the extended reals $\R \cup \{-\infty, \infty\}$. $\Rext$ is a (complete) lattice under the operations $\max$ and $\min$ which we denote with $\join$ and $\meet$. $\Rext \setminus \{-\infty, \infty\}$ has the structure of a group under addition (Definition \ref{def:residuated}). Replacing the ``times'' with ``plus'' and ``times'' in an ordinary euclidean space with ``max'' or ``min'' one obtains a \define{max-plus vector space} $\Rext^m$, an example of a \define{weighted lattice} \cite{maragos2017dynamical}. As in ordinary linear algebra, matrix multiplication defines a transformation. If $A \in \Rext^{m \times m}$ is a matrix and $\mathbf{x} \in \Rext^m$ is a vector, then matrix multiplication is defined in two dual ways 
\begin{align*}
    (A \ovee \mathbf{x})_i &=& \bigjoin_{i =1}^{m} a_{ij} + x_j, \\
    (A \owedge \mathbf{x})_i &=& \bigmeet_{i=1}^m a_{ij} + x_j.
\end{align*}

For illustration, consider the following discrete event system.  Suppose a set of $\{1,2,\dots, m\}$ events begin at time $x_i[t], i \in \{1,2,\dots, m\}$ after $t \in \N$ rounds. Suppose, an event $j$ can only begin after another even $j$ has terminated (e.g.~ironing a shirt after drying). Then, we say $i$ transitions from $j$. With considerations of system design, waiting times depend not only on the ongoing event $i$ but also on the target event $j$. Let $a_{ij}$ denote the duration of time between the start of $i$ and the (subsequent) start of $j$. If $j$ does not depend on process $i$, then set $a_{ij} = -\infty$.  Then,
\begin{align*}
    x_i[t+1] &=& \bigjoin_{j=1}^m x_i[t] + a_{ij} \quad i = 1,2,\dots, m
\end{align*}
is the new start time of event $i$.
Let $A \in \Rext^{m}$ be the matrix with entries $a_{ij}$, then the global dynamics are written in the following form
\begin{align*}
   \mathbf{x}[t+1] &=& A \ovee \mathbf{x}[t].
\end{align*}

Suppose in two interacting discrete event systems, events are collected from the same set $\{1,2,\dots, n\}$. If $\mathbf{x}[t], \mathbf{y}[t] \in \Rext^m$ are vectors of event times after $t$ rounds of System 1 and System 2, respectively, and $A, B \in \R^{m \times m}$ are matrices of transition times, then the \define{synchronization problem} asks whether $(\mathbf{x}[0], \mathbf{y}[0])$ can be chosen so that there exists a $t_0$ such that $\mathbf{x}[t] = \mathbf{y}[t]$ for all $t \geq t_0$. Clearly, this is equivalent to solving the system
\begin{align*}
    A \ovee \mathbf{x} &=& B \ovee \mathbf{y}
\end{align*}
for $(\mathbf{x}, \mathbf{y}) \in \Rext^m \times \Rext^m$.

Cunninghame-Green \& Butkovic proposed a solution to the synchronization problem called the \define{alternating method} \cite{cuninghame2003equation}. Suppose $A^{\dagger} \in \Rext^{m \times m}$ with entries $[a^\dagger_{ij}] = -a_{ji}$. Then, the maps
\begin{align*}
    \mathbf{x} &\mapsto& A \ovee \mathbf{x}, \\
                A^{\ast} \owedge \mathbf{y} &\mapsfrom& \mathbf{y} 
\end{align*}
form a Galois connection (Example \ref{eg:max-plus-galois}). Each iteration of the alternating method is equivalent to applying the Tarski Laplacian (Chapter \ref{ch:tarski})
\begin{align}
    \Laplacian(\mathbf{x},\mathbf{y}) &=& \left( A^{\dagger} \owedge \left( B \ovee \mathbf{y} \right), B^{\dagger} \owedge \left( A \ovee \mathbf{x} \right) \right)
\end{align}
on a particular ``sheaf'' over the graph $\bullet-\bullet$ (Figure 

\begin{figure}
    \centering
    \input{gfx/alternating-method}
    \caption{A max-plus sheaf modeling synchronization of discrete-event systems \cite{cuninghame2003equation}.}
    \label{fig:cunninghame}
\end{figure}

\noindent Synchronization is equivalent to being a section of the sheaf. Therefore, by the Hodge-Tarski Theorem (Theorem \ref{thm:main}), synchronization is also 
equivalent to the condition
\begin{align*}
    A^{\dagger} \owedge \left( B \ovee \mathbf{y} \right) & \succeq \mathbf{x} \\
    B^{\dagger} \owedge \left( A \ovee \mathbf{x} \right) &\succeq \mathbf{y}
\end{align*}
for all $(\mathbf{x}, \mathbf{y}) \in \Rext^m \times \Rext^m$.

% In the general setting of coupled discrete event systems, if $\graph{G} = (\nodes{G}, \edges{G})$ is a network of subsystems with $n = |\nodes{G}|$,the Tarski Laplacian is the operator
% \begin{align}
%     \Laplacian &:& (\Rext^m)^m \to (\Rext^m)^m \\
%     (\Laplacian)_i &=& \bigmeet_{j \in \nbhd{i}} A_i^{\dagger} \owedge \left( A_j \ovee \mathbf{x}_j \right).
% \end{align}

%------------------------
\subsection{Consensus}
%------------------------

Lattice consensus algorithms in the networked multi-agent setting (Algorithm \ref{alg:meet-consensus}, Algorithm \ref{alg:heat-flow}, Algorithm \ref{alg:gossip}) are novel, but consensus on lattices, while obscure, broadly construed, is not. While we have discussed several notions of consensus in multi-agent systems, consensus has other connotations in both computer science as well as economics.

Reaching agreement among remote processes (agents) is a fundamental problem in distributed computing. In a consensus protocol, each process $i \in \{1,2,\dots, n\}$ receives or seeds an input register $x_i$, sends and receives messages to other processes, and, eventually, must determine an output register $y_i$ coinciding with the output registers of all other processes. Unfortunately, it has been shown that, roughly, no asynchronous consensus protocol can tolerate (i.e.~not affect the output) a single fault. For instance, if a process is unexpectedly removed from the system \cite{fischer1985impossibility}, this would be considered a \define{fault}.

Now, suppose input and output registers are collected in a lattice (Definition \ref{def:lattice-alg}). Then, lattice agreement is defined to be a decision $\{y_i\}_{i = 1}^n$ such that $\{y_i\}_{i=1}^n$ is a chain and the following inequality holds for all $i \in \{1,2,\dots, n\}$
\begin{align*}
    \bigmeet_{i=1}^n x_i &\preceq& y_i &\preceq& \bigjoin_{i=1}^n x_i.
\end{align*}
Lattice agreement, centralized branch-and-bound algorithms, have been shown to satisfy fault-tolerance in a variety of settings \cite{zheng2018lattice,zheng2021byzantine,zheng2020byzantine}, unlike lattice consensus.

In the theory of social choice \cite{arrow2012social}, consensus is roughly equivalent to ``choosing in groups'' \cite{munger2015choosing}. At the highest level of generality, if $X$ is a finite set of alternatives (choices), a consensus rule is a map
\[\chi: \bigcup_{k>1} X^k \to \powerset{X} \setminus \emptyset\] sending a tuple $\boldsymbol{\pi}$ of arbitrary individual choices called a profile to a nonempty set of group choices $chi(\boldsymbol{\pi})$ \cite{barthelemy1991formal}. Suppose we fix the number of voters (agents) to be $n$ and replace $X$ with a lattice $\lattice{L}$, then a \define{(latticial) consensus function}\footnote{If such a function is monotone and sends the top (bottom) profile  to the top (bottom) element of $\lattice{L}$, then it has been called a \define{aggregation function} \cite{botur2018generating}, although ``aggregation function'' is sometimes synonymous with a general consensus function.} \cite{barnett1995social} or \define{aggregation function} \cite{janowitz2016aggregation,leclerc2013aggregation} is a map \[\chi: \lattice{L}^n \to \lattice{L}.\]  Lattices frequently arise in consensus settings in social choice. For instance, $\lattice{L}$ could be the lattice of (transitive, reflexive) preference relations on $X$ or a lattice of choice functions \cite{monjardet2004lattices}, maps $c: \powerset{X} \to \powerset{X}$ such that $c(A) \subseteq A$ for all $A \subseteq X$. In the nomenclature, a consensus function $\chi$ is a \define{meet-projection} if there is a subset $N \subseteq \{1,2,\dots, n\}$ such that $\chi(\boldsymbol{\pi}) = \bigmeet_{i \in N} \pi_i$ and, in particular, a \define{Pareto consensus function} if $N= \{1,2,\dots, n\}$.
% Meet-projections have been classified for finite congruence-simple atomistic latticesTechnical conditions on $\lattice{L}$ (finite congruence-simple atomistic) are provided, completely characterizing meet-projections \cite{leclerc2013aggregation}.

A third interpretation of lattice consensus comes about in the general science of classification/taxonomy. In evolutionary biology, phylogenetic trees can be inferred from DNA sequences \cite{nascimento2017biologist}. Having several candidates for the evolutionary tree of a particular organism, consensus is one approach to determining an aggregate phylogenetic tree \cite{bryant2003classification}. In the same spirit, when an ensemble of hierarchical clustering algorithms produces various candidate dendrograms, consensus on the outputs selects a single representative tree \cite{neumann1986lattice}.

Last, but not least, max/min consensus aims to compute the aggregate max or min of a signal on the nodes of a network. In some settings, maximum and minimum are approximated \cite{tahbaz2006one}, or algorithms for reaching consensus on general functions on the inputs were proposed \cite{cortes2008distributed}. Max consensus was also framed as a max-plus-linear system \cite{nejad2009max}. Useful applications of min/max consensus and, arguably, meet/join consensus, include decentralized leader election \cite{borsche2010leader} and minimum-time rendezvous \cite{nejad2009max}.

%---------------------------------
\subsection{Applied sheaf theory}
%---------------------------------

Applied sheaf theory likely originated in the 1970s. Sheaves were argued to be an intuitive way of thinking about certain image segmentation problems \cite{bajcsy1973computer}. In the 1990s, the following audacious claim was made by Goguen, who applied sheaf theory to electrical circuitry \cite{goguen1992sheaf}: 
\begin{displayquote}
The sheaf condition appears to be satisfied by the behaviours of all naturally arising systems from computing science. This ``Sheaf Hypothesis'' is similar to the Church-Turing thesis, that all intuitively computable functions are computable in the precise sense of Turning machines.
\end{displayquote}

Sheaf theory, as a field of applied mathematics, is not yet mature. Some successes include sheaves as a tool for integrating sensors \cite{robinson2017sheaves}, filter design \cite{robinson2014topological}, path optimzation and routing \cite{cormen2022introduction,moy2020path}. The Greatest triumph of applied sheaf may be an ongoing program in the foundations of topological data analysis \cite{ghrist2008barcodes}. The theory of both persistent homology \cite{kashiwara2018persistent,curry2014sheaves,macpherson2021persistent} and Reeb graphs \cite{de2016categorified} has much benefited from a sheaf-theoretic viewpoint. In most settings, the data category of the sheaf (i.e.~what lives ``upstairs'') is the category of finite-dimensional vector spaces and the space is a graph. In this manuscript, we study sheaves over graphs valued in $\cat{Sup}$, the data category of complete lattices (Chapter \ref{ch:lattice-valued}). We briefly review adjacent work on sheaves valued in nonabelian categories.

In quantum mechanics, sheaves have been used to model non-locality: mysterious interactions between fundamental particles, presumably occurring faster than the speed of light \cite{abramsky2011sheaf}. Contextuality, a version of non-locality, is formulated as a sheaf of events, a sheaf in the data category of sets. Cohomology is undefinable here, hence, obstructions to locality, computed as cohomology in degree one, cannot be explicitly calculated from the sheaf of events. One approach is to extend the sheaf of events to a sheaf of abelian groups factoring through the free functor sending a set to its free abelian group. one disadvantage of this approach is that sections may appear that weren't there before. Alternatively, \v{C}ech cohomology was recently defined for sheaves of semimodules \cite{montanhano2021characterization}.

In our opinion, the most significant development in applied sheaf theory is the introduction of sheaf Laplacians \cite{hansen2019toward}. Thus far, sheaf Laplacians have inspired applications in opinion dynamics \cite{hansen2021opinion}, distributed optimization \cite{hansen2019distributed}, sheaf learning \cite{hansen2019learning}, and graph neural networks \cite{bodnar2022neural,barbero2022sheaf}.

% (Co)homology of (co)chain complexes in semiexact categories (the category of lattices and Galois connections incuded) was defined by Marco Grandis \cite{grandis2013homological}. However, a proper theory of sheaf cohomology of $\cat{Sup}$-valued sheaves would require either a notion of projective resolutions \cite{bredon2012sheaf}.

%% file: gfx/alternating-method.tex
% https://q.uiver.app/?q=WzAsMTIsWzAsMSwiXFxSZXh0Xm4iXSxbMiwxLCJcXFJleHRebiJdLFs0LDEsIlxcUmV4dF5uIl0sWzAsMywiXFxtYXRoYmZ7eH0iXSxbMSwzLCJBIFxcb3ZlZSBcXG1hdGhiZnt4fSJdLFs0LDMsIlxcbWF0aGJme3l9Il0sWzMsMywiQiBcXG92ZWUgXFxtYXRoYmZ7eX0iXSxbMiwwLCJcXG1hdGhiZnt6fSJdLFswLDAsIkFeXFxkYWdnZXIgXFxvd2VkZ2UgXFxtYXRoYmZ7en0iXSxbNCwwLCJCXlxcZGFnZ2VyIFxcb3dlZGdlIFxcbWF0aGJie3p9Il0sWzAsMiwiXFxidWxsZXQiXSxbNCwyLCJcXGJ1bGxldCJdLFswLDEsIkEiLDAseyJjdXJ2ZSI6LTF9XSxbMiwxLCJCIiwyLHsiY3VydmUiOjF9XSxbMyw0LCIiLDIseyJzdHlsZSI6eyJ0YWlsIjp7Im5hbWUiOiJtYXBzIHRvIn19fV0sWzUsNiwiIiwyLHsic3R5bGUiOnsidGFpbCI6eyJuYW1lIjoibWFwcyB0byJ9fX1dLFsxLDAsIkFeXFxkYWdnZXIiLDAseyJjdXJ2ZSI6LTF9XSxbMSwyLCJCXlxcZGFnZ2VyIiwyLHsiY3VydmUiOjF9XSxbNyw4LCIiLDIseyJzdHlsZSI6eyJ0YWlsIjp7Im5hbWUiOiJtYXBzIHRvIn19fV0sWzcsOSwiIiwwLHsic3R5bGUiOnsidGFpbCI6eyJuYW1lIjoibWFwcyB0byJ9fX1dLFsxMCwxMSwiIiwyLHsic3R5bGUiOnsiaGVhZCI6eyJuYW1lIjoibm9uZSJ9fX1dXQ==
\[\begin{tikzcd}
	{A^\dagger \owedge \mathbf{z}} && {\mathbf{z}} && {B^\dagger \owedge \mathbf{z}} \\
	{\Rext^n} && {\Rext^n} && {\Rext^n} \\
	\bullet &&&& \bullet \\
	{\mathbf{x}} & {A \ovee \mathbf{x}} && {B \ovee \mathbf{y}} & {\mathbf{y}}
	\arrow["A", curve={height=-6pt}, from=2-1, to=2-3]
	\arrow["B"', curve={height=6pt}, from=2-5, to=2-3]
	\arrow[maps to, from=4-1, to=4-2]
	\arrow[maps to, from=4-5, to=4-4]
	\arrow["{A^\dagger}", curve={height=-6pt}, from=2-3, to=2-1]
	\arrow["{B^\dagger}"', curve={height=6pt}, from=2-3, to=2-5]
	\arrow[maps to, from=1-3, to=1-1]
	\arrow[maps to, from=1-3, to=1-5]
	\arrow[no head, from=3-1, to=3-5]
\end{tikzcd}\]

%% file: Chapters/Chapter02.tex
%*****************************************
\chapter{Lattices}\label{ch:lattices}
%*****************************************

%----------------------------------
\section{From Relations to Posets}
%----------------------------------

Given sets $X$ and $Y$, a binary relation is simply a subset $\rel{R} \subseteq X \times Y$. Sometimes we write $x \rel{R} y$ whenever $(x,y) \in \rel{R}$. Suppose $\rel{R} \subseteq X \times Y$ and $\rel{S} \subseteq Y \times Z$. We can compose relations as follows:
\[ \rel{S} \circ \rel{R} = \{ (x,z)~\vert~\exists~y \in Y~\text{such that}~x \rel{R} y,~y \rel{S} z \}.\]
We may also take unions $\rel{R} \cup \rel{R}'$ and intersections $\rel{R} \cap \rel{R}'$ of relations. Relations have duals. Suppose $\rel{R} \in X \times Y$ is a relation. Then $\rel{R}^{\dagger}$ is the relation $\{ (y,x)~\vert~x \rel{R} y\}$.

\begin{example}[Databases] \label{eg:databases}
	Suppose $X$ is a set of instances and $Y$ is a set of attributes. We set $x \rel{R} y$ if object $x$ has attribute $y$. Most of the time, attributes are not binary labels. One way to accommodate this situation to convert  these labels to integers. Then, a multi-valued relation is a map $\rel{R}: X \times Y \to \N$. Then, we define a filtration of relations
	\begin{equation}
	 	\rel{R}_{\leq m} = \{ (x,y)~\vert~\rel{R}(x,y) \leq m\}.
	 \end{equation}
	This sequence of relations is a filtration in the sense that $\rel{R}_{\leq m} \subseteq \rel{R}_{\leq m'}$ whenever $m \leq m'$.
\end{example}

Motivated by the preceding example, it is convenient to have notation for the set of attributes that is related to a given object. Let $\rel{R} \subseteq X \times Y$ be a relation and $\sigma \subseteq X$. Then, the \define{intent} of $\sigma$, denoted $\galup{\sigma}$, is the set $\{y~\vert~x \rel{R} y~\forall~x \in \sigma\}$. Conversely, the \define{extent} of a $\tau \subseteq Y$, denoted $\galdown{\tau}$, is the set $\{x~\vert~x \rel{R} y~\forall~y \in \tau\}$.

Whenever $X = Y$, we say $\rel{R}$ is an \define{endorelation} (or simply a relation from context). We denote the set of endorelations on $X$ with $\Rel{X}$. Endorelations are the same as directed graphs. For an $x \in X$, the intent $\galup{x}$ is sometimes called the \define{reachability set} or \define{children} of $x$. The extent $\galdown{x}$ is called the \define{predecessor set} or \define{parents} of $x$. Several properties of endorelations are of note.

\begin{definition}
Suppose $\rel{R} \in \Rel{X}$. Then, $\rel{R}$ is 
\leavevmode
\begin{enumerate}
	\item \define{Transitive} if $x \rel{R} y$, $y \rel{R} z$ implies $x \rel{R} z$,
	\item \define{Reflexive}  if $x \rel{R} x$,
	\item \define{Anti-symmetric} if $x \rel{R} y$ and $y \rel{R} x$ implies $x = y$,
	\item \define{Symmetric} if $x \rel{R} y$ and $y \rel{R} x$,
    \item \define{Serial} if for every $x \in X$, there exist $y \in X$ such that $x \rel{R} y$,
    \item \define{Euclidean} if for every $x, y, z \in X$, $x \rel{R} y$ and $x \rel{R} z$ implies $y \rel{R} z$,
    \item \define{Connex} if for every $x,y \in X$, $x \rel{R} y$ or $y \rel{R} x$.
\end{enumerate}
\end{definition}

A \define{equivalence relation} is a relation $\sim~\in~\Rel{X}$ satisfying transitivity, reflexivity and symmetry. It is well known that equivalence relations are in one-to-one correspondence between equivalence relations and partitions. A partition of $X$ is a collection of subsets $\pi = \{X_i\}_{i \in I}$ such that $\bigcup \pi  = X$ and $X_i \cap X_j = \emptyset$ for all $i \neq j$. Let $\Part{X}$ denote the set of partitions of $X$. It is often convenient to represent a partition e.g.~$\pi \in \Part{\{1,2,3,4\}}$ as $ \pi = \{1 2 \vert 3 \vert 4\}$. 

\begin{example}[Semantics]
	This example serves as a prelude to Chapter \ref{ch:semantics}. Suppose $S$ is set of states of a system. For instance, $S$ could be units of time, or, in a multi-agent system, each agent could have local states $S_i$, for instance, detection states $\{pos, neg\}$ in a sensor network, or message statuses $\{sent, received\}$. In the local paradigm, global states are a prodcut of local states $S = \prod_{i} S_i$. A \define{Kripke relation} is an endorelation $\rel{K} \in \Rel(S)$. A Kripke relation defines a ``possible worlds model'' upon which various multi-agent logics \cite{fagin2004reasoning} are built. For instance, given a state $s$, the reachability set $s^{\uparrow} = \{t~\vert~s \rel{K} t\}$ constitutes set of states for which a given proposition must be true in order that the proposition be ``especially'' true at state $s$.
\end{example}

\begin{definition}
	A \define{poset} (\underline{p}artially \underline{o}rdered \underline{set}) $(\poset{P}, \preceq)$ is a set $X$ with a transitive, reflexive and anti-symmetric binary relation $\preceq$.
\end{definition}

We often omit the order $\preceq$ in  $(\poset{P}, \preceq)$ when it is clear from context. If we remove the anti-symmetry requirement, we say $\poset{P}$ is a preordered set or \define{proset}. We write $x \prec y$ if $x \preceq y$ and $x \neq y$. The ordering $\prec$ is called a \define{strict order}.

% \begin{example}[Weak \& Strong Preferences]
% 	Suppose a voter prefers candidate $x$ to a candidate $y$, and believes a candidate $y$ is at least as good as a candidate $z$. We would say that the voter has a \define{weak} preference over $x$ and $y$ and a \define{strong} preference over $y$ and $z$. In a poset, we would write $x \succ y \succeq z$ and conclude by transitivity that $x \succeq z$.
% \end{example}

Before getting to far along, we describe a standard way to visualize posets which will require a definition. We say that an element $y$ of $\poset{P}$ \define{covers} $x$ written $x \covers y$ if there does not exist $z \in \poset{P}$ with $x \prec z \prec y$. We form a covering (endo)relation $\covers$ depicting the elements that cover each other. We can draw the covering relation as a directed graph called a \define{Hasse diagram}. For historical reasons, we do not draw arrows because arrows always point from the bottom to the top of the page. Paths in the hasse diagram correspond to chains. Transitivity and reflexivity is implicit. We do not draw self-loops or arrows that realize transitivity, but they are there (Figure \ref{fig:hasse-dmn}.)

\begin{figure}[h]
	\centering
	\includegraphics[width=0.25\textwidth]{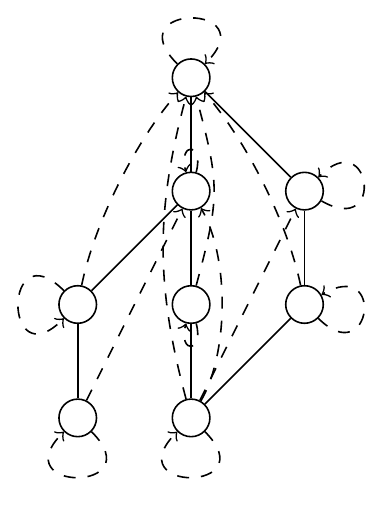}
	\caption{Covering relation (Hasse diagram) v.s.~entire partial order.}
	\label{fig:hasse-dmn}
\end{figure}

%------------------------------
\subsection{Subsets of posets}
%------------------------------

Several types of subsets of $\poset{P}$ are of interest. An \define{upset} is a subset $U \subseteq \poset{P}$ with the property
\[ x \in U,~y \succeq x~\Rightarrow~y \in U.\]
 Similarly, a \define{downset} $D$ is a subset $D \subseteq \poset{P}$ with the property
\[ x \in D,~y \preceq x,~\Rightarrow~y \in D. \]
A \define{principal upset} is generated by a single $x \in \poset{P}$: $\upset{x} = \{y \in \poset{P}~\vert~y \succeq x \}$. A \define{principal downset} is generated by a single $x \in \poset{P}$: $\downset{x} = \{y \in \poset{P}~\vert~y \preceq x \}$.

An \define{interval} in a poset $\poset{P}$ is a subset $I \subseteq \poset{P}$ with the convexity property that for every $x, y \in I$, if $x \preceq z \preceq y$ in $\poset{P}$, then $z \in \poset{P}$ also. An interval is \define{closed} if it is the form $\left[x, y \right] = \{z~\vert~x \preceq z \preceq y\}$. We denote the set of closed intervals in a poset $\Int{\poset{P}}$. A \define{chain} is a subset $C \subseteq \poset{P}$ such that $C$ is totally ordered: for every $x, y \in C$, either $x \preceq y$, or $y \preceq x$ (or both). Similarly, an \define{antichain} is a subset $A \subseteq \poset{P}$ such that for every $x, y \in A$, neither $x \preceq y$ nor $y \preceq x$ (written $x \antichain y$). The \define{height} of $\poset{P}$ is the (path) length of the maximal chain. The \define{width} of $\poset{P}$ is the cardinality of the maximal antichain.

A subset $V \subseteq \poset{P}$ is \define{directed} if for every $x, y \in \poset{P}$, there is a $z$ such that $z \succeq x, z \succeq y$. A subset $F \subseteq \poset{P}$ is \define{filtered} if for every $x, y \in \poset{P}$, there is a $z$ such that $z \preceq x, y$.

\subsection{Convergence}

When computing with lattices, it is useful to have a notion when chains ``converge'' in both a finite and trans-finite sense.

\begin{definition}\label{def:dcc}
$(\poset{P}, \preceq)$ satisfies the descending chain condition (DCC) if every every chain $C \subseteq P$ has a minimum element. Similarly, $\poset{P}$ satisfies the ascending chain condition (ACC) if every chain $C$ has a maximum element. 
\end{definition}

A fortuitous sufficient condition for satisfying DCC (dualy, ACC) is a grading. A poset $\poset{P}$ is \define{graded} if there is a map
\begin{equation}
	r: \poset{P} \to \N
\end{equation}
such that $r(x) < r(y)$ whenever $x \preceq y$ (strictly monotone) and $r(y) = r(x) + 1$ whenever $x \covers y$. A grading on a poset is a particularly restrictive topological sorting of the Hasse diagram. In general, a \define{topological sorting} \cite{cormen2022introduction}assigns a natural number to every node of a directed graph such that the number assigned to the source of an arrow is less than the number assigned to the target.

\begin{examples}[Grading] Suppose $V$ is a vector space. The lattice of finite-dimensional subspaces is a lattice graded by subspace dimension. Similarly, if $X$ is a set, the lattice of finite subsets is a lattice graded by cardinality. The lattice of (finite) partitions of $X$ is graded by the number of parts.
\end{examples}

\begin{proposition}
 	Suppose $\poset{P}$ is graded with grading $r: \poset{P} \to \N$. Then, $\poset{P}$ satisfies DCC.
 \end{proposition}
 \begin{proof}
 	Consider a chain $C \subseteq \poset{P}$. Then, $r(C)$ is a decreasing sequence
	of natural numbers bounded below by $0$. Hence, $r(C)$ converges i.e.~has a minimum element.
 \end{proof}

 Some posets, for instance, the unit interval $[0,1] \subseteq \R$, do not satisfy the ascending/descending chain condition. For instance, the sequence $\{a\}_{n \in \N} = 1- 1/n$ is an ascending chain in $[0,1]$, yet $a_{n+1}$ is strictly greater than $a_n$ for every $n \geq 0$. However, the least upper bound of $\{a_n\}$ provides an answer to the question: \emph{what does $a_n$ converge to?} In general, least upper bounds and greatest lower bounds equip posets with notions of convergence as well as algebraic properties.

%-----------------------------
\subsection{Completeness}
\label{sec:complete-lattice}
%-----------------------------

Suppose $S \subseteq \poset{P}$ is an arbitrary subset. If they exists, we may compute least upper bounds and greatest lower bounds which we hereafter call \define{joins} and \define{meets} respectively:

\begin{align*}
\bigjoin S &=& \min \{y~\vert~y \succeq x~\forall~x \in S\}; \\
\bigmeet S &=& \max \{y~\vert~y \preceq x~\forall~x \in S\}.
\end{align*}

We say that $\poset{P}$ is a \define{complete lattice} if for every subset $S \subseteq \poset{P}$, $\bigjoin S$ exists. $\poset{P}$ is \define{directed complete} if $\bigsqcup V$ exists for every directed subset $V$.\footnote{Of course, there is a dual notion of filtered complete, but, for reasons of tradition, directed complete posets are studied. Also for historical reasons, joins of directed sets are written with $\bigsqcup$.} In fact, it is enough that a poset contains either arbitrary meets or arbitrary joins in order to be a complete lattice, explaining the alternative nomenclature \define{suplattices}.

\begin{proposition}\label{thm:sup-inf}
	Suppose $(\poset{P}, \preceq)$ is a poset and $S \subseteq \poset{P}$ an arbitrary subset. Then, $\bigjoin S$ exists if and only if $\bigmeet S$ exists.
\end{proposition}
\begin{proof}
	We write
	\begin{align*}
	\bigjoin S &=& \bigmeet \{y \in \poset{P}~\vert~y \succeq x,~\forall~x \in S\} \\
	&=& \bigmeet \bigcap_{x \in S} \upset{x}.
	\end{align*}
\end{proof}

\begin{warning}
	Going forward, we make the following abuse of notation. If $\poset{P}$ and $\poset{Q}$ are partial orders, we do not always distinguish between the order relation on $\poset{P}$ and $\poset{Q}$, writing $\preceq$ for both of them. We also do not distinguish between joins/meets in $\poset{P}$ and joins/meets in $\poset{Q}$. Both we write both with the symbols $\bigjoin$ and $\bigmeet$.
\end{warning}

%---------------------------------
\subsection{Maps between posets}
%---------------------------------

A map $f$ between posets $f: \poset{P} \to \poset{Q}$ is \define{order-preserving} or \define{monotone} if $x \preceq y$ implies $f(x) \preceq f(y)$. The set of order-preserving maps between $\poset{P}$ and $\poset{Q}$ is a poset with $f \preceq G$ if $f(x) \preceq g(x)$ for all $x \in \poset{P}$. $f$ is an \define{order embedding} if $f(x) \preceq f(y)$ if and only if $x \preceq y$; an \define{order isomorphism} if $f$ is a surjective order embedding.

The following example is relevant in explainable AI (artificial intelligence).

\begin{example}[Binary Classification]
	Suppose $f_{\theta}: \R^d \to [0,1]$ is a statistical model with parameter $\theta$ that inputs financial data about an applicant (e.g.~income, negative debt, credit score) and outputs the likelihood that an application will not default on a load. In a supervised training, $f_{\theta}(\mathbf{x}_i) = 1$ if $i$ made loan payments and or $f_{\theta}(\mathbf{x}_i) = 0$ if $i$ defaulted. In this scenario, one would desire for $f_{\theta}$ to be monotone to account for the fact that an applicant with a higher income, lower debt, and higher credit score will always be more likely to pay a loan back. This means that $f_\theta(\mathbf{x}) \leq f_\theta(\mathbf{x})$ whenever $\mathbf{x} \preceq \mathbf{y}$.
\end{example}

Posets, and all relations, have a duality baked into their definitions. Given $(\poset{P},\preceq)$, let $(\op{\poset{P}}, \sqsubseteq)$ be the poset called the \define{opposite poset} with the same base set $P$ and order $x \sqsubseteq y$ if and only if $x \succeq y$. Often in order-theoretic arguments, we invoke the \define{Duality Prinicple} which loosely says \emph{if a statement is true for $\poset{P}$, then the dual statement obtained by dualizing all the order-theoretic definitions in the statment is true for $\op{\poset{P}}$}. As an example, a map $f$ is \define{order-reversing} or \define{antitone} if $x \preceq y$ implies $f(y) \preceq f(x)$.  Alternatively, by the Duality Principle, we could say an antitone map between $\poset{P}$ and $\poset{Q}$ is a monotone map $f: \poset{P} \to \op{\poset{Q}}$ or, equivalently, a monotone map $f: \op{\poset{P}} \to \poset{Q}$.

A map $f: \poset{P} \to \poset{Q}$ is join-preserving if for every subset $S \subseteq \poset{P}$
\begin{align*}
	f( \bigjoin S) &=& \bigjoin f(S).
\end{align*}
A map $f: \poset{P} \to \poset{Q}$ is meet-preserving if for every subset $S \subseteq \poset{P}$
\begin{align*}
f( \bigmeet S) = \bigmeet f(S).
\end{align*}
\noindent If $\poset{P}$ and $\poset{Q}$ are directed complete partial orders, then we say $f: \poset{P} \to \poset{Q}$ is \define{continuous} if for every directed subset $V \subseteq \poset{P}$, $f(\bigsqcup V) = \bigsqcup f(V)$.

 \begin{example}[Computer Vision]
 	% Suppose $A \subseteq E$ is a binary image on a grid $E$ (i.e.~pixels are either black or white, although this example is readily generalized to arbitrary images). If $B \subseteq E$ is a structuring element \cite{haralick1987image} (e.g. a black $2\times2$ square), then
 	% \[ A \oplus B = \bigcup_{b \in B} A_b \]
 	% where $A_b = \{a+b~\vert~a \in A\}$. We might view this operation as a map $\delta_B: \powerset{E} \to \powerset{E}$ called a \define{dilation} sending $A \mapsto A \oplus B$. It is clear that $\delta_B$ is join-preserving. On the other hand, define
 	% \[A \ominus B = \{x \in E~\vert~B_x \subseteq A\}\]
 	% and $\epsilon_B: \powerset{E} \to \powerset{E}$ by $A \mapsto A \ominus B$.
 	% This operation is called an \define{erosion} and is meet-preserving.
 	Let $\Rext = \R \cup \{-\infty, +\infty \}$ denote the extended real numbers. Suppose $f: \domain \to \Rext$ is a grayscale image where $\domain$ is the image domain (e.g. $\Z^2$). In image filtering, it is popular to define filters called erosions and dilations \cite{haralick1987image}. If $\graph$ is a second image called the \define{kernel}, define the \define{dilation} of $f$ by
 	\begin{align*}
 		(f \oplus g)(x) &=& \bigjoin_{y \in \domain} f(x-y) + g(y).
 	\end{align*}
 	Define the \define{erosion} of $f$ by
 	\begin{align*}
 		(f \ominus g)(x) &=& \bigjoin_{y \in \domain} f(x+y) - g(y).
 	\end{align*}
 	In the lattice of images \[\Rext^{\domain} = \{f: \domain \to \Rext\},\] the operators
 	\begin{align*}
 		\delta_g &=& f \oplus g \\
 		\epsilon_g &=& f \ominus g \\
 		\delta_g,\epsilon_g &:& \Rext^{\domain} \rightarrow \Rext^{\domain}
 	\end{align*}
 	satisfy
 	\begin{align*}
 		\delta_g( \bigjoin_{i} x_i) &=& \bigjoin_{i} \delta_g(x_i) \\
 		\epsilon_g( \bigmeet_i x_i) &=& \bigmeet_{i} \epsilon_g(x_i)
 	\end{align*}
 	for any $g \in \lattice{L}$. Hence, dilations, as the name suggests, are join-preserving, and erosions are meet-preserving. Other filters called \define{openings} and \define{closing} are defined by composing an erosions and dilations and are useful in denoising, contrast enhancement, reconstruction and edge detection \cite{maragos2009morphological}.
\end{example}

%-------------------------
\section{Fixed Points}
%-------------------------

Frequently, we study maps $f: \poset{P} \to \poset{P}$. Such a map is \define{inflationary} if $f(x) \succeq x$ for all $x \in \poset{P}$ and \define{deflationary} if $f(x) \preceq x$ for all $x \in \poset{P}$. We write $\prefix(f)$ and $\suffix(f)$ for the subsets of deflationary and inflationary points, respectively.
Let $\fixed(f)$ denote the set of fixed points $\prefix(f) \cap \suffix(f)$ (anti-symmetry axiom).

%------------------------------
\subsection{(Co)closure Operators}
%------------------------------

 A map $f: \poset{P} \to \poset{P}$ is \define{idempotent} if $f^2 = f$. Suppose $\poset{P}$ is a poset and $f: \poset{P} \to \poset{P}$ is an order-preserving map. There are various ways to characterize $\fixed(f)$ depending on the properties of $\poset{P}$ and $f$.

 \begin{definition}[Closure Operators]
  Suppose $f: \poset{P} \to \poset{P}$ is a map and $\poset{P}$ is a poset. If $f$ is order-preserving, inflationary, and idempotent, we say $f$ is a \define{closure operator}. Conversely, if $g: \poset{P} \to \poset{P}$ is order-preserving, deflationary, and idempotent we say $g$ is a \define{coclosure operator}. The fixed points of a (co)closure operator are called \define{(co)closed}.
 \end{definition}

 The theory of (co)closure operators \cite{roman2008lattices} allows us to explicitly compute joins and meets of the fixed point lattice.

\begin{theorem}[(Co)closure] \label{thm:closure-fixed}
	Suppose $\lattice{L}$ is a complete lattice and $f: \lattice{L} \to \lattice{L}$ is a closure operator on a complete lattice $(\lattice{L}$. Then, $\fixed(f)$ is a complete lattice with the following meets and joins
\begin{align*}
	\bigsqcap_{i \in I} x_i &=& \bigmeet_{i \in I} x_i \\
	\bigsqcup_{i \in I} x_i &=& f \left( \bigjoin_{i \in I} x_i \right).
\end{align*}
    Dually, suppose $g$ is a coclosure operator on $\lattice{L}$. Then, $\fixed(f)$ is a complete lattice with the following meets and joins
    \begin{align*}
        \bigsqcap_{i \in I} x_i &=&  f \left( \bigmeet_{i \in I} x_i \right ) \\
        \bigsqcup_{i \in I} x_i &=& \bigjoin_{i \in I} x_i.
    \end{align*}
\end{theorem}
\begin{proof}
    See \cite[Theorem 3.8]{roman2008lattices}.
\end{proof}
 
\noindent It is not surprising that closure operators come about naturally in topology.

 \begin{example}[Closure Spaces]
 	A \define{closure space} is a pair $(X, c)$ where $X$ is an (arbitrary) set and $c$ is a closure operator $c: \powerset{X} \to \powerset{X}$. Closed subsets in $X$ are precisely subsets $A$ with $c(A) = A$. A closure space is \define{topological} if
 	\begin{align*}
 		c(A \cup B) &=& c(A) \cup c(B), \\
 		c(\emptyset) &=& \emptyset.
 	\end{align*}
 	Recent work extends persistent homology to new settings (e.g. directed graphs) using closure spaces \cite{bubenik2021homotopy}.
 \end{example}

While (co)closure operators appear frequently throughout various branches of mathematics, the requirement of idempotency and inflation are often too stringent in practical situations.
% One situation that arises frequently in domain theory \cite{abramsky1994domain} and the theory of computation is the following:  $\poset{P}$ is a directed complete partial order (dcpo) and $f$ is a continuous.\footnote{Recall, continuity implies that $f$ is order-preserving.} There is a well-known algorithm to compute the least fixed point of such a function encapsulated by the following theorem.
% \begin{theorem}[Scot Fixed Point Theorem] \label{thm:scot-fp}
% 	Suppose $\poset{P}$ is a dcpo with zero element $0$ and $f$ is a continuous map. Then, $f$ has a fixed point
% 	\begin{align*}
% 		x^{\ast} &=& \bigsqcup_{k \geq 0} f^k(0)
% 	\end{align*}
% \end{theorem}
% \begin{proof}
% 	The set
% 	\begin{align*}
% 	\{f^k(0) \}_{k \geq 0}
% 	\end{align*}
% 	is directed because $f(0) \geq 0$ is satisfied trivially and $f$ is monotone (by continuity). Hence, by continuity,
% 	\begin{align*}
% 	f\left( \bigsqcup_{k \geq 0} f^k(0) \right) &=& \bigsqcup_{k \geq 0} f^{k+1}(0).
% 	\end{align*}
% 	If $f(0) = 0$, we are done. Hence, assume that $x^{\ast} \succ 0$ so that
% 	\begin{align*}
% 		x^{\ast} &=& \bigsqcup_{k \geq 0} f^{k+1}(0) &=& \bigsqcup_{k \geq 0} f^k(0).
% 	\end{align*}	
% \end{proof}
The following fixed point theorem is integral to our study of operators on lattice-valued sheaves.

\begin{theorem}[Tarski Fixed Point Theorem] \label{thm:tfpt}
	Suppose $\lattice{L}$ is a complete lattice and $f: \lattice{L} \to \lattice{L}$ is an order-preserving map. Then, $\fixed(f)$ is a complete lattice. 
\end{theorem}
\noindent We will need a few lemmas.
\begin{lemma}\label{lem:inside-meets}
	Suppose $f: \lattice{L} \to \lattice{L}$ is order-preserving and $\lattice{L}$ is complete. Suppose $\{x_i\}_{i \in I}$ is a subset of $\lattice{L}$ indexed by an arbitrary set $I$. Then,
	\begin{align*}
		f \left( \bigmeet_{i \in I} x_i \right) &\preceq& \bigmeet_{i \in I} f(x_i).
	\end{align*}
	Similarly,
	\begin{align*}
		f \left( \bigjoin_{i \in I} x_i \right) &\succeq& \bigjoin_{i \in I} f(x_i).
	\end{align*}
\end{lemma}
\begin{proof}[Proof of Lemma \ref{lem:inside-meets}]
	$f \left( \bigmeet_{i \in I} x_i \right)$ is \emph{a} lower bound of each $f(x_i)$ because $f$ is order-preserving and $\bigmeet_i x_i \preceq x_i$. Moreover, $\bigmeet_{i \in I} f(x_i)$ is \emph{the} greatest lower bound of the subset $\{f(x_i) \}_{i \in I}$, implying $\bigmeet_{i \in I} f(x_i)$ precedes all other lower bounds. Hence, 
	\begin{align*}
		f \left( \bigmeet_{i \in I} x_i \right) &\preceq& \bigmeet_{i \in I} f(x_i).
	\end{align*}
	The second statement follows from the duality principle.
\end{proof}
\begin{lemma}\label{lem:opposite-fixed}
	Suppose $f: \poset{P} \to \poset{P}$ is order-preserving. Then, $\op{f}: \op{\poset{P}} \to \op{\poset{P}}$ defined to be the same map on sets, but under the opposite order is order-preserving as well. If $f$ is join-preserving, then $\op{f}$ is meet preserving and vice versa. Furthermore, $\prefix(\op{f}) = \suffix(f)$ and $\suffix(\op{f}) = \prefix(f)$.
\end{lemma}
\begin{proof}[Proof of Lemma \ref{lem:opposite-fixed}]
	$x \preceq y$ in $\poset{P}$ if and only if $x \succeq y$ in $\op{\poset{P}}$. $f$ is order-preserving implies $f(x) \preceq f(y)$ in $\poset{P}$, or $f(x) \succeq f(y)$ in $\op{\poset{P}}$. Hence, $\op{f}$ is order-preserving. Meets in $\op{\poset{P}}$ are equivalent to joins in $\poset{P}$.
\end{proof}

\begin{proof}[Proof of Theorem \ref{thm:tfpt}]
	The classical argument \cite{tarski1955lattice} that $\fixed(f)$ is complete has three parts. First, we show $\prefix{f} \subseteq \lattice{L}$ is an invariant subset under $f$. Suppose $x \in \prefix(f)$. Then, 
	\begin{align*}
		f(f(x)) \preceq f(x) \preceq x
	\end{align*}
	which implies $f(x) \in \prefix(f)$ by monotonicity of $f$ and transitivity. Next, we show that $\prefix(f)$ is complete. Let $\{x_i\}_{i \in I} \subseteq \prefix(f)$. Then, Lemma \ref{lem:inside-meets} implies that $\prefix(f)$ is $\bigmeet$-complete:
	\begin{align*}
		f \left( \bigmeet_{i \in I} f(x_i)	\right) \preceq \bigmeet_{i \in I} f(x_i)
	\end{align*}
	implies $\bigmeet_{i \in I} f(x_i) \in \prefix(f)$. By Proposition \ref{thm:sup-inf}, $\prefix(f)$ is also $\bigjoin$-complete.
	Finally, define
	\begin{align*}
	\op{f_{\vert \prefix(f)}}: \op{\prefix(f)} \to \op{\prefix(f)}.
	\end{align*}
	Invariance and Lemma \ref{lem:opposite-fixed} imply $\op{f_{\vert \prefix(f)}}$ is well-defined and order-preserving on $\op{\prefix(f)}$ and that $\prefix(\op{f_{\vert \prefix(f)}}) = \suffix\left(f_{\vert \prefix(f)}\right)$. Now, the argument in the second part shows $\suffix\left(f_{\vert \prefix(f)}\right)$ is complete. Furthermore,
	\begin{align*}
		\fixed(f) &=& \suffix\left(f_{\vert \prefix(f)}\right)
	\end{align*}
	and we are done.
\end{proof}

\noindent The proof of Theorem \ref{thm:tfpt} implies the following corollaries.
\begin{corollary}\label{cor:prefix-suffix}
	The subsets $\prefix(f)$ and $\suffix(f)$ are complete.
\end{corollary}
\begin{proof}
	Let $f \meet \id$ be the map defined $f(x) \meet x$ on $x \in \lattice{L}$. Then, $\fixed(f \meet \id) = \suffix(f)$ because $f(x) \succeq x$ if and only if $f(x) \meet x = x$. Similarly, $\fixed(f \join \id) = \prefix(f)$.
\end{proof}

In particular application domains, it is only necessary to show there is a least fixed point. Hence, the Tarski Fixed Point Theorem is sometimes presented as the following corollary.
\begin{corollary}
	Suppose $\lattice{L}$ is a complete lattice and $f: \lattice{L} \to \lattice{L}$ is order-preserving. Then, $f$ has a least fixed point $x^{\ast}$.
\end{corollary}
\begin{proof}
	$x_{\ast} = \bigmeet \fixed(f)$.
\end{proof}

% Lemma \ref{lem:inside-meets} has a corrolary that can be interpreted as \define{dynamic programming} \cite{bellman1958routing}.

% \begin{corollary}
% 	Suppose $\{f_j: \lattice{L} \to \lattice{L}\}_{j \in J}$ is a set of order-preserving function and $\{x_i\}_{i \in I}$. Then,
% 	\begin{align*}
% 		\bigmeet_{j \in J} f_j \left( \bigmeet_{i \in I} x_i \right) &\preceq& \bigmeet_{i \in I,~j \in J} f_j(x_i).
% 	\end{align*}
% \end{corollary}

%------------------------------------
\section{From Posets to Lattices}
%------------------------------------

Thus far, we have seen that complete lattices are posets with arbitrary meets and joints. We now present an alternative algebraic view of lattices in which we only require binary or nullary (in the case of bounded latttices) meets and joins. 

\begin{definition}[Algebraic Definition]\label{def:lattice-alg}
	A \define{lattice} $\lattice{L}$ is a set $L$ with a pair of binary operations
	\[\join, \meet: \lattice{L} \times \lattice{L} \to \lattice{L}\]
	called \define{join} and \define{meet}
	such that for every $x, y, z \in \lattice{L}$
	\begin{enumerate}
		\item $x \join (y \join z) = (x \join y) \join z$; $x \meet (y \meet z) = (x \meet y) \meet z$ (Associativity),
		\item $x \join y = y \join z$; $x \meet y = y \meet z$ (Commutativity),
		\item $x \join x = x$; $x \meet x = x$ (Idempotence),
		\item $x \meet (x \join y) = x$; $x \join (x \meet y) = x$ (Absorption).
	\end{enumerate}
	A \define{bounded lattice} is a lattice with elements $0, 1 \in \lattice{L}$ with the axiom
	\begin{enumerate}
		\item[5.] $x \meet 1 = x$; $x \join 0 = x$ 
	\end{enumerate}
	A \define{semilattice} $\lattice{S}$ is a set $S$ with a binary operation
	\begin{align*}
	 	\meet: \lattice{S} \times \lattice{S} \to \lattice{S}
	 \end{align*}
	 such that for all $x, y \in \lattice{S}$ 
	\begin{enumerate}
		\item $x \meet (y \meet z) = (x \meet y) \meet z$ (Associativity),
		\item $x \meet y = y \meet z$ (Commutativity),
		\item $x \meet x = x$ (Idempotence).
	\end{enumerate}
	A \define{bounded semilattice} is a semilattice with an element $1 \in \lattice{S}$ with the axiom
	\begin{enumerate}
			\item[4.] $x \meet 1 = x$.
	\end{enumerate}
\end{definition}

\begin{example}[Nesting Poset]
    Suppose $\Space{M}$ is a compact smooth $2$-manifold (e.g.~a closed and bounded surface in $\R^3$) diffeomorphic to the $2$-sphere $\Space{S}^2$. Suppose $f: \Space{M} \to \R$ is a Morse function \cite{matsumoto2002introduction} factoring through an embedding and a projection onto the z-axis
    \[i: \Space{M} \xhookrightarrow{i} \R^3 \xrightarrow{\pi} \R\] By a version of the Implicit Function Theorem \cite[Theorem 2.3]{matsumoto2002introduction}, if $t \in \R$ is a regular value, then $f^{-1}(t)$ is a compact manifold of dimension $1$. By the well-known classification of such manifolds, it follows there is a diffeomorphism
    \begin{align*}
        f^{-1}(t) \simeq \bigsqcup_{i=1}^m \Space{S}^1,
    \end{align*}
    and, furthermore, $f^{-1}$ is embedded in $\R^2$. Hence, $f^{-1}$ is a disjoint union of Jordan curve \cite{hatcher2002algebraic}, say, $\Space{Y} = \gamma_1, \gamma_2, \gamma_m\}$. The complement $\R^2 \setminus \Space{Y}$ consists of exactly $m+1$ connected components, with $1$ unbounded component and $m$ bounded components. It follows that the set of (path) connected components $\pi_0(\Space{M} \setminus \Space{Y})$ have the structure of a join-semilattice under the following order \cite{catanzaro2020moduli} equivalent to the circle containment order \cite{scheinerman1988circle}.
    
    We say two components in $\pi_0(\Space{M} \setminus \Space{Y})$ are adjacent if they they share a boundary. Then, for any pair of adjacent components $c, c' \in \pi_0(\Space{M} \setminus \Space{Y})$, then $\alpha_i \preceq \alpha_j$ if and only if the for the corresponding (unshared) boundaries $\gamma_i and \gamma_j$, we have an inclusion of interiors $\mathrm{int}(\gamma_i) \subseteq \mathrm{int}(\gamma_j)$ (existence, by the Jordan Curve Thoerem \cite{hatcher2002algebraic}) or if $\mathrm{int}(\gamma_j)$ is undbounded.
\end{example}

On the other hand, lattices are ordered sets with binary and empty meets and joins.
\begin{definition}[Combinatorial]\label{def:lattice-comb}
	A \define{lattice} is a partially ordered set $\lattice{L}$ such that for every $x, y \in \lattice{L}$
	\begin{align*}
	x \join y &=& \min \left\{ z~\vert~z \succeq x,~z \succeq y\right\}, \\
	x \meet y &=& \max \left\{ z~\vert~z \preceq x,~z \preceq y \right\}.
	\end{align*}
	A lattice is \define{bounded} if additionally there are elements $0, 1 \in \lattice{L}$ such that $0 \preceq x \preceq 1$ for all $x \in \lattice{L}$.
\end{definition}

The following proposition allows you to move freely between the order-theoretic and algebraic definitions of lattices.

\begin{proposition}\label{prop:alg-combin-equiv}
	The algebraic (Definition \ref{def:lattice-alg}) and combinatorial (Definition \ref{def:lattice-comb}) definitions of lattices are equivalent.
\end{proposition}
\begin{proof}
See \cite[Theorem 3.17]{roman2008lattices}.
\end{proof}
\noindent Consequently, a lattice (bounded) $\lattice{L}$ has the structure of both a (bounded) \define{meet-semilattice} $(\lattice{L}, \meet)$ and a (bounded) \define{join-semilattice} $(\lattice{L}, \join)$.

There are a few specializations of lattices that we take note. A \define{distributive lattice} is a lattice with the property that for all $x, y, z \in \lattice{L}$
\begin{align*}
	x \meet (y \join z) &=& (x \meet y) \join (x \meet z), \\
	x \join (y \meet z) &=& (x \join y) \meet (x \join z).
\end{align*}
A bounded lattice is \define{complemented} if for every $x \in \lattice{L}$ there exists an element $y \in \lattice{L}$ such that $x \join y = 1$ and $x \meet y = 0$. Powersets (Example \ref{eg:powersets}) are an example of a complemented distributive lattice, also know as a \define{Boolean lattice}.

% A lattice is \define{orthocomplemented} if there is an order-reversing map $(-)^{\perp}: \op{\lattice{L}} \to \lattice{L}$ such that $x^{\perp}$ is a complement of $x$ and $x^{\perp \perp} = x$. For example, the lattice of subspaces of an inner-product space, denoted $\subspaces{V}$, is an orthocomplemented lattice; if $W \in \subspaces{V}$ is an element, then
% \begin{align*}
% 	W^{\perp} = \{ v \in V ~\vert~ \langle v, w \rangle = 0~\forall~w \in W \}
% \end{align*}
% is its orthocomplement.

A \define{sublattice} $\lattice{K} \subseteq \lattice{L}$ is a subset $K \subseteq \lattice{L}$ closed under the $\join$ and $\meet$ operation in $\lattice{L}$. A \define{quasisublattice} of $\lattice{L}$ is a lattice $\lattice{K}$ with an order embedding $i: \lattice{K} \hookrightarrow \lattice{L}$ . Clearly, a sublattice is a quasi-sublattice, but not the other way around. If $\lattice{K}$ and $\lattice{L}$ are complete, then we say $\lattice{K}$ is a \define{complete quasisublattice}.

\begin{example}[Lattice of Subgroups]
Suppose $(G, \star, e)$ is a group \cite{dummitfoote}. The (normal, arbitrary) subgroups of $G$ form a lattice $\subspaces{G}$ with
\begin{align*}
	M \join N &=& \{ m \cdot n~\vert~m \in M,~n \in N \}, \\
	M \meet N &=& M \cap N.
\end{align*}
The lattice $\subspaces{G}$ order embeds into the powerset $\powerset{G}$. However, $\subspaces{G}$ is not a sublattice because the join of $M$ and $N$ in $\subspaces{G}$ differs from the join in $\powerset{G}$, $M \cup N$.
\end{example}

A lattice \define{homomorphism} is a map $f: \lattice{K} \to \lattice{L}$ such that for all $x, y \in \lattice{K}$
\begin{align*}
	f(x \join y) &=& f(x) \join f(x), \\
	f(x \meet y) &=& f(x) \meet f(y).
\end{align*}
A \define{bounded lattice homomorphism} is a lattice homomorphism with
\begin{align*}
	f(0) &=& 0, \\
	f(1) &=& 1.
\end{align*}
A homomorphism $f$ is an isomoprhimsm if there is another homomorphism $g: \lattice{L} \to \lattice{K}$ with $g \circ f = \id$ and $f \circ g = \id$. The reader may check that an order isomoprhism between (bounded) lattices is an isomorphism of (bounded) lattices.

\begin{example}[Multidimensional Persistence]
	This example was adapted from recent work on persistence modules over lattice \cite{mccleary2022edit}. Given a lattice $\lattice{L}$ (e.g $\N \times \N$ representing a bifiltration), a \define{persistence module} is an assignment $F: \lattice{L} \to \cat{Vect}$ of lattice elements to finite-dimensional $\field$-vector spaces such that $x \preceq y$ in $\lattice{L}$ induces a linear transformation $F(x) \xrightarrow{F(x \preceq y)} F(y)$ (formally, a functor into the category of vector spaces). A ``change in basis'' of a persistence modules is a lattice homomorpism $f: \lattice{K} \to \lattice{L}$. This induces a lattice homomorphism $\hat{f}: \Int{\lattice{K}} \to \Int{\lattice{L}}$. Certain interesting functions describing births and deaths of $i$-dimensional topological features, e.g. in a metric space or filtration of simplicial complexes \cite{ghrist2008barcodes}, are written as maps $f_i: \Int{\lattice{L}} \to \Z$ which map an interval $[a,b] \mapsto \dim \left(Z_i F(a) \cap B_i F(b) \right)$ where $Z_i F(a)$ and $B_i F(a)$ are the cycles and boundaries of homology \cite{hatcher2002algebraic}.
\end{example}

%-----------------------------------------
\subsection{Lattices with multiplicaiton}
%-----------------------------------------

Quite often, lattices are equipped with an algebraic structure other than binary meets and joins. For instance, the real numbers $\R$ is an (unbounded) lattice that also has the structure of a field. We briefly recall a few definitions. 

\begin{definition}[Monoid]
	A \define{monoid} is a tuple $(\monoid{M}, \star, e)$ such that $\monoid{M}$ is an arbitrary set, $\cdot: \monoid{M} \times \monoid{M} \to \monoid{M}$ is an associative binary operation, and $e \in \monoid{M}$ is an element such that $m \star e = e \cdot m = m$ for all $m \in \monoid{M}$. A monoid is a \define{group} if for every $m \in \monoid{M}$, there exist a unique element $m^{-1} \in \monoid{M}$ such that $m^{-1} \star m = m \star m^{-1} = e$. A monoid is \define{commutative} if $m \star n = n \star m$ for all $m, n \in \monoid{M}$. A monoid is \define{idempotent} if $m \star m = m$ for all $m \in \monoid{M}$.
\end{definition}
\begin{examples}[Monoids]
\leavevmode
\begin{enumerate}
	\item The natural numbers with addition $(\N, +, 0)$ is a monoid, but not a group. On the other hand, $(\Z, + , 0)$ is a group. 
	\item Suppose $X$ is a set, with possible additional structure. Then, the set of maps $f: X \to X$ is a monoid denoted $\End(X)$. The identity $e$ is the idenity map $\id: X \to X$. Multiplication is merely composition of endomorphisms.
	\item A bounded semilattice is a monid with $(\lattice{S}, \meet, 1)$.
\end{enumerate}
\end{examples}

\begin{proposition}
	Suppose $(\monoid{M}, \star, e)$ is a commutative idempotent monoid. Then, $\monoid{M}$ is a semilattice.
\end{proposition}
\begin{proof}
	Define a relation $m \preceq n$ if and only if $m \cdot n = m$. First check $(\monoid{M}, \preceq)$ is a partial order with top element $e$. Then, check $m \meet n$ is in fact the element $m \cdot n$.
\end{proof}

\noindent Residuated lattices are lattices equipped with a compatible monoid structure.

\begin{definition}[Residuated Lattice] \label{def:residuated}
	A \define{residuated lattice} $(\lattice{L}, \meet, \join, 0, 1, \star, e, [-,-])$ is a bounded lattice $(\lattice{M}, \meet, \join, 0, 1)$ such that $(\lattice{L}, \cdot, e)$ is a commutative monoid with a second (non-commutative) binary operation $[-,-]$ on $\lattice{M}$ with the property that
	\begin{align}
		x \star y \preceq z &\quad \text{if and only if} \quad& x \preceq [y, z] \label{eq:heyting}
	\end{align}
	for all $x, y, z \in \lattice{L}$. A residuated lattice is \define{complete} if $\lattice{L}$ is a complete lattice.
\end{definition}

\begin{example}[Fuzzy Logic]
	Consider three residuated lattice structures on the unit interval with the lattice structure $([0,1], \min, \max, 0,1)$, sometimes called \emph{t-norms}. Let $r, s, t \in [0,1]$.
	\leavevmode
	\begin{enumerate}
		\item Suppose $s \star t = s \cdot t$ is ordinary multiplication. Then, $[s, t]$ is calculated
		\[[s, t] = \begin{cases}
		t/s 		& s > t \\
		1 			& s \leq t
		\end{cases}\]
		from
		\begin{align*}
			r \cdot s \leq t &\quad \text{if and only if} \quad& r \leq [s, t].
		\end{align*}
		\item Define $s \star t = \min(s+t-1,0)$ and check
		\[ [s, t] = \begin{cases}
			1-s+t & s>t \\
			1 	  & s \leq t
		\end{cases}. \]
		\item Define $s \star t = \min(s,t)$ and check
			 \[[s, t] = \begin{cases}
			t & s>t \\
			1 	  & s \leq t
		\end{cases}. \]
	\end{enumerate}
	Given a residuated lattice structure on $[0,1]$ and a set $X$, a \define{fuzzy set} is a map $A: X \to [0,1]$. Viewing fuzzy sets as propositions, the logical operations of conjunction ($\meet$), disjunction ($\join$) and implication ($\to$) are defined pointwise according to the residuated lattice structure. For instance, if $A$ and $B$ are propositions, $[A, B]$ is a fuzzy set with $([A, B])(s) = A(s) \to B(s)$. 
\end{example}

\begin{example}[Heyting Algebra]
	Every bounded lattice is a monoid under the meet operation. However, this monoid structure does not \emph{ipso facto} satisfy \eqref{eq:heyting} which now reads
	\begin{align*}
		x \meet y \preceq z &\quad \text{if and only if} \quad& x \preceq [y, z]
	\end{align*}
	for all $x, y, z \in \lattice{H}$, a \define{Heyting algebra}.
	It follows that
	\begin{align*}
		[y, z] &=& \bigjoin \{ x \in \lattice{H}~\vert~x \meet y \preceq z \}
	\end{align*}
	leading to the definition of a \define{pseudocomplement} $\neg x = [x, 0]$. Examples of Heyting algebras are plentiful and run deep in foundation of mathematics \cite{maclane2012sheaves}. 
\end{example}

\begin{example}[Max-Plus Lattice]
	Max-plus algebra has found plentiful applications in locomotion \cite{lopes2014modeling}, control \cite{maragos2009morphological}, scheduling \cite{cassandras2008introduction}, and machine learning \cite{maragos2021tropical}. All of max-plus algebra is based on the idea of ``doing linear algebra'' over a residuated lattice, instead of a field. Consider the extended reals $\Rext = \R \cup \{-\infty, \infty\}$ which forms a complete lattice under $\leq$ with the additional requirement that $-\infty \leq t \leq \infty$ for all $t \in \R$. Then, $(\Rext, \meet, \join, -\infty, \infty, +, 0, [-,-])$ has the structure of a residuated lattice with $[s,t] = \max(t-s,0)$. A \define{max-plus} or \define{tropical} vector space, is then the product lattice $\Rext^d$. In some setting, this product structure is called a \define{weighted lattice} \cite{maragos2017dynamical} because you can define ``scalar multiplication'' component-wise.
\end{example}

%----------------------------------
\section{Additional Examples}
%----------------------------------

Now that most of the background material on lattices is established, we dive into some more examples.

\begin{example}[Subsets] \label{eg:powersets}
	Powersets are fundamental examples of lattices to the degree that the meet and join symbols $\meet$/$\join$ are derivates of the symbols $\cap$/$\cup$ which normally represent intersection and union. Given a set $S$, the \define{powerset} $\powerset{S}$ is a lattice under intersection and union.
	By definition, a union of set of subsets sets is the smallest set containing every subset. Similarly, an intersection of a set of subsets is the largest set contained in each subset. If we let $\bool = \{0 M 1\}$, then powersets are equivalent to the lattice of maps $S \to \bool$, with the inherited lattice structure
	\begin{align*}
	A \meet B (x) &=& \mathtt{AND}\left( A(x), B(x) \right), \\
	A \join B (x) &=& \mathtt{OR} \left( A(x), B(x) \right),
	\end{align*}
	$A, B \in \bool^S$. Similarly, a \define{multiset} is a map $S \to \N$, often denoted $\N \left[S\right]$. The lattice structure on multisets is the following
	\begin{align*}
		A \meet B (x) &=& \min\left( A(x),B(x), \right) \\
		A \join B (x) &=& \max \left(A(x), B(x), \right)
	\end{align*}
	$A, B \in \N \left[ S \right]$.	

	An interesting use-case of powerset lattices is modeling preferences. If $X$ is a set of alternatives, for instance, restaurants), then $\Rel{X} := \powerset{X \times X}$ is the set of all possible preference relations on $X$. Preference relations can be combined with the meet and join axioms, in this case, intersection and union. More interesting lattice structures are obtained by restricting $\Rel{X}$ to transitive-reflexive relations, as it is often argued that preferences are transitive.

	Tangential to the theory of preference relations, is the theory of choice functions. A \define{choice function} is a map $c: \powerset{X} \to \powerset{X}$ such that $c(A) \subseteq A$ for all $A \in \powerset{X}$. Intuitively, a choice function narrows down a subset of choices. It is established that the set of choice functions forms a lattice with pointwise union and intersection; interesting subclasses of choice functions form lattices as well \cite{monjardet2004lattices}.
	\end{example}

\begin{example}[Subspaces]
	Let $V$ be a $\field$-vector space. Let $\subspaces{V}$ consist of the set of all subspaces of $V$. If we write, $\Grass(V,d)$, for the \define{Grassmanian}, the $d$-dimensional subspaces, then we can write
	\begin{align*}
	\subspaces{V} &=& \coprod_{d\geq 0} \mathrm{Gr}(V,d).
	\end{align*}
	$\subspaces{V}$ is a lattice under the subspace ordering $K \leq W$ if $K$ is a subspace of $W$. The lattice operations are the following: $W \join W' = \{w + w'~\vert~w \in W,~w' \in W'\}$ often denoted $W + W'$; $W \meet W' = W \cap W'$ which is trivially a subspace of $V$. Principal Component Analysis (PCA) and its variants can be thought of as a projection of datapoints onto an optimal subspace $W$ of some feature space $V$ \cite{wright2022high}.

	Suppose $T: V \to V$ is a linear transformation. We say a subspace $W \subseteq V$ is \define{invariant} if $T(W) \subseteq W$. For instance, $1$-dimensional invariant subspaces are spanned by vectors $\mathbf{x}$ satisfying $T(\mathbf{x}) = \lambda \mathbf{x}$ for some scalar $\lambda$---these invariant subpaces are eigenspaces. It follows that the invariant subspaces of $T$ denoted $\Inv{T}$ form a lattice, a sublattice of $\subspaces{V}$. Conditions are known for invariant subspaces to be isomorphic \cite{brickman1967invariant}.
\end{example}

\begin{example}[Partitions] \label{eg:partitions}
	The set of partitions $\Part{S}$ of a set $S$ forms a lattice. Suppose $\pi$ and $\pi'$ are partitions. Then, the relation $\pi \preceq \pi'$ means that \emph{$\pi$ is finer than $\pi$'} or \emph{$\pi$ is a refinement of $\pi'$}. Put another way, for every $A \in \pi$, there is an $A' \in \pi'$ with $A \subseteq A'$. The join $\pi \join \pi'$ is the finest common coarsening of $\pi$ and $\pi'$ while $\pi \meet \pi'$ is the coarsest common refinement. Suppose $\pi = \{1  2 \vert 3 4 \vert 5\}$ and $\pi' = \{1 2 3 \vert 4 5\}$, then $\pi \meet \pi' = \{1 2 \vert 3 \vert 4 \vert5 \}$ and $\pi \join \pi' = \{1 2 3 4 5\}$. Partitions are ubiquitous in engineering, from teams of autonomous agents \cite{kantaros2016distributed} to hierarchial clustering algorithms \cite{carlsson2010characterization}. A \define{subpartition} is a partition of some subset $S'$ of the ground set $S$. By similar reasoning, subpartitions form a lattice denoted $\Subpart{S}$. Maps from the time domain to the lattice of subpartitions have been used to model dynamic graph with changing node sets \cite{kim2018formigrams}.
\end{example}

\begin{figure}[h]
	\centering
	\includegraphics[width = 0.5\textwidth]{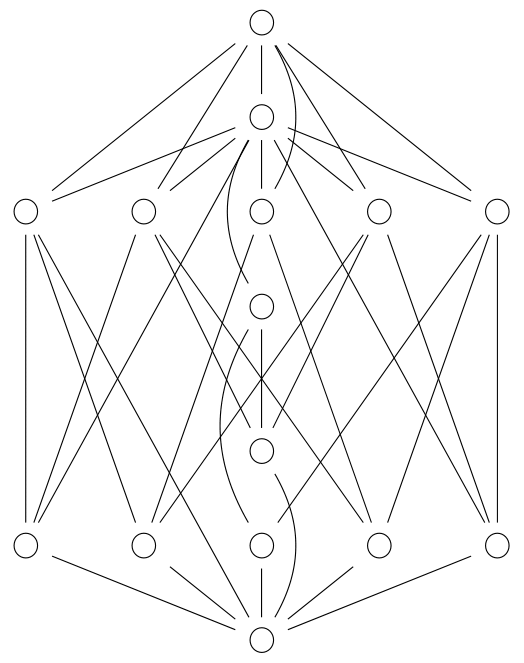}
	\caption{Lattice of partitions $\Part{X}$ on $X = \{1,2,3,4\}$}
	\label{fig:partition-lattice}
\end{figure}

\begin{example}[Information \cite{shannon1953lattice}] \label{eg:information-lattice}
	The following example is due to 
	% Suppose $(\Omega, \Sigma, \prob{-})$ is a probability space consisting of a sample space $\Omega$, $\sigma$-algebra $\Sigma$ and a probability measure $\prob{-}$.
	Recall, a random variable is a measurable map $X: \Omega \to A$ to some set $A$ which, for practical reasons, we call the \define{alphabet}. Depending on the context, we can think of $A$ as a set of symbols or as a set of words. The \define{entropy}
	\[
		H(X) = -\sum_{w \in A} \prob{X^{-1}\{w\}} \log \left( \prob{X^{-1}\{w\}}\right).
	\]
	Entropy rougly represents the amount of information (bits) required to encode the outcome of $X$. A variation of entropy, relative entropy, denoted $H(Y \vert X)$, is the amount of information to describe the outcome of another random variable $Y$ given that the value $X$ generated is already known. Omitting some technical details \cite{li2011connection}, random variables over a fixed probability distribution form a poset with $X \succeq Y$ if and only if $H(Y \vert X) = 0$.

	% It is interesting to point out that if $X$ and $Y$ are independent (i.e.~$H(Y \vert X) =  H(Y)$ and $H(X \vert Y)  = H(X)$), then $X \antichain Y$ as long as $H(X),~H(Y) \neq 0$. (If $H(X)=0$, then $X$ is deterministic.)

	It is know the set of $A$-valued random variables under the information order forms a lattice isomoprhic to $\Part{A}$ \cite{li2011connection}. The \define{information lattice} is a sublattice $\lattice{I}$ of $\Part{S}$ subordinate to the $\sigma$-algebra defining the underlying probability distribution. The join of random variables $X \join Y$ is the \define{joint information content} between then, exactly the joint random variable $(X,Y)$. The meet $X \meet Y$, on the other hand, is the \define{common information} between the two variables.
\end{example}

\subsection{Decomposition}

A recurring theme in algebra is to decompose structures into their smallest possible components. This leads to the following definition.
\begin{definition}
	Let $\lattice{L}$ be a lattice. An element $0 \neq z \in \lattice{L}$ is join-irreducible if $z = x \join y$ implies $x =z$ or $y = z$. Similarly, an element $1 \neq z \in \lattice{L}$ is meet-irreducible if $z = x \meet y$ implies $z = x$ or $z = y$. We denote the set of join-irreducible elements $\mathcal{J}(\lattice{L})$. 
\end{definition}

\noindent Clearly, $(\mathcal{J}(\lattice{L}), \preceq)$ is a poset with $x \preceq y$ if and only if $x \meet y = x$. Less clearly, downsets form distributive lattices.

\begin{theorem}[Birkhoff Duality] \label{thm:birkhoff}
	Suppose $\lattice{L}$ is a finite distributive lattice and $\poset{P}$ is a finite poset. Then,
	\[ \mathcal{D} \left( \mathcal{J}(\lattice{L}) ) \right) \cong \lattice{L}\]
	and
	\[ \mathcal{J} \left( \mathcal{D}(\poset{P}) ) \right) \cong \poset{P}. \]
\end{theorem}
\begin{proof}
	% First we see that $\mathcal{D}$ maps finite posets to finite distributive lattices and $\mathcal{J}$ maps finite distributive lattices to finite posets. Suppose $\poset{P}$ is a finite poset. If $D, D' \in \mathcal{D}(\poset{P})$, then $D \cap D'$ and $D \cup D$ are also down- (closed) sets. Distributivity follows from the distributive propoperty inherited from $\powerset{\poset{P}}$. Likewise, if $\lattice{L}$ is a finite distributive lattice, then $\mathcal{J}(\lattice{L})$ is a subset of $\lattice{L}$ that inherits the order in $\lattice{L}$. Finiteness in both cases is trivial.
	% Now, suppose $\lattice{L} = \mathcal{D}(\poset{P})$. We want to construct an isomorphism between $\poset{P}$ and $\mathcal{J}(\lattice{L})$. Suppose $E$ is join-irreducible in $\lattice{L}$; that is, $E$ is a downset such that if $E = D \cup D'$ for downsets $D$ and $D'$, then $E = D$ or $E = D'$. This implies $E$ covers at most one downset downsets, say $D \covers E$. Furthermore, $E \setminus D$ is a singleton, an element of $\poset{P}$. \todo{finish proof}
 See \cite[Theorem 10.1]{roman2008lattices}.
\end{proof}

Theorem \ref{thm:birkhoff} allows us to move back and forth between these, albeit restrictive, classes of posets and lattices with ease. Because $\mathcal{D}(\poset{Q})$ is a lattice of subsets (i.e.~sublattice of the powerset $\powerset{\poset{Q}}$), Theorem \ref{thm:birkhoff} can also be interpreted as \textit{every finite distributive lattice can be represented as a lattice of subsets}. In general, not restricted to finite posets and distributed lattices, the story is more complicated as illustrated by the following example.

\begin{example}[Natural Numbers]
	Consider the natural numbers $\N$ with the order $m \preceq n$ if and only if $n \vert m$, i.e. there is an $a \in \N$ with $n = a \cdot m$. $(\N, \preceq)$ is a lattice with $m \meet n = \mathrm{lcm}(m,n)$ and $m \join n = \mathrm{gcd}(m,n)$. First, notice that $\N$ is distributive by the classic distributive properties of $\mathrm{gcd}$ and $\mathrm{lcm}$. Hence, $(\N, \preceq)$ is an (infinite) distributive lattice. The join-irreducibles are the prime number $\mathcal{J}(\N)$ with the partial order inherited by $(\N, \preceq)$. It follows that $\mathcal{D}(\mathcal{J}(\N)) = \powerset{\primes}$. However, by the Fundamental Theorem of Arithmetic, $(\N, \preceq) \cong \multiset{\primes}$---every natural number has a unique prime factorization. Hence, $\mathcal{D} \left( \mathcal{J}(\lattice{L}) ) \right) \ncong \lattice{L}$. Because $\N$ is not finite, we could not apply Birkhoff Duality.
\end{example}

%% file: Chapters/Chapter03.tex
%**********************************************
\chapter{Galois Connections}\label{ch:galois}
%**********************************************

Suppose $\poset{P}$ and $\poset{Q}$ are posets. A Galois connection is a pair of compatible order-preserving maps, one from $\poset{P}$ to $\poset{Q}$, another $\poset{Q}$ to $\poset{P}$, such that arbitrary joins and arbitrary meets are preserved in $\poset{P}$ and $\poset{Q}$ respectively. On one hand, Galois connections are share a resemblance to linear adjoints (Definition \ref{def:galois-connection}), but on the other hand, being order approximations of an inverse (Proposition \ref{prop:left-right-left}), share a resemblance to Moore-Penrose pseudoinverses.

%----------------------------------
\section{Defining Properties}
%----------------------------------

\begin{definition}\label{def:galois-connection}
	Let $(\poset{P}, \preceq)$ and $(\poset{Q}, \preceq)$ be posets. A \define{Galois connection} is a pair $(\ladj{f}, \radj{f})$ of order-preserving maps $\ladj{f}: \poset{P} \to \poset{Q}$ and $\radj{f}: \poset{Q} \to \poset{P}$ such that for all $x \in \poset{P}$, $y \in \poset{Q}$,
	\begin{align}
		\ladj{f}(x) \preceq y &\Leftrightarrow& x \preceq \radj{f}(y).
	\end{align}
\end{definition}

\noindent We call $\ladj{f}$ a \define{lower-adjoint} and $\radj{f}$ an \define{upper-adjoint}.

\begin{remark}
The definition of a Galois connection is reminiscent of the definition of the adjoint of a linear transformation. If $T: V \to W$ is a map of inner-product spaces $(V, \langle~,~\rangle_{V})$, $(W, \langle~,~\rangle_{W})$, then $T^{\ast}: W \to V$ is the map defined by the following property
\begin{align}
\langle T(\mathbf{x}), \mathbf{y} \rangle_{W} &=& \langle \mathbf{x}, T^{\ast}(\mathbf{y}) \rangle_{V}.
\end{align}
\end{remark}

Given fixed complete lattices $\lattice{K}$ and $\lattice{L}$, we denote the set of Galois connections between these lattices $\cat{Ltc}(\lattice{K}, \lattice{L})$ ( notation adapted from \cite{grandis2013homological}). The set $\cat{Ltc}(\lattice{K},\lattice{L})$ is a poset with the order $f \leq g$ if and only if $\ladj{f}(x) \preceq \ladj{g}(x)$ and $\radj{g}(y) \preceq \radj{f}(y)$ for all $x, y \in \lattice{L}$ endowing it with the structure of a complete lattice. Galois connections have an equivalent definition highlighting the property that they are the ``best approximation'' of an order inverse.

\begin{proposition} \label{prop:monad-comonad}
	Suppose $\poset{P}$ and $\poset{Q}$ are posets. The following are equivalent.
	\leavevmode
	\begin{enumerate}
		\item $(\ladj{f}, \radj{f}): \poset{P} \to \poset{Q}$ is a Galois connection.
		\item $\ladj{f}: \poset{P} \to \poset{Q}$ and $\radj{f}: \poset{Q} \to \poset{P}$ are order-preserving, and $\radj{f} \ladj{f} \succeq \id_{\poset{P}}$ and $\ladj{f}\radj{f} \preceq \id_{\poset{Q}}$.
	\end{enumerate}
\end{proposition}
\begin{proof}
	Suppose $\ladj{f}(x) \preceq y$ if and only if $x \preceq \radj{f}(y)$ for all $x \in \poset{P},~y \in \poset{Q}$. Substitute $\ladj{f}(x)$ for $y$. Then, $x \preceq \radj{f}\ladj{f}(x)$. Alternatively, if we substitute $\radj{f}(y)$ for $x$, then we obtain $\ladj{f}\radj{f}(y) \preceq y$. 
	Conversely, suppose $\radj{f} \circ \ladj{f} \succeq \id_{\poset{P}}$. If $\ladj{f}(x) \preceq y$, then, because $\radj{f}$ is order-preserving, we have $\radj{f}(\ladj{f}(x)) \preceq \radj{f}(y)$. Because $\radj{f}(\ladj{f}(x)) \succeq x$, it follows by transitivity $\radj{f}(y) \succeq x$. Similarly, if $x \preceq \radj{f}(y)$, then $\ladj{f}(\radj{y}) \preceq y$ implies $\ladj{f}(x) \preceq y$.
\end{proof}

% \noindent The proof of Proposition \ref{prop:monad-comonad} suggests the following.

% \begin{corollary}
% 	Suppose $\radj{f} \ladj{f}(x) \succeq x$ for all $x$. Then,
% 	\begin{align*}
% 		\ladj{f}(x) \preceq y &\Rightarrow& x \preceq \radj{f}(y).
% 	\end{align*}
% 	Suppose $\ladj{f} \radj{f}(y) \preceq y$ for all $y$. Then,
% 	\begin{align*}
% 		x \preceq \radj{f}(y) &\Rightarrow& \ladj{f}(x) \preceq y.
% 	\end{align*}
% \end{corollary}

\noindent Composing lower and upper adjoints stabilizes.

\begin{proposition}\label{prop:left-right-left}
	Suppose $(\ladj{f}, \radj{f})$ is a Galois connection as above. Then,
	\begin{align*}
		\ladj{f}\radj{f}\ladj{f} &=& \ladj{f}; \\
		\radj{f} \ladj{f} \radj{f} &=& \radj{f}.
	\end{align*}
\end{proposition}
\begin{proof}
	For the first, it suffices to show that $\ladj{f}\radj{f}\ladj{f}(x) \succeq \ladj{f}(x)$ for all $x \in \poset{P}$ because $\ladj{f}\radj{f}\ladj{f} \preceq \ladj{f}$ follows from Proposition \ref{prop:monad-comonad}. Also, by Propositon \ref{prop:monad-comonad}, $\radj{f} \ladj{f}(x) \succeq x$. $\ladj{f}$ is order-preserving implies $\ladj{f} \radj{f} \ladj{f} (x) \succeq \ladj{f}(x)$. The second is dual to the first.
\end{proof}

\noindent One consequence of the Proposition \ref{prop:left-right-left} is that $\radj{f}\ladj{f}$ and $\ladj{f} \radj{f}$ are idempotent.

\begin{corollary}\label{cor:idempotence}
	Suppose $(\ladj{f}, \radj{f})$ is a Galois connection as above. Then,
	\begin{align*}
		\radj{f} \ladj{f} \radj{f} \ladj{f} &=& \radj{f} \ladj{f} \\
		\ladj{f} \radj{f} \ladj{f} \radj{f} &=& \ladj{f} \radj{f}
	\end{align*}
\end{corollary}

\noindent Another consequence of Proposition \ref{prop:left-right-left} characterizes Galois connections in which $\ladj{f}$ is an order-embedding.

\begin{proposition} \label{prop:mono-epi}
	Suppose $(\ladj{f}, \radj{f}): \poset{P} \to \poset{Q}$ is a Galois connection and $\ladj{f}$ is an order-embedding. Then, $\radj{f}\ladj{f} = \id_{\poset{P}}$.
	% Conversely, suppose $\ladj{f}$ is surjective. Then, $\ladj{f} \radj{f} = \id_{\poset{Q}}$. Furthermore, if $\ladj{f}$ is an order isomorphism, then 
	% \begin{align*}
	% 	\radj{f}\ladj{f} &=& \id_{\poset{P}}, \\
	% 	\ladj{f}\radj{f} &=& \id_{\poset{Q}}.
	% \end{align*}
\end{proposition}

\noindent We need a lemma.
\begin{lemma}\label{lem:mono-epi}
	Suppose $(\ladj{f}, \radj{f})$ is a Galois connection. Then, $\ladj{f}$ is an order-embedding if and only if $\radj{f} \ladj{f}$ is an order-embedding.
\end{lemma}
\begin{proof}[Proof of Lemma \ref{lem:mono-epi}]
	Suppose $\ladj{f}$ is an order-embedding. Then, by Proposition \ref{prop:left-right-left}, $\ladj{f} \radj{f} \ladj{f}(x) \preceq \ladj{f} \radj{f} \ladj{f}(y)$ if and only if $x \preceq y$. This implies $x \preceq y$ if and only if
	\begin{align*}
		\radj{f}\ladj{f}(x) &\preceq& \radj{f}\ladj{f}\radj{f}\ladj{f}(y) \\
				&\preceq& \radj{f} \ladj{f}(y)
	\end{align*}
	by Corollary \ref{cor:idempotence}.
\end{proof}

\begin{proof}[Proof of Proposition \ref{prop:mono-epi}]
	By Proposition \ref{prop:monad-comonad}, it suffices to show $\radj{f} \ladj{f}(x) \preceq x$ for all $x \in \poset{P}$. By Lemma \ref{lem:mono-epi}, $\radj{f} \ladj{f}(x) \preceq x$ if and only if
	\begin{align*}
		\radj{f} \ladj{f}\left(\radj{f} \ladj{f}(x) \right) &\preceq& \radj{f} \ladj{f}(x).
	\end{align*}
	By Corollary \ref{cor:idempotence}, $\radj{f} \ladj{f}(x) \preceq x$ holds for all $x \in \poset{X}$.
\end{proof}

\noindent A Galois connection gives rise to a closure and coclosure operator.

\begin{corollary}
	$\radj{f}\ladj{f}: \poset{P} \to \poset{P}$ is a closure operator. $\ladj{f}\radj{f}: \poset{Q} \to \poset{Q}$ is a co-closure operator.
\end{corollary}
\begin{proof}
	Proposition \ref{prop:monad-comonad} and Corollary \ref{cor:idempotence}
\end{proof}

\noindent We denote $\fixed(\radj{f}\ladj{f})$ by $\mathbf{Cl}(\lattice{K})$ and $\fixed(\ladj{f}\radj{f})$ by $\mathbf{Cl}(\lattice{L})$.

\begin{definition}\label{def:galois-lattice}
Given a connection $(\ladj{f}, \radj{f}): \lattice{K} \to \lattice{L}$, define a poset
\begin{align*}
	\Gal(f) &=& \{ (x, y)~\vert~ \ladj{f}(x) = y,~x = \radj{f}(y) \} \subseteq \lattice{K} \times \lattice{L}
\end{align*}
called the \define{Galois lattice} of $f$.
\end{definition}

\noindent $\Gal(f)$ is in fact a complete lattice.

\begin{proposition}
Suppose $f: \lattice{K} \to \lattice{L}$ is a Galois connection between complete lattices $\lattice{K}$ and $\lattice{L}$. Then, the following are isomorphic complete lattices: $\mathbf{Cl}(\lattice{K})$, $\mathbf{Cl}(\lattice{L})$, $\Gal(f)$.
\end{proposition}
\begin{proof}
	For the lattice structure, Theorem \ref{thm:closure-fixed}. The isomorphism follows from Proposition \ref{prop:left-right-left}.
\end{proof}

If $f = (\ladj{f}, \radj{f})$ is a Galois connection, $\ladj{f}$ is join-preserving and $\radj{j}$ is meet-preserving. Moreover, if $\ladj{f}$ is a join-preserving map, $\radj{f}$ is uniquely defined. Conversely, if $\radj{f}$ is a join-preserving map, then $\ladj{f}$ is uniquely defined.

\begin{theorem}[Adjoint Functor Theorem\footnote{This result is a special case of the Adjoint Functor Theorem due to Freyd \cite{freyd1964abelian}.}] \label{thm:adjoint-functor-theorem}

Suppose $\poset{P}$ and $\poset{Q}$ are posets and  $(\ladj{f}, \radj{f}): \poset{P} \to \poset{Q}$ is a Galois connection. Then, $\ladj{f}$ preserves joins and $\radj{f}$ preserves preserves meets (whenever they exist). Conversely, if $f: \poset{P} \to \poset{Q}$ is a join-preserving, there there exists a map $f^{\ast}: \poset{Q} \to \poset{P}$ given by
\begin{align}
	f^{\ast}(y) &=& \bigjoin \{x~\vert~f(x) \preceq y \} \label{eq:left-to-right}
\end{align}
such that $(f, f^{\ast})$ is a Galois connection.
\end{theorem}

\noindent It is sometimes convenient to rewrite the formula for $f^{\ast}$ as
\begin{align*}
	f^{\ast}(y) &=& \bigjoin f^{-1}(\downset{y})
\end{align*}
\begin{proof}[Proof of Theorem \ref{thm:adjoint-functor-theorem}]
In the first part, suppose $\ladj{f}$ is  lower adjoint. We show $\ladj{f}$ preserves joins. Suppose $\{x_i\}_i \subseteq \poset{P}$ is an arbitrary subset of $\poset{P}$. Then,
\begin{align*}
	\ladj{f}\left( \bigjoin_i x_i \right) \preceq y \quad &\Leftrightarrow& 	\bigjoin_i x_i \preceq \radj{f}(y)
\end{align*}
{}Set \[y = \bigjoin_i f(x_i).\]
Then, 
\begin{align*}
	\ladj{f}\left( \bigjoin_i x_i \right) \preceq \bigjoin_i f(x_i) \quad &\Leftrightarrow& 	\bigjoin_i x_i \preceq \radj{f}(\bigjoin_i f(x_i))
\end{align*}
By Proposition \ref{lem:inside-meets}, $\bigjoin_i x_i \preceq \radj{f}(\bigjoin_i f(x_i))$ holds tautologically. Hence, $\ladj{f}\left( \bigjoin_i x_i \right) \preceq \bigjoin_i f(x_i)$ holds. By Proposition \ref{lem:inside-meets} again, \[\ladj{f}(\bigjoin_i f(x_i)) \succeq \bigjoin_i f(x_i)\], leading to the equality desired.

In the second part, we show \eqref{eq:left-to-right} is the upper adjoint of $f$. By Proposition \ref{prop:monad-comonad}, it suffices to show $f^{\ast} f (x) \succeq x$ and $f f^{\ast} (y) \preceq y$.
We have
\[
  f^{\ast} f (x) = \bigvee f^{-1}\left(\downarrow f(x) \right).
\]
But $\{x\} \subseteq f^{-1} f (x) \subseteq  f^{-1}\left(\downarrow f(x) \right)$ implies
\[
  f^{\ast} f (x) = \bigvee {f}^{-1}\left(\downarrow f(x) \right) \succeq \bigvee {f}^{-1} f (x)  \succeq \bigvee \{ x\} = x.
\]
For the other inequality, we have
\[
  f f^{\ast} (y) 
  = f \left( \bigvee f^{-1}(\downarrow y) \right) \\
  = \bigvee f f^{-1} (\downarrow y)
\]
since, by the first part, $f$ preserves joins. Notice, $ f f^{-1} (\downarrow y) \subseteq~\downarrow y$. Hence,
\[
  f f^{\ast} (y) = \bigvee f f^{-1} (\downarrow y) \preceq \bigvee \downarrow y = y.
\]
\end{proof}

\noindent The dual result is the following.

\begin{corollary}
	If $f: \poset{Q} \to \poset{P}$ is meet-preserving, there exists a join-preserving map $f_{\ast}: \poset{P} \to \poset{Q}$ such that $(f_{\ast}, f)$ is a Galois connection, defined
	\begin{align*}
		f_{\ast}(x) &=& \bigmeet \{y~\vert~f(y) \succeq x \}.
	\end{align*}
\end{corollary}
%----------------------------------------------
\section{Galois Connections from Relations}
%----------------------------------------------
Some authors define Galois connections to be adjoint order-reversing maps, while others---as do we---call adjoint pairs of order-preserving maps Galois connections. We define a \define{contravariant Galois connection} as a pair $(\galup{f}, \galdown{g})$ of order-reversing maps
\begin{align*}
\galup{f}: \lattice{K} \longleftrightarrow \lattice{L}: \galdown{f}
\end{align*}
such that
\begin{align*}
	\galdown{f} \galup{f}(x) \succeq x &:& \galup{f} \galdown{f}(y) \succeq y.
\end{align*}
or, equivalently,
\begin{align*}
	\galup{f}(x) \succeq y &\Leftrightarrow& \galdown{f}(y) \succeq x
\end{align*}
Whenever there is any ambiguity, we call an order-preserving Galois connection \define{covariant}.
There is an isomorphism
\begin{align*}
	\cat{Ltc} \left( \op{\lattice{K}}, \lattice{L} \right) &\cong& \op{\cat{Ltc}\left(\op{\lattice{K}, \lattice{L}}\right)},	
\end{align*}
and we may view $(\galup{f}, \galdown{f})$ as an ordinary (order-preserving) Galois connection in two ways.

\textit{Caveat lector}. Although there is a one-to-one correspondence between the \textit{sets} of join-reversing maps and meet-reversing maps, the correspondence is not an isomorphism of lattices.

There are canonical ways to construct Galois connections of each type (covariant, contravariant) from a binary relation $\rel{R} \subset X \times Y$. The covariant \define{Galois correspondence} is originally due to  Schmidt \cite{schmidt1953beitrage}, while the contravariant Galois correspondence is due to Ore \cite{ore1944galois}.

\noindent In Chapter \ref{ch:semantics}, we encounter a Galois connection constructed in the above manner.

\begin{theorem}\label{thm:cov-galois}
	Suppose $\rel{R} \in \Rel(X,Y)$ is a binary relation. Then, $\rel{R}$ induces a covariant Galois connection
	\[\begin{tikzcd}
	{\wp(X)} && {\wp(Y)}
	\arrow["\bot", draw=none,from=1-1, to=1-3]
	\arrow["\rel{R}_{\exists}", shift left=1, curve={height=-18pt}, from=1-1, to=1-3]
	\arrow["\rel{R}_{\forall}", shift left=1, curve={height=-18pt}, from=1-3, to=1-1]
\end{tikzcd}\]
with
\begin{align}
	\rel{R}_{\exists}(U) &=& \{y \in Y~\vert~(\exists x) \quad x \in U~\meet~x~\rel{R}~y \}, \\
	\rel{R}_{\forall}(V) &=& \{ x \in X~\vert~(\forall y) \quad x\rel{R}y \Rightarrow y \in V \},
\end{align}
$U \subseteq X, V \subseteq Y$.
\end{theorem}
\begin{proof}
    First, we show that  $\rel{R}_{\exists}$ is join-preserving so we can apply Proposition \ref{thm:adjoint-functor-theorem} in order to compute $\rel{R}_\forall$. Let $\{U_i\}_{i \in I}$ be a family of subsets of $X$. Then,
    \begin{align*}
        \rel{R}_\exists \left( \bigcup_{i \in I} U_i \right) &=& \\
        &=& \{ y \in Y~\vert~(\exists x)x \in \bigcup_{i \in I} U_i~\meet~x \rel{R} y \} \\
        &=& \{ y \in Y~\vert~(\exists x~\exists i \in I) \quad x \in U_i~\meet~x \rel{R} y \} \\
        &=& \bigcup_{i \in I} \{ y \in Y~\vert~(\exists x) \quad x \in U_i~\meet~x \rel{R} y \} \\
        &=& \bigcup_{i \in I} \rel{R}_\exists(U_i).                
    \end{align*}
    Now, we can apply the second part of Proposition \ref{thm:adjoint-functor-theorem} to compute $\rel{R}_{\forall}$. Let $V \subseteq Y$. Then,
    \begin{align*}
        \rel{R}_{\forall}(V) &=& \bigcup \{ U \subseteq X~\vert~ \rel{R}_{\exists}(U) \subseteq V\} \\
        &=& \bigcup \{ U \subseteq X~\vert~ \{y \in Y~\vert~(\exists x) \quad x \in U~\meet~x~\rel{R}~y \} \subseteq V\} \\
        &=& \{x \in X~\vert~(\forall y) \quad x~\rel{R}~y \Rightarrow y \in V  \}
    \end{align*}
    as desired.
\end{proof}

\noindent The following construction has lead to the field of formal concept analysis (FCA) \cite{wille1982restructuring}.

\begin{theorem}\label{thm:contra-galois}
	Suppose $\rel{R} \in \Rel(X,Y)$ is a binary relation. Then, $\rel{R}$ induces a contravariant Galois connection
	\[\begin{tikzcd}
	{{\wp(X)}} && {\op{\wp(Y)}}{}
	\arrow["\bot", draw=none,from=1-1, to=1-3]
	\arrow["\galup{\rel{R}}", shift left=1, curve={height=-18pt}, from=1-1, to=1-3]
	\arrow["\galdown{\rel{R}}", shift left=1, curve={height=-18pt}, from=1-3, to=1-1]
\end{tikzcd}\]
with
\begin{align}
	\galup{\rel{R}}(U) &=& \{y \in Y~\vert~(\forall x) \left(x \in U \Rightarrow x~\rel{R}~y\right) \} \\
	\galdown{\rel{R}}(V) &=& \{ x \in X~\vert~(\forall y) \left(y \in V \Rightarrow x~\rel{R}~y \right) \},
\end{align}
$U \subseteq X, V \subseteq Y$.
\end{theorem}
\begin{proof}
    We show $\galup{\rel{R}}(U) \supseteq V$ if and only if $\galdown{\rel{R}}(V)\supseteq U$ for all $U \subseteq X, V \subseteq Y$ since in $\op{\powerset{X}}$, $V \preceq V'$ if and only if $V \supseteq U'$. Observe,
    \begin{align*}
        \galdown{\rel{R}^{~}} = \galup{\rel{R}^{\dagger}}.
    \end{align*}
    Hence, by simply replacing the role of $X$ and $Y$, we are done.
\end{proof}

% \noindent As promised, there is a prototypical example of a contravariant Galois connection fathomed in this manner.

% \begin{example}[Concept Lattice]
% 	Suppose $X$ is a set of objects and $Y$ is a set of attributes.

% 	$\rel{R} \subseteq X \times Y$
% \end{example}

%-----------------------------
\section{Integral Transforms}
%-----------------------------

Recall the residual operation $[-.-]: \resid{L} \times \resid{L} \to \resid{L}$ of a residuated lattice has the defining property
\begin{align*}
	x \star y \preceq z \quad &\Leftrightarrow& \quad x \preceq [y, z].
\end{align*}
Another way of stating this property is there is a Galois connection
\begin{align*}
	(- \star y, [y, -]): \resid{L} \to \resid{L}.
\end{align*}
We have the following lemma as consequence.
\begin{lemma}\label{lem:resid}
	Suppose $(\resid{L},\meet, \join, 0, 1, \star, e, [-,-])$ is a complete residuated lattice. Then,
	\begin{enumerate}
		\item $\left( \bigjoin_{i} x_i \right) \star y = \bigjoin_i (x_i \star y)$,
		\item $\left[x, \bigmeet_{i} y_i \right] = \bigmeet_{i} \left[x, y_i\right]$.
	\end{enumerate}
\end{lemma}
\begin{proof}
	Apply Theorem \ref{thm:adjoint-functor-theorem}.
\end{proof}

Suppose $\resid{L}$ is a residuated lattice, and suppose $X$ is a set, possibly with additional structure. A \define{kernel} is a map
\begin{align*}
	H: X \times X \to \resid{L}
\end{align*}
Integration over kernels is readily defined as follows.\footnote{We provide a definition slightly more general than the one already introduced by Maragos \cite{maragos2009morphological}, but less general than an analogous categorical one \cite{gutierrez2010fuzzy}.} Suppose $f$ is a map $f: X \to \resid{L}$ and $H: X \times X \to \resid{L}$ is a kernel. Then, the \define{integral transform} of $f$ is the map $\hat{f}: X \to \resid{L}$ defined
\begin{align}
	\hat{f}(x) &=& \bigjoin_{y \in X} H(x,y) \star f(y).
\end{align}

\begin{theorem} \label{thm:kernels}
	Suppose $\resid{L}$ is a residuated lattice and $H: X \times X \to \resid{L}$ is a kernel. Then, the following forms a Galois connection
	\[\begin{tikzcd}
	{\resid{L}^X} && {\resid{L}^X}{}
	\arrow["\bot", draw=none,from=1-1, to=1-3]
	\arrow["\hat{(-)}", shift left=1, curve={height=-18pt}, from=1-1, to=1-3]
	\arrow["\check{(-)}", shift left=1, curve={height=-18pt}, from=1-3, to=1-1]
\end{tikzcd}\]
\begin{align*}
	\left( f: X \to \resid{L} \right) &\mapsto& \hat{f}(x) = \bigjoin_{y \in X} H(x,y) \star f(y) \\
	\check{g}(y) = \bigmeet_{x \in X} \left[ H(y,x), g(y)	\right]	&\mapsfrom& \left( g: X \to \resid{L} \right)
\end{align*}
\end{theorem}
\begin{proof}
	Our strategy is to show $\hat{(-)}$ is join-preserving and apply Theorem \ref{thm:adjoint-functor-theorem} to calculate $\check{(-)}$. Suppose $f = \bigjoin_{i \in I} f_i$ where $f_i \in \resid{L}^X$ are $\resid{L}$-valued functions on $\resid{L}$ with the order $f_i \preceq g_i$ if and only if $f_i(x) \preceq g_i(x)$ for all $x \in X$. Then,
	\begin{align*}
		\hat{f}(x) &=& \bigjoin_{y \in X} H(x,y) \star \left( \bigjoin f_i(y) \right) \\
		&=& \bigjoin_{y \in X} \bigjoin_{i \in I} H(x,y) \star f_i(y) \\
		&=& \bigjoin_{i \in I} \hat{f}_i(x)
	\end{align*}
	by Lemma \ref{lem:resid}.

	To calculate the right adjoint, by Theorem \ref{thm:adjoint-functor-theorem},
	\begin{align*}
		\check{g} &=& \bigjoin \{ f~\vert~\hat{f} \preceq g \},
	\end{align*}
		which implies $f$ is the greatest function whose integral transform is less than $g$
	\begin{align*}
		\check{g}(y) &=& \bigjoin_{f \in \resid{L}^X} \{ f(y)~\vert~ \hat{f}(y) \preceq g(y) \quad \forall x \}.
	\end{align*}
Hence, for all $y \in X$, $f$ satisfies
	\begin{align*}
		\bigjoin_{x \in X} H(y,x) \ast f(x) &\preceq& g(y).
	\end{align*}
	Applying Lemma \ref{lem:resid},
	\begin{align*}
		H(y,x) \ast f(x) \preceq g(y) \quad &\Leftrightarrow& \quad f(x) \preceq \left[ H(y,x), g(y) \right].
	\end{align*}
	Therefore, tightening the bound,
	\begin{align*}
		\check{g}(y) &=& \bigmeet_{x \in X} \left[ H(y,x), g(y) \right].
	\end{align*}
	\end{proof}

\begin{example}[Max-Plus Connection] \label{eg:max-plus-galois}
	Suppose $\resid{L}$ is the residuated lattice $(\Rext, \min, \max, -\infty, \infty, 0, +)$ with
	$[x,y] = y-x$ if $x,y \in \Rext \setminus \{-\infty, \infty\}$. Then, Theorem \ref{thm:kernels} implies
	\begin{align*}
		(A \ovee \mathbf{x})_i &=& \bigjoin_{i = 1}^n a_{ij} + x_j
	\end{align*}
	has an upper adjoint defined by $A^\dagger \owedge \mathbf{x}$ where $A^{\dagger}$ is the matrix with $a^\dagger_{ij} = -a_{ji}$. Hence,
	\begin{align*}
	(A^{\dagger} \owedge \mathbf{x})_j &=& \bigmeet_{i=1}^n a_{ji} - x_j.
	\end{align*}
	Consequently, by Proposition \ref{prop:monad-comonad},
	\begin{align*}
		A^\dagger \owedge \left( A \ovee \mathbf{x} \right) &\geq& \mathbf{x}, \\
		A \ovee \left( A^\dagger \owedge \mathbf{x} \right) &\leq& \mathbf{x},
			\end{align*}
	and
	\begin{align*}
		A \ovee \mathbf{x} \leq \mathbf{y} \quad &\text{if and only if}& \quad \mathbf{x} \leq A^{\dagger} \owedge \mathbf{y}
	\end{align*}
	for all $\mathbf{x}, \mathbf{y} \in \Rext^n$.
\end{example}

%-------------------------------
\section{Category Theory}
\label{sec:cats}
%-------------------------------

We do no assume prior knowledge of category theory. However, as is discussed elsewhere \cite{fong2019invitation}, there are irrefutable parallels between order theory and category theory. Reoccurring constructions pertaining to sheaves and cosheaves are further illuminated if we view them as categorical limits and colimits \cite{riehl2017category} which, here, we present in a more friendly form- as equalizers and coequalizers.

A category $\cat{C}$ consists of a collection of \define{objects} $\Ob{\cat{C}}$ and a collection of \define{morphisms} also called \define{arrows} $\Mor{\cat{C}}$ 
such that for any two objects $x, y \in \Ob{\cat{C}}$, there is a collection of morphisms $\Hom_{\cat{C}}(x,y) \subseteq \Mor{\cat{C}}$ between them with a map
\[ \circ_{x,y,z}: \Hom_{\cat{C}}(x, y) \times \Hom_{\cat{C}}(y,z ) \to \Hom_{\cat{C}} (x, z)\]
satisfying\footnote{Usually, the subscripts in $\circ_{x,y,z}$ are omitted.}
\begin{itemize}
	\item $(h \circ g) \circ f = h \circ (g \circ f)$ (Associativity).
	\item For every $x \in \Ob{C}$, there is a $\id_x \in \Hom_{\cat{C}}(x,x)$ with $f \circ \id = f$ and $\id \circ f = f$ for every composable $f \in \Mor{\cat{C}}$ (Identity).
\end{itemize}

\noindent Colloquially, \textit{if you tell me the objects and arrows, I can tell you the category.}

\begin{examples}
 	\leavevmode
 	\begin{itemize}
 		\item Suppose $(\poset{P}, \preceq)$ is a preorder. Then, we may define a category $\cat{P}$ with $\Ob{\cat{P}} = \poset{P}$ and $\Hom_{\cat{P}}(x, y) = \ast$ (a one-point set) whenever $x \preceq y$ in $\poset{P}$ and $\cat{P}(x,y) = \emptyset$ otherwise.
 		\item Suppose $(\monoid{M}, \star, e)$ is a \define{monoid}. Then, $\monoid{M}$ is a category $\cat{M}$ with $\Ob{\cat{M}} = \ast$, a single object, and $\Hom_{\cat{M}}(\ast, \ast) = \monoid{M}$. Composition of arrows is defined according to the binary operation $\star$.
 	\end{itemize}  
\end{examples}

%------------------------------------
\subsection{Functors}
%------------------------------------

Categories are useful often because we can compare them, a perspective that has lead to connections between different branches of mathematics. Comparing two categories amounts to demonstrating a mapping between them called a functor. A \define{functor} $F$ between categories $\cat{C}$ and $\cat{D}$ consists of a maping $F: \Ob{\cat{C}} \to \Ob{\cat{D}}$ and a map $F_{x,y}: \Hom_{\cat{C}}(x,y) \to \Hom_{\cat{D}}(Fx, Fy)$ such that for every $f \in \Hom_{\cat{C}}(x,y)$ and $g \in \Hom_{\cat{C}}(y, z)$,
\begin{align*}
	F_{x,z}(g \circ f) &=& F_{y,z}(g) \circ F_{x,y}(f), \\
	F(\id_x) &=& \id_{F x}.
\end{align*}

We have already seen functors. For instance, a functor between poset categories $F: \cat{P} \to \cat{Q}$ is precisely an order-preserving map; a functor between monoid categories $F: \cat{M} \to \cat{N}$ is precisely a \define{monoid homomorphism}. Functors also form a category in their own right. If $\cat{D}$ and $\cat{D}$ are categories, then the $\cat{D}^{\cat{C}}$ is a category whose objects are functors $F: \cat{C} \to \cat{D}$ whose morphisms are maps between functors called natural transformations. If $F, G: \cat{C} \to \cat{D}$, then a \define{natural transformation} $\eta: F \Rightarrow G$ is, roughtly, a set of morphisms $\{ \eta_x: Fx \to Gx \}_{x \in \Ob{\cat{C}}}$ such that the following diagram commutes for every $f \in \Hom_{\cat{C}}(x,y)$ 
% https://q.uiver.app/?q=WzAsNCxbMCwwLCJGeCJdLFsxLDAsIkZ5Il0sWzAsMSwiR3giXSxbMSwxLCJHeSJdLFswLDEsIkZmIl0sWzAsMiwiXFxldGFfeCIsMl0sWzEsMywiXFxldGFfeSIsMl0sWzIsMywiR2YiLDJdXQ==
\[\begin{tikzcd}
	Fx & Fy \\
	Gx & Gy
	\arrow["Ff", from=1-1, to=1-2]
	\arrow["{\eta_x}"', from=1-1, to=2-1]
	\arrow["{\eta_y}"', from=1-2, to=2-2]
	\arrow["Gf"', from=2-1, to=2-2]
\end{tikzcd}\]
As expected, the the morphisms $\Hom_{\cat{D}^{\cat{C}}}(F, G)$ are, roughtly, the set of natural transformations between $F$ and $G$.

%---------------------------
\subsection{(Co)equalizers}
%---------------------------

We make use of exactly one technical construction from category theory called a \define{(co)equalizer}. Suppose $\cat{C}$ is a data category and $f,g: A \to B$ are parallel morphisms from between objects $A, B \in \Ob{\cat{C}}$ (e.g.~linear transformations, join-preserving maps). By ``a diagram commutes,'' we mean that if you compose a sequence of arrows in the diagram to form a map from a source object to a target object, you obtain the same map as if you had composed any other sequence of arrows to go from the same source object to the same target object. 
\begin{definition}[Equalizer]
	The \define{equalizer} of a parallel pair $f,g: A \to B$ is an object $E$ and a map $i: E \to A$ such that for any map $j: D \to A$ there is a unique map $j': D \to E$ such that the following diagram commutes
	\[
	\begin{tikzcd}
		E \arrow[r, "i"] & A \arrow[r, "f", shift left] \arrow[r, "g", shift right, swap] & B \\
						 & D \arrow[u, "j"] \arrow[lu, "j'", dashed]
	\end{tikzcd}.
	\]
\end{definition}

\begin{definition}[Coequalizer]
	The \define{coequalizer} of a parallel pair $f,g: A \to B$ is an object $Q$ and a map $p: B \to Q$ such that for any map $q: B \to P$ there is a unique map $q': Q \to P$ such that the following diagram commutes
	\[\begin{tikzcd}
																   & P 								 &					\\
	A \arrow[r, "f", shift left] \arrow[r, "g", shift right, swap] & B \arrow[r, "p"] \arrow[u, "q"] & Q \arrow[ul, "q'", dashed, swap]
	\end{tikzcd}\]
\end{definition}

%--------------------------
\subsection{(Co)products}
%--------------------------

Everyone is familiar with products of sets (cartesian products) or products of vector spaces (direct sum). These constructions often generalize in order categories.

\begin{definition}[Product]
	Suppose $A, B \in \Ob{\cat{C}}$. The \define{product} of $A$ and $B$ is an object denoted $A \prod B$ with maps $p: A \prod B \to A$ and $q: A \prod B \to B$ called \define{projection maps} such that for an object $D$ and maps $p': D \to A$ and $q': D \to B$ there is a unique map $i: D \to A \prod B$ such that the diagram commutes
	\[
	\begin{tikzcd}
	A 		&																							&			B \\
			&	A \prod B \arrow[ul,"p",swap] \arrow[ur,"q"]											&			  \\
			&	D 	\arrow[uul,"p'", bend left] \arrow[uur,"q'", swap, bend right] \arrow[u,"i",dashed]	&
	\end{tikzcd}.
	\]
\end{definition}

\begin{definition}[Coproduct]
Suppose $A, B \in \Ob{\cat{C}}$. The \define{coproduct} of $A$ and $B$ is an object denoted $A \coprod B$ with maps $i: A \to A \coprod B$ and $j: B \to A \coprod B$ called \define{injection maps} such that for an object $D$ and maps $i': A \to D$ and $j': B \to D$ there is a unique map $p: A \coprod B \to D$ such that the diagram commutes
\[
\begin{tikzcd}
												&	D \\
												&	A \coprod B \arrow[u, "p", dashed] & \\
	A \arrow[ur, "i", swap] \arrow[uur, "i'", bend left]	&										& B \arrow[ul, "j"] \arrow[uul, "j'", swap, bend right] 
\end{tikzcd}.
\]
\end{definition}
\begin{example}[Meets \& Joins]
Recall that a poset $\poset{P}$ is a category with an arrow $x \to y$ between two elements (objects) if and only if $x \preceq y$. The product of two elements (if it exists) is the element $x \meet y$ defined by the universal property
	\[
	\begin{tikzcd}
	x 		&																							&			y \\
			&	x \meet y \arrow[ul,"\preceq",swap] \arrow[ur,"\preceq"]											&			  \\
			&	z 	\arrow[uul,"\preceq", bend left] \arrow[uur,"\preceq", swap, bend right] \arrow[u,"\preceq",dashed]	&
	\end{tikzcd}.
	\]
	Similarly, the coproduct of $x$ and $y$ is the element $x \join y$ defined by the universal property
	\[
	\begin{tikzcd}
												&	z \\
												&	x \join y \arrow[u, "\preceq", dashed] & \\
	x \arrow[ur, "\preceq", swap] \arrow[uur, "\preceq", bend left]	&										& y \arrow[ul, "\preceq"] \arrow[uul, "\preceq", swap, bend right] 
	\end{tikzcd}.
	\]
	Products and coproducts of arbitrary arity, if they exist, are similarly defined. For the case of meets, we prove the following universal property from first principles. 
 \begin{lemma} \label{lem:convenience}
	Let $\lattice{L}$ be a complete lattice, $I$ an arbirary indexing set, $\{x_i\}_{i \in I} \subseteq \lattice{L}$, and $y \in \lattice{L}$. Then,
	\begin{align*}
		y &\preceq& \bigmeet_{i \in I} x_i
	\end{align*}
	if and only if
	\begin{align*}
		y &\preceq& x_i
	\end{align*}
 for all $i \in I$.
\end{lemma}
\begin{proof}
	Suppose $y$ is a lower bound of the subset $\{x_i\}_{i \in I}$. By definition, $\bigmeet_{i \in I} x_i$ is the greatest lower bound, hence, $\bigmeet_{i \in I} x_i \succeq y$. Conversely, if $y \preceq \bigmeet_{i \in I} x_i$, then $y$ precedes the greatest lower bound of $\{x_i\}_{i \in I}$ implying $y$ is also a lower bound by transitivity.
\end{proof}
\end{example}

We encourage the reader, as we develop the theory of lattice-valued sheaves and sheaf Laplacians, to imagine how one could extend results from posets, meets, joins and Galois connections to categories, products, coproducts, and adjunctions, respectively.

%------------------------
\subsection{Adjunctions}
%-------------------------

A Galois connection is a specialization of a more general construction in category theory called an adjunction.

\begin{definition}
	Suppose $\cat{C}$ and $\cat{D}$ are categories and $L: \cat{C} \to \cat{D}$ and $R: \cat{D} \to \cat{C}$ are functors. Then, $L$ is a left adjoint (equivalently, $R$ is a right adjoint) written $L \dashv R$ if and only if there is a bijection
	\begin{align*}
	\Hom_{\cat{D}}(Fx, y) \leftrightarrow \Hom_{\cat{D}}(x, Gy) \\
	\end{align*}
	for all $x \in \Ob{\cat{C}}$, $y \in \Ob{\cat{D}}$.
\end{definition}

We have already seen that if $(\ladj{f}, \radj{f}): \lattice{K} \to \lattice{L}$ are Galois connections, the $\ladj{f}$ and $\radj{f}$ are functors. It remains to see
\begin{align}
	\Hom_{\lattice{L}}(\ladj{f}(x), y) \leftrightarrow \Hom_{\lattice{K}}(x, \radj{f}(y)). \label{eq:galois-adjunction}
\end{align}
Since $\lattice{K}$ and $\lattice{L}$ are posets, hom-sets $\Hom(x,y)$ is nonempty if and only if $x \preceq y$. It follows that \eqref{eq:galois-adjunction} is equivalent to Definition \ref{def:galois-connection}.

While we explictly showed some key properties of Galois connections hold, we, could have derived from more far more general properties of adjunctions. For instance, it is widely known that left adjoints preserve colimits, including coproducts, and right adjoints preserve limits, including products (Theorem \ref{thm:adjoint-functor-theorem}), or that adjoints satisfy the triangle identities (Proposition \ref{prop:left-right-left}).

%------------------------------
\subsection{Relevant examples}
%------------------------------

The following are an exhaustive list of categories employed in this manuscript.

\begin{definition}[Relevant Categories] \label{def:relevant-categories}
\leavevmode
\begin{enumerate}
	\item Sets and functions: $\cat{Set}$.
	\item Vector spaces over a field $\field$ and linear transformations: $\cat{Vec}_\field$
	\item Real or complex Hilbert spaces spaces and linear transformations: $\cat{Hil}_\R$ or $\cat{Hil}_\C$.
	\item Posets and order-preserving maps: $\cat{Pos}$.
	\item Monoids and monoid homomorphisms: $\cat{Mon}$.
	\item Complete lattices and join-preserving maps: $\cat{Sup}$
	\item Complete lattices and meet-meet preserving maps: $\cat{Inf}$.
	% These categories, relatively niche, have been extensively studied extensively in the context of topos theory \cite{joyal1984extension} as well as fuzzy logic \todo{more citations};
	% \item Lattices and lattice homomorphisms: $\cat{Lat}$ and its subcategory, bounded lattices and bounded lattice homomorphisms $\cat{Lat}^b$;
	\item Complete lattices and Galois connections: $\cat{Ltc}$. The objects are complete lattices and the morphisms are pairs $(\ladj{f}, \radj{f})$ of lower/upper adjoints. Composition is by
	\begin{align*}
		\circ &:& \Hom_{\cat{Ltc}}(\lattice{L}, \lattice{M}) \times \Hom_{\cat{Ltc}}(\lattice{K}, \lattice{L}) \to \Hom_{\cat{Ltc}}(\lattice{L}, \lattice{M}) \\
		(\ladj{g}, \radj{g}) \circ (\ladj{f}, \radj{f}) &=& (\ladj{g}\ladj{f}, \radj{f}\radj{g})
	\end{align*}
	
	% \item Dcpo's and continuous maps form a category $\cat{Dcpo}$ as does its subcategory, $\cat{Dcpo}_{\ast}$ with whose objects are dcpo's with a bottom element $0$.
\end{enumerate}
\end{definition}

% \begin{definition}[Relevant Functors]
% The functors discussed in this manuscript of the following form
% \begin{align*}
% 	F: (\poset{P}, \fc) &\longrightarrow& \cat{D}
% \end{align*}
% where $(\poset{P}, \fc)$ is an \define{indexing poset} and $\cat{D}$ is a \define{data category}.
% Concretely, this is a choice $F(x) \in (\cat{D})_0$ for every $x \in \poset{P}$ and a morphism $F(x \fc y) \in \cat{D}(F(x),F(y))$ for every $x \fc y$ in $\poset{P}$ such that $x \fc y$ and $y \fc z$ implies
% \begin{align*}
% 	F(y \fc z) \circ F(x \fc y) &=&  F(x \fc z).
% \end{align*}
% We call these functors, \define{poset functors}.

%% file: Chapters/Chapter04.tex
%************************************************
\chapter{Networks \& Signals}\label{ch:networks} %
%************************************************

Networks are ubiquitous, even when we do not realize it. Some examples of networks include social networks, neuronal networks, electrical networks, sensor networks, coauthorship networks, communication networks, transportion networks, and correlation networks. Often the structure of a network determines its function function \cite{christakis2009connected}, while other times it is the other way around \cite{watts2004six}. Dynamical systems supported on networks, \define{networked dynamical systems}, exhibit behavior that is more complex than than dynamical systems on their own \cite{watts2004new}. \define{Signal processing} is one view of networked dynamical systems, and it is the view we adopt.

%-----------------------------------
\section{Networks}
%-----------------------------------

We model networks with graphs. Recall, a \define{multiset}, denoted $\N[S]$ is a set $S$ with a map $S \to \N$ encoding the multiplicity of each element of $S$. A \define{directed (multi)graph} is a tuple $\graph{G} = (\nodes{G}, \edges{G})$ with $\nodes{G}$ a set of \define{nodes} and $\edges{G} \in \N[\nodes{G} \times \nodes{G}]$ is a (multi-)set of \define{edges}.\footnote{Traditionally, a directed multigraph is called a \define{quiver} and, equivalently, presented as the data of $\graph{Q} = (\nodes{Q}, \edges{Q}, h, t)$. \define{Head} and \define{tail} are maps $h,t: \edges{Q} \to \nodes{Q}$ sending an edge (arrow) to its head and tail, respectively.} An \define{undirected (multi)graph} is consists of the data $\graph{G} = (\nodes{G}, \edges{G})$ with $\edges{G} \in \N[\nodes{G} \times \nodes{G}/\sim]$ where we identify the $(i,j) \in \nodes{G} \times \nodes{G}$ with $(j,i)$ by $(i,j)\sim (j,i)$. A graph is \define{simple} if $\edges{G}$ is a set and $\edges{G}$ has an empty intersection with the diagonal $\Delta(\nodes{V}) = \{(i,j) \in \nodes{V} \times \nodes{V}~\vert~ i = j\}$. \emph{We hereafter assume all graphs are simple undirected graphs, unless otherwise noted. We denote nodes with lowercase letters $i, j, i', j'$ etc. and denote edges with pairs of lowercase letters $ij, i'j'$ etc.}

In a graph, a \define{path}  $\gamma$ from $i \in \nodes{G}$ to $j \in \nodes{G}$ is a sequence of nodes $\{i= i_0, i_1, \dots, i_\ell = j\}$ such that $(i_k, i_{k+1}) \in \edges{G}$ for all $0 \leq k \leq \ell$; we write the legnth of $\gamma$ as $|\gamma| = \ell$. The opposite path $\gamma^{-1} = (j = i_{\ell}, i_{\ell-1}, \dots, i_0 = i)$ is a path from $j$ to $i$. A path is a \define{loop} if it has the same source and target nodes. The \define{diameter} of a graph is the length of the longest non-intersecting path $\mathrm{diam}(\graph{G})$. Paths form a category as follows.

\begin{definition}[Free Category]
	Suppose $\graph{G} = (\nodes{G}, \edges{G})$ is a category. The \define{free category} of $\graph{G}$ is a category $\Free(\graph{G})$ with objects $\nodes{G}$ and morphisms between $i, j \in \nodes{G}$ given by paths $\gamma: i \to j$. Composition of paths is concatenation written from left to right $\gamma_1 \cdot \gamma_2$ and the identity is the trivial path denoted $\epsilon_i$ for each node $i \in \nodes{G}$.
\end{definition}

Suppose $i \in \nodes{G}$. The \define{neighbors} of $i$ are the set $\nbhd{i} = \left\{ j ~\vert~ (i,j) \in \edges{G} \right\}$. The cardinality of $\nbhd{i}$ is called the \define{degree} of $i$, denoted $d_i$. Given an edge $(i,j) \in \edges{G}$, the \define{boundary} is the set $\boundary{(i,j)} = \{i, j\}$. Given a node $i \in \nodes{G}$, the \define{coboundary} is the (multi)set $\boundary{i} = \{e \in \edges{G} ~\vert~ i \in \boundary{i}\}$. For an edge $ij \in \edges{G}$, the \define{edge neighbors} consist of $\nbhd{i} = \coboundary{i} \cup \coboundary{j} - ij$.  We say $i$ is \define{incident} to $e \in \edges{G}$ if $i \in \boundary{e}$. $\graph{G}$ is \define{d-regular} if $d_i = d$ for every $i \in \nodes{G}$. $\graph{G}$ is complete \define{complete} if $(i,j) \in \edges{G}$ for all $i, j \in \edges{G}$, $i \neq j$. A graph is \define{connected} if for every $i, j \in \nodes{G}$, there is a path from $i$ to $j$.

If $\graph{G}$ has $n$ nodes, then square $n \times n$ matrices encode graph structure with their sparsity patterns. A weighted graph is a graph $\graph{G}$ with weight function $\weight: \edges{G} \to \R_{+}$. A weight function is frequently encoded as a $n \times n$ matrix $[A]_{ij} = a_{ij}$ called the \define{adjacency matrix} of the weighted graph
\begin{align*} [A]_{ij} &=&
	\begin{cases}
		\weight(i,j), &\text{ if } ij \in \edges{G}\\
		0 			, &\text{otherwise}
	\end{cases}.
\end{align*}
Hence, a weighted graph is given as the data $\graph{G} = (\nodes{G}, \edges{G}, A)$. The \define{degree matrix} $[D]_{ij} = d_{ij}$ of $\graph{G}$ is the $n \times n$
\begin{align*} [D]_{ij} &=&
		\begin{cases}
		\sum_{j=1}^n a_{ij}, &\text{ if } i = j\\
		0 			, &\text{otherwise}
	\end{cases}.
\end{align*}
\begin{remark}[Unweighted Graphs]
	If a $\graph{G}$ is not weighted, then we set
	\[\weight(i,j) = \begin{cases} 1 & ij \in \edges{G} \\ 0 & \text{otherwise} \end{cases}\]
	Then, $d_{ii}$ is the actual degree $d_i$ of $i \in \nodes{G}$.
\end{remark}
%%%%%%%%%%%%%%%%%%%%%%%%%%
\section{Graph Signals}
%%%%%%%%%%%%%%%%%%%%%%%%%%

Graph matrices are highly useful for answering questions about signals on graph, such as, is a signal smooth are there  Graph signal processing \cite{ortega2018graph} is concerned with $\R$-valued signals on graph which are functions $\mathbf{x}: \nodes{G} \to \R$. Graph \define{shift operators} are matrices $S$ that encode the graph topology via their sparsity structure and act on graph signals. Here are a few shift operators.
\begin{itemize}
	\item The \define{adjacency matrix} $A$
	\begin{align*}
		\left(A \mathbf{x} \right)_i &=& \sum_{j \in \nbhd{i}} a_{ij} x_j &&;
	\end{align*}
	\item The \define{Laplacian matrix} $L = D - A$
	\begin{align}
		\left(L \mathbf{x} \right)_i &=& \sum_{j \in \nbhd{i}} a_{ij} (x_i - x_j) \\
		&=& d_{ii} x_i - \sum_{j \in \nbhd{i}} a_{ij} x_j; \label{eq:graph-laplacian}
	\end{align}
	\item The \define{random walk Laplacian matrix} $P = D^{-1} A$
		\begin{align}
			\left(P \mathbf{x} \right)_i &=& \sum_{j \in \nbhd{i}} \frac{1}{d_{jj}} a_{ij} x_j &&; \label{eq:random-walk}
		\end{align}
	\item Matrix normalizations of the above.
\end{itemize}

In all cases, applying the graph shift operator to a graph signal results in a diffused signal. Components of the signal are integrated with the components of neighboring nodes. Iterating the shift operator repeatedly with initial condition $\mathbf{x}[0]$ yields various difference equations. For instance,
\begin{align}\label{eq:adj-dif-eq}
	\mathbf{x}[t+1] &=& \epsilon S \mathbf{x}[t]; \\
	\mathbf{x}[0] &\in& \R^n \nonumber,
\end{align}
or,
\begin{align}\label{eq:lap-dif-eq}
	\mathbf{x}[t+1] - \mathbf{x}[t] &=& -\epsilon S \mathbf{x}[t]; \\
	\mathbf{x}[0] &\in& \R^n \nonumber
\end{align}
for step-size $\epsilon > 0$.

Spectral properties of the graph shift operator reflect on the nature of the diffusion process. Intricate descriptions of the diffusion based on the eigenvalues and eigenvectors of $S$ exist \cite{chung1997spectral}, but coarse results are based on the following facts, for the specific case $S = L$, the graph Laplacian.
\begin{fact}
	The time-independent solution of equation \eqref{eq:lap-dif-eq} are $\Ker L$ whose dimension is equal to the number of connected components of $\graph{G}$.
\end{fact}
\begin{proof}
	Solutions of \eqref{eq:lap-dif-eq} must satisfy
	\begin{align*}
		0 &=& \mathbf{x}[t+1] - \mathbf{x}[t] &=& -L \cdot \mathbf{x}[t]
	\end{align*}
	implying $\mathbf{x}[t]$ is in the kernel of $\lattice{L}$.
\end{proof}
\begin{fact}
	Suppose $\epsilon < 2d_{\max}$, the maximum degree of $\graph{G}$. Then, the discrete-time linear time-invariant system \eqref{eq:lap-dif-eq} is globally exponentially stable.
\end{fact}
\begin{proof}
	Let $P = 1 - \epsilon L$. The graph Laplacian is easily seen to be positive semi-definite. Hence, every eigenvalue is non-negative and has no complex part. By Gershgorin's Theorem \cite{horn2012matrix}, the eigenvalues of $L$ lie in the union of closed disks
	\begin{align*}
		\bigcup_{i = 1}^n \{z~\vert~|z - [L]_{ii}| \leq \sum_{i \neq j} |[L]_{ij}| \} \subseteq \C.
	\end{align*}
	Since the eigenvalues are real non-negative, we can conclude that the eigenvalues of $L$ are between $0 \leq \lambda \leq 2 d_{\max}$. If $\lambda_i$ is the $i$\textsuperscript{th} eigenvalue of $L$, then $\mu_i = 1- \epsilon \lambda_i$ is the $i$\textsuperscript{th} eigenvalue of $P$. Hence, if $\epsilon < 1/2d_{\max}$, then
	\begin{align*}
		\mu_i &=& 1- 1/2 d_{\max} \cdot \lambda_i \\
			&<& 1 - 1/2 d_{\max} \cdot 2 d_{\max} \\
			&=& 1
	\end{align*}
	Since the eignenvalues of $P$ have magnitude less than one, we conclude from elementary linear systems theory \cite{chen1984linear} the system \ref{eq:lap-dif-eq} is globally exponentially stable.
\end{proof}

\noindent Together, these imply the graph Laplacian diffusion process stabilizes at \define{locally constant signals}, signals constant on connected components of the graph.

Graph signal processing \cite{ortega2018graph} harnesses graph diffusion in order to filter graph signals at different scales. (Linear) \define{graph filters} are linear transformations represented by matrices
\begin{align*}
	H &=& \sum_{k=0}^{\infty} h_k S^k
\end{align*}
Filters act on signals by way of $\mathbf{z} = H \mathbf{x}$. Symmetric graph shift operators admit an orthonormal decomposition $S = V {\Lambda} V^{\top}$ where $\Lambda$ is a diagonal matrix of eigenvalues and $V$ is a column matrix of eigenvectors. Thus, maybe the best way to see what filters ``do'' is via the \define{graph Fourier transform}: $\hat{\mathbf{x}} = V^{\top} \mathbf{x}$. The upshot is that \define{filters} are component-wise determined by polynomials $\sum_{k=0}^{\infty} h_k \lambda_i^k$ in the Fourier domain for each eigenvalue $\lambda_i$ of $S$.

% The general principles of graph signal processing also applies to sheaves valued in various data categories. The main difference is that, with sheaves, one can process all sorts of signals. Assumptions about $\R$-valued data types, homogeneity (i.e.~each node has the same number of features), and heterophily no longer exist. In the context of deep learning on graph-structured data, the consequence of these relaxations should mean we separate hypothesis classes that we otherwise could not \cite{bodnar2022neural}.

%-----------------------------------
\subsection{Higher-order networks}
%-----------------------------------
Although this work is focused on ordinary networks, we can't help but mention higher-order networks because the sheaf Laplacian presented in Chapter \ref{ch:tarski} is defined on these objects as well, as even the unimaginative reader may ascertain.

A \define{simplicial complex} $\Space{X}$ is a collection of subsets of a node set $\nodes{G}$ such that if $\sigma \in \Space{X}$ and $\tau \subseteq \sigma$, then $\tau  \in \Space{X}$. As a consequence, individual nodes $i \in \nodes{G}$ are elements of $\Space{X}$. Simply put, a simplicial complex is a \emph{set closed under taking subsets}. An a single subset is called \define{simplex} (plural: \define{simplicies}). Simplices are graded by dimension $\dim \sigma = |\sigma| + 1$. Let $\Space{X}_k$ denote the set of $\sigma \in \Space{X}$ such that $|\sigma| = k+1$. These subsets are called \define{k-simplices}. Simplices, as do graph, have boundaries and coboundaries consisting of $\boundary(\sigma) = \{ \tau \in \Space{X}_{\dim \sigma - 1}~\vert~ \tau \subsetneq \sigma\}$ and $\coboundary(\sigma) = \{\beta \in \Space{X}_{\dim \sigma + 1}~\vert~ \beta \supsetneq \sigma \}$.

Simplicial complexes, like graphs, are purely combinatorial objects, but they can be realized as subsets of Euclidean space, much as graphs can be embedded. The \define{standard} $n$-simplex is the space $\Delta [k] = \left\{ (t_0, \dots, t_k) \in \R^{k+1} ~\vert~ \sum_i t_i = 1,~t_i \geq 0~\forall i \right\}$. The \define{geometric realizaton} of $\Space{X}$ is the topological space $|\Space{X}|$ formed by gluing together standard simplices. In this way, we can interpret a simplicial complex as a collection of points, line segments, filled-in triangles, filled-in tetrehedra, and so forth. It is natural to consider face relations here instead of simple incidence relations. We write $\sigma \fc \tau$ if $\sigma \subseteq \tau$ and $|\tau| = |\sigma|+1$. The face relation poset $(\Face(\Space{X}), \fc)$ is the transitive reflexive closure of the face relation. A simplicial complex $\Space{X}$ is \define{orientated} if every $k$-simplex is endowed with an ordering $\tau = [i_0 i_1 \dots i_k]$. If $\sigma \fc \tau$, then, neccesarily, $\sigma$ consists of the data of a cyclic permutation of $\{i_0, i_1, \dots, i_k\}$. We say the \define{sign} of $\sigma$ is the parity of this permuation.

Hypergraphs generalize simplicial complexes, are far more practical, yet less understood. Hypergraphs model higher-order networks where relationships between agents are not required to be pair-wise (e.g.~group chats), and a relationship between a number of agents does not necessarily mean that each subset of agents also share a relationship on their own. A \define{hypergraph} $\graph{H}$ is a tuple $(\nodes{H}, \hyperedges{H})$ where $\hyperedges{H} \subseteq \powerset{\nodes{G}}$. As such, a hypergraph inherits the partial order on $\powerset{\nodes{G}}$ bestowing a hypergraph with the structure of a poset $(\hyperedges{H}, \subseteq)$. A simplicial complex, then, is synonymous with a down-closed hypergraph.

Finally, (regular) cell complexes are another generalization of simplicial complexes of a more topological flavor. Let $\mathbb{D}^k = \{ x \in \R^{k+1} ~\vert~ \| x\| \leq 1 \}$ be the unit disk (ball) of dimension $k$. A \define{regular cell complex} is constructed from sequence of topological spaces
\begin{align*}
	\emptyset = \Space{X}_{-1} \subseteq \Space{X}_0 \subseteq \Space{X}_1 \subseteq \Space{X}_2 \subseteq \cdots
\end{align*}
such that $\Space{X}_{k}$ is obtained from $X_{k-1}$ by ``gluing''copies of $\mathbb{D}^k$ to $\Space{X}_k$. For more facts about simplicial complexes and cell complexes consult a textbook on algebraic toplology \cite{hatcher2002algebraic}. Parity of faces as well as face relations port from simplicial complexes to this slightly more general setting.

%-----------------------------------------------
\section{Representations of (Hyper)networks}
%-----------------------------------------------

We discuss various order-theoretic and other combinatorial representation of graphs, hypergraphs and simplicial complexes
% , as well as simplicial and topological reprentations of posets.

\begin{enumerate}
	\item  \textit{Incidence poset}. Suppose $\graph{G}$ is a graph. Incidence forms a relation on the nodes $\graph{G}$ denoted with the symbol $\fc$ according to the rule: $i \fc e$ if $i \in \boundary e$. For each edge $ij$, ther eare exactly two incidence relations $i \fc ij \cofc j$. The transitive-reflexive closure  of $\fc$ completes the incidence relation to a partial order on the union $\nodes{G} \cup \edges{G}$ which is denoted, in abuse of notation, $\poset{P}_{\graph{G}}$. We call this poset the \define{incidence poset} of $\graph{G}$ denoted $\poset{P}_{\graph{G}}$
	\item \textit{Face relation poset}. More generally, the simplices (cells) of a simplicial complex $\Space{X}$ (regular cell complex) form a relation with $\sigma \fc \tau$ if and only if $\sigma \in \boundary \tau$. The transitive reflexive closure of this relation is a poset denoted $\Face(\Space{X})$.
\end{enumerate}

Hypergraphs are notoriously unwieldy due to their lack of scalability and uniformity. Hence, it is popular to \emph{represent} hypergraphs with other structures, such as (undirected) graphs.
\begin{enumerate}
	\item[3.] \textit{Clique expansion}. The clique expansion of $\graph{H}$ is the graph $\graph{G_c}$ with $\nodes{G_c} = \nodes{H}$ and $ij \in \edges{G_c}$ if and only if there exists a hyperedges contianing both $i$ and $j$.
	\item[4.] \textit{Line graph}. The line graph of $\graph{H}$ is the graph $\graph{G_\ell}$ with $\nodes{G_\ell} = \hyperedges{H}$ and $ee' \in \edges{G_\ell}$ if and only if $e \cap e' \neq \emptyset$.
	\item[5.] \textit{Star espansion}. The star expansion is the bipartite graph $\graph{G_s}$ with $\nodes{G_s} = \nodes{H} \cup \hyperedges{H}$ and $(i, e) \in \edges{G_s}$ if and only if $i \in \nodes{H}$, $e \in \hyperedges{H}$ and $i \in e$.
\end{enumerate}
A \define{hypergraph signal} is simply a map $\mathbf{x}: \nodes{H} \to \R$. One benefit of these graph representations is that they unleash a variety of ways to process hypergraph signals via graph signal processing.

% We suspect that the data of the three above representations uniquely determines a hypergraph.
% \begin{conjecture}
% 	Up to order isomorphism, a hypergraph can be uniquely recovered from its clique expansion and line graph.
% \end{conjecture}

We turn our attention to various order-theoretic approachs to representing hypergraphs. The lattice of down-closed sets $\mathcal{D}(\hyperedges{H})$ is a finite distributive lattice. If a hypergraph is a simplicial complex, then every hyperedge is down-closed, hence included in $\mathcal{D}(\hyperedges{H})$. 
The data of the star expansion is equivalent to the data of a binary \define{membership relation} $\rel{I} \subseteq \nodes{H} \times \hyperedges{H}$ with $i \rel{I} e$ if and only if node $i \in e$. The Galois lattice $\Gal(\rel{I})$, then, offers another lattice representaiton of a hypergraph. Of note, the Galois lattice $\mathrm{Gal}(\rel{I})$ is isomorphic to the lattice of subsets $\sigma$ of nodes such that the intersection of the set of hyperedges containing every node in $\sigma$ is again $\sigma$. Alternatively $\mathrm{Gal}(\rel{I})$ is isomorphic to the lattice of subsets $\tau$ of hyperedges such that the set of hyperedges containing every node in the joint intersection of nodes belonging to $\tau$ is again $\tau$.

%% file: Chapters/Chapter05.tex
%************************************************

%************************************************
\chapter{Sheaves}\label{ch:sheaves} %
%************************************************

Sheaf theory is the mathematical study of the ``local to global.'' This is why it is not surprising that sheaf theory has been applied in settings where local data (behavior, information) extends to global data \cite{ghrist2017positive,moy2020path}. In the general setting, sheaves are indexed by a poset of open sets. Recall, a $\Space{X}$ be a topological space. Then, the open subsets of $\Space{X}$, denoted $\Open(\Space{X})$, is a poset with inclusion $(\Open(\Space{X}), \subseteq)$. Then, a sheaf valued in a data category $\cat{D}$ over a topological space $\Space{X}$ is a contravariant functor \[\sheaf{F}: \op{\Open(\Space{X})} \to \cat{D}\] In this chapter, we discuss a combinatorial version of sheaf theory called \define{network sheaf theory}. Instead of data being stored over open sets, it is stored over nodes and edges of an undirected graph. More details on general sheaf theory are found in Appendix \ref{ch:appendix-1}.

%---------------------------------------
\section{Network Sheaves \& Cosheaves}
%---------------------------------------

Sheaves are a powerful data structure for interpreting multi-agent systems with near-arbitrary flexibility as far as the type of information/data each agent retrieves and exchanges with her its neighbors. Sheaves encode assignments of data to a geometric or topological structure. Examples of structures that one can build sheaves over are many but some common ones are graphs, simplicial complexes, cell complexes, posets, topological spaces, sites and varieties. While some of these ``sheaf theories'' are subtle, \define{network sheaf theory} is not.

\begin{definition}[Network Sheaves \& Cosheaves] \label{def:network-sheaves}
	Suppose $\graph{G} = (\nodes{G}, \edges{G})$ is a graph, $\poset{P}_{\graph{G}}$ its incidence poset, and $\cat{D}$ a data category. A \define{network sheaf} is a functor $\sheaf{F}: \poset{P}_{\graph{G}} \to \cat{D}$. A \define{network cosheaf} is a functor $\cosheaf{F}: \op{\poset{P}_{\graph{G}}} \to \cat{D}$.
\end{definition}

A sheaf $\sheaf{F}$ assigns a $\sheaf{F}(i)$ in $\cat{D}$ called a \define{stalk} to every $i \in \nodes{G}$, a $\sheaf{F}(ij)$ to every edge $ij \in \edges{G}$, and \define{restriction maps}
\[\begin{tikzcd}
	{} & \sheaf{F}(i) \arrow[r, "\sheaf{F}(i \fc ij)"] & \sheaf{F}(ij) & \sheaf{F}(j )\arrow[l, "\sheaf{F}(j \fc ij)", swap] & {} \\
	{} \arrow[r, no head] & \bullet_i \arrow[rr, no head]  & & \bullet_j \arrow[ru, no head] \arrow[r, no head] & {} \\
	{} \arrow[ur, no head]	&	&	&	& {}
\end{tikzcd}\]
A cosheaf $\cosheaf{F}$ assigns a $\cosheaf{F}(i)$ in $\cat{D}$ called a \define{stalk} to every node $i \in \nodes{G}$, a $\cosheaf{F}(ij)$ to every $ij \in \edges{G}$, and \define{corestriction maps}
\[\begin{tikzcd}
	{} & \cosheaf{F}(i)  & \cosheaf{F}(ij) \arrow[l,"\cosheaf{F}(ij \cofc i)", swap] \arrow[r, "\cosheaf{F}(ij \cofc j)"] & \cosheaf{F}(j ) & {} \\
	{} \arrow[r, no head] & \bullet_i \arrow[rr, no head]  & & \bullet_j \arrow[ru, no head] \arrow[r, no head] & {} \\
	{} \arrow[ur, no head]	&	&	&	& {}
\end{tikzcd}\]
A \define{local (co)section} is a partiular element $x_i \in \sheaf{F}$ or $\cosheaf{F}(i)$.

If a (co)sheaf is a class (data structure), than a \define{(co)chain} in an instantiation of a (co)sheaf.
\begin{definition}[Cochains]
	Suppose $\sheaf{F}$ is a sheaf over $\graph{G}$ valued in a data category with products. The \define{cohains} of $\sheaf{F}$ are the product
	\begin{align*}
		C^\bullet(\sheaf{F}; \graph{G}) &=& \prod_{\sigma \in \nodes{G} \cup \edges{G}} \sheaf{F}(\sigma)
	\end{align*}
	We call the projection of $C^{\bullet}(\graph{G}; \sheaf{F})$ onto nodes the \define{0-cochains}
	\begin{align*}
		C^0(\sheaf{F}; \graph{G}) &=& \prod_{i \in \nodes{G}} \sheaf{F}(i)
	\end{align*}
	and the projection onto edges the \define{1-cochains}
	\begin{align*}
		C^1(\sheaf{F}; \graph{G} ) &=& \prod_{ij \in \edges{G}} \sheaf{F}(ij).
	\end{align*}
\end{definition}

Informally, (co)chains are sometimes called assignments to (co)sheaves \cite{robinson2020assignments}. A (co)chain is like a rough vector field \cite{lee2013smooth} on a manifold in that tangent vectors in a neighborhood need not vary smoothly or continuously.
.
One of the successes of sheaf theory is the formalization of local-to-global phenomena. Historically, sheaf theory \cite{bredon2012sheaf} has been motivated by the problem of extending (continuous, smooth, holomorphic, rational) functions defined on open sets in a (topological, manifold, complex manifold, variety) space to larger open sets. In the specialized case of network sheaves, the \define{local-to-global} problem asks \emph{if a $0$-cochain extend to a $1$-cochain?} In the affirmative, a $0$-cochian is said to be consistent.

\begin{definition}[Sections]\label{def:sections}
	Suppose $\sheaf{F}$ is a sheaf and $\cosheaf{F}$ is a cosheaf over $\graph{G}$. A \define{section} of $\sheaf{F}$ is an assignment $\mathbf{x} \in C^{\bullet}(\sheaf{F}; \graph{G})$
	\begin{align}
		\sheaf{F}_{i \fc ij}(x_i) &=& x_{ij} &=& \sheaf{F}(j \fc ij)(x_j) \label{eq:sections}
	\end{align}
	agreeing over every edge $ij$. The collection of sections of $\sheaf{F}$ are denoted $\sections{\graph{G}; \sheaf{F}}$.
\end{definition}

\noindent The following example illustrates, even in a very simple example, that restriction maps as well as the topology of the graph determine whether or not consistent assignments exist.

\begin{example}[Constant Sheaf]
	Let $\graph{G}$ be an arbitrary graph. Let $V$ be a vector space. The (linear) \define{constnat sheaf} denoted $\underline{V}$ is a sheaf valued in $\cat{Vec}$ with the following data
	\begin{align*}
		\underline{V}(i) &=& V \\
		\underline{V}(ij) &=& V \\
		\underline{V}_{i \fc ij} &=& \id_{V}.
	\end{align*}
	Of course the constnat sheaf makes sense in other data categories. It is the network sheaf with a particular object over every stalk and identity restriction maps.
\end{example}
\begin{example}[A Twisted Sheaf]
	Let $\sheaf{F}$ be the sheaf over a 3-clique (Figure \ref{fig:twisted}). In order for a $0$-cochain to extend to a $1$-cochain,
	\begin{align*}
		x_i &=& -x_j
	\end{align*}
	for every edge $ij$. However, if agreement is reached over two edges, then, agreement is impossible over the remaining edge because the assignemnts to its endpoints are the same. Hence, the only global section is the zero vector $\mathbf{x} = \mathbf{0}$.
\begin{figure}[h]
	\begin{center}
		\includegraphics[width=0.75 \textwidth]{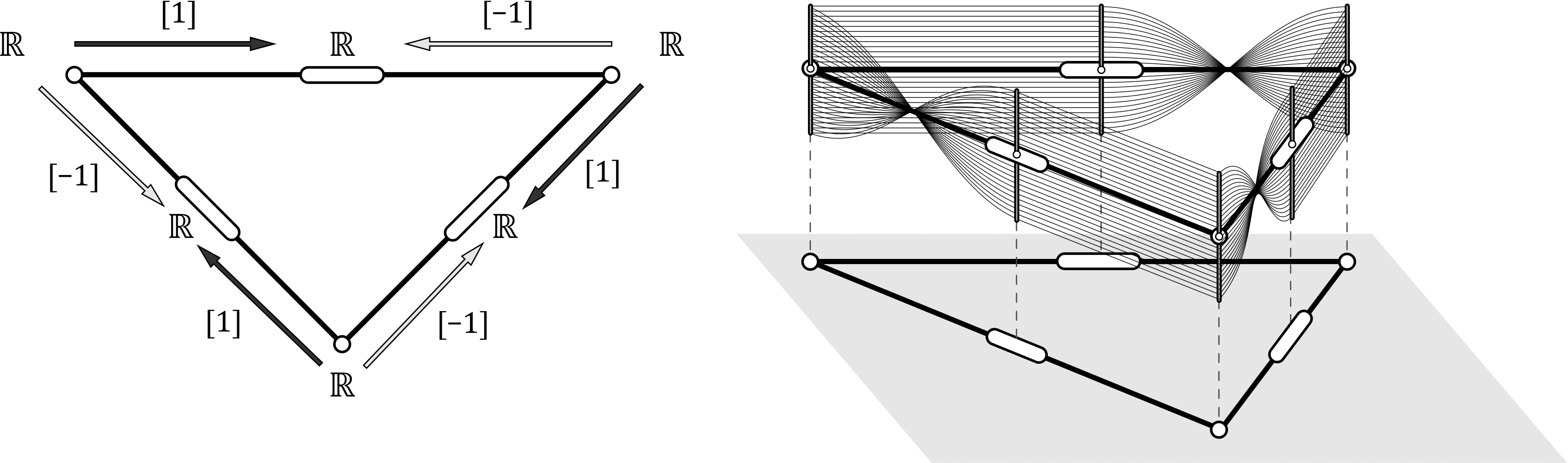}
		\caption{A twisted sheaf.} \label{fig:twisted}
	\end{center}
\end{figure}
\end{example}

How do you calculate sections of a sheaf? An orientation on $\graph{G}$ is a pair of maps
\begin{align*}
	(-)_{+,-}: \edges{G} \to \nodes{G},
\end{align*}
sending an edge $ij$ to its positive and negative endpoints ${ij}_{+} \cup {ij}_{-} = \{ i, j \}$. Given a choice of orientation, sections are equivalent to the following equalizer
\begin{equation}
\begin{tikzcd}
H^0(\graph{G}; \sheaf{F}) \arrow[r, dashed] & C^0(\graph{G};\sheaf{F}) \arrow[r,"\coboundary_{-}",shift left] \arrow[r,"\coboundary_{+}", shift right, swap] & C^1(\graph{G};\sheaf{F})
\end{tikzcd} \label{eq:section-equalizer}
\end{equation}
where $\coboundary_{-}, \coboundary_{+}$ are defined component-wise
\begin{align*}
(\coboundary_{-}\mathbf{x})_{ij} &=& \sheaf{F}({ij}_{-} \fc ij), \\
(\coboundary_{+}\mathbf{x})_{ij} &=& \sheaf{F}({ij}_{+} \fc ij).
\end{align*}
The maps $\coboundary_{-,+}$ are called \define{coboundary maps}. The equalizer condition specifies that $H^0(\graph{G}; \sheaf{F})$ is the largest subobject of $C^0(\graph{G};\sheaf{F})$ such that
\begin{align*}
	(\coboundary_{-}\mathbf{x})_{ij} &=& (\coboundary_{+}\mathbf{x})_{ij}
\end{align*}
for every $\mathbf{x} \in H^0(\graph{G}; \sheaf{F})$. In summary,

\begin{proposition}\label{thm:section-equalizer}
	$\sections{\graph{G}; \sheaf{F}} \cong H^0(\graph{G}; \sheaf{F})$. 
\end{proposition}

An \define{additive category} is a data category in which maps have the structure of an abelian group; maps can be added, subtracted, and there is a zero map between any two objects.

\begin{example}[Additive Category]
	Suppose data category $\cat{D}$ is an additive category. Then, the equalizer \eqref{eq:section-equalizer} can be rewritten as the equalizer
		\begin{equation}
		\begin{tikzcd}
			H^0(\graph{G}; \sheaf{F}) \arrow[r, dashed] & C^0(\graph{G}; \sheaf{F}) \arrow[r,"\coboundary_{+} - \coboundary_{-}",shift left] \arrow[r,"0", shift right, swap] & C^1(\graph{G}; \sheaf{F})
		\end{tikzcd} \label{eq:abelian-section-equalizer}
		\end{equation}
		Then, $H^0(\graph{G}; \sheaf{F}) \cong \Ker(\coboundary_{+} - \coboundary_{-})$.
\end{example}

The consequence of the above example is that calculating cohomology, at least in the category of vector spaces, reduces to calculating the kernel of linear transformation. Of course, many information systems exhibit vector-valued data, but others exhibit other data types. Returning to the equalizer \eqref{eq:section-equalizer}, if the data category is not additive, there is no clear path forward. We cannot write the equalizers as a kernel becuase each restriction map in a $\cat{D}$-valued network sheaf has no additive inverse. Hence, standard practices of linear and homological algebra are of no avail. We have already encountered categories that fail to be additive: the category $\cat{Sup}$ of complete lattices and join-preserving maps, as well as the category $\cat{Inf}$ of complete lattices and meet-preserving maps.

% %--------------------------------------------
% \paragraph{Building a Network Sheaf Lexicon}
% %--------------------------------------------
% It is highly practical to have names for network sheaves that share properties related to their restricition maps. A network sheaf $\graph{F}$ is \define{symmetric} if $\sheaf{F}_{i \fc ij} = \sheaf{F}(j \fc ij)$ for every edge $ij \in \edges{G}$. A network sheaf is \define{star-shaped} if given a node $i$, every restriciton map $\sheaf{F}_{i \fc ij}$ is the same for every neighbor $j$. A network sheaf is an \define{monosheaf} if every restriction map is a monomorphism \cite{weibel1995introduction}; in most data catgories, this is equivalent to every restriction map being injective. A network sheaf is an \define{episheaf} if every restriction map is an epimorphism \cite{weibel1995introduction}; in most data categories, this is equivalent to every restriction map being surjective. A network sheaf is \define{invertible} if every restriction map has an inverse. Similar terminology can be used to describe cosheaves.

%--------------------------------
\section{Homology \& Cohomology}
%--------------------------------

In the past 20 years, applied algebraic topology \cite{ghrist2014elementary} has flourished into a mature research community. Key to this development is \define{persistent homology} and software computing it \cite{otter2017roadmap}. Suppose $\Space{X}$ is a simplicial complex. Then, the \define{homology of $\Space{X}$ in degree $k$}, denoted $H_k(\Space{X})$, is a vector space whose basis elements correspond to independent connected components ($k=0$), cycles ($k=1$), voids ($k=2$), and higher dimensional versions of these ($k>2$). The difficulty with using homology directly to understand data is that it is not robust to perturbations. This leads to the notion of persistent homology introduced in the early 2000s \cite{edelsbrunner2000topological}. A filtration of simplicial complexes (spaces) is a collection of subcomplexes (subspaces) $\{\Space{X}\}_{t \in \poset{P}}$ indexed by some poset $\poset{P}$ such that
\begin{align*}
	\Space{X}_s \subseteq \Space{X}_t &~\text{whenever}~& s \preceq t
\end{align*}
in $\poset{P}$. Homology is functorial in the sense that continuous function between space $f: \Space{X} \to \Space{Y}$ induces a linear transformation between vector spaces $f: H_k(\Space{X}) \to H_k(\Space{X})$ \cite{hatcher2002algebraic}. Hence, the inclusion maps of a filtration induce a series of maps $H_k(\Space{X}_s) \to H_k(\Space{X}_t)$ whenever $s \preceq t$. Filtrations indexed by a chain $\poset{C}$ have point-wise decomposition into subobjects called a barcode \cite{oudot2017persistence} which has been highly successful in summarizing the topology of data \cite{ghrist2008barcodes}.

We have not said anything about how homology or persistent homology is actually computed. This is because network sheaves have cohomology, a dualization of homology. Cohomology is defined in a general algebraic setting, but is equivalent to homology in field coeficients by the usual isomoprhism between a vector space and its dual.

 \begin{definition}[Cochain Complex]
 	Suppose $\cat{D}$ is a category with a zero map. A \define{(bounded) cochain complex} is a sequence of objects $C^{\bullet} = \{ C^k \}_{k \geq 0}$ with maps between them called \define{coboundary maps}
 	\[\begin{tikzcd}
 		C^0 \arrow[r, "\coboundary^0"] & C^1 \arrow[r, "\coboundary^1"] &  C^2  \arrow[r,"\coboundary^{2}"] & C^3 \arrow[r,"\coboundary^3"] & \cdots
 	\end{tikzcd}\]
	such that $\coboundary^{k+1} \coboundary^{k} = 0$ for every $k \geq 0$.
	\end{definition}

If $\cat{D}$ has kernels, (normal) images, and quotients (See \cite{grandis2013homological}), the \define{cohomology} of $C^\bullet$ is 
\begin{align*}
	H^k(C^\bullet) &=& \Ker(\coboundary^k)/\Im(\coboundary^{k-1}).
\end{align*} 

%---------------------------
\section{Hodge Laplacians}
%---------------------------

Motivated by certain problems in multi-agent systems \cite{olfati2007consensus}, decentralized algorithms for computing homology of simplicial complexes have been proposed \cite{muhammad2006control}. In principle, in order to compute $H_k(\Space{X})$, a $k$-simplex (e.g. a node) can pass messages with neighboring $k$-simplices and eventually converge to a homology vector. These algorithms, while promising, have not been widely adopted. The concept, more broadly, of distributed computation of cohomology is pertinent to the study of sheaves, objects more conducive to distributive algorithms than unstructured simplicial complexes.

The following definition goes at least as far back as Eckmann \cite{eckmann1944harmonische}. Suppose $C^{\bullet}$ is a complex in the category of real (or complex) Hilbert spaces. Then, the \define{combinatorial Hodge Laplacian} is the linear transformation
\begin{align}
	\Hodge_\bullet &:& C^{\bullet} \to C^{\bullet} \nonumber \\
	\Hodge_\bullet &=& {\coboundary}^{\ast} \circ \coboundary + \coboundary \circ \coboundary^\ast \label{eq:comb-hodge-lap}
\end{align}
The Hodge Laplacian integrates cochains in $C^k$ with cochains in $C^{k-1}$ and $C^{k+1}$, suggesting the following dynamical system, as a natural generalization of the \emph{heat equation}
\begin{align}
	\dot{\mathbf{x}} &=& - \Hodge_k \mathbf{x}, \label{eq:hodge-heat-flow}
\end{align}
with initial condition $\mathbf{x}(0) \in C^k$.

Moreover, the interpretation of the stable points of the \define{heat equation} is justified by the following theorem.
\begin{theorem}[Hodge Theorem] \label{thm:hodge}
	Suppose $(C^\bullet, \coboundary^\bullet)$ is a cochain complex of inner-product spaces. Then,
	\begin{align*}
		H^k(C^\bullet) &\cong& \Ker(\Hodge^k) \label{eq:hodge-theorem}.
	\end{align*}
\end{theorem}
\noindent We provide a proof for the convenience of the reader.
\begin{proof}[Proof of the Hodge Theorem]
For ease of notation, let $Z_k = \Ker(\coboundary^k)$ and $B_k = \Im(\coboundary^{k-1})$, called the \define{cocycles} and \define{coboundaries}, respectively. The Hodge Laplacian decomposes into a sum of symmetric, and hence positive semi-definite, linear transformations
\begin{align*}
\Hodge_k &=& \Hodge_k^{+} + \Hodge_k^{-} \\
		 &=& {\coboundary^{k}}^{\ast}\coboundary^k + \coboundary^{k-1}{\coboundary^{k-1}}^\ast
\end{align*}
called the \define{upper-} and \define{lower-} \define{Hodge Laplacians}.
Recall, by definition of quotient vector spaces, in the category of finite-dimensional inner-product spaces
\begin{align*}
H^k(C^\bullet) 	&\cong& Z_k \cap B_k^{\bot}.
\end{align*}
We have $x \in \Ker(\Hodge_k^{+})$ if and only if $\langle x,\Hodge_k^{+} x \rangle_{C^k} = 0$ for all $x \in C^k$. Similarly, $x \in \Ker(\Hodge_k^{-})$ if and only if $\langle x, \Hodge_k^{-} x \rangle_{C^k}= 0$. Hence, it suffices to show
\begin{align*}
	\Ker(\Hodge_k^{+} + \Hodge_k^{-}) &=& \\
	\Ker(\Hodge_k^{+}) \cap \Ker(\Hodge_k^{-}) &=& Z_k \cap B_k^{\bot}.
\end{align*}
Thus, $x \in \Ker(\Hodge_k)$ if and only if
\begin{align*}
	\langle \coboundary^{k} x, \coboundary^{k} x \rangle_{C^k} &=& 0 \\
	\langle {\coboundary^{k-1}}^\ast  x, {\coboundary^{k-1}}^\ast x \rangle_{C^k} &=& 0
\end{align*}
which holds if and only if $x \in \Ker(\coboundary^k) \cap \Ker({\coboundary^{k-1}}^{\ast})$.

We claim $\Ker({\coboundary^{k-1}}^{\ast}) = B_k^{\bot}$. First, $x \in \Ker({\coboundary^{k-1}}^{\ast})$ if and only if
\begin{align*}
	\langle {\coboundary^{k-1}}^{\ast} x,  y \rangle_{C^{k-1}} &=& 0.
\end{align*}
for every $y \in C^{k-1}$. Then,
\begin{align*}
	 \langle x, \coboundary^{k-1} y \rangle_{C^k} &=& \langle {\coboundary^{k-1}}^{\ast} x,  y \rangle_{C^{k-1}} &=& 0
\end{align*}
showing that $x$ is orthogonal to every coboundary if and only ${\coboundary^{k-1}}^{\ast} x = 0$.
\end{proof}
By basic theory of ordinary differential equations, the solution to the heat equation \eqref{eq:hodge-heat-flow} is written in closed form
\begin{align*}
	\textbf{x}(t) &=&  e^{-\Delta_k t} \textbf{x}(0).
\end{align*}
Because $\Hodge_k$ is a positive semidefinite linear operator, the Hodge Theorem implies $x(t)$ converges to a projection of $\mathbf{x}(0)$ onto $H^k(C^\bullet)$ as $t \to \infty$. Thus, a n\"{a}ive algorithm for computing $H^k(C^\bullet)$ consists of the following steps.
\begin{enumerate}
	\item Select a $\mathbf{x}(0) \in C^k$. 
	\item Compute the trajectory of the heat equation \eqref{eq:hodge-heat-flow} with initial condition $\mathbf{x}(0)$.
	\item Add $\mathbf{x}(\infty)$ to a growing basis of $H^k(C^\bullet)$.
\end{enumerate}
A strategy to compute cohomology, such as the one above, in prinicple could be parallelized because because the Hodge Laplacian is a local operator. However, it would be our guess that practical difficulties lie in numerical stability. On the other hand, calculating the entire cohomology subspace is not always practical or necessary. Applications often call for calculating a single nearest or optimal cohomology class.

% Or, just as in the case of the graph Laplacian, it is probably even more useful to filter signals on simplicial complexes (or more generally, filtering arbitrary cochains) with Hodge diffusion than finding the locally constant signal on $\graph{G}$, although the latter has found enormous success in consensus problems multi-agent systems (see, for a survey \cite{olfati2007consensus}). The former is often the strategy used by simplical nerural network architectures \cite{simplicialnn}.

% The ultimate benefit to this strategy of computing cohomology is that each computation $(\Hodge_k \mathbf{x})_\sigma$ only involves the neighboring\footnote{The union of the boundary and coboundary of $\sigma$.} simplices of $\sigma$ and, thus, is a local operator. In theory, the entire process could be parallelized.

%------------------------------
\subsection{Sheaf Laplacians}
%------------------------------

We are primarily interested in dynamical systems on sheaves, because, after all, the final destination for this line or research is designing, learning, and controlling multi-agent systems.
% The work of Jakob Hansen \cite{hansenthesis} introduced heat flow dynamics \cite{discourcesheaves} on inner-product space sheaves and solved some of the aforementioned tasks \cite{hansendistributed}. That being said, a fully fledged theory of sheaf dynamics, particularly in other categories besides inner-product spaces, waits in the wings.
For sheaves valued in the data category of real or complex Hilbert spaces, the fundamental ingredients for a Hodge theory are present due to the following fact. For the following, suppose $\Space{X}$ is a simplicial complex or regular cell complex. Then, a \define{cellular sheaf} valued in a category $\cat{D}$ is a functor $\sheaf{F}: \Face(\Space{X}) \to \cat{D}$. As a special case, if $\Space{X}$ is a graph, then a cellular sheaf over $\Space{X}$ is exactly the same data as a network sheaf.

\begin{proposition}
	Suppose $\sheaf{F}: \Face(\Space{X}) \to \cat{Hil}$ is a sheaf of real or complex Hilbert spaces over a simplicial complex. Suppose each simplex $\sigma \in \Space{X}_k$ is encoded as an ordered tuple. Then, the following is a cochain complex
	\[ \label{eq:sheaf-cochain-complex}
 		\begin{tikzcd}
 			C^0(\Space{X}; \sheaf{F}) \arrow[r, "\coboundary^0"] & C^1(\Space{X}; \sheaf{F}) \arrow[r, "\coboundary^1"] & \cdots \arrow[r,"\coboundary^{k-1}"] & C^k(\Space{X};\sheaf{F}) \arrow[r,"\coboundary^k"] & \cdots
 		\end{tikzcd}
 	\]
 	with
 	\begin{align*}
 		C^k(\Space{X}; \sheaf{F}) &=& \bigoplus_{\sigma \in \Space{X}_k} \sheaf{F}(\sigma), \\
 		\left( \coboundary^k \mathbf{x} \right)_{\tau} &=& \sum_{\sigma \fc \tau} {\mathrm{sgn}(\sigma)}\cdot \sheaf{F}(\sigma \fc \tau )(x_{\sigma}).
 	\end{align*}
\end{proposition}
\begin{proof}
	See \cite{curry2014sheaves}.
\end{proof}
\noindent Consequently, the \define{sheaf cohomoloy} $H^k(\Space{X};\sheaf{F})$ is nothing but the cohomology of the complex $ C^\bullet(\Space{X}; \sheaf{F})$, and the \define{linear sheaf Laplacian}
\begin{align*}
	\Laplacian_k: C^k(\Space{X}; \sheaf{F}) \to C^k(\Space{X}; \sheaf{F})
\end{align*}
is nothing but the Hodge Laplacian of $C^\bullet$. Restricting, now, to the case of linear sheaves over graphs, network sheaves, the linear sheaf Laplacian ($k=0$) has has the following form
\begin{align}
	(\Laplacian \mathbf{x})_i &=& \sum_{j \in \nbhd{i}} \sheaf{F}_{i \fc ij}^{\ast}\left( \sheaf{F}_{i \fc ij}(x_i) - \sheaf{F}(j \fc ij)(x_j) \right) \label{eq:sheaf-laplacian}
\end{align}

%------------------------------------------
\section{Examples of Sheaf Laplacians}
%------------------------------------------

The examples to follow show that, indeed, network sheaf Laplacians generalize both graph Laplacians and graph connection Laplacians, each of which, in turn, approximate smooth Laplacians on manifolds. Both the graph Laplacian and the graph connection Laplacian serve as inspiration for Laplacians of lattice-valued sheaves (Chapter \ref{ch:tarski}).

%-----------------------------
\subsection{Graph Laplacians}
%-----------------------------

We now revisit the graph Laplacian with a ``sheafy'' outlook. Suppose $\Space{M} \subseteq \R^N$ is a $k$-dimensional smooth (Riemann) manifold. Recall, a smooth manifold consists of the data of a space $\Space{M}$ and a homeomorphism
\begin{align*}
	\phi_U: U \to \R^k
\end{align*}
for every open set $U$ called a \define{coordinate chart}
such that
\begin{align*}
\phi_V\phi_U^{-1}: \phi_U(U \cap V) \to \phi_V(U \cap V)
\end{align*}
is a differentiable for open sets $U \cap V \neq \emptyset$. A real-valued function $f: U \to \R$ on an open $U$ is said to be \define{smooth} if $f \circ \phi_U^{-1}: \phi_U(U) \to \R$ is differentiable. Real-valued functions on open sets define sheaf\footnote{See \cite{bredon2012sheaf}, for example} $\mathcal{O}_{\Space{M}}$ defined by the following data
\begin{align*}
	\mathcal{O}_{\Space{M}} &:& \op{\Open(\Space{M})} \to \cat{Vec} \\
	\mathcal{O}_{\Space{M}}(U) &=& \{f: U \to \R \} \\ 
	\mathcal{O}_{\Space{M}}(V \subseteq U)(f) &=& f_{\vert V}.
\end{align*}
If $\graph{G}$ is a weighted complete graph constructed from $\Space{M}$ by sampling points $\{x_1, x_2, \dots, x_n\}$ uniformly from $\Space{M}$ and weighting edges inversely proportional to the distances of their boundary in $\R^N$, then the constant sheaf $\underline{\R}$ over $\graph{G}$ approximates the sheaf of smooth functions on $\Space{M}$. This weighted graph is directly weighted to the following network sheaf
\begin{align*}
	\mathcal{F}(i) &=& \R \\
	\mathcal{F}(ij) &=& \R \\
	\mathcal{F}_{i \fc ij}(x_i) &=& e^{-\frac{\| x_i - x_j\|^2}{4t}} \cdot x_i 
\end{align*}
where $t>0$.
The sheaf Laplacian of $\sheaf{F}$ is the calculated
\begin{align*}
	(\Laplacian \mathbf{x})_i &=& \sum_{j \in \nbhd{i}} \left(e^{-\frac{\| x_i - x_j\|^2}{2t}}\right)^{\ast} \left( e^{-\frac{\| x_i - x_j\|^2}{2t}} x_i - e^{-\frac{\| x_i - x_j\|^2}{2t}} x_j \right) \\
								&=& \left( \sum_{j \in \nbhd{i}} e^{-\frac{\| x_i - x_j\|^2}{t}} \right) \cdot x_i - \sum_{j \in \nbhd{i}} e^{-\frac{\| x_i - x_j\|^2}{t}} \cdot x_j, \\
								&=& d_i x_i - \sum_{j \in \nbhd{i}} a_{ij} x_j 
\end{align*}

% The graph Laplacian is shown to approximate the Laplace-Beltrami operator, a classical operator \cite{lee2018introduction} acting on real-valued functions on $\Space{M}$ (i.e.~the global sections of $\mathcal{O}$, i.e.~the scalar fields on $\Space{M}$). Notwithstanding a few technical conditions, in dimension zero, the Laplace-Beltrami operator coincides with the Hodge-DeRahm Laplacian \cite{bott1982differential}, a Hodge Laplacian in the algebraic sense.

We can identify graph Laplacians of arbitrary graphs as sheaf Laplacians in the same manner by setting stalks $\R$ and defining both restriction maps over an edge $ij \in \edges{G}$ to be multiplication by $\sqrt{a_{ij}}$. In the above case, as the parameter $t$ approaches zero, the graph Laplacain, and thus the sheaf Laplacian is shown to approximate the Laplace-Beltrami operator, a classical operator \cite{lee2018introduction} acting on real-valued functions on $\Space{M}$ (i.e.~the global sections of $\mathcal{O}$, i.e.~the scalar fields on $\Space{M}$). For the following, suppose $\Space{M} \subseteq \R^N$ is a $k$-dimensional maniold with Laplace-Beltrami operator $\Delta$, and suppose $\{x_1, \dots, x_n\} \subseteq \Space{M}$ is collection of points sampled uniformly from $\Space{M}$. Set $t_n = n^{-\frac{1}{k+2+\alpha}}$ where $\alpha > 0$, and let $\Laplacian_n^{t_n}$ be the graph Laplacian of the complete weighted graph on $\nodes{G} = \{x_1, \dots, x_n\}$ with weights
	\begin{align*}
		a_{ij} &=& e^{-\frac{{\| x_i - x_j \|}^2}{t_n}}.
	\end{align*}
\begin{theorem}[\cite{belkin2003laplacian}] \label{thm:graph-lap-approximation}
	Suppose $f \in C^\infty(\Space{M})$. Then,
	\begin{align*}
	\lim_{n \to \infty} \frac{1}{t_n (5 \pi t_n)^{\frac{k}{2}}} \Laplacian_n^{t_n} f(x) &=& \frac{1}{\mathrm{vol}(\Space{M})} \Delta f(x)
	\end{align*}
	in probability.
\end{theorem}

% The methods of calculating cohomology from the beginning of this section also apply here to compute cellular sheaf cohomology, although more technical methods involving discrete Morse theory exist \cite{nandacurry}. In the case of network sheaves $\Space{X} = \graph{G}$, it is clear that the equalizer\eqref{eq:abelian-section-equalizer} coincides with the sheaf cohomology $H^0(\Space{X}; \sheaf{F})$. What of cohseaves? Cosheaves of inner-product spaces are not so interesting---or at least they are equivalent to sheaves of inner product spaces---by dualities discussed in the dissertation of Jakob Hansen \cite{jakobthesis}.

%--------------------------------
\subsection{Connection Laplacians}
\label{sec:connection-laplacians}
%--------------------------------

With the graph Laplacian being the first key example of a sheaf Laplacain, we introduce a second example example, the graph connection Laplacian.

Recall, some other elements from differential geometry \cite{jost2008riemannian}. Suppose $\Space{M}$ is a $k$-dimensional Riemanian manifold. To each element $x \in \Space{M}$, the \define{tangent space} $T_x$ is a vector space whose elements are equivalence classes of curves $\gamma: (-1,1) \to \Space{M}$ with $\gamma(0) = x$ whose derivatives coincide. In practice, you can visualize tangent spaces as tangent lines, tangent planes etc. The \define{tangent bundle} is the disjoint union $T \Space{M} = \bigsqcup_{x \in \Space{M}}$ together with the projection map $T \Space{M} \xrightarrow{p} \Space{M}$. A \define{vector field} is a section of the tangent bundle,
\[s: \Space{M} \to T \Space{M} \quad p \circ s = \id_{\Space{M}};\] the map $s$ assigns every point $x \in \Space{M}$ a tangent vector $v \in T_x \Space{M}$. Sections of the tangent bundle define a sheaf\footnote{See, for instance, \cite{bredon2012sheaf}.}
\begin{align*}
	\Gamma &:& \Open(\Space{M}) \to \cat{Vec} \\
	\Gamma(U) &=& \{s: U \to T \Space{M}~\vert~p \circ s = \id_{U}\} \\
	\Gamma(V \subseteq U)(s) &=& s_{\vert V} 
\end{align*}
Indeed, $\Gamma(U)$ is a vector space called \define{local sections}. $\Gamma(\Space{M})$ is the space of vector fields defined on the entire manifold. 

% We might wonder if we may approximate the tangent bundle by a network sheaf. First, we approximate the smooth functions on $\Space{M}$ with a network sheaf. Smooth functions on open subsets of $\Space{M}$ forms a sheaf
% \begin{align*}
% 	\mathcal{O} &:& \Open(\Space{M}) \to \cat{Vec} \\
% 	\mathcal{O}(U) &=& \{f: U \to \R \} \\ 
% 	\mathcal{O}(V \subseteq U)(f) &=& f_{\vert V}.
% \end{align*}
% The sheaf defined in Example \ref{eg:graph-laplacian} approximates $\mathcal{O}$ because the graph Laplacian is shown to approximate the Laplace-Beltrami operator, a classical operator \cite{lee2018introduction} acting on $\mathcal{O}(\Space{M})$ and coinciding with the Hodge-DeRahm Laplacian in dimension zero \cite{bott1982differential}. For the following, suppose $\Space{M} \subseteq \R^N$ is a $k$-dimensional maniold with Laplace-Beltrami operator $\Delta$, and $\{x_1, \dots, x_n\} \subseteq \Space{M}$ is collection of points sampled uniformly from $\Space{M}$. Set $t_n = n^{-\frac{1}{k+2+\alpha}}$ where $\alpha > 0$, and let $\Laplacian_n^{t_n}$ be the graph Laplacian of the complete weighted graph on $\nodes{G} = \{x_1, \dots, x_n\}$ with weights
% 	\begin{align*}
% 		w_{ij} &=& e^{-\frac{{\| x_i - x_j \|}^2}{t_n}}.
% 	\end{align*}
% \begin{theorem}[\cite{belkin2003laplacian}]
% 	Suppose $f \in C^\infty(\Space{M})$. Then,
% 	\begin{align*}
% 	\lim_{n \to \infty} \frac{1}{t_n (5 \pi t_n)^{\frac{k}{2}}} \Laplacian_n^{t_n} f(x) &=& \frac{1}{\mathrm{vol}(\Space{M})} \Delta f(x)
% 	\end{align*}
% 	in probability.
% \end{theorem}

Approximations of the tangent bundle lead to complications. Even if points sampled from a manifold are close, their tangent spaces will differ, in general. High-dimensional data sets are often assumed to have been sampled from low-dimensional manifold. This presents a problem because we cannot relate data in two separate local coordinates. Fortunately, this is a topic widely studied in modern differential geometry, called \define{parallel transport} \cite{knebelman1951spaces}. The following is a definition or a theorem depending on how you look at it. We leave technical details to the curious reader.

\begin{definition}[\cite{jost2008riemannian}] \label{prop:manifold-parallel}
Suppose $\Space{M}$ is a $k$-dimensional manifold and $\gamma$ is a smooth path $\gamma: [0,1] \to \Space{M}$ with $\gamma(0) = x$ and $\gamma(1) = y$.
Then, there is a unique invertible linear transformation
\begin{align*}
	\mathcal{P}_{\gamma}: T_{x} \Space{M} \to T_{y} \Space{M} 
\end{align*}
called \define{parallel transport} sending a tangent vector $v \in T_{x} \Space{M}$ to a tangent vector $\mathcal{P}_{\gamma}(v) \in T_{y} \Space{M}$.
\end{definition} 

\noindent Parallel transport leads to vector diffusion and its approximated by orthogonal linear maps. The following approximates the connection Laplacian \cite{atiyah1973heat}, a close cousin of the Laplace-Beltrami operator and a Hodge Laplacian.

\begin{definition}[Graph Connection Laplacian \cite{bandeira2013cheeger}]
	Fix a dimenison $k$, and suppose $\graph{G} = (\nodes{G}, \edges{G})$ is a weighted graph with the data of an orthogonal linear transformation
		\begin{align*}
			O_{ij}: \R^k \to \R^k
		\end{align*}
	Then, the \define{graph connection Laplacian} is the matrix $\nabla^2$
	\begin{align*}	
		\nabla^2 &=& (\R^k)^n \to (\R^k)^n \\
		(\nabla^2 \mathbf{v})_i &=& d_i x_i - \sum_{j \i \nbhd{i}} a_{ij} O_{ij} v_j
	\end{align*}
	$i \in \nodes{G}$.
\end{definition}

As a consequence of an approximation result \cite{singer2012vector} analogous to Proposition \ref{thm:graph-lap-approximation}, parallel transport between tangent spaces of the manifold are approximated in the graph connection Laplacian by orthogonal linear transformations $O_{ij}$ between the tangent spaces $O_{ij}: T_{x_i} \Space{M} \to T_{x_j} \Space{M}$. The graph connection Laplacian is a sheaf Laplacian. Suppose $\{O_i \in \R^{k \times K}\}_{i \in \nodes{G}}$ is a family of orthogonal matrices indexed by the nodes of a weighted graph $\graph{G}$. Construct the sheaf as follows
\begin{align*}
	\sheaf{F}(i) &=& \R^k, \\
	\sheaf{F}(ij) &=& \R^k, \\
	\sheaf{F}_{i \fc ij}(x_i) &=& \sqrt(a_{ij})O_i x_i
\end{align*}
Then,
\begin{align*}
	(L \mathbf{x})_i &=& \sum_{j \in \nbhd{i}} \sqrt{a_{ij}} O_i^\ast \left( \sqrt{a_{ij}} O_i x_i - \sqrt{a_{ij}} O_j x_j \right) \\
	&=& d_i x_i - \sum_{j \in \nbhd{i}} a_{ij} O_i^{\ast} O_j x_j
\end{align*}
Set $O_{ij} = O_i^{\ast} O_j$ which is orthogonal.

%% file: Chapters/Chapter06.tex
%******************************************************************
\chapter{Lattice-Valued Sheaves}\label{ch:lattice-valued} % 
%******************************************************************

We introduce a novel class of network sheaves, study global sections of these sheaves, as well as mimick the classical sheaf operations \cite{bredon2012sheaf}.

%-----------------------------
\section{Tarski Sheaves}
%-----------------------------

We begin with a definition of a sheaf over a graph whose stalks are complete and whose restriction maps are Galois connections.

\begin{definition}[Tarski Sheaf]
	Suppose $\graph{G} = (\nodes{G}, \edges{G})$ is a graph. A \define{Tarski sheaf} is a functor
	\begin{align*}
		\bisheaf{F}: \poset{P}_{\graph{G}} \to \cat{Ltc}
	\end{align*}
\end{definition}

\noindent A Tarski sheaf assigns a
\begin{enumerate}
\item Complete lattice $\bisheaf{F}(i)$ to every node $i \in \nodes{G}$,
\item Complete lattice $\bisheaf{F}(ij)$ to every edge $ij \in \edges{G}$,
\item Galois connection
% https://q.uiver.app/?q=WzAsMixbMCwwLCJcXGJpc2hlYWZ7Rn0oaSkiXSxbMSwwLCJcXGJpc2hlYWZ7Rn0oaWopIl0sWzAsMSwiXFxzaGVhZntGfShpIFxcZmMgaWopIiwwLHsiY3VydmUiOi0yfV0sWzEsMCwiXFxjb3NoZWFme0Z9KGkgXFxmYyBpaikiLDAseyJjdXJ2ZSI6LTJ9XV0=
\[\begin{tikzcd}
	{\bisheaf{F}(i)} & {\bisheaf{F}(ij)}
	\arrow["{\ladj{\bisheaf{F}(i \fc ij)}}", curve={height=-12pt}, from=1-1, to=1-2]
	\arrow["{\radj{\bisheaf{F}(i \fc ij)}}", curve={height=-12pt}, from=1-2, to=1-1]
\end{tikzcd}\]
		for every $i \in \nodes{G}$ and $j \in \nbhd{i}$.
\end{enumerate}
\begin{remark}[Notation]
Slightly simplifying cumbersome notation, we will write $(\sheaf{F}(i \fc ij), \cosheaf{F}(i \fc ij)$ for the Galois connection $(\ladj{\bisheaf{F}(i \fc ij)}, \radj{\bisheaf{F}(i \fc ij)})$.
\end{remark}

\noindent A Tarski bisheaf contains the data of both a sheaf $\sheaf{F}$ and a cosheaf $\cosheaf{F}$ by
\begin{align*}
	\sheaf{F} &:& \poset{P}_{\graph{G}} \to \cat{Sup}, \\
	\cosheaf{F} &:& \op{\poset{P}_{\graph{G}}} \to \cat{Inf}.
\end{align*}
by selecting restriction maps to be lower adjoints and corestriction maps to be the upper adjoints.
% Recall, if $\cat{C}$ is a category and $\poset{P}$ is a poset, then $\cat{C}^\poset{P}$ is the category of functors and natural transformations.
\begin{theorem}\label{thm:duality}
    Suppose $\poset{P}$ is a poset. Then, the following categories are equivalences of categories
    \begin{align*}
        \cat{Ltc}^{\poset{P}} &\simeq& \cat{Sup}^{\poset{P}} \\ \cat{Ltc}^{\poset{P}} &\simeq& \cat{Inf}^{\op{\poset{P}}}
    \end{align*}
\end{theorem}
\begin{proof}
 % $\cat{Sup}$ and $\cat{Inf}$ are equivalent categories. Let \[\op{(-)}: \cat{Sup} \to \cat{Inf}\] be the functor sending a lattice to its opposite such that $\op{f}(x) \succeq \op{f}(y)$ if and only if $x \preceq y$. Equivalence follows from checking $\op{(-)}$ is full, faithful, and essentially surjective \cite{riehl2017category}.
It suffices to show $\cat{Sup} \simeq \cat{Ltc}$ because we obtain a natural isomorphism $\cat{Ltc}^{\poset{P}} \simeq \cat{Sup}^{\poset{P}}$ from a natural isomoprhism $\cat{Ltc} \simeq \cat{Sup}$.
 \begin{align*}
     \Pi: \cat{Ltc} \to \cat{Sup}
 \end{align*}
 is an equivalence. It is a standard fact in category theory that it suffices to show $\Pi$ is full, faithful, and essentially surjective \cite[Theorem 1.5.9]{riehl2017category}. Full and faithful amounts to demonstrating a bijection
 \begin{align*}
     \Hom_{\cat{Ltc}}(\lattice{K}, \lattice{L}) \cong \Hom_{\cat{Sup}}(\lattice{K}, \lattice{L}).
 \end{align*}
which follows directly from Theorem \ref{thm:adjoint-functor-theorem}. 
The argument that $\op{\cat{Inf}} \simeq \cat{Ltc}$, and hence, that $\cat{Ltc}^{\poset{P}} \simeq \cat{Inf}^{\op{\poset{P}}}$, from duality.
\end{proof}

\noindent The consequence of Theorem \ref{thm:duality} is that the data of $\bisheaf{F}$ is equivalent to the data $\sheaf{F}$ or the data $\cosheaf{F}$.

\begin{remark}[Bisheaves]
	 The notion of a ``bisheaf'' was introduced in the context of sheaves and cosheaves of abelian groups on simplicial complexes \cite{macpherson2021persistent,nanda2020canonical}. Their definition of a bisheaf consists of a sheaf and a cosheaf over the same base (e.g.~graph) together with a compatible homomorphism between the stalks of the sheaf to the stalks of the cosheaf. In some instances, our definition conincides with theirs. For example, if $\sheaf{F}$ is a monosheaf with associated epicosheaf $\cosheaf{F}$, the \emph{Tarski} bisheaf $(\sheaf{F}, \cosheaf{F})$ is a bisheaf in the sense of McPhersen-Patel \cite{macpherson2021persistent}.
 \end{remark}

%--------------------------------------------
\section{New Sheaves from Old}
%-------------------------------------------- 

In the following examples, sheaves of concrete objects (sets, vector spaces, hilbert spaces) will be denoted in plain font; induced Tarski sheaves will be denoted in calligraphic font.

%-------------------------------------------
\subsection{From sheaves of sets}
%-------------------------------------------

Suppose $S: \poset{P}_{\graph{G}} \to \cat{Set}$ is a set-valued network sheaf. Such a structure assigns the data of a set to every node and edge and (arbitrary) maps between these sets over edges. Sheaves of sets are perhaps the most basic examples of sheaves and have found their way to applications in electrical engineering, quantum theory, and topological data analysis \cite{goguen1992sheaf,abramsky2011sheaf, de2016categorified}. The key construction is the following. Suppose $f: X \to Y$ is a function. Then, there is a Galois connection
% https://q.uiver.app/?q=WzAsMixbMCwwLCJcXHBvd2Vyc2V0e1h9Il0sWzIsMCwiXFxwb3dlcnNldHtZfSJdLFswLDEsImYoLSkiLDAseyJjdXJ2ZSI6LTJ9XSxbMSwwLCJmXnstMX0oLSkiLDAseyJjdXJ2ZSI6LTJ9XV0=
\[\begin{tikzcd}
	{\powerset{X}} && {\powerset{Y}}
	\arrow["{f(-)}", curve={height=-12pt}, from=1-1, to=1-3]
	\arrow["{f^{-1}(-)}", curve={height=-12pt}, from=1-3, to=1-1]
\end{tikzcd}\]
sending a subset of $X$ to its image and a subset of $Y$ to its preimage.
Consequently, a sheaf $F: \poset{P}_{\graph{G}} \to \cat{Set}$ factors through the powerset functor, introducing a class of Tarski sheaves
\[\begin{tikzcd}
	\poset{P}_{\graph{G}} \arrow[r,"S"] \arrow[rd, "\bisheaf{S}"] & \cat{Set} \arrow[d, "\powerset{-}"] \\
										   & \cat{Ltc}
\end{tikzcd}\]

\noindent The following extended example can be safely skipped.

\begin{example}[Reeb Graphs]

In topological data analysis, filtrations of a metric space by sublevel sets serve as a canonical example of a topological filtration. This example, roughly, subsumes \v{C}ech filtrations of point clouds \cite{bauer2017morse}. The path-connected components of space is a functor $\pi_0: \cat{Top} \to \cat{Set}$. $\pi_0$, in general, does not have a group structure. Given a cell decomposition of $\R$ with $0$-cells \[Z = \{t_0 < t_1 < \dots < t_{r}\}\] and a map $f: \Space{X} \to \R$, the \define{Reeb cosheaf} is a cosheaf valued in $\cat{Set}$ with stalks $\pi_0 f^{-1} I$ for open intervals $I \subseteq \R$ subordinate to the cell decomposition and with restriction maps induced by the functoriality of $\pi_0$ \cite{de2016categorified}. The global sections of the Reeb cosheaf correspond to path components of $\Space{X}$ that are born and die at particular critical values in $Z$. Sections of the induced Tarski sheaf consist of path components that copersist.
\end{example}

%--------------------------------------------
\subsection{From vector-valued sheaves}
%--------------------------------------------

Suppose $F: \poset{P}_{\graph{G}} \to \cat{Vec}$ is a vector-valued network sheaf. Such a structure assigns a vector space to every node and edge and linear transformations between these vector spaces over edges. Sheaves of vector spaces are key to the foundations of persistent homology \cite[Section 8.2]{curry2014sheaves}. Suppose $T: V \to W$ is a linear transformation. Then, there is a Galois connection
% https://q.uiver.app/?q=WzAsMixbMCwwLCJcXHN1YnNwYWNlc3tWfSJdLFsyLDAsIlxcc3Vic3BhY2Vze1d9Il0sWzAsMSwiVCgtKSIsMCx7ImN1cnZlIjotMn1dLFsxLDAsIlReey0xfSgtKSIsMCx7ImN1cnZlIjotMn1dXQ==
\[\begin{tikzcd}
	{\subspaces{V}} && {\subspaces{W}}
	\arrow["{T(-)}", curve={height=-12pt}, from=1-1, to=1-3]
	\arrow["{T^{-1}(-)}", curve={height=-12pt}, from=1-3, to=1-1]
\end{tikzcd}\]
sending a subspace of $V$ to its image and a subspace of $W$ to its preimage. Hence, a sheaf $F: \poset{P}_{\graph{G}}$ factors through the subspace functor\footnote{Grandis calls this functor the \define{transfer functor} \cite{grandis2013homological}, we previously called it the \define{Grassmanian} \cite{ghrist2022cellular}}, introducing a class of Tarski sheaves
\[\begin{tikzcd}
\poset{P}_{\graph{G}} \arrow[r,"F"] \arrow[rd, "\bisheaf{F}"] & \cat{Vec} \arrow[d, "\subspaces{-}"] \\
									   & \cat{Ltc}
\end{tikzcd}\]
The construction is unaltered in the (full sub)category $\cat{Hil}$ of Hilbert spaces and linear transformations.

% \begin{example}[Network Coding]
% 	Suppose $\mathcal{D}$ is a digraph representing a communication network. In order to send a packet to from a source node to a sink node, instead of trying to find an optimal path to deliver the packet, each node computes a linear combination of its parent nodes, starting from the source node to the target sink node, and sends the linear combination to each child. Ghrist \& Hiroka formulated this problem with network sheaves; the global sections of sheaves of vector spaces over a finite field $\field$ provide a description of the space of network flows in a given network coding problem \cite{ghrist2011network}. In particular, finite fields $\mathbbm{F}_{2^m}$ are useful in network coding problems where messages consist of binary strings of length $m$; a single element of $\mathbbm{F}_{2^m}$ encodes a message, and vectors in the finite dimensional vector space over this field are tuples of messages. Now, enumerating network flows of a graph from a set of sources $S \subseteq \nodes{D}$ to a set of sinks $T \subseteq \nodes{D}$ is a matter of computing sections.
% \end{example}

\begin{example}[Constant Sheaves]
    Suppose $V$ is vector space and $\underline{V}$ is the constant sheaf over $\graph{G}$. Then, the induced Tarski sheaf is again the constant sheaf $\underline{\subspaces{V}}$.
\end{example}

%------------------------------------
\section{New Sup-Lattices from Old}
\label{sec:new-lat-old}
%------------------------------------

As a prerequisite for defining the sheaf operations for $\cat{Sup}$-sheaves, we require categorical operations on complete lattices.

\subsection{Products \& coproducts}

The product in $\cat{Sup}$ is inherited from the product in $\cat{Pos}$. Suppose $\poset{P}$ and $\poset{Q}$ are posets. Then, the \define{product} $\poset{P} \times \poset{Q}$ is the poset with the order $(x,y) \preceq (z,w)$ if and only if $x \preceq y$ and $y \preceq w$. The \define{coproduct} of $\poset{P}$ and $\poset{Q}$, denoted $\poset{P} \coprod \poset{Q}$ is the poset on the disjoint union of $\poset{P}$ and $\poset{Q}$ (as sets) with $x \preceq y$ if $x, y \in \poset{P}$, $x, y \in \poset{Q}$, and $x \preceq y$ in either $\poset{P}$ or $\poset{Q}$. Immediately we run into difficulty: completeness of $\poset{P}$ and $\poset{Q}$ does not imply completeness of  $\poset{P} \coprod \poset{Q}$.

Frequently, we take cartesian products of sets indexed over an arbitrary set $I$
\begin{align*}
	\prod_{i \in I} \lattice{L}_i.
\end{align*}
Elements of the product are denoted with bold letters $\mathbf{x}$. The with individual components denoted $(x_i)_{i \in I}$. The (categorical) \define{product} of a family of lattices $\{\lattice{L}\}_{i \in I}$
is the cartesian product equipped with projection maps $\pi_i: \prod_{i \in I} \lattice{L}_i \to \lattice{L}_i$. Projection maps have lower and upper adjoints
\begin{align*}
\pi_i^{\ast}, {\pi_i}_{\ast}: \lattice{L}_i \to \prod_{i \in I} \lattice{L}_i
\end{align*}
\[(\pi_i^\ast x)_j = \begin{cases}
    x_i & i = j \\
    0 & i \neq j
\end{cases} \quad ({\pi_i}_{\ast} x)_j = \begin{cases}
    x_i & i = j \\
    1 & i \neq j
\end{cases}\]
The (categorical) \define{coproduct} of $\{\lattice{L}_i \}_{i \in I}$ is the cartesian product with maps ${\pi_i}_{\ast}: \lattice{L}_i \to \coprod_{i \in I} \lattice{L}_i$ called inclusion maps.

Products and coproducts in $\cat{Sup}$ define the chains and cochains of a Tarski sheaf. Suppose $\bisheaf{F}$ is a bisheaf over $\graph{G}$. Then, the \define{0-cochains} of $\bisheaf{F}$ are the (complete) product lattice
\begin{align*}
	C^0(\graph{G}; \sheaf{F}) &=& \prod_{i \in \nodes{G}} \sheaf{F}(i)
\end{align*}
together with the projection maps
\begin{align*}
	\pi_i: C^0(\graph{G}; \sheaf{F}) \to \sheaf{F}(i).
\end{align*}
The \define{0-chains} of $\bisheaf{F}$ is the coproduct lattice
\begin{align*}
	C_0(\graph{G}; \cosheaf{F}) &=& \coprod_{i \in \nodes{G}} \cosheaf{F}(i)
\end{align*}
with the inclusion maps
\begin{align*}
{\pi_i}_{\ast}: \cosheaf{F}(i) \to C_0(\graph{G}; \cosheaf{F}).
\end{align*}

\noindent As the product and coproduct are isomorphic, then we idenity both (up to isomorphism) as the \define{biproduct} $\bigoplus_{i \in \nodes{G}} \bisheaf{F}(i)$ \cite[Appendix E.5]{riehl2017category}.

\subsection{Tensor product \& internal hom}

Recall, that the set of join-preserving morphisms between two complete lattice forms a complete lattice called the \define{internal hom} which we hereafter denote $\left[ \lattice{K}, \lattice{L}\right]$. Suppose $\{f_i\}_{i \in I} \subseteq \left[ \lattice{K}, \lattice{L} \right]$. Then,
\begin{align*}
    \left(\bigjoin f_i \right)(x) &=& \bigjoin_{i \in I} f_i(x).
\end{align*}
The \define{tensor product} $\lattice{K} \otimes \lattice{L}$ is defined\footnote{If you prefer a definition with a universal property, see \cite{joyal1984extension}.}
\begin{align*}
 	\lattice{K} \otimes \lattice{L} &\cong& \op{\left[ \mathbf{K}, \op{\mathbf{L}}\right]}.
 \end{align*}
The set $\lattice{K} \otimes \lattice{L}$ consists of order-reversing maps $f: \lattice{K} \to \lattice{L}$ such that $f( \bigjoin S) = \bigmeet f(S)$, or, in other words, contavariant Galois connections.
\begin{example}[Powersets]
If $X$ and $Y$ are sets, then
\begin{align*}
	\powerset{X} \otimes \powerset{Y} &\cong& \powerset{X \times Y}
\end{align*}
\end{example}

% We can idenity individual elements of $\lattice{K} \otimes \lattice{L}$ with the natural map \cite{tensorproducts}
% \begin{align*}
% \lattice{K} \times \lattice{L} &\xrightarrow{\beta}& \lattice{K} \otimes \lattice{L} \\
% (x,y) &\mapsto& (x \otimes y)(z) = 
% \begin{cases}
% 	1 & z = 0 \\
% 	y & 0 \prec z \preceq x \\
% 	0 & z \not \preceq x
% \end{cases}.
% \end{align*}

% Given maps $f: \lattice{K} \to \lattice{L}$ and $g: \lattice{K}' \to \lattice{L}'$, there is naural map $f \otimes g$ making the diagram commute
% \[
% \begin{tikzcd}		
% 	\lattice{K} \times \lattice{L} \arrow[r, "f \times g"] \arrow[d,"\beta" ] 	& \lattice{K}' \times \lattice{L}' \arrow[d, "\beta'"] \\
% 	\lattice{K} \otimes \lattice{L} \arrow[r, dashed, "f \otimes g"]			& \lattice{K}' \otimes \lattice{L}'
% \end{tikzcd}.
% \]

% \begin{fact}
% 	The map $f \otimes g$ is identified as $\beta' \circ (f \times g)\circ \beta^\ast$ where $\beta^{\ast}$ is the upper adjoint in the pair $(\beta, \beta^{\ast})$. 
% \end{fact}

We have the following adjunction between the internal hom and tensor product.
\begin{proposition} \label{prop:hom-tensor}
    Suppose $\lattice{K}, \lattice{L}, \lattice{M}$ are complete lattices. Then, $\left[ \lattice{K} \otimes \lattice{L}, \lattice{M} \right] \cong \left[ \lattice{K}, \left[ \lattice{L}, \lattice{M} \right] \right]$.
\end{proposition}
\begin{proof}
	See \cite[Chapter 1]{joyal1984extension}.
\end{proof}

%---------------------------------------
\subsection{Equalizers \& coequalizers}
%---------------------------------------

``solving equations'' in the $\cat{Sup}$ amounts to boils down to computing equalizers and coequalizers.

\begin{definition}[Equalizer \& Coequalizer]
	An \define{equalizer} of a pair of parrellel join-preserving maps $f,g: \lattice{K} \to \lattice{L}$ is a suplattice $\lattice{E}$ and map $i: \lattice{E} \to \lattice{K}$ with $f \circ i = g \circ i$ such that for any map $j: \lattice{M} \to \lattice{K}$, there is a unique map $i': \lattice{M} \to \lattice{E}$ with $i \circ i' = j$ as in the following diagram
	\begin{equation}
	\begin{tikzcd}
	\lattice{E} \arrow[r, "i"] & \lattice{K} \arrow[r,"f",shift left] \arrow[r,"g", shift right, swap] & 
	\lattice{L} \\
	\lattice{M} \arrow[ur,"j", swap] \arrow[u,"i'", dashed]		& &
	\end{tikzcd} \label{eq:equalizer}.
	\end{equation}
	A \define{coequalizer} of the pair $f,g$ is a suplattice $\lattice{Q}$ and a map $p: \lattice{L} \to \lattice{Q}$ with $p \circ f = p \circ g$ such that for any map $q: \lattice{L} \to \lattice{M}$, there is a unique map $p': \lattice{Q} \to \lattice{M}$ with $p' \circ p = q$ as in the following diagram
	\begin{equation}
	\begin{tikzcd}
	 \lattice{K} \arrow[r,"f",shift left] \arrow[r,"g", shift right, swap] & \lattice{L} \arrow[r,"p"] \arrow[rd,"q", swap] 	 & \lattice{Q} \arrow[d,"p'", dashed]	\\
	& & \lattice{M}
	\end{tikzcd} \label{eq:coequalizer}.
	\end{equation}
\end{definition}

\noindent Universal properties, however, do not suffice for actual computations.

\begin{proposition}\label{prop:exist-equalizer}
$\cat{Sup}$ has equalizers and coequalizers.
\end{proposition}
\begin{proof}
We claim $\lattice{E} = \{ x \in \lattice{K} ~\vert~ f(x) = g(x)\}$ with the inclusion $i$ into $\lattice{K}$ is the equalizer of $f,g: \lattice{K} \to \lattice{L}$. First, $f \circ i = g \circ i$ since $i$ is an inclusion. Second, $\lattice{E}$ is complete, as we now show. Suppose $\{x_i\}_{i \in I} \subseteq \lattice{E}$, then $\bigjoin_{i \in I} x_i \in \lattice{E}$ as
\begin{align*}
    f\left( \bigjoin_{i \in I} x_i \right) &=& \bigjoin_{i \in I} f(x_i) \\
    &=& \bigjoin_{i \in I} g(x_i) \\
    &=& g \left( \bigjoin_{i \in I} x_i \right).
\end{align*}
For coequalizers, let $\lattice{Q} = \{ y \in \lattice{L} ~\vert~ f^{\ast}(y) = g^{\ast}(y)\}$. Then, the coequalizer of $f,g: \lattice{L} \to \lattice{K}$ is the lattice $\op{\lattice{Q}}$ with the projection $i^{\ast}:\op{\lattice{Q}} \to \lattice{L}$.
\end{proof}

% \begin{proposition}\label{prop:exist-coequalizer}
% 	The category $\cat{Sup}$ has coequalizers
% 	\begin{align*}
% 		\lattice{Q} &=& \{ y \in \lattice{L} ~\vert~ f^{\ast}(y) = g^{\ast}(y)\}
% 	\end{align*}
% where $f^{\ast}$ and $g^{\ast}$ are upper adjoints in the Galois pairs $(f, f^{\ast})$ and $(g, g^{\ast})$.
% \issue{need to consult Joyal... not sure if this is right...}
% \end{proposition}
% \begin{proof}
% \todo{write}
% \end{proof}
% The theory of complete lattices and meet-preserving maps is merely a dualization of the theory of $\cat{Sup}$.
% \begin{exercise}
% 	Write down products, coproducts, internal hom, tensor products, equalizers and coequalizers in $\cat{Inf}$.
% \end{exercise}

%---------------------------------
\section{Sections}
%---------------------------------

For a $\cat{Sup}$-sheaf, we define cohomology $H^0$ and $H^1$ as equalizers and coequalizers. We show $H^0$ is isomorphic to the complete lattice of global sections. Recall, \define{sections} are assignments $\mathbf{x} \in C^\bullet(\graph{G}; \sheaf{F})$ such that that satisfy
\begin{align}
		\sheaf{F}(i \fc ij)(x_i) &=& x_{ij} &=& \sheaf{F}(j \fc ij)(x_j) \quad i \fc ij \cofc j
  \label{eq:sections}
\end{align}
for every $i \fc ij \cofc j$.
% \define{cosections} are assignments $\mathbf{x} \in C_\bullet(\graph{G}; \cosheaf{F})$ that satisfy
% \begin{align*}
% 	x_i &=& \cosheaf{F}(i \fc ij)(x_{ij}) &=& x_j  \quad \forall ij \in \edges{G}.
% \end{align*}
% \issue{is this the right definition? need to consult Joyal!}
Sections inherit the order on the product lattice $C^\bullet(\graph{G}; \bisheaf{F})$. Let $\sections{\graph{G}; \sheaf{F}} \subseteq C^\bullet(\graph{G}; \bisheaf{F})$ denote the poset of sections. In fact, $\sections{\graph{G}; \sheaf{F}}$ is a complete lattice calculated as follows.
% Cosections denoted $\cosections{\graph{G}; \cosheaf{F}}$.
% We now justify that $\sections{\graph{G}; \sheaf{F}}$ and $\cosections{\graph{G}; \cosheaf{F}}$ are complete lattices. Alternate proofs, consequence of the Tarski Fixed Point Theorem (Theorem \ref{thm:tfpt}), are presented in Chapter \ref{ch:tarski}.
\begin{definition}
Suppose $\bisheaf{F}$ be a Tarski sheaf over $\graph{G}$. Then, the let $H^0(\graph{G}; \sheaf{F})$ and $H^1(\graph{G}; \sheaf{F})$ to be the equalizer and coequalizer\footnote{See Joyal \cite[Chapter 1]{joyal1984extension} for a proof that $\cat{Sup}$ has all limits and colimits.} of the following diagram.
\begin{equation}
\begin{tikzcd}
H^0(\graph{G}; \sheaf{F}) \arrow[r, dashed] & C^0(\graph{G};\sheaf{F}) \arrow[r,"\coboundary_{-}",shift left] \arrow[r,"\coboundary_{+}", shift right, swap] & C^1(\graph{G};\sheaf{F}) \arrow[r, dashed] & H^1(\graph{G}; \sheaf{F})
\end{tikzcd} \label{eq:sup-section-equalizer}
\end{equation}
with
\begin{align*}
(\coboundary_{-})_{ij} &=& \sheaf{F}({ij}_{-} \fc ij), \\
(\coboundary_{+})_{ij} &=& \sheaf{F}({ij}_{+} \fc ij).
\end{align*} 
\end{definition}
\noindent By Proposition \ref{prop:exist-equalizer},
\begin{align*}
	H^0(\graph{G}; \sheaf{F}) &\cong& \{ \mathbf{x} \in C^0(\graph{G}; \sheaf{F})~\vert~ (\coboundary_{-} \mathbf{x})_{ij} = (\coboundary_{+} \mathbf{x})_{ij} \quad \forall ij \in \edges{G} \} \\
 	&=& \{ \mathbf{x} \in C^0(\graph{G}; \sheaf{F}) ~\vert~ \sheaf{F}(i \fc ij)(x_i) = \sheaf{F}(j \fc ij)(x_j) \quad \forall ij \in \edges{G} \}. \\
 	&\cong& \{ \mathbf{x} \in C^0(\graph{G}; \sheaf{F}) ~\vert~ \sheaf{F}(i \fc ij)(x_i) = x_{ij} = \sheaf{F}(j \fc ij)(x_j) \} \\
 	&=& \sections{\graph{G}; \sheaf{F}}.
\end{align*}

\begin{theorem}[Global Sections Theorem]\label{thm:sections-sublattice}
	Suppose $\bisheaf{F}$ is a Tarski sheaf over $\graph{G}$. Then, $H^0(\graph{G}; \sheaf{F})$ is a complete quasi-sublattice of $C^0(\graph{G}; \sheaf{F})$
\end{theorem}

\begin{warning}
    Beware, $H^0(\graph{G}; \sheaf{F})$ is not in general a sublattice. $H^0(\graph{G}; \sheaf{F})$ is an equalizer. Hence, the unique map $H^0(\graph{G}; \sheaf{F}) \to C^0(\graph{G}; \sheaf{F})$ is injective (see \cite[Exercise 3.1.vi]{riehl2017category}). However, joins and meets in $H^0(\graph{G}; \sheaf{F})$ do not (in general) coincide with joins and meets in $C^0(\graph{G}; \sheaf{F})$.
\end{warning}. 

% In a concrete category where the forgetful functor preserves pullbacks, the monomorphisms are precisely the injective maps.

%%%%%%%%%%%%%%%%%%%%%%%%%%%%%%%%%
\section{Sheaf Operations}
%%%%%%%%%%%%%%%%%%%%%%%%%%%%%%%%%

We construct four of the six classical ``sheaf operations:'' \emph{tensor product sheaves, sheaf hom, pullback and pushforward}. While these four sheaf operations do reappear, we include this exposition to inspire future work.

%---------------------------------------------------
\subsection{Sheaf morphisms \& constant sheaves}
%--------------------------------------------------

In order to define tensor product sheaves and sheaf hom, we first recall morphisms between sheaves. A \define{sheaf morphism}  $\sheaf{F} \xRightarrow{\phi} \sheaf{G}$ between $\sheaf{F}$ and $\sheaf{G}$ consists of the data of join-preserving maps $\phi_i: \bisheaf{F}(i) \to \bisheaf{G}(i)$ for every $i \in \nodes{G}$ and join-preserving maps $\phi_{ij}: \sheaf{F}(ij) \to \sheaf{G}(ij)$ for every $ij \in \edges{G}$ such that for every $i \in \nodes{G}$ and $j \in \nbhd{i}$
\begin{align}
	\sheaf{G}(i \fc ij) \circ \phi_i &=& \phi_{ij} \circ \sheaf{F}(i \fc ij). \label{eq:sheaf-morphism}
\end{align}
The set of all sheaf morphisms is denoted $\Hom\left(\sheaf{F}, \sheaf{G} \right)$.

\begin{proposition}
 	 $\Hom \left( \sheaf{F}, \sheaf{G} \right)$ is a complete lattice.
\end{proposition}
\begin{proof}
	Suppose $I$ is an arbitrary indexing set and $\{ \phi^\alpha \}_{\alpha \in I}$ is a family of sheaf morphisms. For $i \in \nodes{G}$, we write
	\begin{align*}
		\left(\bigjoin_{\alpha \in I} \phi^\alpha_i \right)(x_i) &=& \bigjoin_{\alpha \in I} \phi^\alpha_i(x_i),
	\end{align*}
	and similarly for joins of the family $\phi_{ij}^\alpha$ for $ij \in \edges{G}$.
	It follows the diagram
	\[
\begin{tikzcd}
	\sheaf{F}(i) \arrow[r,"\sheaf{F}(i \fc ij)"] \arrow[d,"\bigjoin \phi_i^\alpha"] & \sheaf{F}(ij) \arrow[d, "\bigjoin \phi_{ij}^\alpha"] &	\sheaf{F}(j) \arrow[l, "\sheaf{F}(j \fc ij)", swap] \arrow[d, " \bigjoin \phi_j^\alpha"] \\
	\sheaf{G}(i) \arrow[r, "\sheaf{G}(i \fc ij)"]					& \sheaf{G}(ij) 					  &	\sheaf{G}(j) \arrow[l, "\sheaf{F}(j \fc ij)", swap]
\end{tikzcd}
\]
commutes for all $ij \in \edges{G}$ because the restriction maps are join-preserving.
\end{proof}

\noindent  Recall, $\bool = \{0 < 1\}$, the boolean lattice with two elements.  The \define{boolean constant sheaf} over $\graph{G}$ is the sheaf $\sheaf{B}$ with $\sheaf{B}(i) = \bool$ for all $i \in \nodes{G}$, $\sheaf{B}(ij) = \bool$ for all $ij \in \edges{G}$, and $\sheaf{B}(i \fc ij) = \id_{\lattice{M}}$.

\begin{proposition} \label{prop:sections-sheaf-hom}
Suppose $\sheaf{F}$ is a $\cat{Sup}$-sheaf over a graph $\graph{G}$. Then, 
	\begin{align*}
	\sections{\graph{G}; \sheaf{F}} &\cong& \Hom \left( \sheaf{B} , \sheaf{F}  \right).
	\end{align*}
\end{proposition}
\begin{proof}
A sheaf morphism $\phi$ is the data of a join-preserving map $\phi_i: \bool \to \sheaf{F}(i)$ for every $i \in \nodes{G}$ and a join-preserving map $\phi_{ij}: \bool \to \sheaf{F}(ij)$ for every $ij \in \edges{G}$ satisfying \eqref{eq:sheaf-morphism}. Therefore,  $\phi_i$ ``picks out'' and $x_i \in \sheaf{F}(i)$ for every node and $\phi_{ij}$ ``picks out'' a $x_{ij} \in \sheaf{F}(ij)$ for every edge, such that $\sheaf{F}(i \fc ij)(x_i) = x_{ij}$. This is precisely the data of a global section.
\end{proof}
\noindent Thus, we have another proof that sections form a complete lattice. Another consequence is that the isomorphism in Proposition \ref{prop:sections-sheaf-hom} implies a sheaf morphism $\sheaf{F} \xRightarrow{\phi} \sheaf{G}$ induces a join-preserving map on the lattice of global sections
\begin{align*}
	\sections{\graph{G}; \sheaf{F}} &\xrightarrow{\Gamma(\graph{G}; \phi) }& \sections{\graph{G}; \sheaf{G}}.
\end{align*}
via precomposition.

\begin{theorem}
    Suppose $\poset{P}$ is a poset. Then, $\Gamma(\graph{G}; -)$ is a functor
    \[\Gamma(\graph{G}; -): \cat{Sup}^{\poset{P}} \to \cat{Sup}.\]
\end{theorem}

%-------------------------------------------------
\subsection{Tensor product \& sheaf hom}
%-------------------------------------------------

\begin{definition}[Tensor Product Sheaf]
	Suppose $\sheaf{F}$ and $\sheaf{G}$ are sheaves over a graph $\graph{G}$. The \define{tensor product sheaf} $\sheaf{F} \otimes \sheaf{G}$ is the functor with
	\begin{align*}
	\left(\sheaf{F} \otimes \sheaf{G}\right)(i) &=& \sheaf{F}(i) \otimes \sheaf{G}(j) \\
 \left(\sheaf{F} \otimes \sheaf{G}\right)(ij) &=& \sheaf{F}(ij) \otimes \sheaf{G}(ij)
	\end{align*}
	for all $i \in \nodes{G}$, $ij \in \edges{G}$ and restriction maps
    \[\left(\sheaf{F} \otimes \sheaf{G}\right)(i) \xrightarrow{\left(\sheaf{F} \otimes \sheaf{G}\right)(i \fc ij)} \left(\sheaf{F} \otimes \sheaf{G}\right)(ij) \]
 sending $\sheaf{F}(i) \xrightarrow{f} \op{\sheaf{G}(i)}$
 to the composition
% https://q.uiver.app/?q=WzAsNSxbMCwwLCJcXGNvc2hlYWZ7Rn0oaWopIl0sWzEsMCwiXFxjb3NoZWFme0Z9KGkpICJdLFszLDAsIlxcb3B7XFxzaGVhZntHfShpKX0iXSxbNCwwLCJcXG9we1xcc2hlYWZ7R30oaWopfSJdLFsyLDAsIlxcc2hlYWZ7Rn0oaSkiXSxbMCwxLCJcXGNvc2hlYWZ7Rn0oaSBcXGZjIGlqKSJdLFsxLDQsIiIsMCx7ImxldmVsIjoyLCJzdHlsZSI6eyJoZWFkIjp7Im5hbWUiOiJub25lIn19fV0sWzQsMiwiZiJdLFsyLDMsIlxcb3B7XFxzaGVhZntHfShpIFxcZmMgaWopfSJdLFswLDMsIihcXHNoZWFme0Z9IFxcb3RpbWVzIFxcc2hlYWZ7R30pKGkgXFxmYyBpaikoZikiLDIseyJjdXJ2ZSI6NH1dXQ==
\[\begin{tikzcd}
	{\cosheaf{F}(ij)} & {\cosheaf{F}(i) } & {\sheaf{F}(i)} & {\op{\sheaf{G}(i)}} & {\op{\sheaf{G}(ij)}}
	\arrow["{\cosheaf{F}(i \fc ij)}", from=1-1, to=1-2]
	\arrow[Rightarrow, no head, from=1-2, to=1-3]
	\arrow["f", from=1-3, to=1-4]
	\arrow["{\op{\sheaf{G}(i \fc ij)}}", from=1-4, to=1-5]
	\arrow["{(\sheaf{F} \otimes \sheaf{G})(i \fc ij)(f)}"', curve={height=24pt}, from=1-1, to=1-5]
\end{tikzcd}\]
\end{definition}
% \begin{exercise}
%     First, show $\left(\sheaf{F} \otimes \sheaf{G}\right)(i \fc ij)$ is a well-defined and join-preserving map. Then, show there is a unique sheaf $\cat{Hom}(\sheaf{F}, \sheaf{G}): \poset{P} \to \cat{Sup}$ over $\poset{P}$ such that
%     \begin{align*}
%         \cat{Hom} \left( \sheaf{F} \otimes \sheaf{G}, \sheaf{H} \right) & \cong & \cat{Hom} \left( \sheaf{F}, \cat{Hom} \left( \sheaf{G}, \sheaf{H} \right) \right)
%     \end{align*}
%     for all sheaves $\sheaf{F}, \sheaf{G}, \sheaf{H} \in \cat{Sup}^{\poset{P}}$.
% \end{exercise}

%------------------------------------
\subsection{Pullback \& pushforward}
%------------------------------------

Suppose $\graph{G} = (\nodes{G}, \edges{G}) $ and $\graph{H} = (\nodes{H}, \edges{H})$. A \define{graph homomorphism} $p: \graph{G} \to \graph{H}$ is an strict order-preserving map
\begin{align*}
 	p : \poset{P}_{\graph{G}} \to \poset{P}_{\graph{H}}.
 \end{align*} 
Suppose $p$ is a graph homomorphism $p: \graph{G} \to \graph{H}$ and $\sheaf{F}$ is a $\cat{Sup}$-sheaf over $\graph{H}$. The \define{pullback} of $\sheaf{F}$ over $p$ is the sheaf denoted $p^{\ast} \sheaf{F}$ with
\begin{align*}
	p^{\ast} \sheaf{F}(i) &=& \sheaf{F}\left( p(i) \right) \\ p^{\ast} \sheaf{F}(ij) &=& \sheaf{F}\left( p(ij) \right), \\
	p^{\ast}\sheaf{F}(i \fc ij)(x_i) &=& \sheaf{F}\left( p(i) \fc p(ij) \right).  
\end{align*}
The definition of the \define{pushforward} of $\sheaf{F}$, denoted $p_{\ast} \sheaf{F}$ is more involved. Stalks are
\begin{align*}
    p_{\ast} \sheaf{F}(i) &=& \lim_{k \fc p(ij)} \sheaf{F}(ij). 
\end{align*}

The cannonical example of a pushforward is the pushforward of the map $p: \graph{G} \to \ast$ which is isomorphic to $\sections{\graph{G}; \sheaf{F}}$. A canonical example of a pullback is the pullback of an embedding $\graph{H} \hookrightarrow \graph{G}$ called the restriction of $\sheaf{F}$ to $\graph{H}$.

%% file: Chapters/Chapter07.tex
%************************************************
\chapter{The Tarski Laplacian}\label{ch:tarski} % 
%************************************************

The graph Laplacian is an operator that acts on scalar fields while the connection Laplacian acts on vector fields. Their common denominator is the sheaf Laplacian \cite{hansen2019toward}. In this chapter, we motivate a few candidate Laplacians for Tarski sheaves from several angles. First, from the vantage of lattice-valued consensus, then as an analogy to parallel transport, then finally as an homage to Hodge theory. Along the way, we present the Hodge-Tarski Fixed Point Theorem (Theorem \ref{thm:main}) and provide algorithms to compute global sections in various settings.

% \[\langle \cdot, \cdot \rangle_{C^k}.\]
% Instantiations of the Hodge Laplacian go hand-in-hand with particular cochain complexes.
% \begin{itemize}
% 	\item The Hodge-de Rham Laplacian is the Hodge Laplacian of the de Rahm complex of a smooth manifold.
% 	\item The simplicial Laplacian is the Hodge Laplacian of the simplicial cochain complex \[C^k(\Space{X}) = \{ \sigma: \Space{X}_k \to \R \}.\]
% 	\item The graph Laplacian is just the simplicial Laplacian of a graph.
% 	\item The linear sheaf Laplacian is the Hodge Laplacian of the cellular sheaf cochain complex of a sheaf of inner-product spaces.
% \end{itemize}

% From the vantage of signal processing, the graph Laplacian is limited as a shift operator because it requires every node to have the same number of features and the Laplacian processes each feature separately; there is no coupling between different features. The linear sheaf Laplacian remedies both of these: we no longer require nodes to have identical number of features, and each node ``passes messages'' to neighbors as in the graph Laplacian, nodes send a transformed version of their features to their neighbors. It has also been suggested that the linear sheaf Laplacian grants more control of diffusion and, thus, convolution, preventing oversmoothing \cite{bodnar2022neural} in the process of minimizing a loss function in a graph neural network. 

%%%%%%%%%%%%%%%%%%%%%
\section{Consensus}
\label{sec:consensus}
%%%%%%%%%%%%%%%%%%%%%%

One practical motivation of the Tarski Laplacian is consensus.

\begin{definition}[Multi-Agent Lattice Consensus]
Let $\lattice{L}$ be a lattice. Suppose $\graph{G} = (\nodes{G}, \edges{G})$ is a graph. A \define{consensus} function is a map $\chi: \lattice{L}^{\nodes{G}} \to \lattice{L}$. A \define{consensus protocol} is a map $\chi: \lattice{L}^{\nodes{G}} \to \lattice{L}$ (consensus function) together with an algorithm whose input is an arbitrary signal $\mathbf{x}: \nodes{G} \to \lattice{L}$ and whose (desired) output is $\chi(\mathbf{x})$.
\end{definition}

\begin{example}[Meet-Consensus]
    Meet-consensus takes place in a sequence of rounds. Each node, every node its state $x_i \in \lattice{L}, i \in \nodes{G}$ to each of its neighbors $j \in \nbhd{i}$. At the conclusion of the round, a node $i$ computes the meet of its received messages as well as (possibly) $x_i$. In a simple example (Figure \ref{fig:consensus}), two nodes (left, right) share a communication link with a common node (center). At $t=0$ (light), each node computes meets with neighbors. At $t=2$ (dark), a consensus is reached.
\begin{figure}[h]
    \centering
    \includegraphics[width=0.8\linewidth]{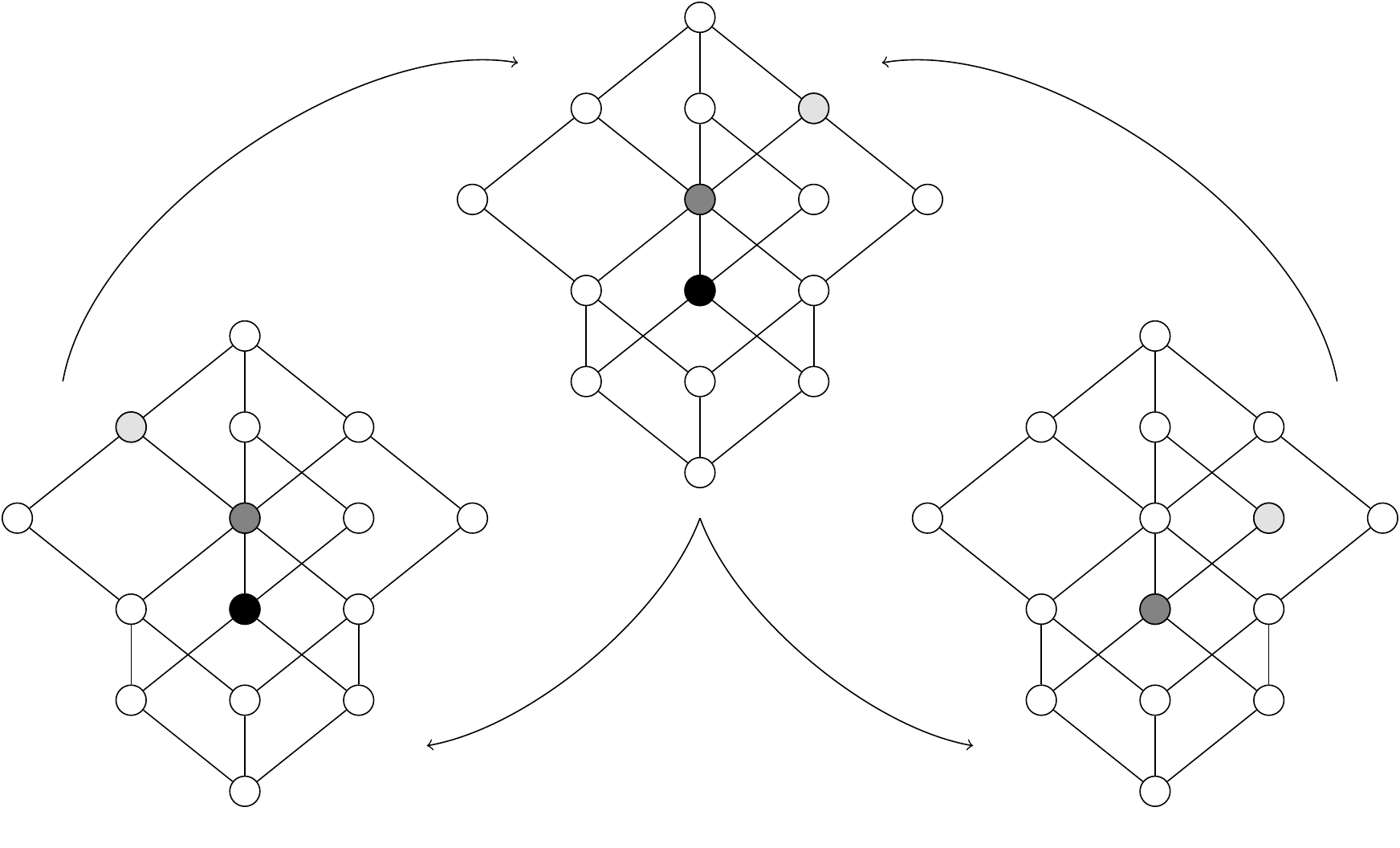}
    \caption{A meet-consensus algorithm.}
    \label{fig:consensus}
\end{figure}
The updates are defined in a recursion equation 
\begin{align*}
		x_i[t+1] &=& x_i[t] \meet \left( \bigmeet_{j \in \nbhd{i}} x_j[t] \right) \label{eq:meet-consensus}
\end{align*}
\begin{algorithm}{}
\DontPrintSemicolon
\SetKwFor{ForPar}{for}{do in parallel}{end}
\KwIn{Graph $\graph{G} = (\nodes{G}, \edges{G})$; $x_i \in \lattice{L}, i \in \nodes{G}$}
\KwOut{$x_i \in \lattice{L}, i \in \nodes{G}$}
\Repeat{$x_i = x_j~\forall~ij \in \edges{G}$}{
\ForPar{$i \in \nodes{G}$}{
\For{$j \in \nbhd{i}$}{
$x_i \leftarrow x_i \meet x_j$ \;
}}}
\caption{Meet-consensus; in join-consensus, replace $\meet$ with $\join$.} \label{alg:meet-consensus}
\end{algorithm}
\end{example}

\begin{proposition}
	Suppose $\graph{G}$ is connected and $\lattice{L}$ is a finite lattice. The meet-consensus algorithm \eqref{eq:meet-consensus} converges to  $\left( \bigmeet_{j \in \nodes{G}} x_j \right)_{i \in \nodes{G}}$ in finite time.
\end{proposition}
\begin{proof}
	After one round, $x_i$ is updated with the meet of $x_i$ with each of the $x_j$ for each neighbor $j \in \nbhd{i}$. Let $\nbhd{i}^k$ denote the $k$-hop neighbors of $i$. It follows that
	\begin{align*}
	x_i[k] &=& \bigmeet_{j \in \nbhd{i}^k \cup i} x_j[0]. 
	\end{align*}
	Because $\graph{G}$ is connected, for any $i \in \nodes{G}$ it is guaranteed that $\nbhd{i}^k \cup i \supseteq \nodes{G}$. Hence,
	\begin{align*}
	x_i[k] &=& \bigmeet_{j \in \nodes{G}} x_j[0]
	\end{align*}
	for $k \geq \mathrm{diam}(\graph{G})$.
\end{proof}

Consensus can easily be recast as a global section problem: consensus is the lattice of sections of the constant sheaf.

\begin{proposition}\label{prop:locally-constant}
Suppose $\underline{\lattice{L}}$ is the constant sheaf over $\graph{G}$. Then,
\begin{align*}
	\sections{\graph{G}; \underline{\lattice{L}}} &=& \{ (x_i)_{i \in \nodes{G}}~\vert~x_i = x_j~\forall~ij \in \edges{G}\}.
\end{align*}
\end{proposition}
\begin{proof}
	The global section condition dictates that $\id(x_i) = \id(x_j)$ for all $ij \in \edges{G}$.
\end{proof}

\noindent The lattice structure of $\sections{\graph{G}; \underline{\lattice{L}}}$ is ``the same'' as $\lattice{L}$ (under an isomorphism) but not equal.

% \begin{example}[Weighted Max-Consensus]
% 	Suppose $\mathbf{x}: \nodes{G} \to \Rext$ is a signal on $\graph{G}$ and $\weight: \edges{G} \to \Rext$ is a (symmetric) weight. If $\graph{G}$ is unweighted, then set $\weight_{ij} = 0$ for all $ij \in \edges{G}$. Let $S$ be the matrix
% 	\[
% 		[S]_{ij} =
% 		\begin{cases} 
% 			\weight_{ij} & ij \in \edges{E} \\
% 			-\infty & ij \not\in \edges{E}
% 		\end{cases}.
% 	\]
% 	Then, the updates 
% 	\begin{align*}
% 	\mathbf{x}[t+1] &=& S \ovee \mathbf{x} \\
% 	\mathbf{x}[0] &\in& \Rext^{|\nodes{G}}.
% 	\end{align*}
% 	converge for some $t \geq T$ \cite{nejad2009max}.  
% \end{example}

\begin{proposition} \label{prop:consensus-sections}
Suppose $\graph{G}$ is a graph, $\lattice{L}$ is a lattice, and
\begin{align}
        S &:& C^0(\graph{G}; \underline{\lattice{L}}) \to C^0(\graph{G}; \underline{\lattice{L}}) \\
	(S \mathbf{x})_i &=& \bigmeet_{j \in \nbhd{j}} x_j 
\end{align}
Then,
 \begin{align*}
 	\suffix(S) &=& \sections{\graph{G}; \underline{\lattice{L}}}.
 \end{align*}
\end{proposition}
% \noindent First, a lemma.
% \begin{lemma}\label{lem:fixed-suffix}
% 	Suppose $\poset{P}$ is a poset and $f: \poset{P} \to \poset{P}$ is order-preserving. Then,
% 	\begin{align*}
% 		\suffix(f) &=& \fixed(f \meet \id).
% 	\end{align*}
% 	Similarly,
% 	\begin{align*}
% 		\prefix(f) &=& \fixed(f \join \id).
% 	\end{align*}
% \end{lemma}
\begin{proof}
 One direction is trivial. By Proposition \ref{prop:locally-constant}, if $\mathbf{x} \in C^0(\graph{G}; \underline{\lattice{L}})$ is locally constant, then $(S \mathbf{x})_i = x_i$, hence, $S \mathbf{x} \succeq \mathbf{x}$. Now suppose $\mathbf{x} \in \suffix(S)$. Then, for all $i \in \nodes{G}$, $x_i = x_i \meet \bigmeet_{j \in \nbhd{i}} x_j$ which is equivalent to $x_i \preceq \bigmeet_{j \in \nbhd{i}} x_j$ for all $i \in \nodes{G}$.
Suppose $ij \in \edges{G}$ is a particular edge.
Then, both
\begin{align*}
	x_i &\preceq& \bigmeet_{j \in \nbhd{i}} x_j &\preceq& x_j, \\
	x_j &\preceq& \bigmeet_{i \in \nbhd{j}} x_i &\preceq& x_i.
\end{align*}
Hence, $x_i = x_j$ for all $ij \in \edges{G}$, by the anti-symmetry axiom of posets, and $\mathbf{x}$ is a global section of the constant sheaf.
 \end{proof}

Beyond meet-consensus (or, dually, join-consensus), are there other consensus algorithms? To answer this question, we first describe sheaves whose sections correspond to as approximations of the constant sheaf. The following definition is due to Hansen \cite[Chapter 7]{hansen2020laplacians}.

\begin{definition}[Approximation] \label{def:approx-constant}
Let $\graph{G}$ be a graph and $\sheaf{G}$ and $\sheaf{F}$ be Tarski sheaves over $\graph{G}$. Then, $\sheaf{F}$ is an approximation to $\sheaf{G}$ if there is there is a sheaf morphism $\sheaf{F} \xRightarrow{\phi} \sheaf{G}$ such that $\sheaf{F}(i) \cong \sheaf{F}(j)$ for all $i \in \nodes{G}$ and such that the induced join-preserving map
	\begin{align*}
		\sections{\graph{G}; \sheaf{F}} &\rightarrow& \sections{\graph{G}; \sheaf{G}}
	\end{align*}
	is an isomorphism.
\end{definition}

\noindent An approximation of the constant sheaf $\underline{\lattice{L}}$ consists of a sheaf whose vertex stalks are $\lattice{L}$ and whose lattice of sections is the same as $\sections{\graph{G}; \underline{\lattice{L}}}$, isomorphic to $\lattice{L}$ if $\graph{G}$ is connected. In other words, approximations to the constant sheaf also model consensus. An approximation of $\underline{\lattice{L}}$ is a sheaf $\sheaf{F}$ characterized by the following diagram
\[\begin{tikzcd}
	\sheaf{F}(i) \arrow[r, "\sheaf{F}(i \fc ij)"] & \sheaf{F}(ij) \arrow[d, "\phi_{ij}"] 		& \sheaf{F}(j) \arrow[l, swap, "\sheaf{F}(j \fc ij)"] \\
	\lattice{L}  \arrow[r, "\id"] \arrow[u, equals]			  & \lattice{L}			& \lattice{L} \arrow[l, swap, "\id"] \arrow[u, equals]
\end{tikzcd}\]. \label{eq:approximation-cd}

% \begin{theorem}[Approximations]
% 	Suppose $\lattice{L}$ is a complete lattice and $\sheaf{F}$ is a symmetric monosheaf over $\graph{G}$ such that $\sheaf{F}(i) = \lattice{L}$ for every node $i \in \nodes{G}$. Then, there is an approximation $\sheaf{F}$ of the constant sheaf $\underline{\lattice{L}}$. 
% \end{theorem}
% \begin{proof}
% Approximations of the constant sheaf $\phi$ are characterized locally by the diagram
% \[\begin{tikzcd}
% 	\sheaf{F}(i) \arrow[r, "\sheaf{F}(i \fc ij)"] & \sheaf{F}(ij) \arrow[d, "\phi_{ij}"] 		& \sheaf{F}(j) \arrow[l, swap, "\sheaf{F}(j \fc ij)"] \\
% 	\lattice{L}  \arrow[r, "\id"] \arrow[u, equals]			  & \lattice{L}			& \lattice{L} \arrow[l, swap, "\id"] \arrow[u, equals]
% \end{tikzcd}\]. \label{eq:approximation-cd}
% Choose
% \begin{align*}
% 	\phi_{ij} &=& \sheaf{F}(i \fc ij)^\ast &=& \sheaf{F}(j \fc ij)^{\ast}
% \end{align*}
% as our proposed approximation. Because $\sheaf{F}(i \fc ij)$ and $\sheaf{F}(j \fc ij)$ are monomorphisms (injective), we have
% \begin{align*}
% 	\sheaf{F}(i \fc ij)^{\ast} \circ \sheaf{F}(i \fc ij) &=& \id &=& \sheaf{F}(j \fc ij)^{\ast} \circ \sheaf{F}(j \fc ij)
% \end{align*}
% verifying that $\phi_{ij}$ is in fact an approximation to the constant sheaf.
% \end{proof}

There are more questions about lattice consensus than answers. If the consensus operator $S$ is a Laplacian for the constant sheaf, what is the ``sheaf Laplacian'' of an approximation to the constant sheaf? or an arbitrary Tarski sheaf? Furthermore, if $\sheaf{F}$ approximates $\underline{\lattice{L}}$ and $\chi: C^0(\sheaf{F}; \graph{G}) \to \lattice{L}$ is a consensus function, does there exist ``sheaf dynamics'' converging to the $\chi(\mathbf{x}[0])$?
% If $\chi$ is the consensus function \[\chi(\mathbf{x}) = \bigmeet_{i \in \nodes{G}} x_i\], called the \define{Pareto} consensus function \cite{monjardet2004lattices}, or the respective the join-projection, sometimes called the \define{co-Pareto} consensus function, the answer is in the affirmative.
There is also a utility for consensus ``in-between'' meet- and join- consensus. In some applications, for instance, earliest rendezvous \cite{nejad2009max}, min-consensus is desired. On the other hand, consider the task of forming teams amongst a set of agents. Cast as a consensus problem on the lattice of subpartitions, it would not be sensible to achieve either the finest or the coarsest organization possible.

\begin{example}[Causal Inference]
Suppose a set of random variables is collected in a partial order $\poset{P}$. For instance, random variables could indicate (transitive, reflexive, anti-symmetric) causal relationships between features. From join-consensus on the lattice of downsets $\mathcal{D}(\poset{P})$, a subset of $\poset{P}$ can be recovered from the poset of join-irreducibles of $\mathcal{D}(\poset{P})$ (isomorphic to $\poset{P}$, Theorem \ref{thm:birkhoff}), thus defining a weak consensus function $\chi: \poset{P}^n \to \powerset{\poset{P}}$.
% Birkhoff Duality (\ref{thm:birkhoff}) allows us to construct a lattice $\mathcal{D}(\poset{P})$ that represents the possible consensus problems. Now, the consensus operators $S$ and $S_{-}$ act on $C^0(\graph{G}; \underline{\mathcal{D}(\poset{P})}$. If the consensus is a join-irreducible, then the unique isomorphism $\mathcal{J}(\mathcal{D}(\poset{P})) \cong \poset{P}$ recovers one of the original labels. If the consensus is not join-irreducible, there fails to be a consensus, although other interpretations are possible in this scenario.
% \begin{figure}
% \begin{center}
% \scalebox{0.75}{\input{gfx/Chapter05-Weather}} \quad
% \scalebox{0.5}{\input{gfx/Chapter05-Weather-2}}
% \end{center}
% \caption{Partial order $\poset{P}$ on temperature labels (left); lattice of downsets $\mathcal{D}(\poset{P})$ (right).}\label{fig:weather}
% \end{figure}
\end{example}

\begin{example}[Distributed Optimization]
A multitude of agent-based tasks are formulated or equivalent to distributed optimization \cite{djuric2018cooperative}. Suppose $\{1,2, \dots, n\}$ is a group of agents with objective function $f_i,~i \in \{1,2,\dots, n\}$. In the simplest setting, \define{distributed optimzation} is unconstrained optimization problem
\begin{align}
    \min_{\mathbf{x}} \sum_{i = 1}^n f_i(\mathbf{x}) \label{eq:dist-optim}
\end{align}
Consensus and distributed optimization are intertwined through the observation that \eqref{eq:dist-optim} is equivalent to
\begin{align}
    \min_{\mathbf{x}_1, \mathbf{x}_2, \dots, \mathbf{x}_n} \sum_{i = 1}^n f_i(\mathbf{x}_i) \nonumber \\ 
    \text{subject to} \label{eq:dist-optim} \\ 
    \mathbf{x}_i = \mathbf{x}_j \quad \forall i, j \in \{0,1,\dots, n\} \nonumber
\end{align}
Distributed optimization, thus, reduces to consensus and optimization of the local objective functions $f_i, i \in \{1,2,\dots, n\}$.

Recent work generalizes distributed optimization to network sheaves over $\graph{G}$ with euclidean stalks \cite{hansen2019distributed}:
\begin{align}
	\min_{\mathbf{x} \in C^0(\graph{G}; \sheaf{F})} \sum_{i \in \nodes{G}} f_i(\mathbf{x}_i) \nonumber \\
	\text{subject to} \label{eq:sheaf-dist-op} \\
	\mathbf{x} \in  H^0(\graph{G}; \sheaf{F}) \nonumber.
\end{align}
where each $f_i$ is a convex function $f_i: \sheaf{F}(i) \to \R$.

Optimization algorithms maximizing \cite{nakashima2019subspace} or minimizing \cite{topkis1978minimizing} real-valued functions on lattices are restricted to certain classes of functions (e.g.~submodular, supermodular, increasing difference). An algorithm solving \eqref{eq:sheaf-dist-op} for $\sheaf{F}$, a Tarski sheaf, is an open problem whose solution, we speculate, will involve clever modulations between the tasks of running lattice consensus and performing ``gradient'' updates of the local objective functions $f_i, i \in \nodes{V}$.

% \begin{example}[Maximum Entropy Sampling]
% Let $\mathcal{X} = \{X_i\}_{i \in S}$ be a collection of random variables indexed by a finite set $S$. Let $\Sigma$ be the cooresponding covariance matrix with
% \begin{align*}
% 	\Sigma[i,j] &=& \Cov{(X_i,X_j)};
% \end{align*}
% if $A \subseteq S$, let $\Sigma_A$ be the submatrix of covariances. Then, the following map
% \begin{align*}
% 	H &:& \powerset{S} \to \R \\
% 	H(A) &=& \log \det \Sigma_A 
% \end{align*}
% is the \define{entropy} of a set of random variables $\{X_i\}_{i \in A \subseteq S}$. The optimization problem
% \begin{align*}
% 	\max_{A \in \powerset{S}} && H(A) \\
% 	\text{subject} && \text{to} \\
% 	|A| &\leq& k
% \end{align*}
% can be interpreted as selecting a $k$-element subset of a dataset containing the most representative information possible \cite{lee2004first}.
% \end{example}

% \begin{theorem}[Topkis's Theorem \cite{topkis1978minimizing}]
% 	Suppose $f: \lattice{L} \to \R$ is a submodular function on a lattice $\lattice{L}$. Then, the optimal solutions $\lattice{L}^{\ast} = \argmin f$ form a sublattice of $\lattice{L}$.
% \end{theorem}

\end{example}

%----------------------------
\section{Parallel Transport}
\label{sec:parallel}
%----------------------------

Schreiber and Waldorf \cite{schreiber2009parallel} and others \cite{krishnan2020invertibility} have argued parallel transport on a manifold $\Space{M}$ valued in an arbitrary category $\cat{D}$ is modeled by a functor
\begin{align*}
	P: \mathrm{Path}(\Space{M}) \to \cat{D}.
\end{align*}
from the path groupoid of $\Space{M}$. $\Path(\Space{M})$ is a category whose objects are points from $\Space{M}$ and whose morphisms are, roughly, homotopy classes of paths between points \cite{hardie2000homotopy}. In the sequel, given a graph $\graph{G}$ and a Tarski sheaf $\bisheaf{F}$ on $\graph{G}$, we construct a functor
\begin{align*}
	P^{\bisheaf{F}}: \Free(\graph{G}) \to \cat{Pos}
\end{align*} 
sending a path $ \gamma: i \to j$ to a order-preserving map $\Parallel{F}{\gamma}: \bisheaf{F}(i) \to \bisheaf{F}(j)$ simulating parallel transport. In our view, parallel transport models message-passing of local sections of agents connected by a path in a communication network.

\begin{definition}[Parallel Transport]
Suppose $\graph{G}$ is a graph and $\bisheaf{F}$ is a Tarski sheaf over $\graph{G}$. Suppose $j \in \nbhd{i}$. Then, the \define{parallel transport} along the path $i\to j$ is the map
\begin{align*}
    \Parallel{F}{i \to j} &:& \bisheaf{F}(i) \to \bisheaf{F}(j) \\
    \Parallel{F}{i \to j}(x_i) &=& \cosheaf{F}(j \fc ij)\sheaf{F}(i \fc ij)(x_i).
\end{align*}
If $\gamma$ is an arbitrary path from $i$ to $j$, then $\Parallel{F}{\gamma}$ is defined as the (unique) composition of parallel transport maps along paths of length one. The parallel transport of the trivial path $\epsilon_i$ based at $i \in \nodes{G}$ is defined to be the identity map on $\bisheaf{F}(i)$.
\end{definition}

Parallel transport is functorial by construction (i.e.~$\Parallel{F}{\gamma_1 \cdot \gamma_2} = \Parallel{F}{\gamma_2}\circ \Parallel{F}{\gamma_1 }$). Furthermore, parallel transport is a functor into the catgory of posets and order-preserving maps because $\Parallel{F}{i \to j}$, as a composition of an join- and meet- preserving map, is not (in general) join- or meet- preserving. If $\gamma$ and $\gamma'$ are two distinct paths from $i$ to $j$, in is not the case (in general) that
\begin{align*}
	\Parallel{F}{\gamma}(x_i) = \Parallel{F}{\gamma}(x_j) \quad \forall i, j \in \nodes{G}.
\end{align*}
Whenever the above holds, $\bisheaf{F}$ is said to be flat.

%-----------------------
\subsection{Holonomy}
%-----------------------

Having defined parallel transport of Tarski sheaves, an abundance of potential geometric construction wait at the door including connections, curvature, characteristic classes, and more. Immediately, we are in a position to define holonomy, a construction which, roughly speaking, measures the degree to which parallel transport around a closed loop fails to preserve orientations of tangent vectors.

\begin{definition}[Holonomy]
Suppose $\graph{G}$ is a graph. Let $\Omega_i(\graph{G})$ denote the loops in $\graph{G}$ based at at particular $i \in \nodes{G}$. Then, the \define{holonomy} of $\bisheaf{F}$ is the ordered monoid
\begin{align*}
		\Hol_i(\graph{G};\sheaf{F}) &=& \{ \Parallel{F}{\gamma}~\vert~\gamma \in \Omega_i(\graph{G}) \}
\end{align*}
\end{definition}

Holonomy is a monoid because the composition of loops yields a composition of transport maps and the trivial path $\epsilon_i$ yields an identity. It is a poset because $\Hol_i(\graph{G}; \sheaf{F})$ inherits the order of the poset of order-preserving maps $\End_{\cat{Pos}}(\sheaf{F}(i))$.

%----------------------------------
\subsection{Global sections}
%----------------------------------

Parallel transport characterizes sections. We first show that sections facilitate the transport of ``tangent vectors'' (lattice elements) to compatible "tangent vectors" (lattice elements) along any (network) path.

\begin{lemma}[Section Path Lemma] \label{lem:section-path}
	Suppose $\bisheaf{F}$ is a Tarski sheaf over $\graph{G}$ and $\mathbf{x} \in \sections{\graph{G}; \sheaf{F}}$ is a section. Then, for every path $\gamma$ from $i$ to $j$,
	\begin{align*}
		\Parallel{F}{\gamma}(x_i) &\succeq& x_j.
	\end{align*}
\end{lemma}
\begin{proof}
    We proceed by induction on path length:
    
    \textit{Base Case.} Suppose $\gamma$ is the path $i \to j$ (of length one. Then, as $\mathbf{x}$ is a section,
    \begin{align*}
        \sheaf{F}(i \fc ij)(x_i) &\succeq& \sheaf{F}(j \fc ij)(x_j)
    \end{align*}
    implies
    \begin{align*}
        \cosheaf{F}(j \fc ij)\sheaf{F}(i \fc ij)(x_i) &\succeq& x_j.
    \end{align*}

    \textit{Inductive Step.} Suppose $\gamma$ is a path of length $\ell$. Suppose, without loss of generality, that $\gamma$ is a path from $i$ to $j$, and $\gamma'$ is a path from the same $i$ to $j'$. By the inductive hypothesis, $\Parallel{F}{\gamma}(x_i) \succeq x_j$. Then,
    \begin{align*}
        \Parallel{F}{\gamma'}(x_i) &\succeq&  \bigmeet_{j' \in \nbhd{j}} \Parallel{F}{j \to j'}\Parallel{F}{\gamma}(x_i) \\
        &\succeq& \bigmeet_{j' \in \nbhd{j}} \Parallel{F}{j \to j'}(x_j) \\
        &\succeq& x_{j'}.
    \end{align*}     
\end{proof}

\noindent For path of length one, \define{neighborhood paths}, Lemma \ref{lem:section-path} has a partial converse.

\begin{lemma}[Neighborhood Path Lemma] \label{lem:neighborhood}
	Suppose $\bisheaf{F}$ is a Tarski sheaf over $\graph{G}$, and supoose 
\begin{align*}
\Parallel{F}{i \to j} (x_i) &\succeq& x_j, \quad i \in \nodes{G}, j \in \nbhd{i}
\end{align*}
for all $\mathbf{x} \in C^0(\graph{G}; \sheaf{F})$. Then, $\mathbf{x} \in \sections{\graph{G}; \sheaf{F}}$.
\end{lemma}
\begin{proof}
	By assumption, both of the following hold for an $\mathbf{x} \in C^0(\graph{G}; \sheaf{F})$ over all $ij \in \edges{G}$
	\begin{align*}
		\cosheaf{F}(j \fc ij)\sheaf{F}(i \fc ij)(x_i)  &\succeq& x_j, \\
            \cosheaf{F}(i \fc ij)\sheaf{F}(j \fc ij)(x_j) &\succeq& x_i.
        \end{align*}
    Therefore, because by the property of Galois connection,
	\begin{align*}
		\sheaf{F}(j \fc ij)(x_j) &\succeq& \sheaf{F}(i \fc ij)(x_i), \\
		\sheaf{F}(i \fc ij)(x_i) &\succeq& \sheaf{F}(j \fc ij)(x_j),
	\end{align*}
 which implies equality.
\end{proof}
\begin{remark}[Proof of Lemma \ref{lem:neighborhood}]
    The argument in the proof of the Neighborhood Lemma breaks downs when you try to generalize it to arbitrary paths. For instance, if $\gamma = (i, j, k)$ is a path between distinct nodes, then
	\begin{align*}
		\Parallel{F}{\gamma} (x_i) &=& \cosheaf{F}(k \fc jk)\sheaf{F}(j \fc jk)\cosheaf{F}(j \fc ij)\sheaf{F}(i \fc ij)(x_i) &\succeq& x_k
	\end{align*}
	does not imply that 
	\begin{align*}
		\Parallel{F}{\gamma} (x_k) &=& \cosheaf{F}(i \fc ij)\sheaf{F}(j \fc ij)\cosheaf{F}(j \fc jk)\sheaf{F}(k \fc jk)(x_k) &\succeq& x_i,
	\end{align*}
	only that
	\begin{align*}
		\sheaf{F}(j \fc jk)\cosheaf{F}(j \fc ij)\sheaf{F}(i \fc ij)(x_i) &\succeq& \sheaf{F}(k \fc jk)(x_k).
	\end{align*}
\end{remark}

%--------------------------------
\section{The Hodge-Tarski Theorem}
%--------------------------------

Finally, in this section, we define a sheaf Laplacian for Tarski sheaves culminating in a Hodge-style fixed point theorem computing global sections.

%----------------------------------
\subsection{The Tarski Laplacian}
%----------------------------------

% Motivated by the graph connection Laplacian \cite{singer2012vector}, which coincides with the graph Laplacian in dimension $0$, we define a connection Laplacian which we call the \define{Tarski Laplacian}.

\begin{definition}
	Suppose $\bisheaf{F}$ is a Tarski sheaf over $\graph{G}$. The \define{Tarski Laplacian} is an operator defined
	\begin{align*}
		\Laplacian &:& C^0(\graph{G}; \sheaf{F}) \longrightarrow C^0(\graph{G}; \sheaf{F}) \\
		\left(\Laplacian \mathbf{x} \right)_i &=& \bigmeet_{j \in \nbhd{i}} \cosheaf{F}(i \fc ij)\sheaf{F}(j \fc ij)(x_j) .
	\end{align*}
\end{definition}

Recall, the \define{meet projection} of a lattice $\lattice{L}$ onto a subset $N \subseteq 
 \{1,2,\dots, n\}$ is the map
 \begin{align*}
     \chi &:& \lattice{L}^n \to \lattice{L} \\
     \chi(\mathbf{x}) &=& \bigmeet_{i \in N} x_j
 \end{align*}
Then, we note immediately, the Tarski Laplacian is a composition of two operations: (1) parallel transport, and (2) meet-projection. Therefore,
\begin{equation}
    \left(\Laplacian \mathbf{x} \right)_i = \bigmeet_{j \in \nbhd{i}} \Parallel{F}{j \to i}(x_j). \label{eq:tarski-lalacian}
\end{equation}

\noindent Also, recall, the conditions of the Tarski Fixed Point Theorem \ref{thm:tfpt}.

\begin{lemma}\label{lem:tarski-mono}
    The Tarski Laplacian is an order-preserving map on a complete lattice.
\end{lemma}
\begin{proof}
    By a similar argument that lattice polynomials are order-preserving \cite{birkhoff1940lattice}, $\Laplacian$ is a composition of order-preserving function. $C^0(\graph{G}; \sheaf{F})$ is the product of complete lattices, hence, complete.
\end{proof}

%-----------------------
\subsection{Diffusion}
%-----------------------

For the following proposition, let $\nbhd{i}^k$ denote the \define{$k$-hop neighborhood} of $i \in \nodes{G}$.\footnote{Let $\nbhd{i}^0 = \{i\}$.} Suppose $W \subseteq \nodes{G}$ is a subset of nodes. Then, $\pi_W$ is the  projection of $C^0(\graph{G}; \sheaf{F})$ onto  $W$.

\begin{proposition}[Locality] \label{prop:locality}
    Suppose $\bisheaf{F}$ is a Tarski sheaf over $\graph{G}$ with Tarski Laplacian $\Laplacian$. Suppose $i \in \nodes{G}$. Then, the Tarski Laplacin is restricted to the complement of the $k$-hope neighborhood
    \begin{align*}
   \Laplacian^k_{\nodes{G} \setminus \nbhd{i}^k} &:& \prod_{j \in \nodes{G} \setminus \nbhd{i}^k} \sheaf{F}(j) \to \prod_{j \in \nodes{G} \setminus \nbhd{i}^k} \sheaf{F}(j)   \end{align*}
    is the identity.
    % where $\pi_{\nodes{G} \setminus \nbhd{i}^k}$ and $\pi_{\nodes{G} \setminus \nbhd{j}^k}$ are the projections of $C^0(\graph{G}; \sheaf{F})$ onto $\prod_{i' \in \nodes{G} \setminus \nbhd{i}^k} \bisheaf{F}(i')$ and $\prod_{j' \in \nodes{G} \setminus \nbhd{j}^k} \bisheaf{F}(j')$, respectively.
\end{proposition}

As an operator, the Tarski Laplacian is said to be \define{local} in the sense that $(\Laplacian^k \mathbf{x})_i$ only depends only on the assignment to the nodes $j \in \nbhd{i}^k$. Familiar examples of other local operators include the graph Laplacian, graph connection Laplacian, and the linear sheaf Laplacian.

%-----------------------------------------------------
\subsection{Similarities to the connection Laplacian}
%-----------------------------------------------------

Recall, in the connection Laplacian,
\begin{align*}
	(\graphLaplacian \mathbf{x})_i &=& \frac{1}{d_i} \sum_{j \in \nbhd{i}} w_{ij} O_{ij} x_j,
\end{align*}
tangent vectors of neighbors $\{\mathbf{x}_j \}_{j \in \nbhd{i}}$ are transformed by orthogonal maps $O_{ij}$ approximating parallel transport between $T_{x_i} \Space{M}$ and $T_{x_i} \Space{M}$. Next, node messages $\{ O_{ij} \mathbf{x}_j \}_{j \in \nbhd{i}}$ are received and aggregated with via a weighted sum.

In the Tarski Laplacian \eqref{eq:tarski-lalacian}, lattice elements of neighbors $\{ x_j \}_{j \in \nbhd{i}}$ are transformed by order-preserving maps $\Parallel{F}{j \to i}$ subordinate to $\bisheaf{F}$ that model parallel transport . Next, node messages $\{ \Parallel{F}{j \to i}(x_j) \}_{j \in \nbhd{i}}$ are received and aggregated with a  meet-projection onto $\nbhd{j}$.

%----------------------------------
\subsection{A fixed point theorem}
%----------------------------------

Recall, the Hodge Theorem (Theorem \ref{thm:hodge}) characterizes the fixed points of heat equation \eqref{eq:hodge-heat-flow}. We now state a culminating result characterising the suffix points of the Tarski Laplacian with the sections $\sections{\graph{G}; \sheaf{F}}$ of a Tarski sheaf $\sheaf{F}$ over $\graph{G}$.

\begin{theorem}[Tarski-Hodge Theorem] \label{thm:main}
	Suppose $\bisheaf{F}$ is a Tarski sheaf over a graph $\graph{G}$ with Tarski Laplacian $\Laplacian: C^0(\graph{G}; \sheaf{F}) \to C^0(\graph{G}; \sheaf{F})$. Then,
	\begin{align}
		\suffix(\Laplacian) &=& H^0(\graph{G}; \sheaf{F}).
	\end{align}
\end{theorem}
\noindent The proof relies on the theory of parallel transport from Section \ref{sec:parallel}.
% \noindent The proof relies on the theory of parallel transport we developed in Section \ref{sec:parallel} and the following lemma.
% \begin{lemma} \label{lem:main}
% 	Suppose $\mathbf{x} \in C^0(\graph{G}; \sheaf{F})$. Then,
% 	\begin{align*}
% 		\mathbf{x} \in \sections{\graph{G}; \sheaf{F}}
% 	\end{align*}
% 	 if and only if
% 	 \begin{align*}
% 		\Parallel{F}{j \to i} (x_j) &\succeq& x_i.
% 	\end{align*}
% 	for all $i \in \nodes{G}$, $j \in \nbhd{i}$.
% \end{lemma}
% \begin{proof}[Proof of Lemma \ref{lem:main}]
% 	Follows from the Section Path Lemma (Lemma \ref{lem:section-path}) and the Neighborhood Lemma (\ref{lem:neighborhood}).
% \end{proof}
\begin{proof}[Proof of Tarski-Hodge Theorem]
By the Section Path Lemma (Lemma \ref{lem:section-path}) and the Neighborhood Lemma (Lemma \ref{lem:neighborhood}),
	\begin{align*}
		\mathbf{x} \in \sections{\graph{G}; \sheaf{F}}
	\end{align*}
	 if and only if
	 \begin{align*}
		\Parallel{F}{j \to i} (x_j) &\succeq& x_i
	\end{align*}
	holds for all $i \in \nodes{G}$, $j \in \nbhd{i}$. By the universal property of meets (Lemma \ref{lem:convenience}), 
	\begin{align*}
		\Parallel{F}{j \to i} (x_j) &\succeq& x_i
	\end{align*}
 for all $j \in \nbhd{i}$
 if and only if
	\begin{align*}
		\bigmeet_{j \in \nbhd{i}} \Parallel{F}{j \to i} (x_j) &\succeq& x_i.
	\end{align*}
 Therefore, $(\Laplacian \mathbf{x})_i \succeq x_i$ for every $i \in \nodes{V}$ precisely if $\mathbf{x} \in H^0(\graph{G}; \sheaf{F})$. 
\end{proof}
\noindent Equivalently, we may rewrite the above as the fixed points of the operator $\Laplacian \meet \id$ because $\Laplacian(\mathbf{x}) \succeq \mathbf{x}$ if and only if $\Laplacian (\mathbf{x}) \meet \mathbf{x}$.
\begin{corollary}
    $\fixed(\Laplacian \meet \id) = H^0(\graph{G}; \sheaf{F})$.
\end{corollary}

\begin{remark}
Indeed, the Hodge-Tarski Theorem (Theorem \ref{thm:main}) is a ``Hodge-style fied point theorem.'' The Euler discretization of the heat equation \eqref{eq:hodge-heat-flow} (step-size $\alpha>0$) is the following
\begin{align}
    \mathbf{x}[t+1] - \mathbf{x}[t] &=& - \alpha\Delta_k \mathbf{x} \label{eq:fixed-hodge}
\end{align}
It follows
\begin{align*}
    \fixed(\id - \alpha \Delta) &\cong& H^k(C^\bullet)
\end{align*}
by the isomorphism between $\Ker(\Delta_K)$ and $H^0(C^\bullet)$ in the Hodge Theorem (Theorem \ref{thm:hodge}). Then, compare \eqref{eq:fixed-hodge} to
\begin{align}
    \fixed(\id \meet \Laplacian) &=& H^k(\graph{G}; \sheaf{F}).
\end{align}
\end{remark}

\noindent A second (also obvious) corollary is a consequence of the isomorphism between global sections and $H^0$ (Proposition \ref{thm:section-equalizer}).
\begin{corollary} \label{cor:fixed-point}
$\fixed(\Laplacian \meet \id) \cong \sections{\graph{G}; \sheaf{F}}$.
\end{corollary}

\noindent A final corollary offers a third an alternate proof to Theorem \ref{thm:sections-sublattice} that the global sections of a Tarski sheaf are a complete quasi-sublattice of $\sections{\graph{G}; \sheaf{F}}$.

\begin{corollary} \label{cor:h0sections}
	$H^0(\graph{G}; \sheaf{F})$ is a complete lattice. Furthermore, there is an order-embedding  of $\sections{\graph{G}; \sheaf{F}}$ into the complete lattice $C^0(\graph{G}; \sheaf{F})$ of $0$-cochains.
\end{corollary}
\begin{proof}
The Tarski Laplacian is an order-preserving map on a complete lattice $C^0(\graph{G}; \sheaf{F})$ (Proposition \ref{lem:tarski-mono}). Therefore, by Corollary \ref{cor:h0sections} and the Tarski Fixed Point (Theorem \ref{thm:tfpt}), the isomorphism between $\sections{\graph{G}; \sheaf{F}}$ and $H^0(\graph{G}; \sheaf{F})$ followed by the inclusion of $H^0(\graph{G}; \sheaf{F})$ into $C^0(\graph{G}; \sheaf{F})$ is an order-embedding. 
\end{proof}

%---------------------------------
\section{Towards Operator Theory}
%---------------------------------

By Corollary \ref{cor:h0sections}, global sections of Tarski sheaves have arbitrary meets and joins, forming a complete lattice. In the context of multi-agent systems, this implies the set of consistent assignments of lattice-valued data over a network of agents forms a lattice and results in two merging operations (meet/join) between consistent assignments, sections. In this section, we introduce two new operators: the closure and the Helmholtzian. We discuss a strategy for computing the explicit lattice structure of $H^0(\graph{G}; \sheaf{F})$ of a Tarski sheaf.

%-----------------------------------
\subsection{A closure operator}
%-----------------------------------

We introduce another local operator on $C^0(\graph{G}; \sheaf{F})$ that is \define{not} a Laplacian.

\begin{definition}[Closure]
	The \define{closure} is the endomorphism
	\begin{align}
		\Closure &:& C^0(\graph{G}; \sheaf{F}) \rightarrow C^0(\graph{G}; \sheaf{F}) \\
		\left( \Closure  \mathbf{x}\right)_i &=& \bigmeet_{j \in \nbhd{i}} \cosheaf{F}(i \fc ij)\sheaf{F}(i \fc ij)(x_i).
	\end{align}
\end{definition}

\noindent Recall, if $(\ladj{f}, \radj{f})$ is a Galois connection, then $\radj{f}\ladj{f} \succeq \id$  (Proposition \ref{prop:monad-comonad}). The closure, then, ``intensifies'' local data. To be precise, $\Closure$ is a closure operator on the complete lattice $C^0(\graph{G}; \sheaf{F})$.

\begin{theorem}[Closure Theorem]\label{thm:exander-closure}
	The operator $\Closure$ is a closure operator on $C^0(\graph{G}; \sheaf{F})$.
\end{theorem}
\begin{proof}
    It is trivial to check $\Closure$ is order-preserving. By Proposition \ref{prop:monad-comonad} and the universal property  of meets (Lemma \ref{lem:convenience}), $(\Closure\mathbf{x})_i \succeq x_i$ for all $i \in \nodes{G}$. Hence, $\Closure$ is inflationary. To show $\Closure$ is idempotent, we show $\Closure^2(\mathbf{x}) \preceq \mathbf{x}$ and $\Closure^2(\mathbf{x}) \succeq \mathbf{x}$. For the first inequality,
	\begin{align*}
		(\Closure^2 \mathbf{x})_i &=& \bigmeet_{j \in \nbhd{i}} \cosheaf{F}(i \fc ij)\sheaf{F}(i \fc ij) \bigl( \bigmeet_{j \in \nbhd{i}} \cosheaf{F}(i \fc ij)\sheaf{F}(i \fc ij) (x_i) \Bigl) \\
		&\preceq& \bigmeet_{j \in \nbhd{i}} \cosheaf{F}(i \fc ij)\sheaf{F}(i \fc ij) \cosheaf{F}(i \fc ij)\sheaf{F}(i \fc ij)(x_i) \\
		&=& \bigmeet_{j \in \nbhd{i}} \cosheaf{F}(i \fc ij)\sheaf{F}(i \fc ij) (x_i) \\
		&=& (\Closure \mathbf{x})_i,
	\end{align*}
	where the penultimate equality follows from Lemma \ref{prop:left-right-left}. For the second inequality, 
	\begin{align*}
		(\Closure^2 \mathbf{x})_i &=& \bigmeet_{j \in \nbhd{i}} \cosheaf{F}(i \fc ij)\sheaf{F}(i \fc ij) \bigmeet_{j \in \nbhd{i}} \cosheaf{F}(i \fc ij)\sheaf{F}(i \fc ij) (x_i) \\
		&\succeq& \bigmeet_{j \in \nbhd{i}} \cosheaf{F}(i \fc ij)\sheaf{F}(i \fc ij)(x_i)
	\end{align*}
	by $\Closure \mathbf{x} \succeq \mathbf{x}$ and monotonicity. Hence, $\Closure^2 = \Closure$.
\end{proof}

\noindent The fixed points of the closure operator $E$ is a quasi-sublattice of $C^0(\graph{G}; \sheaf{F})$.
\begin{corollary}
	$\fixed(\Closure)$ is a complete lattice.
\end{corollary}
\begin{proof}
    Apply Theorem \ref{thm:closure-fixed}.
\end{proof}
\noindent Moreover, we can explicitly compute meets and joins with
\begin{align*}
	\mathbf{x} \meet_{\fixed(E)} \mathbf{y} &=& \mathbf{x} \meet \mathbf{y}, \\
	\mathbf{x} \join_{\fixed(E)} \mathbf{y} &=& E(\mathbf{x} \join \mathbf{y}).
\end{align*}

\noindent The following establishes a connection between the closure and the Tarski Laplacian.
\begin{lemma}\label{lem:closure-laplacian}
    Suppose $\mathbf{x} \in H^0(\graph{G}; \sheaf{F})$. Then, $\Laplacian(\mathbf{x}) = \Closure(\mathbf{x})$.
\end{lemma}
\begin{proof}
    Let $\mathbf{x} \in H^0(\graph{G}; \sheaf{F})$. Then, for all $i \in \nodes{G}$,
    \begin{align*}
        (\Laplacian \mathbf{x})_i &=& \bigmeet_{j \in \nbhd{i}} \cosheaf{F}(i \fc ij)\sheaf{F}(j \fc ij)(x_j) \\ &=& 
        \bigmeet_{j \in \nbhd{i}} \cosheaf{F}(i \fc ij)\sheaf{F}(i \fc ij)(x_i)
    \end{align*}
    as $\sheaf{F}(j \fc ij)(x_j) = \sheaf{F}(i \fc ij)(x_i)$ for every $ij \in \edges{G}$.
\end{proof}

%-----------------------------------------------------------
\subsection{Meets \& joins of sections} 
%-----------------------------------------------------------

Lemma \ref{lem:closure-laplacian} suggests that one may attempt to restrict $\Laplacian$ to $H^0(\graph{G}; \sheaf{F})$ to obtain a closure operator. The following shows that $H^0(\graph{G}; \sheaf{F})$ is invariant under $\Laplacian$. A similar result was shown by Cunninghame-Green \& Butkovic \cite{cuninghame2003equation} in a highly-specialized case (Figure \ref{fig:cunninghame}) of max-plus restriction maps over the network with two nodes and one edge.
\begin{lemma}\label{lem:cunninghame}
    Suppose $\mathbf{x} \in H^0(\graph{G}; \sheaf{F})$. Then, $\Laplacian(\mathbf{x}) \in H^0(\graph{G}; \sheaf{F})$.
\end{lemma}
\begin{proof}
    By the Hodge-Tarski Theorem (Theorem \ref{thm:main}), it suffices to show that $\Laplacian(\mathbf{x}) \succeq \mathbf{x}$ which follows by monotonicity of $\Laplacian$ and transitivity.
\end{proof}

\begin{warning}
	It would be tempting to conclude that $\Laplacian \vert_{H^0(\graph{G}; \sheaf{F})}: H^0(\graph{G}; \sheaf{F}) \to H^0(\graph{G}; \sheaf{F})$ is a closure operator, and, then, one may compute meets and joints in $H^0(\graph{G}; \sheaf{F})$ explicitly by Theorem \ref{thm:closure-fixed}. However, this argument is invalid because, in the definition of the Tarski Laplacian, the product lattice structure on $C^0(\graph{G}; \sheaf{F})$ does not, in general, coincide with the lattice structure on $H^0(\graph{G}; \sheaf{F})$. Hence, this argument is circular.
\end{warning}

\subsection{The Tarski Helmholtzian}
%------------------------------------

The Tarski Laplacian acts on $C^0(\graph{G}; \sheaf{F})$. We define another operator, analogous to the graph Helmholtzian \cite{lim2020hodge} which acts on $C^1(\graph{G}; \sheaf{F})$.

\begin{definition}[Helmholtzian]
	Suppose $\bisheaf{F}$ is a Tarski sheaf over $\graph{G}$. The \define{Tarski Helmholtzian} is the operator $\Helmholtz : C^1(\graph{G}; \cosheaf{F}) \to C^1(\graph{G}; \cosheaf{F})$ defined
	\begin{align*}
		(\Helmholtz \mathbf{x})_{ij} &=& \left( \bigjoin_{j' \in \nbhd{i} - j} \sheaf{F}(i \fc ij)\cosheaf{F}(i \fc ij')(x_{ij'}) \right) \join \left( \bigjoin_{i' \in \nbhd{j} - i} \sheaf{F}(j \fc ij)\cosheaf{F}(j \fc i'j)(x_{i'j})  \right)
	\end{align*}
\end{definition}

\noindent We conjecture that the lattice $H^1(\graph{G}; \sheaf{F})$ constructed as a coequalizer \eqref{eq:section-equalizer} coincides with the prefix points of the Tarski Helmholtzian.

\begin{conjecture}
    Suppose $\bisheaf{F}$ is a Tarski sheaf over $\graph{G}$ with Tarski-Helmholtz operator $\Helmholtz: C_1(\graph{G}; \sheaf{F}) \to C_1(\graph{G}; \sheaf{F})$. Then,
	\begin{align}
		\prefix(\Helmholtz) &=& H^1(\graph{G}; \sheaf{F}).
	\end{align}
\end{conjecture}

\section{Algorithms}
%--------------------

We end this chapter with a discussion of two different algorithms that are inspired by the Tarski Laplacian.

%----------------------
\subsection{Heat Flow}
%----------------------
The Tarski Laplacian inspires a discrete-time heat equation on $C^0(\graph{G}; \sheaf{F})$:

\begin{align}
	\mathbf{x}[t+1] &=& (\Laplacian \meet \id) \mathbf{x}[t] \label{eq:heat-equation}
\end{align}
 with initial condition $\mathbf{x}[0] \in C^0(\graph{G}; \sheaf{F})$.
Locally, these updates can be cast as a distributed algorithm (Algorithm \ref{alg:heat-flow}). In order to guarantee the algorithm terminates, we require an assumption.

\begin{assumption}(DCC) \label{as:dcc}
	The product lattice $C^0(\sheaf{F}; \graph{G})$ satisfies the descending chain condition (Definition \ref{def:dcc}).
\end{assumption}

\begin{algorithm}[h]
\caption{Heat Flow} \label{alg:heat-flow}
\DontPrintSemicolon
\SetKwFor{ForPar}{for}{do in parallel}{end}
\KwIn{Bisheaf $\bisheaf{F}$ over graph $\graph{G} = (\nodes{G}, \edges{G})$;  initial $x_i \in \sheaf{F}(i), i \in \nodes{G}$}
\KwOut{$x_i \in \sheaf{F}(i), i \in \nodes{G}$}
\Repeat{$\sheaf{F}(i \fc ij)(x_i) = \sheaf{F}(j \fc ij)(x_j)~\forall~ij \in \edges{G}$}{
\ForPar{$i \in \nodes{G}$}{
\For{$j \in \nbhd{i}$}{
$x_i \leftarrow x_i \meet \cosheaf{F}(i \fc ij)\sheaf{F}(j \fc ij)(x_j)$ \;
}}}
\end{algorithm}

\begin{proposition}[Completeness] \label{thm:alg-heat}
	Suppose $\bisheaf{F}$ is a Tarski sheaf over $\graph{G}$ satisfying Assumption \ref{as:dcc}. Then, Algorithm \ref{alg:heat-flow} converges in finitely many iterations.
\end{proposition}
\begin{proof}
	$\Laplacian \meet \id \preceq \id$ (trivially) and the recursion $\mathbf{x}[t] = (\Laplacian \meet \id)^t\mathbf{x}[0]$ imply $\mathbf{x}[s] \preceq \mathbf{x}[t]$ whenever $s \leq t$. Hence, $\{ \mathbf{x}[t] \}_{t \geq 0}$ forms a descending chain in $C^0(\graph{G}; \sheaf{F})$.
\end{proof}

\begin{corollary}
	Suppose $\bisheaf{F}$ is a Tarski sheaf over $\graph{G}$ satisfying Assumption \ref{as:dcc}, then Algorithm \ref{alg:heat-flow} converges to some $\mathbf{x} \in H^0(\graph{G}; \sheaf{F})$.
\end{corollary}
\begin{proof}
    By Proposition \ref{thm:alg-heat}, the Algorithm \ref{alg:heat-flow} converges. Hence, there is a $t_0$ such that $(\Laplacian \meet \id)^t \mathbf{x}[t] = (\Laplacian \meet \id)^t \mathbf{x}[t_0]$ for all $t \geq t_0$. Therefore, by the Tarski-Hodge Theorem,  $(\Laplacian \meet \id)^t \mathbf{x}[t_0] \in H^0(\graph{G}; \sheaf{F})$.
\end{proof}

\noindent If a poset is graded, then it satisfies the descending chain condition.

\begin{corollary}
	Suppose $C^0(\graph{G}; \sheaf{F})$ is graded. Then, Algorithm \ref{alg:heat-flow} also terminates in finite time and converges to a section.
\end{corollary}

\noindent The heat flow alogithm (Algorithm \ref{alg:heat-flow}) is depicted visually (Figure \ref{fig:heat-flow}).

\begin{figure}[h]
    \centering
    \includegraphics[width=0.6\textwidth]{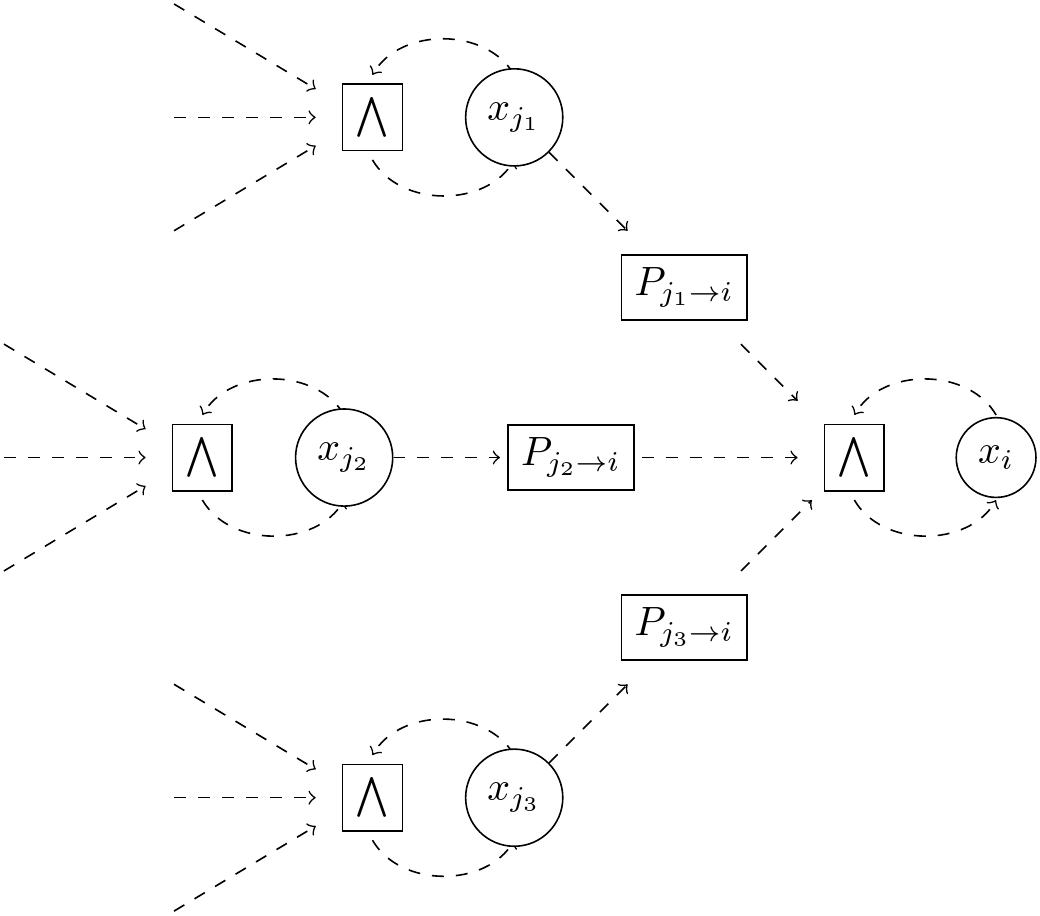}
    \caption{Heat flow on a Tarski Sheaf.}
    \label{fig:heat-flow}
\end{figure}

%--------------------
\subsection{Gossip}
%--------------------
Gosip algorithms are a general class of decentralized algorithms that are broadly characterized as having each exchange information with one or a few neighboring nodes each round. The advantages of gossip over Laplacian-based methods include fault tolerance and a relaxation of assumption that nodes update their states synchronously \cite{kempe2003gossip}. In order to describe an asynchronous version of the Heat Flow algorithm, we will need some notation. Let $\tau: \N \to \powerset{\nodes{G}}$ be a \define{broadcast sequence}. At time $t$, the nodes $\tau(t) \subseteq \nodes{G}$ broadcast to all neighbors $j \in \nbhd{i}$. This is quite general. For instance, if $\tau(t) = \nodes{G}$, then we recover Algorithm \ref{alg:heat-flow}. However, if $\tau(t)$ is a single node, then we have a truly asynchronous algorithm. Additionally, we require nodes to continually exchange information.

\begin{assumption}(Liveness) \label{as:liveness}
	For every $i \in \nodes{G}$ and $t_0 \in \N$, there exist $t \geq t_0$ such that $i \in \tau(t)$. 
\end{assumption}

\noindent The Heat Flow algorithm (Algorithm \ref{alg:heat-flow}) is modified (Algorithm \ref{alg:gossip}) by aggregating parallel transported lattice-valued information from only the nodes that are broadcasting at a given time.

\begin{algorithm}[h]
\SetKwFor{ForTo}{to}{do}{end}
\caption{Gossip}\label{alg:gossip}
\DontPrintSemicolon
\KwIn{Bisheaf $\bisheaf{F}$ over graph $\graph{G} = (\nodes{G}, \edges{G})$;  initial $x_i \in \bisheaf{F}(i), i \in \nodes{G}$; broadcast sequence $\tau: \N \to \powerset{\nodes{G}}$.}
\For{$t \in \N$}{\For{$i \in \tau(t)$}{\For{$j \in \nbhd{i}$}
{$\mathtt{sendTo}(j) \leftarrow \sheaf{F}(i \fc ij)(x_i)$ \;
$\mathtt{receiveFrom}(i) \leftarrow \mathtt{sendTo}(j)$ \;
$x_j \leftarrow x_j \meet \cosheaf{F}(j \fc ij)\left(\mathtt{receiveFrom}(i)\right)$
}}}
\end{algorithm}

\noindent Alternatively, we may rewrite the algorithm as a recursion equation
\begin{align}
	\mathbf{x}[t+1] &=& L_t \bigl( \mathbf{x}[t] \bigr) \meet \mathbf{x}[t] \label{eq:gossip}
\end{align}
called the \define{gossip equation} with the time-varying Tarski Laplacian $\Laplacian_t$ defined as follows.

\begin{definition}[Time-Varying Tarski Laplacian] \label{def:varrying-tarski}
   Suppose $\graph{G}$ is a graph and $\bisheaf{F}$ is a Tarski sheaf over $\graph{G}$ and $\tau: \N \to \powerset{\nodes{G}}$ satisfies Assumption \ref{as:liveness}. The \define{asynchronous Tarski Laplacian}:\footnote{If $\tau(t) \cap \nbhd{i} = \emptyset$, then $(\tilde{\Laplacian}_t)_i = 1$ because $\bigmeet \emptyset = 1$.} is the time-varying operator
\begin{align}
	\tilde{\Laplacian}_t &:& \N \times C^0(\graph{G}; \sheaf{F}) \to C^0(\graph{G}; \sheaf{F}) \\
	\left(\tilde{\Laplacian}_t \right)_i &=& \bigmeet_{j \in \tau(t) \cap \nbhd{i}} \cosheaf{F}(i \fc ij)\sheaf{F}(j \fc ij)(x_j). \nonumber
\end{align} 
\end{definition}

\noindent We now state a equilibrium result for global sections for gossip.

\begin{theorem}[Theorem 1 \cite{riess2022asynchronous}]
	Suppose $\graph{G}$ is a graph and $\bisheaf{F}$ is a Tarski sheaf over $\graph{G}$. Suppose $\tau: \N \to \powerset{\nodes{G}}$ is a firing sequence satisfying Assumption \ref{as:dcc}. Suppose $\tilde{\Laplacian}_t$ is the asynchronous Tarski Laplacian. Then, if $\tau$ is a broadcast sequence satisfying Assumption \ref{as:liveness}, the sections $H^0(\graph{G}; \sheaf{F})$ coincide with the time-independent solutions of the recursion equations
	\begin{align}
		\mathbf{x}[t+1] &=& (\Laplacian_t \meet \id) \mathbf{x}[t]
	\end{align}
  with initial condition $\mathbf{x}[0] \in C^0(\graph{G}; \sheaf{F})$.
\end{theorem}

\begin{corollary}
	Suppose $\graph{G}$, $\bisheaf{F}$ and $\tau$ satisfy Assumptions \ref{as:dcc} and \ref{as:liveness}. Then, Algorithm \ref{alg:gossip} terminates in a finite number of iterations and converges to an $\mathbf{x} \in H^0(\graph{G}; \sheaf{F})$.
\end{corollary}

\noindent In Appendix \ref{ch:appendix-2}, we review some experimental results on the convergence rate of the gossip algorithm (Algorithm \ref{alg:gossip}).

%% file: Chapters/Chapter08.tex
%************************************************
\chapter{Signals}\label{ch:signals} % 
%************************************************

Signal processing is a broad area of engineering devoted to the study of various sorts of signals. Some examples of signals include audio, video, speech, image, communication, GPS, sonar, lidar, and radar, among others.  In our view, signal processing is synonymous with information processing.

At a high level of abstraction, a signal is a map $x: \domain \to \channels$ where $\domain$ is a domain and $\channels$ is an object in a data category. Classically, $\channels$ is a finite-dimensional vector space and each dimension of $\channels$ is called a \define{channel} or \define{feature}, however, the same terminology may be used if $\channels$ is an arbitrary object in a data category. The set of signals over $\domain$ valued in $\channels$ is denoted $\signals(\domain, \channels)$.
\begin{examples}
    Some traditional examples:
    \begin{itemize}
        \item[] \emph{Waveforms}. Suppose $T$ is the period of the waveform. Then, $\domain$ is a the interval $[0,T]$ with and $\channels = \R$ (electric potential, displacement etc.).
        \item[] \emph{Images}. $\domain$ is a grid $\Z \times \Z$ (pixels) and $\channels = \R^d$ (intensities) where usually $d = 1$ (grayscale) or $3$ (RBG).
        \item[] \emph{Networks}. $\domain$ is a graph $\graph{G} = (\nodes{G}, \edges{G})$ and $\channels = \R^d$; each channel is a feature. In practice, channels are often features generated from data and the graph is constructed from a similarity measure between nodes/agents (e.g. Pearson correlation).
    \end{itemize}
\end{examples}

In Chapter \ref{ch:lattice-valued} we introduced lattice-valued sheaves as a way to generalize sheaves valued in sets and vector spaces. If every stalk in a lattice-valued sheaf over a graph $\graph{G}$ is the same lattice, say $\lattice{L}$, then it is perfectly reasonable to identify $0$-cochains $C^0(\graph{G}; \sheaf{F})$ with signals $\signals(\nodes{G}, \lattice{L})$. We might, then, hope for some of the themes of signal processing to port to lattice-valued sheaves. 

%--------------------------------------
\section{Algebraic Signal Processing}
%--------------------------------------

Algebraic signal processing (ASP) is a unified signal processing framework based on the representation theory of associative algebras \cite{etingof2011introduction}. One limitation of ASP, however, is that $\signals(\domain, \channels)$ must be a vector space in this perspective. We first basics of representation theory.

\begin{definition}[Algebras]
    Let $\field$ be a field. An \define{associative algebra} $\field$, or $\field$-algebra for short, is a vector space $A$ equipped with an associative bilinear map\footnote{Equivalently, a linear map $A \otimes A \to A$}
\begin{align*}
    \cdot: A \times A \to A
\end{align*}
A $\field$-algebra is \define{unital} if there exists an element $1 \in A$ with $a \cdot 1 = 1 \cdot a = a$. An \define{algebra homomorphism} between algebras $(A, \cdot)$ and $(A', \cdot')$ is a linear transformation $\rho: A \to A'$ such that $\rho(a \cdot b) = \rho(a) \cdot' \rho(b)$.
\end{definition}

\noindent We provide several examples of algebras whose representationsa are of interest in engineering.

\begin{example}[Polynomials]
        Polynomials with coeficients in $\field$ form an algebra denoted $\field[t]$; multiplication is given by multiplication of polynomials. Multinomials, in the same way, form an algebra $\field[t_1,t_2, \dots, t_n]$. The unit is the constant polynomial $1 \in \field[t]$ or $\field[t_1,t_2, \dots, t_n]$.
\end{example}
\begin{example}[Endomorphism Algebra]
    Suppose $V$ is a vector space. Then, $\End{V}$ is a $\field$-algebra with composition. If $V$ is finite-dimensional, $\End{V}$ is isomorphic to the algebra $\mathrm{Mat}_{n \times n}(\field)$ of $n \times n$ matrices with coefficients in $\field$. The unit is the identity linear transformation.
\end{example}
\begin{example}[Monoid Algebra]
Let $(\monoid{M}, \star, e)$ be a monoid. The \define{monoid algebra} denoted $\field[\monoid{M}]$ (sometimes called the \define{monoid ring}) is the free vector space generated by formal linear combinations of elements of $\monoid{M}$
    \begin{align*}
        \sum_{m \in \monoid{M}} h_m [m], \quad |\{h_m \in \field ~\vert~ h_m \neq 0\}| < \infty.
    \end{align*}
Multiplication in $\field[\monoid{M}]$ is defined on basis elements
    \begin{align*}
        [m] [n] &=& [m \cdot n],
    \end{align*}
then extended by linearity. The algebra $\field[\monoid{M}]$ is unital with unit $[e]$. Equivalently, the monoid algebra consists of the vector space of functions $\field^{\monoid{M}}$ with multiplication
\begin{align*}
    (f \oast g)(m) &=& \sum_{p \cdot q = m} f(p)g(q) \quad f, g \in \field^{\monoid{M}}, \quad m, p, q \in \monoid{M}.
\end{align*}
The monoid algebra is integral in the study of representations of finite monoids \cite{steinberg2016representation}, a generalization of the theory of representations of finite groups \cite{serre1977linear}.
\end{example}

\begin{example}[Monoid Algebra]
Suppose $\monoid{M}$ is the semilattice $M = \{x,y,z,1\}$ with the multiplication table
\begin{center}
\begin{tabular}{l|llll}
$\meet$ & $x$ & $y$ & $z$ & $1$ \\
    \hline
$x$ & $x$ & $z$ & $z$ & $x$ \\
$y$ & $z$ & $y$ & $z$ & $y$ \\
$z$ & $z$ & $z$ & $z$ & $z$ \\
$1$ & $x$ & $y$ & $z$ & $1$
\end{tabular}
\end{center}
Then, suppose $h = [x] - [y]$ and $h' =  [1] - 2[y] - [z]$, then
\begin{align*}
    h \cdot h' &=& \\
            &=& ([x] - [y] ) \cdot ([1] - 2[y] - [z]) \\
            &=& [x \meet 1] - 2 [ x \meet y] - [x\meet z] - [y \meet 1] + 2[y \meet y] + [y \meet z] \\
            &=& [x] - 2 [z] - [z] - [y] + 2[y] + [z] \\
            &=& [x] - 2 [z] + [y]
\end{align*}
\end{example}

\begin{example}[Incidence Algebra]
    Suppose $\poset{P}$ is a locally finite\footnote{A poset is locally finite if every closed interval $[x,y]$ is finite poset. A finite poset is of course locally finite} poset. Then, the \define{incidence algebra} of $\poset{P}$, denoted $\field \poset{P}$ is a $\field$-algebra over the vector space of functions
    \[f: \Int{\poset{P}} \to \field\]
    with multiplication
    \begin{align*}
        (f \oast g)([x,y]) &=& \sum_{x \preceq y \preceq z} f([x,z])g([z,y]).
    \end{align*}
    Equivalently, $\field \poset{P}$ is a monoid algebra on $\Int{\poset{P}}$ with the appropriate monoid stucture.
\end{example}
\begin{example}[Path Algebra]
    A \define{quiver} is a tuple $Q = (\nodes{Q}, \edges{Q}, h,t)$ consisting of sets $\nodes{Q}, \edges{Q}$ (vertices, arrows) and maps $h,t: \edges{Q} \to \nodes{Q}$ (head, tail). A path in a quiver is a sequence $\gamma = a_0 a_1 \dots a_{\ell}$ of elements of $\edges{Q}$ such that $h(a_{i}) = t(a_{i+1})$. Trivial paths $e_i$ are indexed by $\nodes{Q}$. Let $\mathrm{Path}(Q)$ denote the set of paths in $Q$. Then, the \define{path algebra} $\field Q$ is the free vector space generated by $\mathrm{Path}(Q)$ with a product of basis elements
   \[ \gamma_1 \cdot \gamma_2 = \begin{cases}
       \gamma_1\gamma_2 & h(\gamma_1) = t(\gamma_2) \\
       0                & h(\gamma_1) \neq t(\gamma_2)
   \end{cases}\]
\end{example}

Suppose $A$ is a $\field$-algebra. A \define{representation of $A$} is a vector space $V$ together with an algebra homomorphism $\rho: A \to \End(V)$. Representations of $A$ form a category $\cat{Rep}_A$ with the following morphisms.

\begin{definition}[Intertwining Maps]
    Suppose $(V, \rho)$ and $(V, \rho')$ are representations of $A$. Then, an \define{intertwining map} is a linear transformation $\theta: V \to V'$ such that
    \begin{align*}
        \rho_a' \circ \theta &=& \theta \circ \rho_a
    \end{align*} 
    for all $a \in A$.
\end{definition}

\noindent In the above examples, signal spaces are naturally paired with each algebra with a representation, leading to the following definition.

\begin{definition}[Algebraic Signal Model]
    An  \define{algebraic signal model} is a triple $(A, \signals, \rho)$. $A$ is a $\field$-algebra whose elements are called \define{filters}; $\signals$ is vector space whose elements are called \define{signals}; $\rho$ is a representation
    \begin{align*}
        \rho &:& A \to \End(\mathcal{X}). 
    \end{align*}
    Given a filter $h \in A$, the linear tranformation\footnote{We often write $\rho_h: \signals \to \signals$ for the linear transformation that is the image of an element $h \in A$ under $\rho$.} $\rho_h: \signals \to \signals$ is an \define{instantiation} of $h$. Acting on a signal $x \in \signals$ by $\rho_h$ is called \define{convolution}.
\end{definition}

\begin{remark}
    The original formulation of algebraic signal processing regards an algebraic signal model as an algebra $A$ of filters together with an $A$-module. However, the data of an $A$ module is the same as the data of a representation $\rho: A \to \End(\signals)$ by standard facts of representation theory \cite{etingof2011introduction}. We prefer to think of representations as maps to emphasize relationships between models.
\end{remark}

\noindent We now show that familiar and unfamiliar signal processing frameworks fall under the banner of algebraic signal processing.

\begin{example}[Graph Signal Processing \cite{parada2021algebraic}]
    Let $\graph{G} = (\nodes{G}, \edges{G})$ be a weighted graph with (graph) shift operator $S: \R^\nodes{G} \to \R^\nodes{G}$. There are several common choices for $S$, including the (normalized) adjacency or Laplacian matrix. A \define{graph filter} is an element of the algebra $\field[t]$. The signal space $\signals(\nodes{G}, \R)$ consists is the vector space of functions $\R^{\nodes{G}}$. Convolution, then, is a representation
    \begin{align*}
        \rho: \R[t] \to \End\left(\R^{\nodes{G}}\right).
    \end{align*}
    Because $\R[t]$ is generated by the single mononomial $t$, a representation is uniquely determined by $\rho(t)$ which we chose to be the graph shift operator. Hence, a graph filter is given by
    \begin{align*}
        \mathbf{z} &=& \mathbf{h} \oast \mathbf{x} \\
        &=& \rho_{\mathbf{h}} \cdot \mathbf{x} \\
        &=& \left(\sum_{k=0}^{\infty} h_k S^k\right) \cdot \mathbf{x}
    \end{align*}
    Let's consider a specific example of a graph filter, the \define{heat kernel} \cite{ortega2018graph}. The heat kernal is a parameterized linear graph filter denoted $H_t$ which acts on graph signals as
    \begin{align*}
        H_t \cdot \mathbf{x} &=& \left(\sum_{k=0}^{\infty} \frac{t^k e^{-t}}{k!} S^k\right) \cdot \mathbf{x}, \quad t \geq 0.
    \end{align*}
    A graph filter, in practice, will typically have $h_k = 0$ for all $k \geq K$ where $K \ll n$. While there is a computational justification for this, there is also a theoretical one. Suppose $n = |\nodes{G}|$. Then, by the Cayley-Hamilton Theorem, we can always write $S^n$ as a linear combination of $\{I, S, S^2, \dots, S^{n-1}\}$. Hence, we have an alternate description of graph filters as a representation of a smaller algebra
    \begin{align*}
        \rho: \field[t]/(t^n-1) \to  \End(\field^{\nodes{G}})
    \end{align*}
    where $\field[t]/(t^n-1)$ is the algebra of polynomials of degree at most $n-1$.  
\end{example}

\begin{example}[Quiver Signal Processing \cite{parada2020quiver}]
    Suppose $Q = (\nodes{Q}, \edges{Q}, h,t)$ is a quiver. A \define{representation} $\boldsymbol{\pi}_{\bullet}$ of $Q$ is an assignment of a vector space $\boldsymbol{\pi}(v)$ for every $e \in \nodes{Q}$ and a linear transformation $\boldsymbol{\pi}(e): \boldsymbol{\pi}(t(e)) \to \boldsymbol{\pi}(h(e))$ for every $e \in \edges{Q}$. The \define{total space} of a representation is the direct sum
    \begin{align*}
        \mathrm{Tot}(\boldsymbol{\pi}) &=& \bigoplus_{i \in \nodes{Q}} \boldsymbol{\pi}(i).
    \end{align*}
     It is a fact that representations of $Q$ are equivalent to representations of $\field Q$ \cite{oudot2017persistence}. In one direction, suppose $\boldsymbol{\pi}$ is a representation of $Q$. Then, we define
    \begin{align*}
        \rho^{\boldsymbol{\pi}}: \field Q \to \End\left( \mathrm{Tot}(\boldsymbol{\pi}) \right)
     \end{align*}
     by sending paths $\gamma = e_1 e_2 \dots e_\ell$ to compositions of linear maps
    \begin{align*}
    \boldsymbol{\pi}(\gamma) &=& \boldsymbol{\pi}(e_{\ell})\cdots \boldsymbol{\pi}(e_2)\boldsymbol{\pi}(e_1),
    \end{align*}
    and extending $\rho^{\boldsymbol{\pi}}$ linearly to the entire domain $\field Q$. Acting on $\mathrm{Tot}(\boldsymbol{\pi})$,
    \[
        (\rho^{\boldsymbol{\pi}}([\gamma]) \cdot \mathbf{x})_j =
        \begin{cases}
            \boldsymbol{\pi}(\gamma) \cdot x_i & h(\gamma) = j,~t(\gamma) = i \\
            0 & \text{otherwise}
        \end{cases}\]
    Hence, repesentations of $Q$ are in coorespondence with algebraic signal models, $(A, \signals, \rho) = (\field Q, \mathrm{Tot}(\boldsymbol{\pi}), \rho^{\boldsymbol{\pi}})$.
\end{example}

\begin{example}[Monoid Algebra]
    Suppose $\monoid{M}$ is a monoid. Then, in monoid signal processing, $A = \field \left[ G \right]$ and $\signals = \field^\monoid{M}$. Monoid convolution is a representation
    \begin{align}
        \rho: \field \left[ M  \right] \to \End(\field^\monoid{M}).
    \end{align}
    Shift operators coorespond to generators of $\field \left[ \monoid{M} \right]$ which, in turn, coorespond to the generators of $\monoid{M}$ (if they exist).
\end{example}

%---------------------------------------
\section{Lattice Signal Processing}
%---------------------------------------

Narrowing our focus to lattices, we first study signals whose domain is a (semi)lattice, then signals whose channels are lattices, and, finally, signals whose domain and channels are lattices.

%-------------------------------------
\subsection{Signals on semilattices}
%-------------------------------------

\textpuschel and Wendell proposed \cite{puschel2019discrete} a signal processing framework on meet- or join- semilattices. In this case, $\domain$ is a semilattice and $\channels$ is the real or complex number. In our view, their key definition is convolution of semilattice signals, although, in a series of paper, they build up a complete theory including Fourier transforms \cite{puschel2021discrete}, filters for denoising \cite{seifert2021wiener}, and an algorithm for sampling \cite{wendler2019sampling}.

\begin{definition}
Suppose $(\poset{S}, \meet, 1)$ is a (meet-) semilattice.\footnote{Without loss of generality, since $\op{\poset{S}}$ is a bounded join-semilattice.} Suppose $\mathbf{f} \in {\field}^{\poset{S}}$ is a signal and $\mathbf{h} \in {\field}^{\poset{S}}$ is a filter. Then, \define{convolution} of $\mathbf{f}$ by $h$ is the following operation
\begin{align}
    (\mathbf{h} \conv \mathbf{f})_y &=& \sum_{x \in \poset{S}} h_x f_{x \meet y}. \label{eq:puschel-conv}
\end{align}
\end{definition}

\noindent If there is any ambiguity, convolutions are written $\mathbf{h} \conv_\meet \mathbf{f}$ or $\mathbf{h} \conv_\join \mathbf{f}$ depending on whether $\poset{S}$ is a meet- or join- semilattice.

Convolution can be written in more compact form by means of a \define{shift operator}
\begin{align*}
    T_x &:& \field^{\poset{S}} \to \field^{\poset{S}} \nonumber \\
    (T_x \mathbf{f})_y &=& f_{x \meet y}. \label{eq:shift-lattice}
\end{align*}
Thus,
\begin{align}
    \mathbf{h} \conv - &=& \sum_{x \in \poset{S}} h_x T_x.
\end{align}
is a map sending a filter $\mathbf{j}$ to an endomorphism on $\field^{\poset{S}}$. Let us consider some properties of the shift operators.

\begin{lemma} \label{lem:shift-properties}
    Suppose $x, y \in \lattice{L}$. Then,
    \leavevmode
    \begin{enumerate}
        \item $T_x$ is linear.
        \item $T_{x \meet y} = T_{x} T_{y}$.
        \item $T_x$ is idempotent.
        \item $T_x T_y = T_y T_x$.
    \end{enumerate}
\end{lemma}
\begin{proof}
    Linearity is trivial. For (2),
    \begin{align*}
        (T_x T_y \mathbf{f})_{z}    &=& (T_y \mathbf{f})_{z \meet x} \\
                                    &=& f_{ (z \meet x) \meet y} \\
                                    &=& f_{(x \meet y) \meet z} \\
                                    &=& (T_{x \meet y} \mathbf{f})_z.
    \end{align*}
    Because the meet operation in a semilattice is idempotent and commutative, (3) and (4) follow from (2).
\end{proof}
\noindent The following established associativity of filtering.
\begin{lemma} \label{lem:shift-commute}
    Suppose $\mathbf{h}, \mathbf{h}' \in \field^{\poset{S}}$ are filters and $f \in \field^{\poset{S}}$ is a signal. Then,
    \begin{align*}
        \mathbf{h} \conv (\mathbf{h}' \conv \mathbf{f}) &=& \mathbf{h}' \conv (\mathbf{h} \conv \mathbf{f}).
    \end{align*}
\end{lemma}
\begin{proof}
    Let $x, y \in \poset{S}$.
    \begin{align*}
        \mathbf{h} \conv \mathbf{f}  &=& \sum_{x \in \poset{{}S}} h_x T_x \mathbf{f} 
    \end{align*}
    so that
    \begin{align*}
        \mathbf{h} \conv (\mathbf{h}' \conv \mathbf{f}) &=& \sum_{y \in \poset{S}} h_y T_y \left( \sum_{x \in \poset{S}} h_x' T_x \mathbf{f} \right) \\
                                &=& \sum_{y \in \poset{S}} \sum_{x \in \poset{S}} h_y h_x' T_x T_y \mathbf{f}
    \end{align*}
    because $T_y$ is linear. Then, by Lemma \ref{lem:shift-commute},
    \begin{align*}
    \sum_{y \in \poset{S}} \sum_{x \in \poset{S}} h_y h_x' T_x T_y &=& \sum_{y \in \poset{S}} \sum_{x \in \poset{S}} h_x' h_y T_y T_x.
    \end{align*}
    Reversing the process implies $\mathbf{h} \conv (\mathbf{h}' \conv \mathbf{f}) = \mathbf{h}' \star (\mathbf{h} \conv \mathbf{f})$.
\end{proof}

It follows that semilattice convolution is an algebraic signal model with $A  = \field [\poset{S}]$, $\signals = \field^{\poset{S}}$, and
\begin{align}
    \rho &:& \field [ \poset{S} ] \to \End\left(\field^\poset{S}\right) \label{eq:puschel-conv-2} \\
    \rho\left(\sum_{x \in \poset{S}} h_x [x]\right) &=& \sum_{x \in \poset{S}} h_x T_{x} \nonumber.
\end{align}

\begin{proposition}
    Convolution acting on the left (or right by commutativity) $\mathbf{h} \conv - : \field^{\poset{S}} \to \field^{\poset{S}}$ \eqref{eq:puschel-conv} coincides with the algebra homomorphism $\rho_h$ \eqref{eq:puschel-conv-2} applied to an algebraic filter $h \in A$.
\end{proposition}
\begin{proof}
    Identity a filter $\mathbf{h} \in \field^\poset{S}$ with the coefficients of multinomial
    \begin{align*}
    h &=& \sum_{x \in \poset{S}} h_x [x].
    \end{align*}
    Then, for $h \in \field[\poset{S}]$, $\mathbf{f} \in \field^{\poset{S}}$ define $\rho_h \cdot \mathbf{f}$ by $\mathbf{h} \conv \mathbf{f}$.
    By Lemma \ref{lem:shift-properties}, $\rho_h \in \End(\field^{\poset{S}})$. However, we still need to show that $\rho$ is an algebra homomorphism. Suppose $h, h' \in \field[\poset{S}]$. Then,
    \begin{align*}
        h \cdot h' &=& \left( \sum_{x \in \poset{S}} h_x [x] \right) \cdot \left( \sum_{y \in \poset{S}} h_y' [y] \right) \\
        &=& \sum_{z = x \meet y} h_x h_y' T_z
    \end{align*}
    Lemma \ref{lem:shift-commute} and applying $\rho$ yields
    \begin{align*}
        \rho_{h \cdot h'} &=& \sum_{z = x \meet y} h_x h_y' \rho([z]) \\
                    &=& \sum_{x, y \in \poset{S}} h_x h_y' T_{x \meet y} \\
                    &=& \sum_{x \in \poset{S}} \sum_{y \in \poset{S}} h_x h_y' T_x T_y \\
                    &=& \sum_{x \in \poset{S}} h_x T_x \left( \sum_{x \in \poset{S}}  h_y' T_y \right) \\
                    &=& \rho_{h} \rho_{h'}.
    \end{align*}
    Linearity in $h$ is trivial.
\end{proof}

%--------------------------------
\subsection{Signals on lattices}
%--------------------------------

We argue signal processing on lattices is not equivalent to signal processing on semilattices. Suppose $(\lattice{L}, \meet, \join, 0, 1)$ is a bounded lattice. We view $\lattice{L}$ as a bounded semilattice, and thus a monoid, in two seperate ways: a meet-semilattice denoted $(\monoid{L}_\meet, \meet, 1)$ and a join-semilattice $(\monoid{L}_\join, \join, 0)$. However, the absorption identities (Definition \ref{def:lattice-alg}) further restrict the structure of a lattice than merely a poset that is simulatneously a meet- and join- semilattice. For one, there are two families of shift operators, one for each element $x \in \lattice{L}$
\begin{align*}
    (T_x^\join \mathbf{f})_y &=& f_{x \join y} \\
    (T_x^\meet \mathbf{f})_y &=& f_{x \meet y}
\end{align*}
Each family instantiates filters
\begin{align*}
    \rho^\meet &:& \field[\monoid{L}_\meet] \to \End(\field^{\lattice{L}}) \\
    \rho^\join &:& \field[\monoid{L}_\join] \to \End(\field^{\lattice{L}})
\end{align*}
Given a bounded lattice $(\lattice{L}, \meet, \join, 0, 1)$, we investigate the relationship between the algebraic signal models
\begin{align*}
    \mathcal{M}_{\meet} &=& (\field[\monoid{L}_\meet], \field^\lattice{L}, \rho^\meet), \\
    \mathcal{M}_{\join} &=& (\field[\monoid{L}_\join], \field^\lattice{L}, \rho^\join).
\end{align*} 
% \begin{lemma}
%     Suppose $\alpha: \monoid{M} \to \monoid{N}$ is a monoid homomorphism. Then, there is an induced algebra homomorphism
%     \begin{align*}
%         \alpha: \field[\monoid{M}] \to \field[\monoid{N}].
%     \end{align*}
% \end{lemma}   
% \begin{proof}
%     The free vector space construction is functorial, hence, there is a linear transformation 
%     \begin{align*}
%         \hat{\alpha}: \field[\monoid{M}] \to \field[\monoid{N}];
%     \end{align*} 
%     on basis vectors $\hat{\alpha}([m]) = [\alpha(m)]$. It is obvious that $\hat{\alpha}$ is also an algebra homomorphism.
% \end{proof}

% \noindent Moreover, if $\alpha$ is an isomorphism, then $\hat{\alpha}$ is as well.

\begin{proposition} \label{prop:shift-meet-join}
    Suppose $\lattice{L}$ is a lattice with shift operators $\{T^{\meet}_a$, $T^{\join}_a\}_{a \in \lattice{L}}$ defined
    \begin{align*}
        (T^{\meet}_a \mathbf{f})_x &=& f_{x \meet a}, \\  
        (T^{\join}_a \mathbf{f})_x &=& f_{x \join a}.
    \end{align*}
      Then,
    \begin{align*}
        T_x^{\join}T_x^{\meet} &=& T_x^{\meet}T_x^{\join}.
    \end{align*}
    Furthermore,
    \begin{align*}
        T_x^{\join}T_x^{\meet} \mathbf{f} &=& f_x \mathbf{1}
    \end{align*}
    where $\mathbf{1} \in \field^\lattice{L}$ is the constant unit signal.
\end{proposition}
\begin{proof}
    Suppose $\mathbf{f} \in \field^{\lattice{L}}$. The absorption identities (Ch imply
    \begin{align*}
        (T_x^{\join} T_x^{\meet} \mathbf{f})_y &=& (T_x^{\meet} \mathbf{f})_{x \join y} \\
                                               &=& f_{(x \join y) \meet x} \\
                                               &=& f_x.
    \end{align*}
    Clearly, $(T_x^{\meet} T_x^{\join} \mathbf{f})_y = f_x$ as well.
\end{proof}
\begin{corollary}
    $T_x^{\meet} T_x^{\join}$ and $T_x^{\join} T_x^{\meet}$ are idempotent.
\end{corollary}
\begin{proof}
    Notice $T_x^{\meet} \mathbf{1} = \mathbf{1}$ and $T_x^{\join} \mathbf{1} = \mathbf{1}$. Then, by Proposition \ref{prop:shift-meet-join},
    \begin{align*}
         T_x^{\join} T_x^{\meet} T_x^{\join}T_x^{\meet} \mathbf{f} &=& T_x^{\join} T_x^{\meet} f_x \mathbf{1} \\
                                         &=& f_x \mathbf{1}  \\
                                         &=& T_x^{\join}T_x^{\meet} \mathbf{f}
    \end{align*} 
\end{proof}

% This gives us a sense that we cannot design more expressible filters by composing $T_x^{\meet}$ and $T_x^{\join}$.
We will now introduce some of Fourier analysis presented in the original paper \cite{puschel2021discrete}. Suppose $\lattice{L}$ is a bounded lattice. Let
\[1_{\{y \preceq x\}} = \begin{cases}
    1 & y \preceq x \\
    0 & y \not\preceq x
\end{cases}\]
be the indicator function on the boolean variable $\{y \preceq x\}$. Define a signals $\mathbf{f}^{\preceq y} = (1_{\{y \preceq x\}})_{x \in \lattice{L}}$ and $\mathbf{f}^{\succeq y} = (1_{\{y \succeq x \}})_{x \in \lattice{L}}$.

\begin{lemma}[Theorem 1 \cite{puschel2021discrete}] \label{lem:puschel-basis}
    Suppose $\lattice{L}$ is a lattice with shifts $\{T_x^{\meet}, T_x^{\join}\}_{x \in \lattice{L}}$. Then, $T_x^{\meet} \mathbf{f}^{\preceq y} = 1_{y \preceq x} \mathbf{f}^{\preceq y}$ and $T_x^{\join} \mathbf{f}^{\succeq y} = 1_{y \succeq x} \mathbf{f}^{\succeq y}$. Furthermore, $\{ \mathbf{f}^{\preceq y} \}_{y \in \lattice{L}}$ is an eigenbasis of $\field^{\lattice{L}}$ with respect to $T_x^{\meet}$. Dually, $\{ \mathbf{f}^{\succeq y} \}_{y \in \lattice{L}}$ is an eigenbasis of $\field^{\lattice{L}}$ with respect to $T_x^{\join}$.
\end{lemma} 
\begin{proof}
    See \cite{puschel2021discrete} for the $T_x^{\meet}$ case. The $T_x^{\join}$ cases follows by the duality principle.
\end{proof}

Every basis vector $\mathbf{f}^{\preceq y}$ (dually, $\mathbf{f}^{\succeq y}$) is an eigenvector of $T_{x}^{\meet}$ with eigenvalue $1_{\{x \preceq y\}}$. Similarly, for $T_x^{\join}$. Hence, the only eigenvalues of $T_{x}^{\meet}$ and $T_{x}^{\join}$ are $0$ and $1$. Define a change-of-basis transformation $\theta: \field^{\lattice{L}} \to \field^{\lattice{L}}$ by sending each $\mathbf{f}^{\preceq y} \in \{ \mathbf{f}^{\preceq y}_{y \in \lattice{L}}$ to $\mathbf{f}^{\succeq y} \in \{ \mathbf{f}^{\succeq y}_{y \in \lattice{L}}$. For the remainder of this section, suppose $\lattice{L}$ is finite. Then, we can place a total order on the elements $\lattice{L}$ (e.g. topological sort), labeling elements $\lattice{L} = \{y_1, y_2, \dots, y_m\}$. Now, in the finite case, the linear transformation $\theta$ is represented by an invertible matrix $\Theta = [\theta_{ij}]$ uniquely determined by writing each element of the basis $\{\mathbf{f}^{\succeq y}\}_{y \in \lattice{L}}$ as a linear combination of basis elements $\{\mathbf{f} \}^{\preceq y}_{y \in \lattice{L}}$. 
\begin{align*}
    \mathbf{f}^{\succeq y_i} &=& \sum_{i = 1}^{m} \theta_{i,j} \mathbf{f}^{\preceq y_j}.
\end{align*}

\begin{example}
Suppose $\lattice{L}$ is the lattice
\begin{center}
   \includegraphics[width=0.2\textwidth]{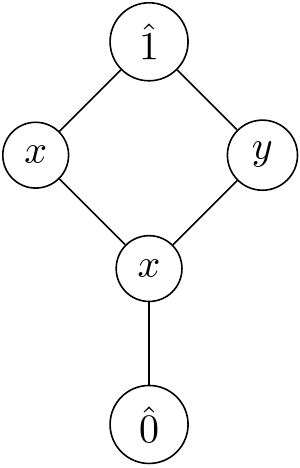} 
\end{center}
We order the elements of $\lattice{L} = \{ \hat{0}, z, x, y, \hat{1} \}$. We write $B_{\meet}$ for the matrix whose columns are the meet-basis vectors
\[B_{\meet} = \begin{bmatrix}
    \mathbf{f}^{\preceq \hat{0}} & \mathbf{f}^{\preceq z} & \mathbf{f}^{\preceq x} & \mathbf{f}^{\preceq y} & \mathbf{f}^{\preceq \hat{1}}
    % \mathbf{f}^{\preceq \hat{0}} & \mathbf{f}^{\preceq z} & \mathbf{f}^{\preceq x} & \mathbf{f}^{\preceq y} & \mathbf{f}^{\preceq \hat{1}} \\
    % \vline & \vline & \vline & \vline & \vline \\
    % \vline & \vline & \vline & \vline & \vline \\
    % \vline & \vline & \vline & \vline & \vline \\
    % \vline & \vline & \vline & \vline & \vline
\end{bmatrix} = \begin{bmatrix}
    1 & 1 & 1 & 1 & 1 \\
    0 & 1 & 1 & 1 & 1 \\
    0 & 0 & 1 & 0 & 1 \\
    0 & 0 & 0 & 1 & 1 \\
    0 & 0 & 0 & 0 & 1
\end{bmatrix}, \]
Clearly, $B_{\join} = B_{\meet}^{\top}$ because $[B_\meet]_{i,j} = 1_{\{x_i \preceq x_j\}}$
% \[B_{\join} = \begin{bmatrix}
%     \mathbf{f}^{\succeq \hat{0}} & \succeq{f}^{\preceq z} & \succeq{f}^{\preceq x} & \succeq{f}^{\preceq y} & \succeq{f}^{\preceq \hat{1}} \\
%     \vline & \vline & \vline & \vline & \vline \\
%     \vline & \vline & \vline & \vline & \vline \\
%     \vline & \vline & \vline & \vline & \vline \\
%     \vline & \vline & \vline & \vline & \vline
% \end{bmatrix} = \begin{bmatrix}
%     1 & 0 & 0 & 0 & 0 \\
%     1 & 1 & 0 & 0 & 0 \\
%     1 & 1 & 1 & 0 & 1 \\
%     0 & 0 & 0 & 1 & 1 \\
%     0 & 0 & 0 & 0 & 1
% \end{bmatrix}, \]
Then, $\Theta B_{\meet} = B_{\join}$ implies $\Theta = B_{\join} B_{\meet}^{-1}$. In our case, $\theta$ sends the signal
\begin{align*}
    \mathbf{f} &=& \left(f_{\hat{0}}, f_z, f_x, f_y, f_{\hat{1}}\right)
\end{align*}
to the signal
\begin{align*}
    \theta(\mathbf{f}) &=& \left( f_{\hat{0}} - f_z, f_{\hat{0}} - f_x - f_y + f_{\hat{1}}, f_{\hat{0}} - f_y, f_{\hat{0}} - f_x, f_{\hat{0}} \right).
    \end{align*}
\end{example}

We now present an ``impossibility result'' demonstrating an isomorphism between the representations $\rho^{\meet}$ and $\rho^{\join}$ implies the underlying lattice is the trivial lattice $\lattice{L} = \{ 0 = 1\}$.

\begin{theorem}
    Suppose $\lattice{L}$ is a lattice. Then,
    \begin{align}
        \theta \left(T_x^{\meet} \mathbf{f} \right) &=& T_x^{\join} \theta(\mathbf{f}) \label{eq:shift-intertwine}
    \end{align}
    if and only if $\lattice{L} = \{0 = 1\}$.
\end{theorem}
\begin{proof}
    One direction is trivial. Suppose \eqref{eq:shift-intertwine} holds. Then, by Lemma \ref{lem:puschel-basis}, we can write $\mathbf{f} = \sum_{y \in \lattice{L}} h_y \mathbf{f}^{\preceq y}$ with $h_y \in \field$. On one hand,
    \begin{align*}
        \theta \left(T_x^{\meet} \mathbf{f} \right) &=& \theta T_x^{\meet} \left( \sum_{y \in \lattice{L}} h_y \mathbf{f}^{\preceq y} \right)  \\
            &=& \sum_{y \in \lattice{L}} h_y \theta T_x^{\meet y} \mathbf{f}^{\preceq y}  \\
            &=& \sum_{y \in \lattice{L}} h_y 1_{y \preceq x} \theta(\mathbf{f}^{\preceq y}) \\
            &=& \sum_{y \in \lattice{L}} h_y 1_{y \preceq x} \mathbf{f}^{\succeq y}
    \end{align*}
    On another hand,
    \begin{align*}
        T_x^{\join} \theta(\mathbf{f}) &=& T_x^{\join} \theta \left( \sum_{y \in \lattice{L}} h_y \mathbf{f}^{\preceq y} \right) \\
        &=& \sum_{y \in \lattice{L}} h_y T_x^{\join} \theta ( \mathbf{f}^{\preceq y}) \\
        &=& \sum_{y \in \lattice{L}} h_y T_x^{\join} \mathbf{f}^{\succeq y} \\
        &=& \sum_{y \in \lattice{L}} h_y 1_{y \succeq x} \textbf{f}^{\succeq y}
    \end{align*}
    Therefore, it must hold that, given a fixed $x \in \lattice{L}$, $y \preceq x$ if and only if $y \succeq x$. This is only possible if $\lattice{L} = \{ \hat{0} = \hat{1}\}$ is a one-element lattice.
\end{proof}

% Recall the join-irreducibles of $\lattice{L}$ is a poset denoted $\irred{\lattice{L}}$. By duality, the meet-irreducibles of $\lattice{L}$ is the poset $\irred{\op{\lattice{L}}}$. 

\noindent We provide two examples of signals over lattices and their convolutions.

\begin{example}[Encodings]
    Given an element of a lattice $x \in \lattice{L}$ we can encode it as a signal $\mathbf{1}_x: \lattice{L} \to \R$ with
    \[\mathbf{1}_x (y) =
    \begin{cases}
        1 & x = y \\
        0 & x \neq y
    \end{cases}.\]
    In the signal-processing literature, this is called a \define{one-hot} encoding since exactly one component of the signal is ``hot,'' while the rest are ``cold.'' In practice, noise is sometimes added to one-hot encodings for numerical stability. The one-hot encoding is useful for applying continuous information processing techniques to lattices (algebraic structures). For instance, if $\lattice{L}$ is finite, $\lattice{L} = \{ \hat{0} = y_1, y_2, \dots, y_m = \hat{1} \}$ is an total ordering of the elements of $\lattice{L}$, and $\bisheaf{F}$ is an approximation of the constant sheaf $\underline{\lattice{L}}$ over $\graph{G}$, then the Tarski Laplacian can be identifies as a map $\lattice{L}: \R^{n \times m} \to \R^{n \times m}$. In general, the Tarski Laplacian of a lattice-valued sheaf whose stalks are finite lattices can be expressed as a map \[\Laplacian: \R^{\sum_{i \in \nodes{G}} |\bisheaf{F}(i)|} \to \R^{\sum_{i \in \nodes{G}} |\bisheaf{F}(i)|}\] 
    More efficient encodings of lattice elements is a subject of interest. For instance, if a lattice is finite, then every element is a join of join irreducible elements. Hence, you could represent an element of a finite lattice with a few-hot encoding, however, not uniquely.  
\end{example}
\begin{example}[Hypergraph Convolution]
    With mimicking graph signal processing \cite{zhang2019introducing} as one approach to filtering signals on hypergraphs, lattice convolution, as was suggested in the introductory paper \cite{puschel2021discrete}, is another approach. Suppose $\graph{H} = (\nodes{H}, \hyperedges{H})$ is a hypergraph. Recall, a hypergraph signal is a map $\mathbf{f}: \nodes{H} \to \R$. Suppose $\rel{I}$ is the incidence relation of $\graph{H}$. Then, the Galois lattice is the following
    \begin{align*}
        \Gal(\rel{I}) &=& \{ (\sigma, \tau) \subseteq \nodes{H} \times \hyperedges{H}~\vert~ \cap_{i \in \sigma} \galup{i} = \tau, \sigma = \cap_{e \in \tau} \galdown{e} \},
    \end{align*}
    which we label $\Gal(\graph{H})$ in this context. Recall, $\Gal(\graph{H})$ is isomorphic to the following fixed-point lattices
    \begin{align*}
        \Closed(\powerset{\nodes{H}}) &=& \{\sigma \subseteq \nodes{H}~\vert~\galdown{\galup{\sigma}} = \sigma \} &\subseteq& \powerset{\nodes{H}}, \\
        \Closed(\powerset{\hyperedges{H}}) &=& \{\tau \subseteq \hyperedges{H}~\vert~\galup{\galdown{\tau}} = \tau \} &\subseteq& \powerset{\hyperedges{H}}.
    \end{align*}
    Signals on $\graph{H} = (\nodes{H}, \hyperedges{H})$ are ported to signals on $\Gal(\graph{H})$ via aggregations. We hypothesize maximum is a good choice, because the hypergraph diffusion energy functional \cite{chan2018spectral} is of the form
    \begin{align*}
        Q(\mathbf{f}) &=& \sum_{e \in \hyperedges{H}} \max_{i,j \in e}(f_i - f_j)^2.
    \end{align*}

     \begin{figure}
        \centering
        \input{gfx/hypergraph} \\
        \input{gfx/hypergraph-2}
        \caption{Hypergraphs and their Galois lattices.} \label{fig:hypergraphs}
     \end{figure}
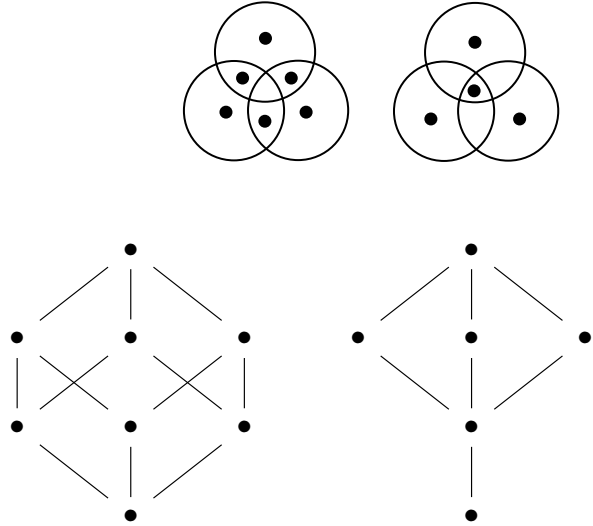
\end{example}

%-------------------------------
\subsection{Signals on posets}
%-------------------------------

If $\Omega$ is a finite poset $\poset{P}$ and $\channels$ is again the real or complex number, then one may port the signal processing framework on lattices to a new framework on finite posets under Birkhoff Duality (Theorem \ref{thm:birkhoff}). In one direction, given a finite poset $\poset{P}$ there is an isomorphism
\begin{align}
    \eta &:& \poset{P} \xrightarrow{\cong} \irred{\downsets{\poset{P}}} \label{eq:birkhoff-iso} \\
         \eta(p) &=& \downset{p} \nonumber
 \end{align}
which implies the join-irreducibles of the distributive lattice $\downsets{\poset{P}}$ are in one-to-one correspondence with the elements of $\poset{P}$. Hence, each $p \in \poset{P}$ corresponds to a unique join-irreducible $\downset p$. This implies the monoid algebra $\field[\downsets{\poset{P}}_{\join}]$ has a generating set $\mathcal{G} = \{[\downarrow p] \}_{p \in \poset{P}}$. Hence, the shift operators in the representation framework are given by the images of $g \in \mathcal{G}$ under the representation
\begin{align*}
    \rho_{\join} &:& \field\left[\mathcal{G}\right] \to \End\left(\field^{\downsets{\poset{P}}}\right).
\end{align*}

Recent work, takes a similar approach and generalizes (semi)lattice signal processing to signal processing on directed acyclic graphs (DAGs) \cite{seifert2022causal}. The authors of this work do not, however, construct convolutions via Birkhoff Duality as we suggest.

%---------------------------------------
\subsection{Lattice-valued signals}
%---------------------------------------

On the other hand, suppose $\domain$ is an abelian group (e.g. $\Z$ or $\Z/n\Z$) and $\channels$ is a lattice. A lattice-valued signal, then, is a map $\mathbf{f}: A \to \lattice{L}$. Lattice-valued signals are a natural taxum of signal to study. For instance, if $\channels$ is the lattice of partitions of the nodes of a dynamic graph, a signal $f \in \signals(\Z, \Part{\nodes{G}})$ would model the merging and splitting of connected components of a network as connections between nodes are established/disconnected \cite{kim2018formigrams}.

\begin{definition}[Lattice-Valued Convolution \cite{maragos2017dynamical}]
    Suppose $(\resid{L}, \meet, \join, 0, 1, \star, e, \to)$ is a residuated lattice and $(A, +, 0)$ an abelian group. Define the following \define{sup-convolution}, $f, h \in \signals(A,\lattice{L})$, 
    \begin{align}
        (h \conv f)(x) &=& \bigjoin_{s \in A} h(s) \star f(t - s).
    \end{align}
\end{definition}

\noindent Here is an example leading to nonlinear filters in the time domain.

\begin{example}[1-D Max-Plus Convolution ]
    Suppose $\domain$ is $\Z$ and $\channels$ is the lattice $\Rext$. If $h, f: \Z \to \Rext$ are signals, then 
    \begin{align*}
        (h \conv f)(t) &=& \max_{s \in \Z} h(s) + f(t-s).
    \end{align*}
    Of course, $\channels$ could have multiple channels (i.e.~$\channels = \Rext^d$) in which case each channel can be processed in parallel.
\end{example}

Let $\alpha \in \resid{L}$ be an element of a residuated lattice. If $f \in \signals(\Omega, \resid{L})$ is a residuated-lattice -valued signal over some domain $\domain$, we define a \define{scalar multiplication} of a signal $f$ by an $a \in \lattice{L}$
\begin{align*}
    (a \star f)(t) &=& a \star f(t) \quad \forall t \in \domain.
\end{align*}
If $f : \Z \to \channels$ is a signal, then its domain translation is $T_sf: \Z \to \channels$ given by $T_sf(t) = f(t-s)$.
We say an operator $D: \lattice{L}^\Z \to \lattice{L}^\Z$ is \define{dilation-translation invariat} if the following conditions are satisfied
\begin{align*}
    D \left(\bigjoin_{i \in I} f_i \right) &=& \bigjoin_{i \in I} D(f_i) \\
    D(a \star f) &=& a \star D(f) \\
    T_s D(f) &=& D(T_s f)
\end{align*}
for all $f \in \resid{L}^\Z$, $a \in \resid{L}$, $s \in \Z$. The first requirement states that $D$ is a join-preserving map with respect the lattice structure on $\resid{L}^\Z$. The second and third specify an equivariance property.

The following is the closest to an ``algebraic result'' for this flavor of signal processing.

\begin{theorem}[Theorem 3 \cite{maragos2009morphological}]
    An operator $D: \lattice{L}^\Z \to \lattice{L}^\Z$ is DTI dilation-translation invariant if and only if every signal $f \in \resid{L}^\Z$ is expressed as
    \begin{align*}
        f(t) &=& \bigjoin_{s \in \Z} f(s) \star D(q)(t-s)
    \end{align*}
    where $q: \Z \to \lattice{L}$ is the signal called the \define{impulse} defined
    \[ q(t) = \begin{cases}
        e & t = 0 \\
        0 & t \neq 0
    \end{cases}\]
\end{theorem}

\noindent Our contribution here is the following definition.

\begin{definition}[\pusmar Convolution]
Suppose $\poset{P}$ is a finite distributed lattice and $\resid{L}$ is a residuated lattice. Then, the \pusmar convolution of $h,f: \poset{P} \to \resid{L}$ is the map
\begin{align*}
    h \conv f &:& \poset{P} \to \resid{L} \\
    (h \conv f)(x) &=& \bigjoin_{y \in \lattice{K}} h(y) \star f(x \meet y)
\end{align*}
\end{definition}

\begin{theorem}
    Suppose $f: \poset{P} \to \resid{L}$ is join-preserving. Then, $h \conv f: \poset{P} \to \resid{L}$ is join-preserving, also. Furthermore, the map
    \begin{align*}
        - \conv f &:& \resid{L}^\poset{P} \to \resid{L}^\poset{P}, \\
        (h \conv f)(x) &=& \bigjoin_{y \in \poset{P}} f(x \meet y) \ast h(y) 
    \end{align*}
    has a right adjoint
    \begin{align*}
        - \conv' f &:& \resid{L}^\poset{P} \to \resid{L}^\poset{P}, \\
        (h \conv' f)(y) &=& \bigmeet_{x \in \poset{P}} f(x \meet y) \to h(y).
    \end{align*}
\end{theorem}
\begin{proof}
    Follows from Theorem \ref{thm:kernels}.
\end{proof}

%---------------------------------------
\section{Convolutional Neural Networks}
%---------------------------------------

Convolution neural networks (neural networks in general) consist of a series of layers. Each layer $\ell$  has an input signal $\textbf{f}^{(\ell-1)}$ and an output signal $\textbf{f}^{(\ell)}$ so that the output to layer $\ell$ is the input of layer $\ell+1$. The input layer is $\mathbf{x} = \mathbf{f}^{(0)}$; the output layer is denoted $\mathbf{y}$. In each layer, the following operations may be performed.
\begin{enumerate}
    \item Filter of the input $\textbf{h}^{\ell}$ signal with a bank of filters.
    \item Aggregate filtered signals.
    \item Apply a nonlinearity $\sigma_{\ell}$ (e.g.~agregation,ReLU or sigmoid).
    \item (Optional) Pooling.
\end{enumerate}
In our view, convolutional neural networks fall under the guise of signal processing.

%----------------------------------------------------------
\subsection{Lattice neural networks}
%----------------------------------------------------------

Convolutional layers in a lattice convolutional neural network (L-CNN) have two flavors. Meet- and join- convolutional layers. Fix a finite lattice $\lattice{L}$, and choose an (arbitrary) total order on the elements $\lattice{L} = \{\hat{0} = y_1, x_2, \dots, x_m = \hat{1}\}$. Hence, we can identify a lattice signal $f: \lattice{L} \to \R$ as a vector $\mathbf{f} \in \R^m$. For layer $\ell$, let $F_{\ell}$ denote the number of features in that layer. Each feature has an associated lattice signal, stored in a matrix \[\mathbf{F}^{(\ell)} = \{\mathbf{f}_i\}_{i = 1}^{F_{\ell-1}} \in \R^{m \times F_{\ell}}.\] Similarly, filters are stored in a tensor \[\mathbf{H}^{(
\ell)} = \{\mathbf{h}_{i,j}^{(\ell)}\}_{\substack{i=1,\dots, F_{\ell-1} \\ j = 1, \dots, F_{\ell}}} \in \R^{m \times F_{\ell-1} \times F_{\ell}}.\] Then, a meet-convolution layer is
\begin{align*}
    \mathbf{F}^{\ell}_j &=& \sigma_{\ell} \left( \sum_{i = 1}^{F_{\ell-1}} \mathbf{h}_{i,j}^{(\ell)} \conv_{\meet} \mathbf{f}_j  \right) \\
    j &\in& \{1,2,\dots,F_{\ell} \},
\end{align*}
while a join-convolutional layer is
\begin{align*}
    \mathbf{F}^{\ell}_j &=& \sigma_{\ell} \left( \sum_{i = 1}^{F_{\ell-1}} \mathbf{h}_{i,j}^{(\ell)} \conv_{\join} \mathbf{f}_j  \right) \\
    j &\in& \{1,2,\dots,F_{\ell} \}.
\end{align*}
Either choice consists of lattice convolution by filters $\mathbf{h}_i^j$, aggregation (sum), followed by nonlinearity $\sigma_{\ell}$. A note on implementation, since signals in every layer are supported on the same lattice, a Fourier basis $\{f^{\leq y_i}\}_{i = 1, \dots, m}$ suffices to write every convolution operation as a matrix multiplication.
In practice, you may implement a L-CNN with a few fully-connected (MLP \cite{goodfellow2016deep}) layers, for example, in a classifier. 

\begin{example}[Point Cloud Classification]
    A point cloud $\Space{Y}$ is a finite subset of $\R^d$. For instance, point clouds can be collected from LiDAR (Light Detection and Ranging) systems. There is a well-known pipeline to ``approximate'' $\Space{Y}$ as a simplicial complex. Given resolution $\epsilon > 0$, define the \define{Vietoris-Rips complex}
    \begin{align*}
        \mathrm{VR}_{\epsilon}(\Space{Y}) &=& \{ \sigma \in \powerset{\Space{Y}}~\vert~ B_{\epsilon}(x) \cap B_{\epsilon}(y) \neq \emptyset \quad \forall x,y \in \sigma \}.
    \end{align*}
    This leads to a filtration $\Space{X}_{\epsilon} = \mathrm{VR}_{\epsilon}$. It has been argued that, while this filtration is robust to some kinds of noise, it is not robust to outliers. It was originally suggest by Carlsson and Zomorodian \cite{carlsson2009theory} to filter $\Space{Y}$ also by density $\rho$ (e.g. estimated with k-nearest neighbors). The intuition is that certain topological features can be captured at different sampling resolutions. Let
    \begin{align*}
        \mathrm{VR}_{\epsilon}(\Space{Y}^{\geq \rho})
    \end{align*}
    denote the Vietoris-Rips complex of the sub- point cloud of point with density greater than or equal to $\rho$. Then, $\Space{X}_{(\epsilon, \rho)} = \mathrm{VR}_{\epsilon}(\Space{Y}^{\geq \rho})$ clearly forms a bifiltration of simplicial complexes. Bi-filtrations have several known topological invariants: the Hilbert function and bi-graded Betti numbers \cite{lesnick2015interactive}. Upon discretizing and rescaling the domain of $\epsilon$ and $\rho$, each invariant is seen as a signal on the grid lattice $\lattice{L} = [m] \times [n]$ where $[n] = \{0 < 1 < \cdots <n\}$. To this end, we trained a L-CNN \cite{riess2020multidimensional} in order to classify point clouds in the ModelNet10 dataset \cite{wu20153d} strictly using topological invariants produced with the RIVET software package \cite{lesnick2015interactive} as features. The results are summarized in Appendix \ref{ch:appendix-2}.
\end{example}

%-------------------------------------------------
\subsection{Towards graph neural networks (GNNs)}
%-------------------------------------------------

Graph Neural Networks (GNNs) are a powerful models offering more efficient and generalizable solutions to node classification, network classification, and representation learning tasks. It is widely argued that graph diffusion is the driving force behind graph neural networks, the impetus for our study of foundational models of graph neural networks based upon the Tarski Laplacian (as well as other linear sheaf Laplacians \cite{bodnar2022neural,barbero2022sheaf,hansen2020sheaf}).

We will review three types of GNNs, each of which brings to light different aspects of neural information processing on graphs. First we establish some common notation. Suppose $\graph{G} = (\nodes{G}, \edges{G}, \weight)$ is a weighted graph with $n = |\nodes{G}|$ nodes, (normalized) adjacency matrix $A \in \R^{n \times n}$, degree matrix $D \in \R^{n \times n}$, and (normalized) graph Laplacian $L\in \R^{n \times n}$. Suppose each layer $\ell$ of the GNN has $F_\ell$ features so that $\mathbf{X}^{(\ell)} \in \R^{n \times \ell}$ is a feature matrix of layer $\ell$. Each layer is associated another matrix $\mathbf{H}^{(\ell)} \in \R^{F_{\ell} \times F_{\ell+1}}$ called the \define{weight matrix} or \define{channel-mixing} matrix. The entries of each $\mathbf{H}^{(\ell)}$ are learnable parameters. Finally, each layer has a nonlinear function $\sigma_\ell$ that is either pointwise (e.g.~ReLU, sigmoid) or a localized nonlinear function such as local maximum or median \cite{ruiz2019invariance}.

\begin{enumerate}
\item Graph Convolutional Networks (GCNs) \cite{kipf2016semi} are deep neural networks that incorporate elements of graph diffusion and mutli-layer perceptrons. Layers are of the form
\begin{align}
    \mathbf{X}^{(\ell + 1)} &=& \sigma_{\ell} \left( \Laplacian \mathbf{X}^{(\ell)} \mathbf{H}^{(\ell)} \right),
\end{align}
or, locally,
\begin{align*}
    x_{i,g}^{(\ell+1)} &=& \sigma_{\ell} \left( \sum_{f = 1}^{F_{\ell}} h_{f,g}^{(\ell)} \left(\Laplacian \mathbf{x}^{(\ell)}_f\right)_i \right), \quad \forall i \in \nodes{G},~g \in \{1,2,\dots,F_{\ell+1}\}.
\end{align*}
Only one diffusion step takes place in each layer. Graph convolution takes place, but each graph filter only has one non-zero parameter.
\item Message-Passing Neural Networks (MPNNs) \cite{gilmer2017neural} take quite a general form.
\begin{enumerate}
    \item Each node $i \in \nodes{G}$ receives messages from neighboring nodes consisting via a function $\psi_{i,j}^{(\ell)}\left(\mathbf{f}_i^{(\ell)}, \mathbf{f}_j^{(\ell)}, a_{ij}\right)$ (learned parameters possible).
    \item Node $i \in \graph{G}$ aggregates the messages received by neighbors with a map \[\aggregate: \R^{\nbhd{i}} \to \R,\] for instance, sum, maximum, or mean.
    \item Then, a node updates its state with an update function $\phi^{(\ell)}$ (learned parameters possible) depending on the previous state $\mathbf{f}_i^{(\ell)}$ and the aggregated message $\mathbf{m}_i^{(\ell+1)}$.
\end{enumerate}
    In summary,
    \begin{align}
        \mathbf{m}^{(\ell+1)}_i &=& \aggregate_{w \in \nbhd{i}} \psi_{i,j}^{(\ell)} \left( \mathbf{x}_i^{(\ell)}, \mathbf{x}_j^{(\ell)}, a_{ij} \right), \\ \nonumber
        \mathbf{x}^{(\ell+1)}_i &=& \phi^{(\ell)}_i \left( \mathbf{x}_i^{(\ell)}, \mathbf{m}_i^{(\ell+1)} \right). \label{eq:pmnn}
    \end{align}
    \item Graph Convolutional Neural Networks (GCNNs) \cite{gama2018convolutional} epitomize a perspective that \emph{GNNs are build from filters} and emphasize the spectral domain. Each layer can have multiple diffusions take place, aggregating information in a $k$-hop neighorhood, as opposed to a $1$-hop neighborhood.
\begin{align*}
    x_{i,g}^{(\ell+1)} &=& \sigma_{\ell} \left( \sum_{f = 1}^{F_{\ell}} \sum_{k=0}^{K_{\ell}} h_{f,g,k}^{(\ell)} \left(\Laplacian^k \mathbf{x}^{(\ell)}_f\right)_i \right), \quad \forall i \in \nodes{G},~g \in \{1,2,\dots,F_{\ell+1}\}.
\end{align*}
\end{enumerate}

An iteration of heat flow \refeq{eq:heat-equation} can be viewed as a layer in a message-passing neural network as follows. Suppose $\sheaf{F}$ is a Tarski sheaf over $\graph{G}$. Then, with the following,  
\begin{align*}
    \psi_{i,j}^{(\ell)}(\mathbf{x}_i^{(\ell)}, \mathbf{x}_j^{(\ell)}, a_{ij}) &=& \Parallel{F}{j \to i}(\mathbf{x}_j^{(\ell)}) \\
    \aggregate_{j \in \nbhd{i}} \mathbf{x}_j &=& \bigmeet_{j \in \nbhd{i}} \mathbf{x}_j \\
    \phi^{(\ell)}_i(\mathbf{x}^{(\ell)}_i, \mathbf{m}^{(\ell+1)}_i) &=& \mathbf{x}^{(\ell)}_i \meet \mathbf{m}^{(\ell+1)}_i,
 \end{align*}
 heat flow is a message-passing neural network. 

An open question remains whether the Tarski Laplacian inspires convolutional neural networks. With respect to an underlying lattice-valued sheaf over $\graph{G}$, we hypothesize that signals on lattice $\bisheaf{F}(i)$ are correlated related to signals on neighboring lattices $\{ \bisheaf{F}(j) \}_{j \in \nbhd{i}}$. This leads to a possible notion of signal convolution relative to sheaf.

\begin{definition}[Convolution Relative to a Sheaf]
Suppose $\graph{G} = (\nodes{G}, \edges{G})$ is a graph and $\bisheaf{F}$ is a Tarski sheaf over $\graph{G}$. Suppose $\mathbf{F} = \{ \mathbf{f}_i \}_{i \in \nodes{G}}$ is a collection of signals over lattices indexed by $\nodes{G}$, suppose $\mathbf{H} = \{\mathbf{h}^{i,j} \}_{i, j \in \nodes{H}}$ in a filter bank, and suppose $\Laplacian: C^0(\graph{G}; \sheaf{F}) \to C^0(\graph{G}; \sheaf{F})$ is the Traski Laplacian. Then, \define{convolution of $\mathbf{F}$ by a filterbank $\mathbf{H}$ relative to $\bisheaf{F}$} is the operation
\begin{align}
    (\mathbf{H} \star_{\bisheaf{F}} \mathbf{F})_{i} &=& \sum_{j \in \nodes{G}} \sum_{\mathbf{y} \in C^0(\graph{G}; \sheaf{F})} h^{i,j}_{y_i} T_{(\Laplacian \mathbf{y})_i}^{\meet} \mathbf{f}_i
\end{align}
\end{definition}

%% file: gfx/hypergraph.tex
\tikzset{every picture/.style={line width=0.75pt}} %set default line width to 0.75pt        

\begin{tikzpicture}[x=0.75pt,y=0.75pt,yscale=-1,xscale=1]
%uncomment if require: \path (0,315); %set diagram left start at 0, and has height of 315

%Shape: Circle [id:dp014185129789872741] 
\draw   (99.5,155) .. controls (99.5,141.19) and (110.69,130) .. (124.5,130) .. controls (138.31,130) and (149.5,141.19) .. (149.5,155) .. controls (149.5,168.81) and (138.31,180) .. (124.5,180) .. controls (110.69,180) and (99.5,168.81) .. (99.5,155) -- cycle ;
%Shape: Circle [id:dp41242872033056455] 
\draw   (115,125.5) .. controls (115,111.69) and (126.19,100.5) .. (140,100.5) .. controls (153.81,100.5) and (165,111.69) .. (165,125.5) .. controls (165,139.31) and (153.81,150.5) .. (140,150.5) .. controls (126.19,150.5) and (115,139.31) .. (115,125.5) -- cycle ;
%Shape: Circle [id:dp7370949189294205] 
\draw   (131.5,154.86) .. controls (131.5,141.05) and (142.69,129.86) .. (156.5,129.86) .. controls (170.31,129.86) and (181.5,141.05) .. (181.5,154.86) .. controls (181.5,168.66) and (170.31,179.86) .. (156.5,179.86) .. controls (142.69,179.86) and (131.5,168.66) .. (131.5,154.86) -- cycle ;
%Shape: Circle [id:dp24491800274672126] 
\draw  [fill={rgb, 255:red, 0; green, 0; blue, 0 }  ,fill opacity=1 ] (137.17,160.31) .. controls (137.17,158.79) and (138.4,157.56) .. (139.92,157.56) .. controls (141.44,157.56) and (142.67,158.79) .. (142.67,160.31) .. controls (142.67,161.83) and (141.44,163.06) .. (139.92,163.06) .. controls (138.4,163.06) and (137.17,161.83) .. (137.17,160.31) -- cycle ;
%Shape: Circle [id:dp20406618751526762] 
\draw   (204.36,155.29) .. controls (204.36,141.48) and (215.55,130.29) .. (229.36,130.29) .. controls (243.16,130.29) and (254.36,141.48) .. (254.36,155.29) .. controls (254.36,169.09) and (243.16,180.29) .. (229.36,180.29) .. controls (215.55,180.29) and (204.36,169.09) .. (204.36,155.29) -- cycle ;
%Shape: Circle [id:dp3522354737136386] 
\draw   (219.86,125.79) .. controls (219.86,111.98) and (231.05,100.79) .. (244.86,100.79) .. controls (258.66,100.79) and (269.86,111.98) .. (269.86,125.79) .. controls (269.86,139.59) and (258.66,150.79) .. (244.86,150.79) .. controls (231.05,150.79) and (219.86,139.59) .. (219.86,125.79) -- cycle ;
%Shape: Circle [id:dp8902224963314618] 
\draw   (236.36,155.14) .. controls (236.36,141.34) and (247.55,130.14) .. (261.36,130.14) .. controls (275.16,130.14) and (286.36,141.34) .. (286.36,155.14) .. controls (286.36,168.95) and (275.16,180.14) .. (261.36,180.14) .. controls (247.55,180.14) and (236.36,168.95) .. (236.36,155.14) -- cycle ;
%Shape: Circle [id:dp24169638028191276] 
\draw  [fill={rgb, 255:red, 0; green, 0; blue, 0 }  ,fill opacity=1 ] (117.74,155.74) .. controls (117.74,154.22) and (118.98,152.99) .. (120.49,152.99) .. controls (122.01,152.99) and (123.24,154.22) .. (123.24,155.74) .. controls (123.24,157.26) and (122.01,158.49) .. (120.49,158.49) .. controls (118.98,158.49) and (117.74,157.26) .. (117.74,155.74) -- cycle ;
%Shape: Circle [id:dp13486960782583313] 
\draw  [fill={rgb, 255:red, 0; green, 0; blue, 0 }  ,fill opacity=1 ] (157.74,155.74) .. controls (157.74,154.22) and (158.98,152.99) .. (160.49,152.99) .. controls (162.01,152.99) and (163.24,154.22) .. (163.24,155.74) .. controls (163.24,157.26) and (162.01,158.49) .. (160.49,158.49) .. controls (158.98,158.49) and (157.74,157.26) .. (157.74,155.74) -- cycle ;
%Shape: Circle [id:dp9087826922642508] 
\draw  [fill={rgb, 255:red, 0; green, 0; blue, 0 }  ,fill opacity=1 ] (137.46,118.6) .. controls (137.46,117.08) and (138.69,115.85) .. (140.21,115.85) .. controls (141.73,115.85) and (142.96,117.08) .. (142.96,118.6) .. controls (142.96,120.12) and (141.73,121.35) .. (140.21,121.35) .. controls (138.69,121.35) and (137.46,120.12) .. (137.46,118.6) -- cycle ;
%Shape: Circle [id:dp874584284481152] 
\draw  [fill={rgb, 255:red, 0; green, 0; blue, 0 }  ,fill opacity=1 ] (242.03,120.6) .. controls (242.03,119.08) and (243.26,117.85) .. (244.78,117.85) .. controls (246.3,117.85) and (247.53,119.08) .. (247.53,120.6) .. controls (247.53,122.12) and (246.3,123.35) .. (244.78,123.35) .. controls (243.26,123.35) and (242.03,122.12) .. (242.03,120.6) -- cycle ;
%Shape: Circle [id:dp3167518155954585] 
\draw  [fill={rgb, 255:red, 0; green, 0; blue, 0 }  ,fill opacity=1 ] (220.03,159.17) .. controls (220.03,157.65) and (221.26,156.42) .. (222.78,156.42) .. controls (224.3,156.42) and (225.53,157.65) .. (225.53,159.17) .. controls (225.53,160.69) and (224.3,161.92) .. (222.78,161.92) .. controls (221.26,161.92) and (220.03,160.69) .. (220.03,159.17) -- cycle ;
%Shape: Circle [id:dp5763197570956211] 
\draw  [fill={rgb, 255:red, 0; green, 0; blue, 0 }  ,fill opacity=1 ] (264.32,159.17) .. controls (264.32,157.65) and (265.55,156.42) .. (267.07,156.42) .. controls (268.59,156.42) and (269.82,157.65) .. (269.82,159.17) .. controls (269.82,160.69) and (268.59,161.92) .. (267.07,161.92) .. controls (265.55,161.92) and (264.32,160.69) .. (264.32,159.17) -- cycle ;
%Shape: Circle [id:dp8051679040161526] 
\draw  [fill={rgb, 255:red, 0; green, 0; blue, 0 }  ,fill opacity=1 ] (241.59,145.01) .. controls (241.59,143.49) and (242.82,142.26) .. (244.34,142.26) .. controls (245.86,142.26) and (247.09,143.49) .. (247.09,145.01) .. controls (247.09,146.52) and (245.86,147.76) .. (244.34,147.76) .. controls (242.82,147.76) and (241.59,146.52) .. (241.59,145.01) -- cycle ;
%Shape: Circle [id:dp37631988726217847] 
\draw  [fill={rgb, 255:red, 0; green, 0; blue, 0 }  ,fill opacity=1 ] (150.32,138.6) .. controls (150.32,137.08) and (151.55,135.85) .. (153.07,135.85) .. controls (154.59,135.85) and (155.82,137.08) .. (155.82,138.6) .. controls (155.82,140.12) and (154.59,141.35) .. (153.07,141.35) .. controls (151.55,141.35) and (150.32,140.12) .. (150.32,138.6) -- cycle ;
%Shape: Circle [id:dp5728967883463594] 
\draw  [fill={rgb, 255:red, 0; green, 0; blue, 0 }  ,fill opacity=1 ] (126.03,138.6) .. controls (126.03,137.08) and (127.26,135.85) .. (128.78,135.85) .. controls (130.3,135.85) and (131.53,137.08) .. (131.53,138.6) .. controls (131.53,140.12) and (130.3,141.35) .. (128.78,141.35) .. controls (127.26,141.35) and (126.03,140.12) .. (126.03,138.6) -- cycle ;

\end{tikzpicture}

%% file: gfx/hypergraph-2.tex
% https://q.uiver.app/?q=WzAsMTQsWzEsMywiXFxidWxsZXQiXSxbMCwyLCJcXGJ1bGxldCJdLFsxLDIsIlxcYnVsbGV0Il0sWzIsMiwiXFxidWxsZXQiXSxbMCwxLCJcXGJ1bGxldCJdLFsxLDEsIlxcYnVsbGV0Il0sWzIsMSwiXFxidWxsZXQiXSxbMSwwLCJcXGJ1bGxldCJdLFs0LDMsIlxcYnVsbGV0Il0sWzQsMiwiXFxidWxsZXQiXSxbMywxLCJcXGJ1bGxldCJdLFs0LDEsIlxcYnVsbGV0Il0sWzUsMSwiXFxidWxsZXQiXSxbNCwwLCJcXGJ1bGxldCJdLFswLDFdLFswLDJdLFswLDNdLFsxLDVdLFsyLDRdLFsyLDZdLFszLDZdLFszLDVdLFsxLDRdLFs0LDddLFs1LDddLFs2LDddLFs4LDldLFs5LDEwXSxbOSwxMV0sWzksMTJdLFsxMCwxM10sWzExLDEzXSxbMTIsMTNdXQ==
\[\begin{tikzcd}
	& \bullet &&& \bullet \\
	\bullet & \bullet & \bullet & \bullet & \bullet & \bullet \\
	\bullet & \bullet & \bullet && \bullet \\
	& \bullet &&& \bullet
	\arrow[from=4-2, to=3-1, no head]
	\arrow[from=4-2, to=3-2, no head]
	\arrow[from=4-2, to=3-3, no head]
	\arrow[from=3-1, to=2-2, no head]
	\arrow[from=3-2, to=2-1, no head]
	\arrow[from=3-2, to=2-3, no head]
	\arrow[from=3-3, to=2-3, no head]
	\arrow[from=3-3, to=2-2, no head]
	\arrow[from=3-1, to=2-1, no head]
	\arrow[from=2-1, to=1-2, no head]
	\arrow[from=2-2, to=1-2, no head]
	\arrow[from=2-3, to=1-2, no head]
	\arrow[from=4-5, to=3-5, no head]
	\arrow[from=3-5, to=2-4, no head]
	\arrow[from=3-5, to=2-5, no head]
	\arrow[from=3-5, to=2-6, no head]
	\arrow[from=2-4, to=1-5, no head]
	\arrow[from=2-5, to=1-5, no head]
	\arrow[from=2-6, to=1-5, no head]
\end{tikzcd}\]

%% file: Chapters/Chapter09.tex
%************************************************
\chapter{Semantics}\label{ch:semantics}
%************************************************

Our second application domain can be loosely described as \define{semantics}. In this chapter, we explore semantics in two different contexts: temporal logic and multi-agent logic.

In logic, syntax consists of strung-together symbols, while semantics is found ``under the hood.'' Semantics also has a linguistic connotation as something like ``subsurface layers of meaning,'' and, of course, meaning is highly dependent on context. Thus, semantics is synonymous with a formal context which we call a \define{model} and others call a \define{structure}. A model consists of a set of \define{states} and rules for verifying whether a given state satisfies a logic formula.
Model checking \cite{baier2008principles} is a subfield of formal verification in computer science dedicated to solving the following problems.
\begin{problem}[Model Checking]
	Given a model, a particular state $s$, and a formula $\phi$ (in an appropriate language), decide whether $s$ satisfies $\phi$, written $(\model, s) \models \phi$.
\end{problem}
\begin{problem}[Validity]
	Given a model $\model$ and a formula $\phi$, decide whether every state satisfies a given formula or, equivalently, check every $s$ satisfies $\neg \phi$. If $\phi$ is valid, it is common to write $\model \models \phi$.
\end{problem}
\begin{problem}[Planning]
	Given a $\model$ and a formula $\phi$, produce a state $s$ with $(\model,s) \models \phi$ (if such a state exists).
\end{problem}

For models whose rules are specified by satisfying a given relation, we observe that Galois connections induced by relations (Theorem \ref{thm:semantic-connection}) transform semantic content. Thus, we argue, as others have done \cite{landman2012structures},  that order theory is practical in semantics. To date, we see formal theories of semantics having been applied to single- and multi-agent autonomous systems as well as artificial intelligence. It is in the area of autonomy that we hope our work can make an impact.

%-----------------------------------------------
\section{Temporal Logic}
%-----------------------------------------------

Temporal logic is a modal logic with much success in interpreting abstract systems by specifying logical constraints a system must satisfy as well as the time or duration. Two critical properties that temporal logic model checking address are \define{saftey}, something ``bad'' will never happen, and \define{liveness}, something ``good'' will always happen, at least eventually.  More recently, temporal logic control systems have been utilized in path-planning \cite{smith2011optimal,fainekos2005hybrid} and multi-agent autonomous systems \cite{kantaros2016distributed}.

%-------------------------------------------
\subsection{Transition systems}
%-------------------------------------------

The logic of time, temporal logic, is powerful for modeling systems. We utilize a popular method of abstracting away a system. We begin with a set of \define{atomic propositions}, denoted $\Phi$. Elements of $\Phi$ are labeled $p_1, p_2$ etc. In the most elementary setting, these propositions have two possible truth values $\{\mathtt{false}, \mathtt{true}\}$ and typically designate facts about a system, for instance, whether an agent is located in particular region of a map, or whether there is a communication link between a pair of agents. In temporal logic, the truth values of atomic propositions are time-variant.

\begin{definition}
A (labeled nondeterministic) \define{transition system} is a tuple $\TS = (Q, Q_0, \Sigma, \delta, \AP, \eval)$ with
\begin{enumerate}
	\item $Q$, a set of states,
	\item $Q_0 \subseteq Q$, a set of initial states,
	\item $\Sigma$, a set of control inputs,
	\item $\delta$, a transition map sending a states $q \in Q$ and a control input $\sigma \in \Sigma$ to a subset of possible next states
	\begin{align}
		\delta: Q \times \Sigma \to \powerset{Q} \label{eq:state-transition},
	\end{align}
	\item $\AP$, a set of atomic propositions (true or false),
	\item $\eval$, a map sending a state $q \in Q$ to a set of atomic propositions true at $q$
	\begin{align*}
		\eval: Q \to \powerset{\AP}.
	\end{align*}	
\end{enumerate}
\end{definition}
\noindent Notice, we could have write the transition map as a labeled relation
\begin{align*}
	\to~\subseteq~Q \times \Sigma \times Q
\end{align*}
by currying $\delta$. Thus, transition systems are conveniently drawn as directed graphs with arrows labeled by subsets of $\Sigma$. In an \define{deterministic transition system}, \eqref{eq:state-transition} is replaced by
\begin{align*}
	\delta: Q \times \Sigma \to Q.
\end{align*}
and we assume $Q_0 = \{Q_0\}$. In an \define{unlabeled transition system}, we ignore $\Sigma$ and replace \eqref{eq:state-transition} by
\begin{align}
	\delta: Q \to \powerset{Q}, \label{eq:unlabeled-state-transition}
\end{align}
equivalently, a relation $\to~\subseteq Q \times Q$. If $Q$, $\Sigma$ and $\AP$ are finite, we say the transition system is \define{finite}.

\begin{remark}
	In a transition system, uncertainty is not (ordinarily) modeled with probability, but agnostic to the likelihood of a given state $q$ transitioning to any number or states $q' \in \delta(q,\sigma)$. We can form an unlabeled transition system from a labeled one via
	\begin{align*}
		\tilde{\delta}(q) &=& \bigcup_{\sigma \in \Sigma} \delta(q,\sigma).
	\end{align*}
	An unlabeled transition system is equivalent to a binary relation $\to ~\subseteq~ Q \times Q$.
\end{remark}

\begin{example}[Control Systems]
Consider the discrete time-invariant linear system
\begin{align}
	\mathbf{x}[t+1] &=& A \mathbf{x}[t] + B \mathbf{u}[t] \label{eq:linear-system} \\
	\mathbf{y}[t+1] &=& C \mathbf{x}[t+1]
\end{align}
with $A \in \R^{n \times n}$, $\mathbf{x} \in \R^{n}$, $B \in \R^{n \times p}$, $\mathbf{u} \in \R^p$, $C \in \R^{m \times n}$, $\mathbf{y} \in \R^m$. Linear systems are ubiquitous in control theory, for instance, model-predictive control \cite{camacho2013model}. The system \eqref{eq:linear-system} is rewritten as the following transition system:
\begin{align*}
	Q &=& \R^n \\
	\Sigma &=& \R^p \\
	\AP &=& \R^m \\
	\delta(\mathbf{x},\mathbf{u}) &=& \{A \mathbf{x} + B \mathbf{u}\} \\
	\eval(x) &=& \{C \mathbf{y}\}.
\end{align*}
\end{example}

Let $Q^\omega$ denote the set of infinite sequences of elements in $Q$
\begin{align*}
	Q^{\omega} = \{\alpha: \N \to Q\}
\end{align*}
which we call traces.
If $\tau \in Q^{\omega}$, the \define{suffix} of $\tau$ is the trace $\tau[t_0]$ starting at $t_0 \geq 0$.
Suppose $\sigma_{in} \in \Sigma^{\omega}$ is a sequence of control inputs, for instance, driving directions: \emph{turn left at the light, go straight at the stop sign}. Then,
\begin{align*}
	\tau(t) &=& \delta(\sigma(t),\tau(t-1)) \\
	\tau(0) &\in& Q_0
\end{align*}
defines a \define{trace} $\tau \in {\powerset{Q}}^\omega$. Finally, a trace specifies a sequence of subsets of $\AP$ via
\begin{align*}
	\sigma_{out}(t) &=& \eval\left( \tau(t) \right).
\end{align*}

Before we describe how to subject an output $\sigma_{out}$ to temporal logical constraints, consider the following critical example.

\begin{example}[Network Connectivity \cite{riess2021temporal}] \label{eg:network-connectivity}
	Suppose $\graph{G} = (\nodes{G}, \edges{G})$ is an undirected graph.\footnote{Without loss of generality, we can think of $\graph{G}$ as the complete graph on vertices $\nodes{G}$} A \define{dynamic graph} is map
	\begin{align}
		\graph{G}_t: \N \to \mathrm{Subgraph} \left(\graph{G}\right) \label{eq:dynamic-graph}
	\end{align}
	 is a dynamic graph with edges $\edges{G}_t \subseteq \edges{G}$ collecting edges of the subgraph $\graph{G}_t$ at time $t$. At each discrete time instant, nodes in the edge $ij \in \edges{G}$ toggle their status from on to off or from off to on or are on standby. These actions are represented by the following control inputs
	\begin{align*}
	 	\Sigma &=& \prod_{ij \in \edges{G}} \{\sigma_{switch}, \quad \sigma_{sleep}\}.
	 \end{align*}
	 States are represented
	 \begin{align*}
	 	\mathbf{Q} &=& \{-1, +1\}^{|\edges{G}|}.
	 \end{align*}
	 State transitions are governed by the function
	 \[\delta(\mathbf{s}, \boldsymbol\sigma)_{ij} =
	 	\begin{cases}
	 		-q_{ij} & \sigma^{ij} = \sigma_{switch}   \\
	 		q_{ij}  & \sigma^{ij} = \sigma_{sleep}
	 	\end{cases}\]
	 given $\mathbf{q} \in Q$, $\boldsymbol \sigma \in \Sigma$.
	 Let
	 \begin{align*}
	 	\AP &=& \{ p^{ij} \}_{ij \in \edges{G}}
	 \end{align*}
	 where $p^{ij}$ is the proposition
	 \begin{displayquote}
	\emph{There is a connection between node $i$ and node $j$.}
	 \end{displayquote}
	 Then, of course,
	 \begin{align*}
	 	\eval(\mathbf{s}) &=& \bigcup_{q_{ij} > 0} \pi^{ij}.
	 \end{align*}
	 Hence, a control input $\boldsymbol \sigma_{in} \in \Sigma^\omega$ determines an output $\boldsymbol \sigma_{out} \in {\powerset{\AP}}^\omega$ specifying a subset of edges $\edges{G}_t$ that are connected at time $t$. So we see the defined transitions system $\TS = (\mathbf{Q}, \mathbf{q}_0, \Sigma, \AP, \eval)$ is equivalent to a dynamic graph \eqref{eq:dynamic-graph}. Shortly, we will find it helpful to be able to reason about dynamic graphs with transition systems.
\end{example}

%-------------------------------------------
\subsection{Linear temporal logic}
%-------------------------------------------

Linear temporal logic (LTL) has the following syntax
\begin{align*}
	\mathtt{true} \quad \vert \quad p \quad \vert \quad \phi \meet \psi \quad \vert \quad \neg \phi \quad \vert \quad \nexttime \phi \quad \vert \quad \phi \until \psi \quad \vert \quad \eventually \phi \quad \vert \quad \always \phi.
\end{align*}
There needs explanation. If $p$ is an atomic proposition in $\AP$, then $p$ is a formula. If $\phi$ and $\psi$ are formulas, then all of the operations of propositional calculus also constitute formulas. The other symbols are interpreted: $\nexttime \phi$ \emph{next time}, $\phi \until \psi$ \emph{until}, $\eventually \phi$ \emph{eventually}, $\always \phi$ \emph{always}. Derived operators include $\always \eventually \phi$ \emph{infinitely often} and $\eventually \always \phi$ \emph{at some point forever}. The collection of all formulas formed with this syntax is $\lang^{LTL}(\Phi)$.

The semantics of LTL is the following. Let $\sigma \in \powerset{\AP}^{\omega}$ be an input word, $p \in \AP$, and $\phi, \psi \in \lang^{LTL}(\AP)$. Then,
\begin{align*}
	\sigma &\models& p  \quad &\Leftrightarrow& \quad p \in \sigma(0), \\
	\sigma & \models& \phi \meet \psi \quad &\Leftrightarrow& \quad \sigma \models \phi,~\sigma \models \psi, \\
	\sigma & \models& \neg\phi \quad &\Leftrightarrow& \quad \sigma \not\models \phi, \\
	\sigma &\models& \nexttime \phi &\Leftrightarrow& \quad \sigma[1] \models \phi, \\
	\sigma &\models& \phi \until \psi &\Leftrightarrow& \quad \exists~t_0 \geq 0 \quad \text{such that} \quad \sigma[t_0] \models \psi \quad \text{and} \quad \sigma[t'] \models \phi \quad \forall 0 \leq t' < t_0, \\
	\sigma &\models& \eventually \phi &\Leftrightarrow& \quad \exists t_0 \geq 0 \quad \text{such that} \quad \sigma[t_0] \models \phi, \\
	\sigma &\models& \always \phi &\Leftrightarrow& \quad \sigma[t'] \models \phi \quad \forall t' \geq 0.
\end{align*}
For a trace $\tau \in \powerset{Q}^{\omega}$, we write $\tau \models \phi$ if for the cooresponding output work $\sigma_{in} \models \phi$.

Going back to syntax, it is useful to have a notion of which words satisfy a given proposition. Given $\phi$, the \define{intension} of $\phi$ is
\begin{align*}
	\intension{\phi} &=& \{ \tau \in \powerset{Q}^{\omega}~\vert~\tau \models \phi \}
\end{align*}
and the \define{words} satisfying $\phi$ is collected in the set $\mathrm{Words}(\phi) = \eval \left( \intension{\phi} \right)$.

Given a transition system $\TS$ and a proposition $\phi$ there is a computation pipeline to verify that $\phi$ is valid in $\TS$ \cite{baier2008principles}. The process involves ``translating'' a formula $\phi$ into a  \textbuchi automaton. 

\begin{definition}
	A \define{nondeterministic \textbuchi automaton (NBA)} is a tuple $\Buchi = (S, S_0, \Sigma, \delta, F)$ with
	\leavevmode
	\begin{enumerate}
		\item $S$, a set of states,
		\item $S_0$, a set of initial states,
		\item $\Sigma$, a set of control inputs,
		\item $\delta$, a transition map
		\begin{align*}
			\delta: S \times \Sigma \to \powerset{S},
		\end{align*}
		\item $F \subseteq S$, a set of accepting states.
	\end{enumerate}
	An \define{input word} $\sigma \in \Sigma^{\omega}$ is \define{accepted} if there is at least one corresponding sequence $\mathbf{s} = s(0)s(1)s(2)\dots$ such $s(0) \in S_0$ and the set $\{t~\vert~s(t) \in F  \}$ is countably infinite. The \define{language} of $\Buchi$, denoted $\lang(\Buchi)$ consists of the set of input words $\sigma$ that are accepted by $\Buchi$.
\end{definition}

\begin{theorem}[Theorem 5.37 \cite{baier2008principles}]
	For any $LTL$ formula $\phi \in \lang^{LTL}(\AP)$, there is a (generalized\footnote{A generalized \textbuchi nondeterministic automaton (GNBA) is lightly more general than a plain NBA. However, GNBA can be transformed into an NBA \cite[Theorem 4.56]{baier2008principles}.}) nondeterministic \textbuchi automaton $\mathbf{G}_{\phi}$ with $\Sigma = \powerset{\AP}$ such that
	\begin{align*}
		\mathrm{Words}(\phi) &=& \lang(\mathbf{G}_\phi).
	\end{align*}
\end{theorem}

Once we have $\phi$ translated into $\Buchi_{{}\phi}$, we construct the \define{product \textbuchi automaton} (PBA) $\mathbf{P} = \TS \otimes \Buchi_{\phi}$ \cite[p.~200]{baier2008principles} which can be viewed as a directed graph via \eqref{eq:unlabeled-state-transition}. After taking a quotient of the digraph of $\mathbf{P}$ by strongly connected components, use Djistra's algorithm \cite{cormen2022introduction} to search for paths beginning at an initial state $(q_0,s_0) \in Q \times S$ that visit an accepting state in $Q \times F$ infinitely often. A projection of this path in the product automaton $\mathbf{P}$ onto the original transition system $\mathbf{TS}$ yields a trace $\tau \models \phi$. We say $\TS \models \phi$ if there does not exist a trace $\tau$ such that $\tau \models \neg \phi$.

\begin{proposition}[Complexity \cite{baier2008principles}]
	The complexity of the model checking problem $\TS \models \phi$ is
	\begin{align*}
		\mathcal{O}\left( |\TS| 2^{|\phi|} \right)
	\end{align*}
	where $|\TS|$ is the number of states and transitions in $\TS$ and $|\phi|$ is (roughly) the number of symbols in $\phi$.
\end{proposition}

%-----------------------------------------
\subsection{Distributed model checking}
%-----------------------------------------

In the following, we consider a network $\graph{G} = (\nodes{G}, \edges{G})$ of agents and a number of atomic propositions indexed by the agents $i \in \nodes{G}$. In stead of model-checking a formula $\phi$ for the entire system, a strategy is to model-check \define{local formulas} $\phi_i$.

\begin{problem}[Distributed LTL Model-Checking] \label{prob:dist-model-checking}
	Suppose $\graph{G} = (\nodes{G}, \edges{G})$ is a graph and $\TS = (Q, Q_0, \Sigma, \delta, \AP)$ is a transition system such that $\AP$ is indexed by $\nodes{G}$ and a (possibly empty) set $\mathcal{A}$ not dependent on the graph $\graph{G}$. 
	\begin{align*}
		\AP &=& \{ p^i_A \}_{i \in \nodes{G}, A \in \mathcal{A}}. 
	\end{align*}
	Let $\AP_{i} = \{ p^j_A \}_{j \in \nbhd{i} \cup i, A \in \mathcal{A}}$ be the set of atomic propositions supported on the the neighborhood of $i \in \nodes{G}$. Let
	\begin{align*}
		\phi &=& \bigmeet_{i \in \nodes{G}} \phi_i
	\end{align*}
	where $\phi_i \in \lang(\Phi_i)$. Given $\sigma \in \Sigma^\omega$, decide whether $\sigma \models \phi$.
	% in time less than $\mathcal{O}(2^{|\edges{G}|})$. \todo{or should it be $\mathcal{O}(2^{|\nodes{G}|})$?}
\end{problem}

\noindent We include some examples of distributive model-checking problems.

\begin{example}[Intermittent Connectivity]
	The following example is taken from \cite{kantaros2016distributed}. Suppose robots with limited communication (e.g.~submarine robots) each travel in between two rendezvous locations $\mathbf{v}_i$ and $\mathbf{v}_j$ in order to commuicate with other agents. Then, let $\graph{G} = (\nodes{G}, \edges{G})$ be the graph with nodes indexing rendevous locations $\mathbf{v}_{ij}$ and edges indexing robots $r_{ij}$. This intermittent connectivity problem is a distributed model-checking problem with
	\begin{align*}
		\phi_i &=& \always\eventually\bigmeet_{j \in \nbhd{i}} p^{ij}_{\mathbf{v}_i}
	\end{align*}
	where $p^{ij}_{\mathbf{v}_i}$ is the proposition:
	\begin{displayquote}
		\emph{Robot $r_{ij}$ is located in the region $\mathcal{R}_{i} = \{\mathbf{x}~\vert~ \| \mathbf{x}- \mathbf{x}_{ij} \| \leq \epsilon\}$.}
	\end{displayquote}
	The local formulas $\phi_i$ specify:
	\begin{displayquote}
		 \emph{All robot assigned to the rendevous point $\mathbf{v}_i$ are simultaneously located in the region $\mathcal{R}_{i}$ infinitely often.}
	\end{displayquote}
\end{example}
\begin{example}[Epidemics]
	Let $\mathcal{A} = \{S, I, R\}$ represent the status \define{susceptible}, \define{infected}, and \define{recovered}. Suppose $\AP = \{ p^i_{A \in \mathcal{A}} \}$. For instance, $p^i_I$ is the proposition \emph{agent $i$ is infected}. The popular susceptible-infected-recovered (SIR) model can be viewed as a distributed model-checking problem with
	\begin{align*}
	\phi_i &=& \always \left( \bigjoin_{j \in \nbhd{i}} p_S^i \Rightarrow \eventually p_I^j  \right) \meet \always \left( p_I^i \Rightarrow \eventually p_R^i \right).
	\end{align*}
	We interpret $\phi_i$:
	\begin{displayquote}
		\emph{If at least one neighbor of agent $i$ is infected, then agent $i$ will eventually be infected; if agent $i$ is infected, she will eventually recover.}
	\end{displayquote}
	Of course, there are similar models such as susceptible-infected-susceptible (SIS), susceptible-infected-recovered-diseased (SIRD). Note that epidemic models can be used for modeling rumors as well as physical diseases.
\end{example}
\begin{example}[Non-Interfering Connectivity \cite{riess2021temporal}]
	Consider an arbitrary undirected graph $\graph{G} = (\nodes{G}, \edges{G})$ and the transition system $\TS$ defined in Example \ref{eg:network-connectivity} modeling network connectivity. Recall, $p^{ij}, ij \in \edges{G}$ is the proposition
	\begin{displayquote}
		\emph{Communication link $ij$ is active.}
	\end{displayquote}
	Consider the local formula
	\begin{align*}
		\phi_i &=& \bigmeet_{j \in \nbhd{i}} \phi_{ij}
	\end{align*}
	where $\phi_{ij}$ is the formula
	\begin{align}
		\phi_{ij} &=& \always \left( p^{ij} \Rightarrow \bigjoin_{i'j' \in \nbhd{ij}} \neg p^{i'j'} \right) \meet \always\eventually p^{ij}. \label{eq:noninterence}
	\end{align}
	$\phi_{ij}$ can be interpreted:
	\begin{displayquote}
		\emph{If communication link $ij$ is active, then every adjacent communication link must be off. Moreover, every communication link must become active infinitely often.}
	\end{displayquote}
\end{example}

We proposed \cite{riess2021temporal} a partial solution to the distributed model-checking problem, which applies to the above examples.
\begin{enumerate}
\item Select a subset $\nodes{C} \subseteq \nodes{G}$ of nodes which we call \define{command nodes}.
\item After selecting a global proximity radius $R \in \N$ and a local proximity radius $r \leq \lfloor R/2 \rfloor$, form the \define{command graph} with nodes $\nodes{C}$ and $ij \in \edges{C}$ if and only if $i$ and $j$ are connected by a shortest path of length no greater than $r$ in $\graph{G}$. Construct a family of subgraphs $\graph{G_i}, i \in \nodes{C}$ induced by the node sets
\begin{align*}
	\nodes{G_i} &=& \{j \in \nodes{G}~\vert~ d(i,j) \leq r \};
\end{align*}
$\graph{G_i}$ is the smallest subgraph of $\graph{G}$ containing $\nodes{G_i}$.
\end{enumerate}

For $i \in \nodes{C}$, let $\TS_i$ be the subtransaction system of $\TS$ (Problem \ref{prob:dist-model-checking}) restricted to $\graph{G_i}$. Model checking each $\TS_i \models \bigmeet_{j \in \nodes{G_i}} \phi_j,~i \in \nodes{C}$ is now substantially less expensive than model-checking $\TS \models \phi$.

%-----------------------------------------------
\section{Muli-Agent Logic}
%-----------------------------------------------

For more details, we suggest the reader consult the standard reference \cite{fagin2004reasoning}. For the following, suppose a system has agents labeled $\nodes{G} = \{1, 2, \dots, n\}$. In multi-agent logic, we develop a mathematical language to reason about information held by individuals and groups of individuals in the system. Information, in our setting, consists of formulas and their truth values, contrasting the perspective of Chapter \ref{ch:signals} which views information as signals, boxes of numbers. In temporal logic, there are two modalities $\always$ and $\eventually$. In multi-agent logic, we have a modality for every agent $K_1, K_2, \dots, K_n$.

\begin{definition}[Kripke Model] 
	An \define{Kripke frame} consists of a set $S$ of \define{states} with binary relations
	\begin{align*}
		\rel{K}_1, \rel{K}_2, \dots, \rel{K}_n.
	\end{align*}
	A \define{Kripke model}, denoted $\model$, is a frame $(S, \rel{K}_1, \rel{K}_2, \dots, \rel{K}_n)$ together with an \define{evaluation}
	\begin{align*}
		\pi: S \to \powerset{\AP}.
	\end{align*}
\end{definition}
It is obvious, by currying, that evaluations are in bijection with relations $\models$ between $S$ and $\AP$ which we call \define{interpretations}; we liberally go back and forth between evaluations and interpretations
\begin{align*}
(\model, s) \models p \quad &\Leftrightarrow& \quad p \in \pi(s).
\end{align*}
Multi-agent logic distinguishes itself from ordinary propositional logic with the set of relations $\rel{K}_i$ that encode the semantics for syntactic modal operators written $K_i$. Together with the connectives of propositional logic, the entire syntax of \define{multi-agent logic} is given by the following
\begin{align*}
\mathtt{true} \quad \vert \quad p \quad \vert \quad \phi \meet \psi \quad \vert \quad \neg \phi \quad \vert \quad i \quad \vert \quad K_i \varphi.
\end{align*}
The set of all sentences formed via the above syntax over the atomic propositions $\Phi$ is called the language $\lang_n(\AP)$. The set of all $n$-agent Kripke models over $\AP$ is denoted $\Models_n(\AP)$. We can \define{interpret} a $\phi \in \lang_n(\AP)$ with the following inductive rules:
\begin{align}
	(\model, s) &\models& p  \quad &\Leftrightarrow& \quad p \in \eval(s) \\
	(\model, s) &\models& \phi \meet \psi \quad &\Leftrightarrow& \quad (\model, s) \models \phi, (\model, s) \models \psi \\
	(\model, s) &\models& \neg \phi \quad &\Leftrightarrow& \quad (\model, s) \not \models \phi \\
	(\model, s) &\models& K_i \phi \quad &\Leftrightarrow& \quad (\model,t) \models \phi \quad \forall t \in S \quad \text{such that} \quad s\rel{K}_it
\end{align}

It is the last rule that gives a Kripke model the nickname, a ``possible worlds'' model because $K_i \phi$ is satisfied at a state if every neighboring state $t$ with $s~\rel{K}_i t$ satisfies $\phi$. As suggested by notation, the recommended mental model for the formula $K_i \phi$ is \emph{agent $i$ knows $\phi$}. We will see shortly why this is just a ``mental model.'' The modal logic $\lang_n(\Phi)$ could then be coined a \define{epistemic logic}. We write $\model \models \phi$ if $(\model, s) \models \phi$ for all $s \in S$. In this case, we would say $\phi$ is \define{valid} in $\model$.

Derived from the basic syntax and semantics, are additional logical operators. Suppose $A \subseteq \nodes{G}$ is a subset of agents. Then,
\begin{align*}
    E_A\phi &=& \bigmeet_{i \in A} K_i \phi
\end{align*}
which in our mental model specifies that \emph{everyone in $A$ knows $\phi$.} Define
\begin{align*}
    C_A\phi &=& \bigmeet_{k \geq 1} E^k_A \phi
\end{align*}
which in our mental model specifies that $\phi$ is \define{common knowledge}, or, in other words, everyone knows that everyone knows that everyone knows $\phi$ \emph{ad infinitum}. Note that $E_G$ is not in the language $\lang_n(\AP)$. To include common knowledge, we will need to augment our language to include the symbols $C_G$. Finally, let $D_A$ be the operator with semantics
\begin{align*}
    (\model, s) &\models& D_A \phi \quad &\Leftrightarrow& (\model, t) \models \phi \quad (\forall t \in S) \quad (s,t) \in \bigcap_{i \in A}\rel{K}_i.
\end{align*}
In our mental mode, $D_A$ supplies a notion of \define{distibutive knowledge} with the intuition that the ``combined'' knowledge of all the agents in the group $A \subseteq \nodes{G}$ implies $\phi$. We write $\lang_n^C$, $\lang_n^D$, $\lang_n^{CD}$ for the augmented languages with common knowledge, distributive knowledge, and both common and distributive knowledge, respectively. 

For the time being, we will narrow our focus on $\lang_n$. Let us briefly describe this logic axiomatically. An axiom system consists of axioms and inference rules. An axiom is a valid formula while an inference rule is a reduction of a set of valid formulas in a context to another formula. The axiom system $\mathsf{K}_n(\Phi)$ consists of the following rules and axioms:\footnote{We label axioms and inference rules per Fagin et.~al.~for convenience \cite{fagin2004reasoning}.}
\begin{itemize}
	\item[(A1)] Tautologies of propositional calculus.
	\item[(A2)] $\left(K_i \phi \meet K_i(\phi \meet \psi)\right) \Rightarrow K_i \psi \quad \forall i \in \nodes{G}$ (Distribution Axiom)
	\item[(R1)] From $\phi$ and $\phi \Rightarrow \psi$ infer $\psi$ (Modus Potens)
	\item[(R2)] From $\phi$ infer $K_i \phi$ (Knowledge Generalization)
\end{itemize}
The following theorem guarantees, given $\AP$, every formula provable in $\mathsf{K}_n$ is valid in $\Models_n$ and conversely every formula valid in $\Models_n$ is provable in $\mathsf{K}_n$.

\begin{theorem}[Theorem 3.1.3 \cite{fagin2004reasoning}]
The axiom system $\mathsf{K}_n(\AP)$ is sound and complete axiomatization of the language $\lang_n(\AP)$ with respect to $\mathcal{M}_n(\Phi)$. 
\end{theorem}

Additional axioms distinguishing knowledge over belief, result in alternative axiom systems which correspond semantically to subclasses of $\Models_n$ \cite{fagin2004reasoning}.

\begin{itemize}
    \item[(A3)] $K_i \phi \Rightarrow \phi $ (Knowledge Axiom)
    \item[(A4)] $K_i \phi \Rightarrow K_i K_i \phi$ (Positive Introspection)
    \item[(A5)] $\neg K_i \phi \Rightarrow K_i \neg K_i \phi$ (Negative Introspection)
    \item[(A6)] $\neg K_i(\mathtt{false})$ (Consistency )
    \item[(A7)] $\phi \Rightarrow K_i \neg K_i \neg \phi$ (Monadicity)
\end{itemize}

We will not say much more about the resulting correspondences between classes of Kripke models and axiom systems except that axiom systems with more axioms than $\mathsf{K}_n$ restrict the class of allowable Kripke relation $\rel{K}_i$. For instance, the axiom system called $\mathsf{S5}_n$ ($\mathsf{K}_5$ plus A3, A4, A5) restricts $\Models_n$ to the class of models $\Models_n^{rst}$ whose Kripke relations $\rel{K}_i$ are equivalence relations. A case can be made, thus, that $\Models_n^{rst}$ models knowledge. Hence, we interpret the formula $K_i \phi$ as \emph{agent $i$ knows $\phi$}.

In another instance, the axiom system called $\mathsf{KD45}_n$ ($\mathsf{K}_5$ plus A4, A5, A6) restricts $\Models_n$ to the class of models $\Models_n^{elt}$ whose Kripke relations $\rel{K}_i$ are Euclidean, serial and transitive. $\mathsf{KD45}_n$ is reasonably interpreted as modeling belief with $K_i$ interpreted as \emph{agent $i$ believes $\phi$} and, of special note, $\neg K_i \neg \phi$ as \emph{agent $i$ doesn't disbelieve $\phi$}. With this interpretation in mind, notice the absence of A3 in $\mathsf{KD45}_n$; believing something does not make it true. 
	
%---------------------------------------
\subsection{Kripke-Galois connections}
%---------------------------------------

Recall from our discussion of temporal logic, intension allows you to pass from syntax and semantics. In Kripke semantics, the \define{intension} of a formula describes the semantic content of that formula, the subset of states for which a given formula is true. In general, we refer to a subset $e \subseteq \powerset{S}$ as an \define{event}.
\begin{definition}[Intension]
	Suppose $\model$ is a Kripke model and $\phi \in \lang_n(\Phi)$ a formula. Then,
	\begin{align*}
		\intension{\phi} &=& \{s \in S~\vert~(\model,s) \models \phi \}
	\end{align*}
	is the \define{intension} of $\phi$. We say $\phi$ has more \define{semantic content} than $\psi$ if $\intension{\phi} \supseteq \intension{\psi}$. We say two formulae $\phi, \psi \in \lang_n(\Phi)$ are \define{semantically equivalent} if $\intension{\phi} = \intension{\psi}$ and write $\phi \equiv \psi$.
\end{definition}
\begin{proposition}
	Given a model $\model$, semantic equivalence is an equivalence relation on $\lang_n(\AP)$.
\end{proposition}
Let $\lang_n(\AP)/\simeq$ denote the set of equivalence classes of formulas where individual classes are denoted $[\phi]$.
Semantic equivalence is not to say that two formulas $\phi$ and $\psi$ are the same, but that they are the same in the context $\model$.
We now begin an investigation of how syntax and semantics interact in multi-agent logic. Let $(-)^c$ denote complement in the lattice $\powerset{S}$.

\begin{lemma}\label{lem:intension}
    Suppose $\phi, \psi \in \lang_n(\AP)$ and $\model = (S, K_1, \dots, K_n, \AP, \eval)$.  Then,
    \begin{align*}
        \intension{\phi \meet \psi} &=& \intension{\phi} \cap \intension{\psi} \\
        \intension{\neg\phi} &=& \intension{\phi}^{c} \\
        \intension{\mathtt{true}} &=& S 
    \end{align*}
\end{lemma}
\begin{proof}
    Exercise.
\end{proof}

Of course, other identities relating the syntax of propositional logic and the structure of the lattice $\powerset{S}$ are easily shown. However, lattice-theoretic interpretations of Kripke modalities are much less trivial. In the following theorem, we establish a Galois connection on $\powerset{S}$ which ``respects'' semantic equivalence.

\begin{theorem}[Kripke-Galois Connection] \label{thm:semantic-connection}
Suppose $\model$ is a Kripke model $\model = S, \rel{K}_1, \dots, \rel{K}_n, \AP, \eval)$, Then, there is a family of Galois connections indexed by $i \in \{1,2,\dots,n\}$
\[\begin{tikzcd}
\powerset{S} \arrow[r,"\rel{K}_i^\exists", bend left] & \powerset{S} \arrow[l,"\rel{K}_i^\forall", bend left]
\end{tikzcd}\]{}
given by
\begin{align}
	\rel{K}_i^{\exists}(e) &=& \{t \in S~\vert~ (\exists s) \left( s \in e \meet s\rel{K}_it \right)\} \label{eq:semantic-galois}\\ 
	\rel{K}_i^{\forall}(e) &=& \{s \in S~\vert~(\forall t) \left( s \rel{K}_i t \Rightarrow t \in e \right) \} \nonumber
\end{align}
such that
\begin{align}
\rel{K}_i^{\forall}(\intension{\phi}) &=& \intension{K_i \phi} \label{eq:preserve-semantics} \\
\rel{K}_i^{\exists}(\intension{\phi}) &=& \intension{\neg K_i \neg \phi}. \nonumber
\end{align}
for all $\phi \in \lang_n(\AP)$.
\end{theorem}
\begin{proof}
    By Theorem \ref{thm:cov-galois}, the pair $\left(\rel{K}_i^{\exists}, \rel{K}_i^{\forall} \right)$ is a Galois connection. It remains to show \eqref{eq:preserve-semantics}.
    First,
    \begin{align*}
        \rel{K}_i^{\forall}( \intension{\phi} ) &=& \\
        &=& \{ s \in S~\vert~ (\forall t) \quad s \rel{K}_i t \Rightarrow t \models \phi \} \\
        &=& \intension{ K_i \phi}.
    \end{align*}
    Second,
    \begin{align*}
        \rel{K}_i^{\exists}( \intension{\phi}) &=& \\
        &=& \{t \in S~\vert~ (\exists s) \quad s \models \phi~\meet~s \rel{K}_i t \} \\
        &=& \{t \in S~\vert~\neg \left( (\forall s) \quad s \not\models \phi \join \neg(s \rel{K}_i t) \right)   \} \\
        &=& \{t \in S~\vert~\neg \left( (\forall s) \quad s \rel{K}_i t \Rightarrow s \not\models \phi \right)   \} \\
        &=& \{t \in S~\vert~\neg \left( (\forall s) \quad s \rel{K}_i t \Rightarrow s \in \neg\phi \right)   \} \\
        &=& S - \{t \in S~\vert~ (\forall s) \quad s \rel{K}_i t \Rightarrow s \in \neg\phi   \}  \\
        &=& S - \intension{K_i \neg \phi} \\
        &=& \intension{\neg K_i\neg\phi}
    \end{align*}
    where the last equality is by Lemma \ref{lem:intension}.
\end{proof}

\begin{corollary}
    The semantic content of $K_i \neg K_i \neg \phi$ is greater than the semantic content of $\phi$, whereas, the semantic content of $\neg K_i \neg K_i \phi$ is less than the semantic content of $\phi$.
\end{corollary}
\begin{proof}
    By Proposition \ref{prop:monad-comonad},
    \begin{align*}
        \rel{K}_i^{\forall} \rel{K}_i^{\exists}(\intension{\phi}) &\supseteq& \intension{\phi} \\
        \rel{K}_i^{\exists} \rel{K}_i^{\forall}(\intension{\phi}) &\subseteq& \intension{\phi}.
    \end{align*}
    Furthermore,
    \begin{align*}
        \rel{K}_i^{\forall} \rel{K}_i^{\exists}(\intension{\phi}) &=& \intension{K_i \neg K_i \neg \phi} \\
        \rel{K}_i^{\exists} \rel{K}_i^{\forall}(\intension{\phi}) &=& \intension{\neg K_i \neg K_i \phi}
    \end{align*}
    by Lemma \ref{lem:intension}.
\end{proof}
\noindent Thus, applying $\rel{K}_i^{\forall} \rel{K}_i^{\exists}$ or $\rel{K}_i^{\exists} \rel{K}_i^{\forall}$ increases or decreases the information content, respectively. It may not be obvious, but the reader may check for herself than $\rel{K}^{\forall}_i \rel{K}^{\exists}_j$ in general neither increases or decreases semantic content.

\begin{example}[Alice, Bob \& Eve]\label{eg:alice-bob-eve}
Consider a network $\graph{G}$ with three agents (Alice, Bob and Eve) represented by nodes $i,j,k$ respectively. Suppose Alice and Bob can communicate $ij \in \edges{G}$, Bob and Eve can communicate $jk \in \edges{G}$, but not Alice and Eve $ik \not \in \edges{G}$. There are three states in our system $S = \{r,s,t\}$. The Kripke relations $\rel{K}_i, \rel{K}_j, \rel{K}_k$ for Alice, Bob and Eve are depicted by the following directed graphs.
% https://q.uiver.app/?q=WzAsMTMsWzEsMSwiXFxMYXJnZVxcY2lyY2xlYXJyb3dyaWdodF5yIl0sWzMsMCwiXFxMYXJnZVxcY2lyY2xlYXJyb3dsZWZ0XnMiXSxbNSwxLCJcXExhcmdlXFxjaXJjbGVhcnJvd2xlZnRedCJdLFswLDMsInt9Il0sWzMsMiwiXFxMYXJnZVxcY2lyY2xlYXJyb3dsZWZ0XnMiXSxbMSwzLCJcXExhcmdlXFxjaXJjbGVhcnJvd3JpZ2h0XnIiXSxbNSwzLCJcXExhcmdlXFxjaXJjbGVhcnJvd2xlZnRedCJdLFsxLDUsIlxcTGFyZ2VcXGNpcmNsZWFycm93cmlnaHReciJdLFs1LDUsIlxcTGFyZ2VcXGNpcmNsZWFycm93cmlnaHRedCJdLFszLDQsIlxcTGFyZ2VcXGNpcmNsZWFycm93cmlnaHRecyJdLFsxLDQsIlxcbWF0aGNhbHtLfV9rIl0sWzEsMiwiXFxtYXRoY2Fse0t9X2oiXSxbMSwwLCJcXG1hdGhjYWx7S31faSJdLFswLDEsIiIsMCx7Im9mZnNldCI6MX1dLFsxLDAsIiIsMCx7Im9mZnNldCI6MX1dLFsxLDIsIiIsMCx7Im9mZnNldCI6LTF9XSxbMiwxLCIiLDAseyJvZmZzZXQiOi0xfV0sWzUsNiwiIiwwLHsib2Zmc2V0IjotMX1dLFs2LDUsIiIsMCx7Im9mZnNldCI6LTF9XSxbNyw4XSxbOCw5XSxbOSw3XV0=
\[\begin{tikzcd}
	& {\mathcal{K}_i} && {\Large\circlearrowleft^s} \\
	& {\Large\circlearrowright^r} &&&& {\Large\circlearrowleft^t} \\
	& {\mathcal{K}_j} && {\Large\circlearrowleft^s} \\
	{{}} & {\Large\circlearrowright^r} &&&& {\Large\circlearrowleft^t} \\
	& {\mathcal{K}_k} && {\Large\circlearrowright^s} \\
	& {\Large\circlearrowright^r} &&&& {\Large\circlearrowright^t}
	\arrow[shift right=1, from=2-2, to=1-4]
	\arrow[shift right=1, from=1-4, to=2-2]
	\arrow[shift left=1, from=1-4, to=2-6]
	\arrow[shift left=1, from=2-6, to=1-4]
	\arrow[shift left=1, from=4-2, to=4-6]
	\arrow[shift left=1, from=4-6, to=4-2]
	\arrow[from=6-2, to=6-6]
	\arrow[from=6-6, to=5-4]
	\arrow[from=5-4, to=6-2]
\end{tikzcd}\]
Each relation induces a Galois connection, summarized in the following table.

\begin{tabular}{l|ll|ll|ll}
$e$         & $\mathcal{K}^{\exists}_i(e)$ & $\mathcal{K}_i^{\forall}(e)$ & $\mathcal{K}_j^{\exists}(e)$ & $\mathcal{K}_j^{\forall}(e)$ & $\mathcal{K}_k^{\exists}(e)$ & $\mathcal{K}_k^{\forall}(e)$ \\
\hline
$\emptyset$ & $\emptyset$                  & $\emptyset$                  & $\emptyset$                  & $\emptyset$                  & $\emptyset$                  & $\emptyset$                  \\
$\{r\}$     & $\{r,s\}$                    & $\emptyset$                  & $\{r,t\}$                    & $\emptyset$                  & $\{r,t\}$                    & $\emptyset$                  \\
$\{s\}$     & $\{r,s,t\}$                  & $\emptyset$                  & $\{s\}$                      & $\{s\}$                      & $\{r,s\}$                    & $\emptyset$                  \\
$\{t\}$     & $\{s,t\}$                    & $\emptyset$                  & $\{r,t\}$                    & $\emptyset$                  & $\{s,t\}$                    & $\emptyset$                  \\
$\{r,s\}$   & $\{r,s,t\}$                  & $\{r\}$                      & $\{r,s,t\}$                  & $\{s\}$                  & $\{r,s,t\}$                  & $\{s\}$                      \\
$\{r,t\}$   & $\{r,s,t\}$                  & $\{r,s,t\}$                  & $\{r,t\}$                    & $\{r,t\}$                    & $\{r,s,t\}$                  & $\{r\}$                      \\
$\{s,t\}$   & $\{r,s,t\}$                  & $\{r,s,t\}$                      & $\{r,s,t\}$                  & $\{r,s,t\}$                      & $\{r,s,t\}$                  & $\{t\}$                      \\
$\{r,s,t\}$ & $\{r,s,t\}$                  & $\{r,s,t\}$                  & $\{r,s,t\}$                  & $\{r,s,t\}$                  & $\{r,s,t\}$                  & $\{r,s,t\}$                 
\end{tabular}
\end{example}

%-----------------------------------------
\subsection{Kripke sheaves}
%-----------------------------------------

We now introduce a class of lattice-valued network sheaves and cosheaves whose stalks are powersets representing the semantic content of logical propositions and whose restriction and corestriction maps are Kripke-Galois connections.

\begin{definition}[Kripke Sheaves] \label{def:kripke-sheaf}
    Suppose $\graph{G}$ is a graph with nodes labeled $\nodes{G} = \{1,2, \dots, n\}$. Suppose $M = (S, \rel{K}_1, \dots, \rel{K}_n, \Phi, \pi)$ is a Kripke model. A \define{Kripke bisheaf} is a Tarski sheaf $\bisheaf{F}_{\model}$ characterized by the following.
    \begin{enumerate}
        \item Each stalk $\bisheaf{F}_{\model}(i)$ is the powerset $\powerset{S}$.
        \item Each stalk $\bisheaf{F}_{\model}(ij)$ is the powerset $\powerset{S}$.
        \item $\sheaf{F}_{\model}(i \fc ij) = \rel{K}_i^{\exists}$ for all $i \in \nodes{G}, j \in \nbhd{i}$.
        \item $\cosheaf{F}_{\model}(i \fc ij) = \rel{K}_i^{\forall}$ for all $i \in \nodes{G}, j \in \nbhd{i}$.
    \end{enumerate}
\end{definition}
\noindent The $0$-cochains of $\sheaf{F}_{\model}$ are the product $C^0(\graph{G}; \sheaf{F}_{\model}) = \powerset{S}^n$. For cochains, thereby passing from syntax to semantics, let
\begin{align*}
    \intension{\boldsymbol{\phi}} &=& \left( \intension{\phi_i}\right)_{i \in \nodes{G}} \in \powerset{S}^n
\end{align*}
 for a tuple of formulas $\boldsymbol{\phi} \in \lang_{n}(\Phi)^n$.
\noindent Let's consider some special cases.
\begin{itemize}
    \item Suppose $\rel{K}_i$ is the identity relation on $S$; $s \rel{K}_i s$ for all $s \in S$, no other elements are in $\rel{K}_i$. Then, $(\rel{K}_i^{\exists}, \rel{K}_i^{\forall}) = (\id, \id)$ and $\sheaf{F}_{\model}$ is the constant sheaf $\underline{\powerset{S}}$. Note, in this case, $\intension{\phi} = \intension{K_i \phi}$ so that the constant sheaf reduces a network sheaf model of the semantics for ordinary propositional logic.
    \item Suppose $\rel{K}_i$ is a (possibly different) equivalence relation for each $i$, so that $\Models_n^{rst}$ is a class of models for $\mathsf{S5}_n$. The restriction map $\rel{K}^{\exists}_i$ sends an event $e$ to the union of equivalence classes of elements of $e$. In particular, $\rel{K_{\exists}}$ sends equivalence classes to equivalence classes. The corestriction map $\rel{K}_i^{\forall}$ sends an event $e$ to the set of elements of $S$ whose equivalence class is contained in $e$.
\end{itemize}

As our notation suggests, $0$-cochains model tuples of formulas through their intension. We investigate the sections of this class of sheaves.
\begin{proposition} \label{prop:sections-intension}
    Suppose $\bisheaf{F}_{\model}$ is a Kripke sheaf with the above data. Suppose $\boldsymbol{\phi} \in \lang_{n}(\Phi)^n$ is a tuple of formulas indexed by the nodes of $\graph{G}$. Then, $\intension{\boldsymbol{\phi}} \in H^0(\graph{G}; \sheaf{F}_{\model})$ if and only if
    \begin{align*}
        K_i \neg \phi_i \equiv K_j \neg \phi_j \quad \forall ij \in \edges{G.}
    \end{align*}
\end{proposition}
\begin{proof}
    Recall, by the global sections criterion, $\mathbf{e} \in \sections{\graph{G}; \sheaf{F}_{\model}}$ if and only 
    \begin{align*}
        \sheaf{F}_{\model}(i \fc ij)(e_i) &=& \sheaf{F}_{\model}(j \fc ij) \quad \forall ij \in \edges{G}.
    \end{align*}
    Here, with $e_i$ defined to equal $\intension{\phi_i}$ for some formula $\phi_i \in \lang_n(\AP)$, so that
    \begin{align*}
        \rel{K}_i^{\exists}(\intension{\phi_i}) &=& \rel{K}_j^{\exists}(\intension{\phi_j}) \quad \forall ij \in \edges{G}.
    \end{align*}
    By Lemma \ref{lem:intension} and Theorem \ref{thm:semantic-connection},
    \begin{align*}
        \rel{K}_i^{\exists}(\intension{\phi_i}) &=& \intension{\neg K_i \neg \phi_i} \\
        &=& \intension{K_i \neg \phi_i}^c \\
        &=& \rel{K}_i^{\forall}\left(\intension{\neg \phi_i}\right)^c \\
        &=& \rel{K}_i^{\forall}\left( \intension{\phi_i}^c \right)^c
    \end{align*}
    This implies,
    \begin{align*}
        \rel{K}_i^{\forall}\left( \intension{\phi_i}^c \right)^c &=& \rel{K}_j^{\forall}\left( \intension{\phi_j}^c \right)^c,
    \end{align*}
    which, taking complements of both sides and applying $\intension{\phi}^c = \neg \intension{\phi}$ again, yields
    \begin{align*}
        \rel{K}_i^{\forall}\left( \intension{\neg \phi_i} \right) &=& \rel{K}_j^{\forall}\left( \intension{\neg \phi_j} \right).
    \end{align*}
    Hence,
    \begin{align*}
        \intension{K_i \neg \phi_i} &=& \intension{K_j \neg \phi_j}.
    \end{align*}
\end{proof}
\noindent The proof of Proposition \ref{prop:sections-intension} suggest that under certain conditions, $K_i \phi_i$ is semantically equivalent to $K_j \phi_j$ as well.
\begin{proposition}\label{prop:sections-intension-2}
    Suppose, in addition to the above assumptions, that $\rel{K}_i$ is serial for all $i \in \nodes{G}$. Then, $\intension{\phi} \in H^0(\graph{G}; \sheaf{F}_{\model})$ if and only if
    \begin{align*}
        K_i \phi_i \equiv K_j \phi_j.
    \end{align*}
\end{proposition}
\noindent We will need a lemma.
\begin{lemma}\label{lem:forall-empty}
    Suppose $\rel{K}_i$ is serial. Then, $\rel{K}_i^{\forall}(\emptyset) = \emptyset$.
\end{lemma}
\begin{proof}[Proof of Lemma \ref{lem:forall-empty}]
By Proposition \ref{thm:adjoint-functor-theorem} and Theorem \ref{thm:semantic-connection},
\begin{align*}
    \rel{K}_{\forall}(\emptyset) &=& \bigcup \{e \subseteq S ~\vert~ \rel{K}_i^{\exists}(e) = \emptyset \}.
\end{align*}
Because $\rel{K}_i$ is serial, for every $s \in S$, there is at least one $t \in S$ with $s \rel{K}_i t$. Hence, $\rel{K}_i^{\exists}(e) = \emptyset$ if and only if $e = \emptyset$. Thus, $\rel{K}_{\forall}(\emptyset) = \emptyset$.
\end{proof}
\begin{proof}[Proof of Proposition \ref{prop:sections-intension-2}]
By Proposition \ref{prop:sections-intension}, $\intension{\boldsymbol{\phi}}$ is a section if and only if $\intension{K_i \neg \phi_i} = \intension{K_j \neg \phi_j}$ for all $ij \in \edges{G}$, which holds if and only if $\intension{K_i \phi_i} = \intension{K_j \phi_j}$ for all $ij \in \edges{G}$ by the following.
\begin{claim}
    $\rel{K}_i^{\forall}(\intension{\phi_i}^c) = \rel{K}_i^{\forall}(\intension{\phi})^c$.
\end{claim}
\noindent By Proposition \ref{thm:adjoint-functor-theorem} (again) and Theorem \ref{thm:semantic-connection},
\begin{align*}
    \rel{K}_i^{\forall}\left( \intension{\phi_i} \cap \intension{\phi_i}^c\right) &=& \rel{K}_i^{\forall}\left( \intension{\phi_i}\right) \cap \rel{K}_i^{\forall}\left( \intension{\phi_i}^c\right).
\end{align*}
On the other hand,
\begin{align*}
    \rel{K}_i^{\forall}\left( \intension{\phi_i} \cap \intension{\phi_i}^c\right) &=& \rel{K}_i^{\forall}(\emptyset) \\
    &=& \emptyset
\end{align*}
Hence,
\begin{align*}
    \rel{K}_i^{\forall}\left( \intension{\phi_i}^c\right) &=& \rel{K}_i^{\forall}\left( \intension{\phi_i}\right)^c,
\end{align*}
or,
\begin{align}
    \intension{K_i \neg \phi_i} &=& \intension{K_i \phi_i}^c.
\end{align}
Finally, $\intension{K_i \neg \phi_i} = \intension{K_j \neg \phi_j}$ is equivalent to  
$\intension{K_i \phi_i}^c =\intension{K_j \phi_j}^c$ which is of course equivalent to $\intension{K_i \phi_i} =\intension{K_j \phi_j}$ by taking set complements of both sides.
\end{proof}

Proposition \ref{prop:sections-intension} and Proposiotin \ref{prop:sections-intension-2} suggest we can interpret global sections of Kripke sheaves as local \define{knowledge consensus}. We say $\bisheaf{F}$ is \define{star-shaped} if for all $i \in \nodes{G}$ and for all $j, j' \in \nbhd{i}$, $\sheaf{F}(i \fc ij) = \sheaf{F}(i \fc ij')$. However, by the following Lemma, local knowledge consensus necessarily extends to global knowledge consensus.

\begin{lemma}
    Suppose $\bisheaf{F}$ is a star-shaped Tarski sheaf over a $\graph{G}$, and suppose $\mathbf{x} \in \sections{\graph{G}; \sheaf{F}}$. Then, $x_{ij} = x_{i'j'}$ for every pair of edges $ij, i'j' \in \edges{G}$ in the same connected component of $\graph{G}$.
\end{lemma}
\begin{proof}
    By the star-shaped property,
    \begin{align*}
        \sheaf{F}(i \fc ij)(x_i) &=& \sheaf{F}(i \fc ij')(x_i)
    \end{align*}
    for all $i \in \nodes{G}$, $j, j' \in \nbhd{i}$. Yet, because $\mathbf{x} \in \sections{\graph{G}; \sheaf{F}}$,
    \begin{align*}
        \sheaf{F}(j \fc ij)(x_j) &=& \sheaf{F}(i \fc ij)(x_i) \\
        \sheaf{F}(j' \fc ij')(x_{j'}) &=& \sheaf{F}(i \fc ij')(x_i)
    \end{align*}
    which implies $x_{ij} = x_{ij'}$. By connectivity, this implies sections of $\sheaf{F}$ are constant when projected onto $C^1(\graph{G}; \sheaf{F})$.
\end{proof}
\noindent All the preceding results lead up to the following theorem.
\begin{theorem}
    Suppose $\model = (S, \rel{K}_1, \dots, \rel{K}_n, \AP, \eval)$ is a Kripke model supported on a graph $\graph{G} = (\nodes{V}, \edges{E})$ such that every $\rel{K}_i$ is serial and $\bisheaf{F}_{\model}$ its corresponding Kripke bisheaf. Then, $\intension{\boldsymbol{\phi}} \in H^0(\graph{G}; \sheaf{F}_\model)$ if and only if for every $i, j \in \nodes{G}$ in the same connected component
    \begin{align*}
        K_i \phi_i \equiv K_j \phi_j.
    \end{align*}
\end{theorem}
\begin{corollary}
    In particular, if $\boldsymbol{\phi}$ is constant with value $\phi$, then
     \begin{align*}
        \intension{K_i \phi} &=& \intension{K_j \phi}
    \end{align*}
    for all $i, j \in \nodes{G}$ in the same connected component.
\end{corollary}

%----------------------------------------
\subsection{Kripke Laplacian}
%----------------------------------------

Our focus will now shift to semantic dynamics over a network. Diffusion dynamics on the $0$-cochains $\powerset{S}^n$ are defined by means of the Tarski Laplacian of $\sheaf{F}_{\model}$.

\begin{definition}[Kripke Laplacian] \label{def:kripke-laplacian}
    Suppose $\model = (S, \rel{K}_1, \dots, \rel{K}_n, \AP, \eval)$ is  Kripke model and $\graph{G}$ is a graph with nodes $\nodes{G} = \{1,2,\dots, n\}$ as before. Then, the \define{Kripke Laplacian} is the operator
    \begin{align}
    \Laplacian &:& \powerset{S}^n \to \powerset{S}^n \nonumber \\
        (\Laplacian \mathbf{e})_i &=& \bigcap_{j \in \nbhd{i}} \rel{K}_i^{\forall} \rel{K}_i^{\exists}(e_i)
    \end{align}
\end{definition}

\begin{example}[Alice, Bob \& Eve]\label{eg:alice-bob-eve-2}
   Let us revisit the network $\graph{G}$, frame $(S, \rel{K}_i, \rel{K}_j, \rel{K}_j)$, and Galois connections $(\rel{K}_i^{\exists}, \rel{K}_i^{\forall}), (\rel{K}_j^{\exists}, \rel{K}_j^{\forall}), (\rel{K}_k^{\exists}, \rel{K}_k^{\forall})$ from Example \ref{eg:alice-bob-eve}. Then, the semantic sheaf $\bisheaf{F}_{\model}$ is depicted by the following diagram
    % https://q.uiver.app/?q=WzAsOCxbMiwwLCJcXHdwKFMpIl0sWzMsMCwiXFx3cChTKSJdLFs0LDAsIlxcd3AoUykiXSxbMSwwLCJcXHdwKFMpIl0sWzAsMCwiXFx3cChTKSJdLFswLDEsIlxcYnVsbGV0X2kiXSxbMiwxLCJcXGJ1bGxldF9qIl0sWzQsMSwiXFxidWxsZXRfayJdLFs0LDMsIlxcbWF0aGNhbHtLfV9pXntcXGV4aXN0c30iLDAseyJvZmZzZXQiOi0xfV0sWzIsMSwiXFxtYXRoY2Fse0t9X2tee1xcZXhpc3RzfSIsMix7Im9mZnNldCI6MX1dLFswLDMsIlxcbWF0aGNhbHtLfV9qXntcXGV4aXN0c30iLDIseyJvZmZzZXQiOjF9XSxbMyw0LCJcXG1hdGhjYWx7S31faV57XFxmb3JhbGx9IiwwLHsib2Zmc2V0IjotMX1dLFszLDAsIlxcbWF0aGNhbHtLfV9qXntcXGZvcmFsbH0iLDIseyJvZmZzZXQiOjF9XSxbMSwwLCJcXG1hdGhjYWx7S31fal57XFxmb3JhbGx9IiwwLHsib2Zmc2V0IjotMX1dLFswLDEsIlxcbWF0aGNhbHtLfV9qXntcXGV4aXN0c30iLDAseyJvZmZzZXQiOi0xfV0sWzEsMiwiXFxtYXRoY2Fse0t9X2tee1xcZm9yYWxsfSIsMix7Im9mZnNldCI6MX1dLFs1LDYsIiIsMix7InN0eWxlIjp7ImhlYWQiOnsibmFtZSI6Im5vbmUifX19XSxbNyw2LCIiLDAseyJzdHlsZSI6eyJoZWFkIjp7Im5hbWUiOiJub25lIn19fV1d
\[\begin{tikzcd}
	{\wp(S)} & {\wp(S)} & {\wp(S)} & {\wp(S)} & {\wp(S)} \\
	{\bullet_i} && {\bullet_j} && {\bullet_k}
	\arrow["{\mathcal{K}_i^{\exists}}", shift left=1, from=1-1, to=1-2]
	\arrow["{\mathcal{K}_k^{\exists}}"', shift right=1, from=1-5, to=1-4]
	\arrow["{\mathcal{K}_j^{\exists}}"', shift right=1, from=1-3, to=1-2]
	\arrow["{\mathcal{K}_i^{\forall}}", shift left=1, from=1-2, to=1-1]
	\arrow["{\mathcal{K}_j^{\forall}}"', shift right=1, from=1-2, to=1-3]
	\arrow["{\mathcal{K}_j^{\forall}}", shift left=1, from=1-4, to=1-3]
	\arrow["{\mathcal{K}_j^{\exists}}", shift left=1, from=1-3, to=1-4]
	\arrow["{\mathcal{K}_k^{\forall}}"', shift right=1, from=1-4, to=1-5]
	\arrow[no head, from=2-1, to=2-3]
	\arrow[no head, from=2-5, to=2-3]
\end{tikzcd}\]
The Kripke Laplacian is the operator
\begin{align*}
    \Laplacian(e_i,e_j,e_k) &~& \\
    &=& \left( \rel{K}_i^{\forall}\rel{K}_j^{\exists}(e_j),  \rel{K}_j^{\forall}\rel{K}_i^{\exists}(e_i) \cap  \rel{K}_j^{\forall}\rel{K}_k^{\exists}(e_k), \rel{K}_k^{\forall} \rel{K}_j^{\exists}(e_j) \right)
\end{align*}
We compute the value of the Kripke Laplacian for a few initial $0$-cochains.

\begin{tabular}{lll|lll}
$e_i$       & $e_j$       & $e_k$       & $(\Laplacian \mathbf{e})_i$ & $(\Laplacian \mathbf{e})_j$ & $(\Laplacian \mathbf{e})_k$ \\
\hline
$\emptyset$ & $\emptyset$ & $\emptyset$ & $\emptyset$                 & $\emptyset$                 & $\emptyset$                 \\
$\{r\}$     & $\{s\}$     & $\{t\}$     & $\emptyset$                 & $\{s\}$                     & $\emptyset$                 \\
$\{s\}$     & $\{r\}$     & $\{t\}$     & $\{r,s,t\}$                 & $\{r,s,t\}$                 & $\{r\}$                     \\
$\{t\}$     & $\{t\}$     & $\{t\}$     & $\{r,s,t\}$               & $\{r,s,t\}$                 & $\{r\}$                     \\
$\{s\}$     & $\{s\}$     & $\{s\}$     & $\emptyset$                 & $\{s\}$                     & $\emptyset$                 \\
$\{r,s,t\}$ & $\{r,s,t\}$ & $\{r,s,t\}$ & $\{r,s,t\}$               & $\{r,s,t\}$                 & $\{r,s,t\}$                
\end{tabular}
\end{example}
\noindent We can now apply the theory built up in Chapter \ref{ch:lattice-valued} and \ref{ch:tarski}.
\begin{theorem}
    Given the data $\model$ and $\graph{G}$, the global sections $\sections{\graph{G}; \sheaf{F}_{\model}}$ form a complete quasi-sublattice of $\powerset{S}^n$.
\end{theorem}
\begin{proof}
    Apply Theorem \ref{thm:sections-sublattice}.
\end{proof}
We argue that the lattice $\sections{\graph{G}; \sheaf{F}_{\model}}$ is an algebraic model for the pair $(\model, \graph{G})$. Suppose $\rel{K}_i$ is the identity relation for every $i \in \nodes{G}$. Then, $\sections{G; \sheaf{F}_{\model}}$ is isomorphic to $\powerset{S}$. This makes sense because in this case, modalities $\rel{K}_i$ act as the identity, which reduces the semantics to propositional calculus, a Boolean algebra.

\begin{theorem}
    Suppose $L$ is the Kripke Laplacian of $(\model, \graph)$. Then,
    \begin{align*}
        \suffix(L) &=& \sections{\graph{G}; \sheaf{F}_{\model}}.
    \end{align*}
\end{theorem}
\begin{proof}
    Apply Theorem \ref{thm:main}.
\end{proof}
\begin{theorem}
    The global sections coincide with the global time-invariant solutions of
    \begin{align*}
        e_i[t+1] &=& (L \mathbf{e})_i \cap e_i \quad \forall i \in \nodes{G}.
    \end{align*}
\end{theorem}
\begin{proof}
    Powersets satisfy the descending chain condition. Apply Propositon \ref{thm:alg-heat}.
\end{proof}

\noindent In Appendix \ref{ch:appendix-2}, we run an experiment by randomly generating a network and local Kripke relations. We compute heat flow using the gossip algorithm (Algorithm \ref{alg:gossip}).

%----------------------------------------
\subsection{Diffusive knowledge}
%----------------------------------------

As before, suppose $\graph{G}$ is a graph and $\AP$ is a set of atomic propositions.
\begin{definition}
    Suppose $\phi \in \lang_n(\AP)$. Then, for $i \in \nodes{V}$, the \define{syntactic Laplacian} is the formula in $\lang_n(\AP)$ defined
    \begin{align*}
        \left(\tilde{\Laplacian} \boldsymbol{\phi}\right)_i &=& \bigmeet_{j \in \nbhd{i}} K_j \neg K_i \neg \phi.
    \end{align*}
\end{definition}

\begin{lemma}
Suppose $\boldsymbol{\phi} \in \lang_n(\AP)^n$ and $\model \in \Models_n(\AP)$. Then,
    \begin{align*}
        \intension{(\Laplacian \varphi)_i} &=& \left( \Laplacian \intension{\boldsymbol{\phi}}\right)_i.
    \end{align*}
\end{lemma}
\begin{proof}
    By Lemma \ref{lem:intension} and Theorem \ref{thm:semantic-connection},
    \begin{align*}
        \intension{\Laplacian_i \phi_i} &=& \\
        &=& \intension{\bigmeet_{j \in \nbhd{i}} K_i \neg K_j \neg \phi_j} \\
        &=& \bigmeet_{j \in \nbhd{i}} \intension{K_i \neg K_j \neg \phi_j} \\
        &=& \bigmeet_{j \in \nbhd{i}} \rel{K}_i^{\forall} \left( \intension{\neg K_j \neg \phi_j}  \right) \\
        &=& \bigmeet_{j \in \nbhd{i}} \rel{K}_i^{\forall} \rel{K}_j^{\forall} \intension{\phi_j}.
    \end{align*}
\end{proof}

Suppose $\boldsymbol{\phi}$ is constant with value $\phi$. Then, it makes sense to define $L_i \phi$
 to be the formula
 \begin{align*}
     \bigmeet_{j \in \nbhd{j}} \neg K_j \neg \phi.
 \end{align*}
In $\mathsf{KD45}_n$, modeling belief, $L_i \phi$ could be interpreted
    \begin{displayquote}
        \emph{If Bob is a neighbor of Alice, then Alice believes that Bob doesn't disbelieve $\phi$.}
    \end{displayquote}
Similarly, in $\mathsf{S5}_n$, modeling knowledge, $L_i \phi$ could be interpreted
    \begin{displayquote}
        \emph{If Bob is a neighbor of Alice, then Alice knows that Bob doesn't not know $\phi$.}
    \end{displayquote}
In both cases, and generally in our mental model, we could call $L_i \phi$ \define{diffusive knowledge} of $\phi$.

By mimiking the common knowledge operators $E_A \phi$ and $C_A \phi$, we can both extend the notion of diffusive knowledge to an arbitrary subset of nodes and integrate higher order information by iterating the $L_i$ operator. Suppose $A \subseteq \nodes{G}$ is a group of agents. Then, define
    \begin{align*}
        \Laplacian_A \phi &=& \bigmeet_{i \in A} \Laplacian_i \phi \\
        \Laplacian^{\omega}_A \phi &=& \bigmeet_{k \geq 1} \Laplacian_A^k \phi. 
    \end{align*}
Many questions remain. For instance, can you infer $E_{\nodes{G}} \phi$, $C_{\nodes{G}} \phi$ or $D_{\nodes{G}} \phi$ from $\Laplacian_{\nodes{G}}$ or $\Laplacian^{\omega} \phi$? What happens in $\mathsf{K}_n$ or other axiom systems? Does it depend on the structure of $\graph{G}$?

% %---------------------------
% \subsection{Model Checking}
% %---------------------------
% Another future direction of this work is developing more efficient model-checking algorithms for multi-agent systems based on diffusive knowledge. First we will review model checking in multi-agent logic.

% Given the sizes of a Kripke Model and a formula, we can reason about the relative difficulty of model-checking. The size of a model is given by the sum
% \begin{align*}
% 	\| \model \| &=& |S| + \sum_{i = 1}^n | \rel{K}_i |.
% \end{align*}
% The size of a formula $|\phi|$ is the number of symbols, including parentheses.

% \begin{proposition}[Proposition 3.2.1 \cite{fagin2004reasoning}]
% 	Given a model $\model$, a state $s \in S$, and a formula $\phi \in \lang_n(\Phi)$, there is an algorithm that determines whether or not $(\model, s) \models \phi$ in time $\mathcal{O}(\|M \|\cdot |\phi|)$.
% \end{proposition}

%% file: Chapters/Chapter0A.tex
%***************************************************
\chapter{Sheaf Theory}
\label{ch:appendix-1}
%***************************************************

In Chapter \ref{ch:sheaves} and \ref{ch:lattice-valued} we were introduced to network sheaves and, more specifically, lattice valued network sheaves. However, those with some degree of familiarity with sheaves, may not even recognize that a functor $\sheaf{F}: \Ind(\graph{G}) \to \cat{Sup}$ is properly a sheaf. We may also wonder what lattice-valued sheaf theory looks like in the general case.

Recall, a topological space is the data $\left( \Space{X}, \Open(\Space{X}) \right)$ with a poset of subsets called \define{open sets}\footnote{The notion of an open set is quite general, but it coincides with familiar intuition. For instance, if $(\Space{M}, d)$ is a metric space, then a subset $U \subseteq \Space{M}$ is open in the \define{metric space topoloy} if and only for every $x \in U$ there is an $\epsilon > 0$ such that the (open) ball $B_\epsilon(x) = \{y \in \Space{M}~\vert~d(x,y) < \epsilon \}$ is contained in $U$. The topology formed by arbitrary unions and finite intersections of balls in$\Space{M}$ is called the \define{metric space topology}.} satisfying the following conditions
\begin{enumerate}
    \item $\emptyset \in \Open(\Space{X})$.
    \item $\Space{X} \in \Open(\Space{X})$.
    \item If $\{U\}_{i \in I} \subseteq \Open(\Space{X})$ is an arbitrary collection of opens, then their union $\bigcup_{i \in I} U_i \in \Open(\Space{X})$ is also open. 
\end{enumerate}

\begin{definition}[Presheaf]
    Suppose $\Space{X}$ is a topologcal space and $\cat{D}$ is a category. Then, a \define{presheaf over $\Space{X}$ valued in $\cat{D}$} is a contravariant functor
        \[\sheaf{F}: \op{\Open(\Space{X})} \to \cat{D}\]
\end{definition}
Specifically, for every open set $U$, a presheaf is the data
\begin{align*}
    \sheaf{F}(U) &\in& \mathrm{Ob}(\cat{D}) \\
    \sheaf{F}(V \subseteq U) &\in& \Hom_{\cat{D}}(\sheaf{F}(V), \sheaf{F}(U))
\end{align*}
for $V \subseteq U$ open. The architypical example of a presheaf with $\cat{D} = \cat{Set}$ over a space $\Space{X}$ is the following. For $U \in \Open(\Space{X})$, assign the set of $\R$-valued continuous functions to $U$ and assign the restriction of a particular map on $U$ to the open subset $V \subseteq U$
\begin{align*}
    \sheaf{F}(U) &=& \{ f: U \to \R \} \\
    \sheaf{F}(V \subseteq U) &=& f_{\vert V}
\end{align*}

Suppose $\cat{A}$ and $\cat{B}$ are categories. A functor $\cat{A} \to \cat{B}$ is \define{faithful}
if for every pair of objects $x, y \in \mathrm{Ob}(\cat{A})$, the map $\hom_{\cat{A}}(x,y) \hookrightarrow \hom_{\cat{B}}\left(T(x), T(y)\right)$ is injective. Then, a \define{concrete category} $\cat{D}$ is a category that admits a faithful functor $ \cat{D} \to \cat{Set}$.

\begin{example}[Suplattices]
    The category $\cat{Sup}$ is concrete. The functor $\cat{Sup} \to \cat{Set}$ sends a complete lattice $\lattice{L}$ to its underlying set $L$ and a join-preserving map $f: \lattice{K} \to \lattice{L}$ to the ordinary function $f: K \to L$. Faithfulness is trivial because morphisms in $\cat{Sup}$ are maps between sets that preserve structure; if two functions are equal after forgetting the join-preserving structure, then, of course, the structure-preserving maps were equal. 
\end{example}

We define sheaves with arbitrary algebraic strucutre by first defining sheaves valued in sets. An \define{open cover} of $U \subseteq \Space{X}$ consists of a collection of open sets $\{U_i\}_{i \in I}$
such that $\bigcup_{i \in I} U_i \supseteq U$. Let $U_{ij}$ denote the open set $U_i \cap U_j$. The following definition is taken from \cite{maclane2012sheaves}, while not standard, is one of a handful of equivalent definitions of a sheaf.

\begin{definition}[Sheaf]
    Suppose $\sheaf{F}$ is a presheaf of sets over $\space{X}$. Then, $\sheaf{F}$ is a \define{sheaf} if for every open set $U \in \Open(\Space{X})$ and open cover $\mathcal{U} = \{ U_{i} \}_{i \in I}$ of $U$, the following is an equalizer
% https://q.uiver.app/?q=WzAsMyxbMCwwLCJcXHNoZWFme0Z9KFUpIl0sWzEsMCwiXFxwcm9kX3tpIFxcaW4gSX1cXHNoZWFme0Z9KFVfaSkiXSxbMywwLCJcXHByb2Rfe2ksIGogXFxpbiBJfSBcXHNoZWFme0Z9KFVfe2lqfSkiXSxbMSwyLCJcXHByb2Rfe2ogXFxpbiBJfSBcXHNoZWFme0Z9KFVfe2lqfSBcXHN1YnNldGVxIFVfaSkiLDAseyJvZmZzZXQiOi0zfV0sWzEsMiwiXFxwcm9kX3tpIFxcaW4gSX0gXFxzaGVhZntGfShVX3tpan0gXFxzdWJzZXRlcSBVX2opIiwyLHsib2Zmc2V0IjozfV0sWzAsMSwiIiwyLHsic3R5bGUiOnsiYm9keSI6eyJuYW1lIjoiZGFzaGVkIn19fV1d
\[\begin{tikzcd}
    {\sheaf{F}(U)} & {\prod_{i \in I}\sheaf{F}(U_i)} && {\prod_{i, j \in I} \sheaf{F}(U_{ij})}
    \arrow["{\prod_{j \in I} \sheaf{F}(U_{ij} \subseteq U_i)}", shift left=3, from=1-2, to=1-4]
    \arrow["{\prod_{i \in I} \sheaf{F}(U_{ij} \subseteq U_j)}"', shift right=3, from=1-2, to=1-4]
    \arrow[dashed, from=1-1, to=1-2]
\end{tikzcd}\]
Then, a presheaf $\sheaf{F}$ valued in a concrete category $\cat{D}$ is \define{sheaf} if the composite presheaf in the following diagram
% https://q.uiver.app/?q=WzAsMyxbMCwwLCJcXE9wZW4oXFxTcGFjZXtYfSkiXSxbMSwwLCJcXGNhdHtEfSJdLFsxLDEsIlxcY2F0e1NldH0iXSxbMCwxLCJcXHNoZWFme0Z9Il0sWzEsMl0sWzAsMiwiXFxoYXR7XFxzaGVhZntGfX0iLDJdXQ==
\[\begin{tikzcd}
    {\op{\Open(\Space{X})}} & {\cat{D}} \\
    & {\cat{Set}}
    \arrow["{\sheaf{F}}", from=1-1, to=1-2]
    \arrow[from=1-2, to=2-2]
    \arrow[dashed, from=1-1, to=2-2]
\end{tikzcd}\]
is a sheaf.
\end{definition}

Global sections of a sheaf over a space are simply the object $\sheaf{F}(\Space{X})$. However, this definition is not very useful because it is not constructable. In certain situations, $\Space{X}$ is presented with a cell decomposition or, in real-life examples, is approximated with as a graph or simplicial complex.

\begin{definition}
    Suppose $(\poset{P}, \fc)$ is a poset. Then, the \define{Alexandrov topology} on $\poset{P}$, denoted $\Alex(\poset{P})$, consists of up-sets: subsets $U \subseteq \poset{P}$ with $x \in U$, $y \cofc u$ implies $y \in U$.
\end{definition}

\noindent Recall, a category is \define{complete} or \define{cocomplete} if it has all (small) limits and colimits, respectively.\footnote{We will not review limits and colimits here. For a good introduction, see \cite[Chpater 3]{riehl2017category}.}

\begin{theorem}\label{thm:apendix-complete}
    The category $\cat{Sup}$ of complete lattices and join-preserving maps is complete and cocomplete.
\end{theorem}
\begin{proof}
    See Joyal \cite[Chapter 1]{joyal1984extension} for a detailed explanation. Alternatively, $\cat{Sup}$ has products, equalizers, coproducts, and coequalizers, making $\cat{Sup}$ complete and cocomplete \cite[Theorem 3.4.12]{riehl2017category}.
\end{proof}

\noindent We now state a result by Curry  which adresses the common (invalid) critique that functors out of posets (e.g.~Tarski sheaves) are not really sheaves. The follwing theorem requires the notion of a Kan extension \cite{riehl2017category} and is stated in a form dual to the original result.

\begin{theorem}[Theorem 4.8 \cite{curry2019functors}] \label{thm:kan}
    Suppose $\poset{P}$ is a poset and $\cat{C}$ is a complete category. Suppose $\sheaf{F}: \poset{P} \to \cat{C}$ is a functor (e.g.~network sheaf), and suppose $\iota: \poset{P} \rightarrow \op{\Alex(\poset{P})}$ is the functor associating an element $x \in \poset{P}$ to the open set $\uparrow x$. Then, the right Kan extension $\mathrm{Ran}_\iota (\sheaf{F})$ of $\sheaf{F}$ along $\iota$ is a sheaf over the Alexandrov topology $\Alex(\poset{P})$.
\end{theorem}

\begin{corollary}
    A Tarski sheaf $\bisheaf{F}: \poset{P}_{\graph{G}} \to \cat{Ltc}$ is (naturally isomorphic to) a sheaf over the space $\Alex(\poset{P}_{\graph{G}})$.
\end{corollary}
\begin{proof}
    Theorem \ref{thm:duality}, Theorem \ref{thm:apendix-complete} and Theorem \ref{thm:kan}.
\end{proof}

%% file: Chapters/Chapter0B.tex
%************************************************
\chapter{Experiments}
\label{ch:appendix-2}
%************************************************

%---------------------------------------------
\section{Gossip}
%---------------------------------------------

In the following example, we demonstrate the validity of the gossip algorithm introduced in Chapter \ref{ch:tarski} for an example sheaf introduced in Chapter \ref{ch:semantics}. Suppose $\bisheaf{F}$ is a Tarski sheaf, such that for every $ij \in \edges{G}$, $\sheaf{F}(ij)$ is not only a complete lattice but also a metric space $(\sheaf{F}(ij), d_{ij})$. We measure ``how close'' an $\mathbf{x} \in C^0(\graph{G}; \sheaf{F})$ is to being a section with the following notion of \define{Dirichlet energy} or \define{total variation}.
\begin{align*}
    V(\mathbf{x}) &=& \sum_{ij \in \edges{G}} d_{ij} \left( \sheaf{F}(i \fc ij)(x_i), \sheaf{F}(j \fc ij)(x_j)\right).
\end{align*}
The following is evident
\begin{fact}
    Suppose $\bisheaf{F}$ is a Tarski sheaf over $\graph{G}$ and $\mathbf{x} \in C^0(\graph{G}; \sheaf{F})$. Then, $V(\mathbf{x}) = 0$ if and only if $\mathbf{x} \in H^0(\graph{G}; \sheaf{F})$.
\end{fact}

Recall the definition of the time-varying Tarki Laplacian (Definition \ref{def:varrying-tarski}). In the following experiment, we generate random geometric graphs $\graph{G} = (\nodes{G}, \edges{G})$ with parameters $|\nodes{G}| = 40$ and $r=0.08$. Specifically, locations $y{}_i \in \R^2$ of each $i \in \nodes{G}$ are sampled uniformly at random from $[0,1]^2$; $ij \in \edges{G}$ if and only if $\| x_i - x_j \|\leq r$.

We construct a Kripke sheaf $\bisheaf{F}_M$ (Definition \ref{def:kripke-sheaf}) by generating a Kripke model $M$ (really, just a frame) with randomized Kripke relations $\rel{K}_i,~i \in \nodes{G}$. In our experiment, we choose a fixed number of states $S = \{1,2,\dots, 10\}$ and wire each $(a,b) \in \rel{K}_i$ randomly (with probability $p=0.9$ if $a = b$, $p=0.1$ otherwise). The relations $\rel{K}_i$ are generated independently. We select a random firing sequence $\tau: \N \to \powerset{\nodes{G}}$ which serves the purpose of determining which nodes broadcast to their neighbors at time $t$. In our experiment, we randomly select a single node to broadcast each round. For a series of random initial assignments $\mathbf{x}[0]$, we run the heat flow dynamics using the asynchronous Kripke Laplacian (Definition \ref{def:kripke-laplacian}). The results are summarized below (Figure \ref{fig:gossip-experiment}).

\begin{figure}[h!]
    \centering
    \includegraphics[width=0.75\textwidth]{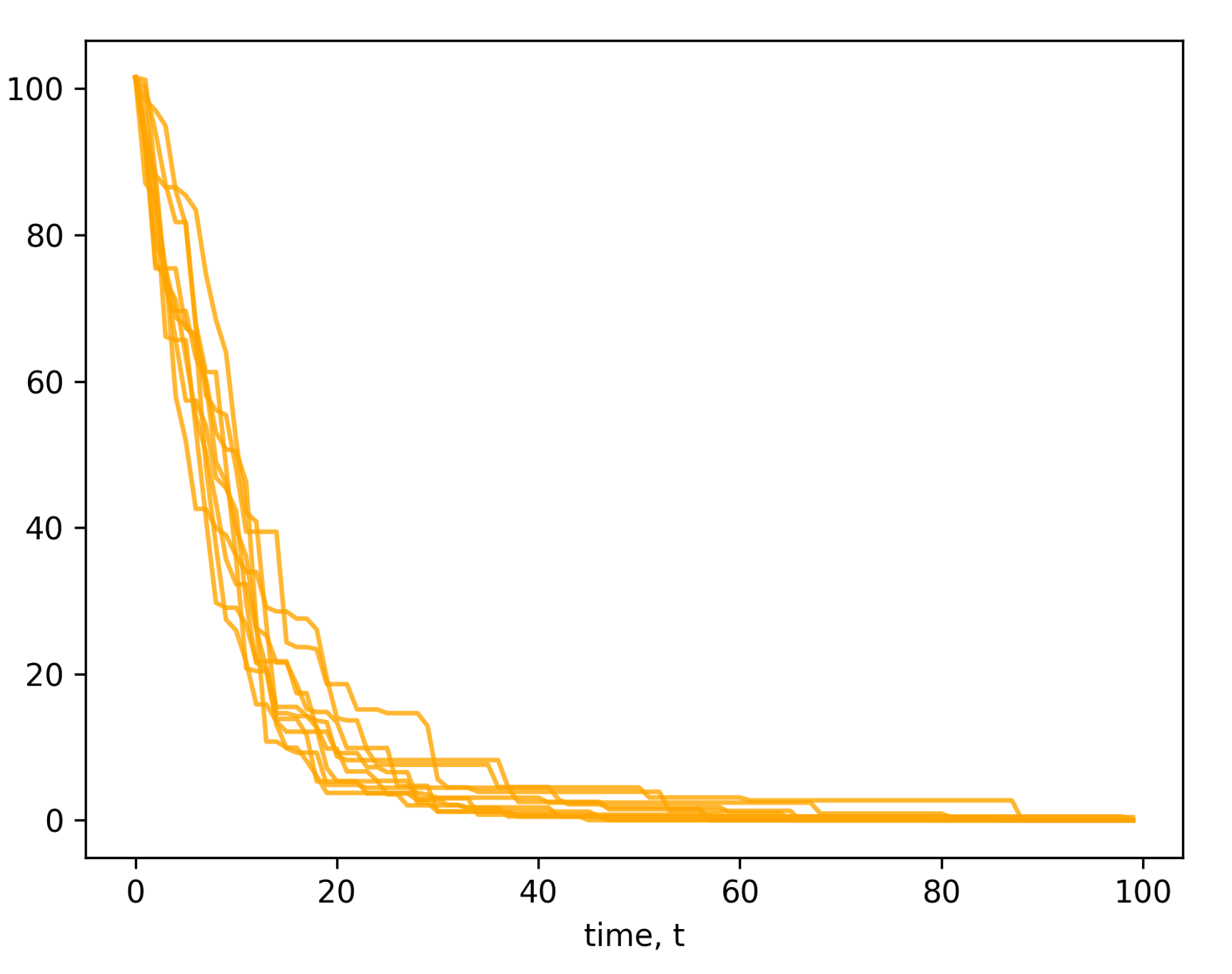}
    \caption{Dirichlet energy of heat flow for the gossip algorithm on a geometric random graph with a randomized Kripke sheaf for $10$ randomized initial conditions.} \label{fig:gossip-experiment}
\end{figure}

%--------------------------------------------------------
\section{A Lattice Convolutional Neural Network (L-CNN)}
%--------------------------------------------------------

We use a portion of the ModelNet10 dataset \cite{wu20153d} as a source of point clouds. This dataset consists of 100s of 3D CAD models which represent objects from 1 of 10 classes. We sample points from the 3D models to produce finite metric spaces embedded in $\R^3$, then compute the corresponding multidimensional persistence modules \cite{lesnick2015interactive}, from which we produce features used as an input to a convolutional neural network classifier.

The pipline thus begins with a 3D polyhedral model, of which $3000$ vertices are sampled to produce a point cloud in $\R^3$. This point cloud then produces a bifiltered simplicial complex, whose degree-$0$ persistent homology we calculate using RIVET~\cite{lesnick2015interactive}, sampled at a discrete grid of $40 \times 40$ points, producing lattice-indexed signals given by the Hilbert function and the multi-graded Betti numbers $\xi_0, \xi_1, \xi_2$; four features in total. These are then passed to the classifier, which produces a class prediction, in this case one of 10 possible objects. As the filter function on these data sets, we use a codensity function that is the inverse sum distance from a point $x$ to the $k$-nearest neighbors ($k=100$). The name codensity is appropriate because the points in the densest regions of $\mathcal M$ appear earlier in the filtration.

We compare the performance of two convolutional networks on this classification
task. One uses is a L-CNN and the other is a standard CNN. Each has three convolutional layers followed by two fully connected layers. The lattice-based convolution layers are of the form $\alpha \cdot \text{MeetConv}(x) +
(1-\alpha)\cdot\text{JoinConv}(x)$ for a hyperparameter $\alpha \in [0,1]$ ($\alpha=0.5$) Thus, \define{lattice layers} are a hybrid of join- and meet- convolutions. All convolution kernels have dimension $4 \times 4$, hidden convolution layers have 16 features, and the final convolution layer has 8 features. The first two convolution layers are followed by max-pooling layers with a $2 \times 2$ kernel. For the lattice convolutional layers, the support of the kernel lay in an evenly spaced $4 \times 4$ grid of points in $[m]\times [n]$, while the standard convolutional layers had a standard kernel.
The inner fully connected layer has 32 features. We use a cross-entropy loss function with a softmax in the final layer. The lattice-based convolutional architecture is summarized in Figure \ref{fig:lcnn}. The networks are trained with Adam gradient algorithm with learning rate $2 \times 10^{-4}$ for a total of 300 epochs. We hold out
10\% of the data for testing. Results are shown in Figure \ref{fig:lcnn}.

\begin{figure}[h!]
    \centering
    \includegraphics[width=0.75\textwidth]{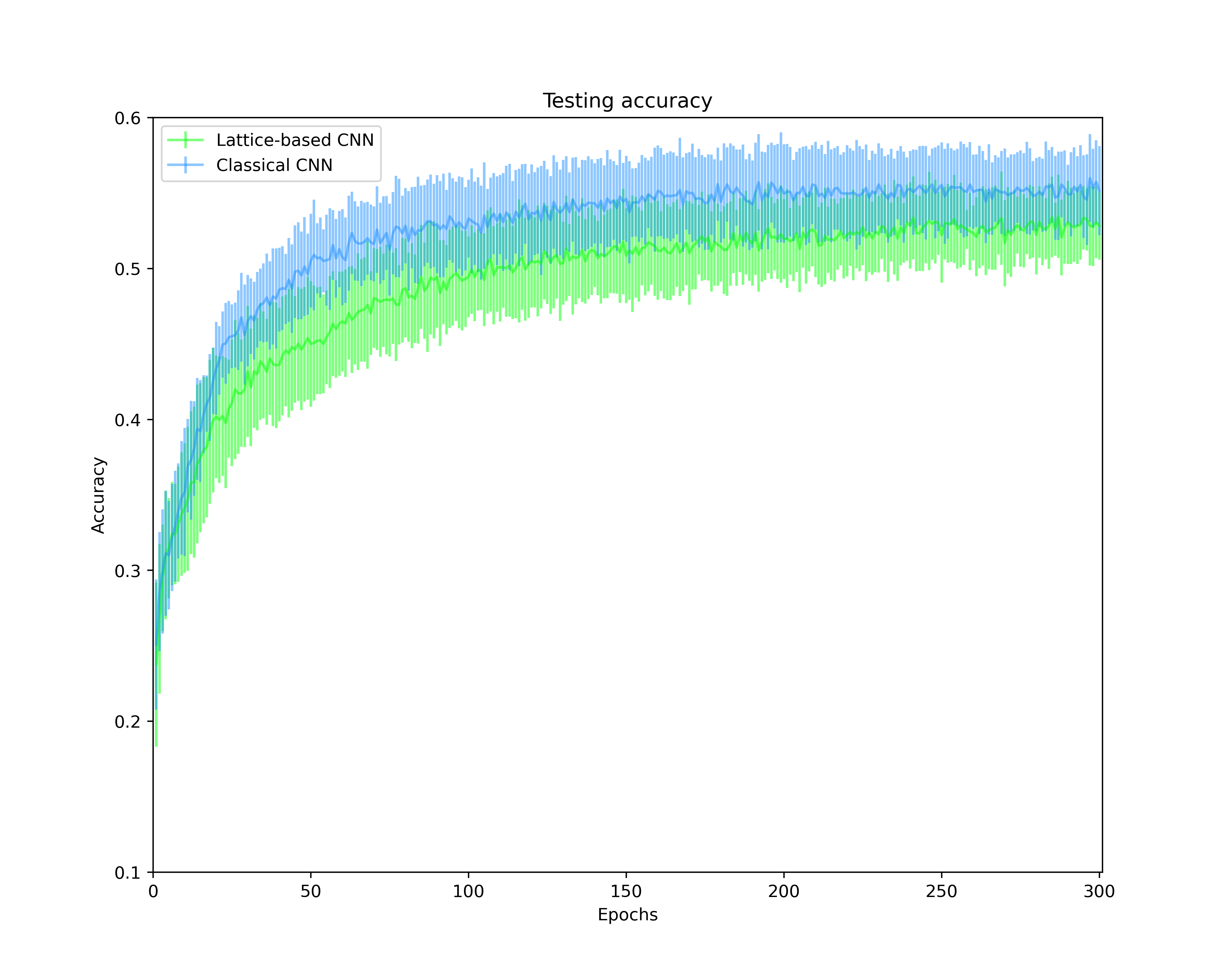}
    \caption{Testing/training accuracy of L-CC (Lattice Convolutional Neural Network) compared to a classical CNN.} \label{fig:lcnn}
\end{figure}

%% file: FrontBackmatter/Bibliography.tex
%********************************************************************
% Bibliography
%*******************************************************
% work-around to have small caps also here in the headline
% https://tex.stackexchange.com/questions/188126/wrong-header-in-bibliography-classicthesis
% Thanks to Enrico Gregorio
\defbibheading{bibintoc}[\bibname]{%
  \phantomsection
  \manualmark
  \markboth{\spacedlowsmallcaps{#1}}{\spacedlowsmallcaps{#1}}%
  \addtocontents{toc}{\protect\vspace{\beforebibskip}}%
  \addcontentsline{toc}{chapter}{\tocEntry{#1}}%
  \chapter*{#1}%
}
\printbibliography[heading=bibintoc]

% Old version, will be removed later
% work-around to have small caps also here in the headline
%\manualmark
%\markboth{\spacedlowsmallcaps{\bibname}}{\spacedlowsmallcaps{\bibname}} % work-around to have small caps also
%\phantomsection
%\refstepcounter{dummy}
%\addtocontents{toc}{\protect\vspace{\beforebibskip}} % to have the bib a bit from the rest in the toc
%\addcontentsline{toc}{chapter}{\tocEntry{\bibname}}
%\label{app:bibliography}
%\printbibliography